\PassOptionsToPackage{table,x11names}{xcolor}
\documentclass[acmsmall,screen,nonacm]{acmart}

\usepackage{amssymb}
\usepackage[inline]{enumitem}
\usepackage{graphicx}
\usepackage{listings}
\usepackage{soul}
\usepackage{xspace}
\usepackage{algorithm}
\usepackage[noend]{algpseudocode}
\usepackage{subcaption}
\usepackage{mathtools}
\usepackage[createShortEnv, conf={restate,text proof={}}]{proof-at-the-end}
\usepackage{hyperref}
\usepackage{caption}
\usepackage{tabularx}
\usepackage{booktabs}
\usepackage{makecell}
\usepackage{stmaryrd}
\usepackage{pifont}
\usepackage{tikz}
\usetikzlibrary{positioning}
\usetikzlibrary{trees, shapes.geometric}
\usepackage[prefixflatinterpret,prefixflatbracketmodalinterpret,prefixflatsetinterpret,sidenotecalculus,longseqcontext,silentconst]{logic}
\usepackage[prefixflatinterpret,prefixflatbracketmodalinterpret,differentialdL,simplenames,silentconst]{dL}
\usepackage{mathrsfs}
\usepackage{yfonts}
\usepackage{stackengine}
\allowdisplaybreaks[4]

\theoremstyle{remark}
\newtheorem{remark}{Remark}

\usepackage{prettyref}
\newcommand{\rref}[2][]{\prettyref{#2}}
\newrefformat{sec}{Section\,\ref{#1}}
\newrefformat{subsec}{Section\,\ref{#1}}
\newrefformat{def}{Def.\,\ref{#1}}
\newrefformat{thm}{Theorem\,\ref{#1}}
\newrefformat{prop}{Proposition\,\ref{#1}}
\newrefformat{lem}{Lemma\,\ref{#1}}
\newrefformat{cor}{Corollary\,\ref{#1}}
\newrefformat{ex}{Example\,\ref{#1}}
\newrefformat{tab}{Table\,\ref{#1}}
\newrefformat{fig}{Fig.\,\ref{#1}}
\newrefformat{app}{Appendix\,\ref{#1}}
\newrefformat{alg}{Algorithm\,\ref{#1}}
\newrefformat{model}{Model\,\ref{#1}}
\newrefformat{line}{Line\,\ref{#1}}

\newrefformat{eq}{Eq.\,\eqref{#1}}
\newrefformat{case}{Case\,\eqref{#1}}
\newrefformat{appendix}{Appendix\,\ref{#1}}

\usepackage{float}
\floatstyle{ruled}
\newfloat{model}{tbp}{model}
\floatname{model}{Model}
\newcounter{modelline}
\makeatletter
\newcommand{\mline}[1]{{\refstepcounter{modelline}\ltx@label{#1}}~\text{\scriptsize{\themodelline}}\quad}
\makeatother

\makeatletter
\newcommand{\boxedSep}[2][\fboxsep]{{%
  \setlength{\fboxsep}{#1}\fbox{\m@th$\displaystyle#2$}}}
\makeatother

\newcommand{\Sem}[1]{\llbracket #1 \rrbracket}

\newcommand{\existdown}{%
  \mathrel{\ooalign{%
    \raisebox{-0.4ex}{\scalebox{0.5}{\(\,\downarrow\)}}\cr
    \hfil\raisebox{0.4ex}{\scalebox{0.7}{\(\exists\)}}\hfil\cr
  }}%
}
\newcommand{\aproj}[3]{\ensuremath{{#1\!\existdown\!#2}}\xspace}
\newcommand{\dproj}[3]{\ensuremath{{#1 {[]}\!\!\existdown\!#2}}\xspace}

\newcommand{\aruntimem}[2]{\ensuremath{{#1 \lightning #2}}\xspace}
\newcommand{\druntimem}[2]{\ensuremath{{#1 \lightning #2}}\xspace}

\newcommand{\foralldown}{%
  \mathrel{\ooalign{%
    \raisebox{-0.4ex}{\scalebox{0.5}{\(\mspace{2.5mu}\downarrow\)}}\cr
    \hfil\raisebox{0.4ex}{\scalebox{0.7}{\(\forall\)}}\hfil\cr
  }}%
}
\newcommand{\asubst}[3]{\ensuremath{{#1 \foralldown #2}}\xspace}
\newcommand{\dsubst}[3]{\ensuremath{{#1 {[]}\!\!\foralldown #2}}\xspace}
\newcommand{\fwd}[2]{\ensuremath{{#2_{#1:}}}\xspace}

\newcommand{\prefix}[2]{\ensuremath{{#2_{:#1}}}\xspace}

\newcommand{\invexpr}{\kwd{Inv}\xspace}
\newcommand{\invgen}{\kwd{genInv}\xspace}
\newcommand{\exec}{\kwd{simpl}\xspace}

\newcommand{\namedGame}[2]{{\scalebox{0.8}{\ensuremath{\textcolor{gray}{{#2\mspace{-6mu}:\mspace{-1mu}}}}}\ensuremath{#1}}\xspace}
\newcommand{\namedGameSpace}[2]{\ensuremath{{\textcolor{gray}{#2\mspace{-1mu}:\mspace{-1mu}}{#1}}}\xspace}
\newcommand{\avalid}[4]{\ensuremath{\namedGame{#1}{#2} \models {#3}}\xspace}
\newcommand{\dvalid}[4]{\ensuremath{\namedGame{#1}{#2} {[]}\!\!\models {#3}}\xspace}

\newcommand{\asolve}{{\ensuremath{{\langle\!\rangle}\mspace{-2mu}\kwd{map}\xspace}}\xspace}
\newcommand{\dsolve}{{\ensuremath{{[]}\mspace{-2mu}\kwd{map}\xspace}}\xspace}
\newcommand{\nodes}[1]{\ensuremath{{\kwd{subgames}(#1)}}}
\newcommand{\restrict}[2]{\ensuremath{#1 |_#2}}
\newcommand{\successor}[2]{\ensuremath{\kwd{succ}(#1, #2)}}

\newcommand{\actleft}{\ensuremath{\mathfrak{l}}\xspace}
\newcommand{\actright}{\ensuremath{\mathfrak{r}}\xspace}
\newcommand{\actstop}{\ensuremath{\mathfrak{s}}\xspace}
\newcommand{\actgo}{\ensuremath{\mathfrak{g}}\xspace}
\newcommand{\actseq}{\ensuremath{\mathfrak{c}}\xspace}
\newcommand{\act}{\ensuremath{\mathfrak{a}}\xspace}

\newcommand{\opSemantics}[2]{\ensuremath{\mathcal{T}({#1}, #2)}\xspace}
\newcommand{\nodeLabel}[1]{\textcolor{gray}{\ensuremath{#1}}}

\newcommand{\eraseLabel}{\ensuremath{\pi_{\act}}}
\newcommand{\skipActions}{\ensuremath{\mathfrak{S}}\xspace}
\newcommand{\eraseSkip}{\ensuremath{\pi_{\neg\skipActions}}}

\newcommand{\apolicy}[2]{\ensuremath{\mathcal{P}_{\mspace{-5mu}\namedGameSpace{#1}{#2}}}\xspace}
\newcommand{\astrategize}[2]{\ensuremath{\mathcal{S}_{#1}(#2)}\xspace}
\newcommand{\caret}{\ensuremath{\char`^}\xspace}
\newcommand{\then}{\ensuremath{\cdot}\xspace}
\newcommand{\runstrategy}[2]{\ensuremath{\lfloor #1 \rfloor_{#2}}\xspace}

\newcommand{\finalNode}{\kwd{end}\xspace}
\newcommand{\leaf}{\kwd{leaf}\xspace}

\newcommand{\deriveWrt}[2]{\ensuremath{\frac{\kwd{d}#1}{\kwd{d}#2}}\xspace}
\newcommand{\modEnd}[2]{\ensuremath{#1(\finalNode\mapsto{#2})}\xspace}
\newcommand{\projToStrat}[1]{\ensuremath{h_{#1}}\xspace}
\newcommand{\projToStratI}[1]{\ensuremath{i_{#1}}\xspace}

\newcommand{\red}[1]{\textcolor{red}{#1}}
\newcommand{\blue}[1]{\textcolor{blue}{#1}}

\newcommand{\define}{\, \equiv \,}

\newcommand{\kwd}[1]{{\normalfont \textsf{#1}}\xspace}

\newcommand{\safe}{\kwd{safe}}

\newcommand{\assumptions}{\kwd{assumptions}}

\newcommand{\skp}{\kwd{skip}}

\newcommand{\reduce}{\kwd{reduce}}

\renewcommand{\dbox}[2]{[#1] \, #2}  
\renewcommand{\ddiamond}[2]{\langle #1 \rangle \, #2}

\newcommand{\seq}{\,;\,}
\newcommand{\D}[1]{\ensuremath{#1'}}

\newcommand*{\props}{\ensuremath{\mathcal{P}_{\hspace{0.1em}\reals}}\xspace}

\def\leftrule{L}%
\def\rightrule{R}%

\newcommand{\I}{\dLint[state=\omega]}

\begin{document}

\title{Hybrid Game Control Envelope Synthesis}
\author{Aditi Kabra}
\orcid{0000-0002-2252-0539}
\affiliation{%
  \institution{Carnegie Mellon University}
  \city{Pittsburgh}
  \state{PA}
  \country{USA}
}
\email{akabra@cs.cmu.edu}

\author{Jonathan Laurent}
\orcid{0000-0002-8477-1560}
\affiliation{%
  \institution{Karlsruhe Institute of Technology}
  \city{Karlsruhe}
  \country{Germany}
}
\email{jonathan.laurent@kit.edu}

\author{Stefan Mitsch}
\orcid{0000-0002-3194-9759}
\affiliation{%
  \institution{DePaul University}
  \city{Chicago}
  \state{IL}
  \country{USA}
}
\email{smitsch@depaul.edu}

\author{André Platzer}
\orcid{0000-0001-7238-5710}
\affiliation{%
  \institution{Karlsruhe Institute of Technology}
  \city{Karlsruhe}
  \country{Germany}
}
\email{platzer@kit.edu}

\begin{abstract}
Control problems for embedded systems like cars and trains can be modeled by two-player hybrid games.
Control envelopes, which are families of safe control solutions, correspond to nondeterministic winning policies of hybrid games, where each deterministic specialization of the policy is a control solution.
This paper synthesizes nondeterministic winning policies for hybrid games that are as permissive as possible.
It introduces \emph{subvalue maps}, a compositional representation of such policies that enables verification and synthesis along the structure of the game.
An inductive logical characterization in differential game logic (\dGL) checks whether a subvalue map induces a sound control envelope which always induces a winning play.
A policy is said to win if it always achieves the desirable outcome when the player follows it, no matter what actions the opponent plays.
The maximal subvalue map, which allows the most action options while still winning, is shown to exist and satisfy a logical characterization.
A family of algorithms for nondeterministic policy synthesis can be obtained from the inductive subvalue map soundness characterization.
An implementation of these findings is evaluated on examples that use the expressivity of \dGL to model a range of diverse control challenges.
\end{abstract}


\maketitle

\definecolor{vblue}{rgb}{.1,.15,.62}

\newsavebox{\Rval}%
\sbox{\Rval}{$\scriptstyle\mathbb{R}$}
\newsavebox{\backiterateb}%
\sbox{\backiterateb}{$\scriptstyle\overleftarrow{\dibox{{}^*}}$}

\section{Introduction}
\label{sec:introduction}

Imperative program-like games serve as specifications for program synthesis \cite{10.1145/1706299.1706339,10.5555/184737}, where the game specifies the \emph{shape} of the program, leaving the details as nondeterministic choices for the player to resolve.
The synthesis of a single correct solution consists of identifying a single winning policy, indicating how to resolve choices to obtain a deterministic program that meets the specification (wins the game).
For some programs, such as embedded system control software, it is useful to know a maximal \emph{set} of control solutions, which can then be specialized for secondary or evolving requirements \cite{10.1007/978-3-030-44051-0_34,10.1145/2883817.2883842,saferl}, e.g., train control software that ensures speed limit adherence can be specialized for fuel efficiency depending on current operating conditions.
This corresponds to the synthesis of a \emph{set} of winning policies for the specification game, which is to say, a \emph{nondeterministic policy} whose every deterministic specialization corresponds to a winning policy.
This paper introduces \emph{subvalue maps} as a compositional representation of nondeterministic policies that can be verified and synthesized recursively along the syntactic structure of the game.
Via this representation we develop the theory to compare, verify, and synthesize nondeterministic policies.
Our formalization is in differential game logic (\dGL) \cite{DBLP:journals/tocl/Platzer15}, which models two-player games with imperative program constructs like assignments, loops and branching.
A central contribution is to develop a representation of nondeterministic game policies whose verification/synthesis composes at the level of such program constructs.

This paper synthesizes nondeterministic policies for \emph{hybrid games}.
It is able to solve problems that previous hybrid games nondeterministic policy synthesis work does not because the subvalue map representation lets us leverage program analysis techniques like loop variants, invariants and their continuous analogs.
\emph{Hybrid} games \cite{DBLP:journals/tocl/Platzer15} extend games \cite{parikhGL} with differential equations.
Nondeterministic policy synthesis for hybrid games is important because it solves the problem of identifying families of correct control solutions (also called \emph{control envelopes}) for embedded systems like cars and trains which have discrete as well as continuous physical behavior.
Hybrid games serve as controller synthesis specifications \cite{DBLP:conf/cdc/NerodeY92,DBLP:journals/IEEE/TomlinLS00,DBLP:conf/hybrid/MoorD01} which characterize the control possibilities within a system's existing hardware and environment.
One player, canonically called Angel, resolves the nondeterminism that the controller can resolve by making a control decision, while the other player, canonically called Demon, resolves nondeterminism resulting from environment behavior.
This results in a two-player, zero-sum \emph{game} where Angel's policy to win, however Demon may resolve non-determinism, is exactly what the controller must do to maintain the desired properties (e.g., safety, liveness, reach-avoid) regardless of environmental conditions.
Embedded controllers are safety-critical, ubiquitous, but hard to design correctly, so synthesizing their control envelopes addresses an important challenge.

After defining subvalue maps, this paper introduces the theory to verify that subvalue maps induce only policies that don't lose.
To ensure that induced policies additionally never let the player get stuck, always allowing some action, the refined concept of \emph{inductive} subvalue map is introduced.
We order subvalue maps such that more permissive maps that permit more control options are greater.
We synthesize inductive subvalue maps.
First, we show completeness: it is always possible to construct an optimal inductive subvalue map for a given game.
Then, since our ultimate goal is to use subvalue maps as control envelopes/nondeterministic policies, we consider a restricted category of inductive subvalue maps where checking if a given controller lies within the envelope is efficient (polynomial time).
We present an algorithmic framework to synthesize such subvalue maps.
The theoretical development needed is subtle.
For example, at loops, a nondeterministic policy should allow \emph{unboundedly} many repetitions so long as there is still a \emph{finite} exit strategy \emph{within} the same nondeterministic policy.
The first few sections of the paper build up the theoretical tools that make the theory for subvalue maps possible, and largely internalized within \dGL.
These include:
\begin{enumerate*}
  \item a \dGL analog to Brzozowski derivatives \cite{10.1145/321239.321249} that let us step through a game and characterize gameplay after a given subgame (\rref{sec:stepping-through-dgl}), and
  \item \dGL characterizations checking whether a set of states is reachable by following \emph{some} strategy or \emph{any} strategy in a given subvalue map (\rref{sec:subvalue-maps}).
\end{enumerate*}

We synthesize inductive subvalue maps by recursively computing and composing them while syntactically stepping through a game.
In \dGL, \emph{subgames} specify game decision points where at most one player makes a play.
These subgames compose together per a recursive syntax to form large, arbitrarily complex games.
Subvalue maps are analogous to value functions in reinforcement learning \cite{rl}.
A value function maps states to utility.
Once it is known, an agent policy can be derived from it.
Similarly, an Angelic subvalue map maps every subgame to a formula denoting a set of states from which Angel can win the rest of the game (\emph{sub}value because some maps map to a set of \emph{known} winning states smaller than the theoretical maximum winning set).
A player's policy set can then be reconstructed by tracing through the game: at every subgame where the player must take a decision, it accepts only those decisions that reach the next subgame in a mapped winning state for that subgame per the subvalue map.

This paper solves these challenges for \emph{all} of \dGL:
Synthesis for safety properties, liveness properties, and any combination of these are all solved by a single, natural framework without special handling of either.
The flexibility of our final synthesis framework is demonstrated through examples with diverse control challenges, and performance is evaluated on benchmarks from the literature and a procedurally generated benchmark suite.
Our main contributions are to:
\begin{enumerate}
  \item Introduce \emph{inductive subvalue maps} as symbolic, composable representations of nondeterministic winning policies for hybrid games (\rref{sec:solution-representation}).
  \item Prove completeness for optimal hybrid game inductive subvalue map synthesis via the construction of a \emph{maximal subvalue map} (\rref{sec:ordering-solutions}).
  \item Present an algorithmic framework (\rref{sec:solving-algorithm}) for \emph{synthesis} of inductive subvalue maps where checking that control is within the induced nondeterministic policy is \emph{computationally efficient}.
  \item Evaluate the approach, demonstrating its application to representatives of various control theory problem classes (\rref{sec:evaluation}).
\end{enumerate}

\section{Differential Game Logic}
\label{sec:background}

This paper synthesizes policies for games written in differential game logic (\dGL), a logic for two-player hybrid games that has a relatively complete axiomatization.
A detailed explanation can be found elsewhere \cite{DBLP:journals/tocl/Platzer15,Platzer18}, but we provide an overview of \dGL and set up the notation used in this paper.

A \dGL game is played by two adversarial players, canonically called Angel and Demon.
Various operators give either player opportunities to make control decisions to try to attain a winning condition.
This paper uses the convention that Angel is responsible for the controller's decisions, and wins when she maintains the desired system properties (e.g. safety), and Demon plays for the environment.
\dGL games are generated by a recursive grammar consisting of the game constructs that follow.
In many of the game constructs, either Angel or Demon can make a choice, deciding how the game should proceed, and try to choose in a way that is most favorable to them.
These decisions are represented using \emph{actions}, which will later be used to keep track of state change, construct game trees and to represent policies\footnote{The action notation presented here is a simplified form of the operational semantics of \dGL \cite{DBLP:journals/tocl/Platzer15}[Appendix C]. Sequential compositions are given an action for uniformity but do not have an action in the original operational semantics.}.
The game constructs are as follows:
\begin{enumerate}
    \item Angelic free assignment $x:=*$ models Angel assigning any real value of her choice to variable $x$, and can be used, e.g., to model a controller's ranged choices.
    Angel's available actions are of the form $(x:=e)$, where $e$ is a (real polynomial) term.
    Dually, Demonic free assignment $x:=\otimes$ models Demon assigning a real value of his choice to variable $x$ instead.
    It models, e.g., the environment introducing disturbances making control harder.
    Demon's available actions are of the form $\pdual{(x:=e)}$, where $e$ is a term.
    \item Continuous evolution $\{x'=f(x)\ \&\ Q\}$ modifies $x$ per the solution of differential equation $x'=f(x)$.
    How long the ODE runs is determined by Angel, who must also maintain that domain constraint formula $Q$ is true throughout.
    This can model, for example, a control mode running till the controller interrupts.
    Angel's available actions are of the form $(x' = f(x) \, \&\, Q \,@\, t)$, where $t$ is a term representing the time Angel chooses to run the ODE for.
    Dually, in $\{x'=f(x)\ \&\ Q\}^d$, Demon chooses how long to run the ODE while ensuring that $Q$ holds throughout.
    Demon's available actions are of the form $\pdual{(x' = f(x) \, \&\, Q \, @ \, t)}$, where $t$ is a term for the time Demon runs the ODE.
    Demon-controlled ODEs can model, e.g., control loop latency in controllers that periodically poll to take decisions.
    \item Loop $\alpha^*$ runs $\alpha$ as many times as Angel wants.
    Before starting a fresh iteration of the loop, Angel gets to choose whether to run the loop again or to exit it.
    This can model, for example, a controller that repeatedly performs a task until a criterion is met.
    Her available actions are \(\actgo\) to \emph{go} repeat the loop for one more iteration, or \(\actstop\) to \emph{stop} and exit the loop.
    Dually, $\alpha^\times$ runs $\alpha$ as many times as Demon chooses.
    Demon loop can model a control loop that runs arbitrarily many times, and must stay safe forever.
    Demon's available actions are \(\pdual{\actgo}\) to repeat the loop for one more iteration, or \(\pdual{\actstop}\) to exit the loop.
    \item Branch choice $\alpha\cup\beta$ lets Angel choose to play either game $\alpha$ or game $\beta$.
    This can model a controller choosing between two modes of operation, for example accelerating or braking.
    Angel's available actions are \(\actleft\) to play \emph{left} game $\alpha$, or \(\actright\) to play \emph{right} game $\beta$.
    Dually, $\alpha\cap\beta$ lets Demon choose to play either game $\alpha$ or game $\beta$.
    This can model different modes of environment behavior, such as a car in the environment, that the controller must avoid colliding with, accelerating or braking.
    Demon's available actions are \(\pdual{\actleft}\) to play game $\alpha$, or \(\pdual{\actright}\) to play game $\beta$.
    \item Test $\ptest{Q}$ makes Angel immediately lose the current game if formula $Q$ is false, but has no effect if $Q$ is true.
    It can be used to limit the actions available to Angel modeling control system constraints.
    The action representing Angel going through the test (and possibly losing) is written as $(\ptest{Q})$.
    Dually, $!Q$ makes Demon lose the game immediately when $Q$ is false, and has no effect if $Q$ is true.
    It can model physical environment constraints.
    Demon's only available action, to take the test, is written as $(!Q)$.
\end{enumerate}
Additionally, some game constructs do not involve Angel or Demon making any decisions but are still important to model overall functioning of the game.
\begin{enumerate}
    \item Assignment $x:=e$ assigns the expression $e$ to variable $x$. The corresponding action that Angel performs to change state is written as $(x:=e)$.
    \item Sequential composition $\alpha;\beta$ runs game $\alpha$ followed by game $\beta$.
    The corresponding action that Angel performs is \(\actseq\) to start \(\alpha\).
\end{enumerate}
\dGL formula $\ddiamond{\alpha}{\phi}$ is true in any state in which player Angel has a winning strategy to play hybrid game $\alpha$ so that regardless of what Demon does, in the end \dGL formula $\phi$ holds true.
As usual, state \(\sigma\) is a mapping from variables to real numbers, and \(\sigma(x)\) is the real value associated with variable \(x\).
Dually, formula $[\alpha] \phi$ means Demon has a winning strategy to play game $\alpha$ so as to reach a state satisfying $\phi$ in the end. \rref{app:dl-semantics} recalls the denotational semantics.
\rref{eq:dgl-grammar} shows the grammar of \dGL hybrid games.
\begin{equation}
\label{eq:dgl-grammar}
\begin{aligned}
   \alpha \equiv\  x:=e \mid \alpha;\beta \mid \ptest{Q} &\mid \{x'=f(x) \& Q\}\phantom{^d} \mid \alpha^\ast \mspace{4mu} \mid \alpha\cup\beta \mid x:=\ast \\
   \mid \mspace{1.5mu} ! \mspace{1.5mu} Q &\mid \{x'=f(x) \& Q\}^d \mid \alpha^\times \mid \alpha\cap\beta \mid x:=\otimes
\end{aligned}
\end{equation}

The \emph{subgames} of a \dGL game $\alpha$ are the games that are recursively composed to construct $\alpha$, and correspond to the nodes of the abstract syntax tree (AST) induced by the grammar of \rref{eq:dgl-grammar}.

We briefly describe the formal representation of actions and policies.
Sections \ref{sec:stepping-through-dgl} and \ref{sec:subvalue-maps} will revisit these ideas with further detail.
In any given initial state, a \dGL game generates a game tree (\rref{sec:labeled-subgames}) where each node corresponds to an action.
A \emph{play} of a game is a sequence of actions corresponding to a path from a root of the game tree to a leaf.
Actions map old state to new state.
For example, the action \((x:=e)\) maps state \(\sigma\) to \(\sigma(x \mapsto e)\), i.e., \(\sigma\) with value \(e\) replaced for \(x\).
Since agents and their gameplay are symmetric, in this paper we often focus on the Angelic perspective without loss of generality.

An Angelic (nondeterministic) \emph{policy} is a function that maps every subgame and state to the set of actions that Angel can take at that subgame and state.
\emph{Henceforth in this paper, policy will be used to mean nondeterministic policy}.
Angel \emph{follows} an Angelic policy if at every subgame she must take a decision, she plays one of the actions permitted by the policy. 
An Angelic policy \emph{wins} if Angel following it guarantees her reaching the end of the game in her winning region.

This paper synthesizes a winning policy for Angel, which characterizes the space of control actions that ensures her victory, given a \dGL game and winning condition.
We introduce a way to represent winning policies by capturing the player's \emph{winning regions}, i.e., the set of game states from which the player has a way to win regardless of what the opponent does.
As usual, formulas can be interpreted as the set of states in which they are true.
Thus, formula $\ddiamond{\alpha}{\phi}$ represents the winning region of Angel for game $\alpha$ and winning condition $\phi$.
Notation $\models \phi$ means formula $\phi$ is valid (true in all states).

\newcommand{\preci}[1]{\phi_{#1}}

\section{Overview}
\label{sec:overview}%

\newcommand{\astSpacing}{0.7cm}%
\newcommand{\precAstAnnot}[2]{$#1$ \ \textcolor{gray}{\scriptsize $\preci{#2}$}}%
\newcommand{\AlgGuess}{\textbf{Guess }}%
\newcommand{\AlgCheck}{\textbf{Check }}%
\newcommand{\AlgNoCheck}{\textbf{\st{Check} }}%
\newcommand{\AlgDeduce}{\textbf{Deduce }}%
\newcommand{\ExRelaxedLoop}{n \!:=\! * \seq\! ?n \!\ge\! 0 \seq\! v \!:=\! v\!+\!an}
\newcommand{\AlgDef}{\equiv}%
\newcommand{\AlgSubApprox}{\Leftarrow}%
\newcommand{\AlgEqv}{\Leftrightarrow}%
\begin{figure}[!hb]

    \hrule

    \medskip
    
    \scriptsize

   \begin{minipage}[t]{0.53\textwidth}
   \noindent \textbf{1. Input} \\
    \textbf{Game: }
      $(((v := v + a)^* \cup (v:=v-1 \seq a := *)) \seq \{x'=v\}^d)^{\times}$
    \textbf{Angel's Goal: } $x > 0$

   \medskip
   \hrule
   \bigskip

    \noindent \textbf{3. Policy for Angel derived from subvalue map}
    
    \smallskip
    
    \begingroup
    \renewcommand{\arraystretch}{1.5}
   \begin{tabular}{m{0.17\textwidth}m{0.73\textwidth}}
    $\boxed{\cup}\,$ &
        Taking the {left} (resp. right) branch is allowed if $\preci{6}$ (resp. $\preci{5}$) is true. \\
    $\boxed{*}\,$ & (Re-)entering the loop is allowed if $\preci{7}$ is true. Exiting is allowed if $\preci{2}$ is.  \\
    $\boxed{a := *}\,$ & The value assigned to $a$ must satisfy $\preci{2}$. \\
    \end{tabular}
    \endgroup

    \medskip
   \end{minipage}
   \hfill
   \vrule
   \hfill
   \begin{minipage}[t]{0.43\textwidth}
    \noindent \textbf{2. Compute subvalue map}\\
     \scriptsize
     \vspace{-5pt}
    \begin{tikzpicture}
      [
        baseline=(current bounding box.north),
        level distance={\astSpacing},
        every node/.style={draw,rectangle,align=center},
        sibling distance=3cm
      ]
      \node (root) {\precAstAnnot{\times}{1}}
        child {
            node {\precAstAnnot{\,;\,}{9}}
            [sibling distance=1.8cm]
            child {
                node {\precAstAnnot{\,\cup\,}{8}}
                [sibling distance=2.80cm]
                child {
                    node {\precAstAnnot{\,*\,}{6}}
                    child {node {\precAstAnnot{v\!:=\!v\!+\!a}{7}}}
                }
                child {
                    node {\precAstAnnot{\,;\,}{5}}
                    [sibling distance=1.75cm]
                    child {node {\precAstAnnot{v\!:=\!v\!-\!1}{4}}}
                    child {node {\precAstAnnot{a\!:=\! *}{3}}}
                }
            }
            child { node {\precAstAnnot{\{x'=v\}^d}{2}} }
        };
        \node[right=0.5cm of root] {\precAstAnnot{\finalNode}{\finalNode}};
    \end{tikzpicture}
   \end{minipage}

   \smallskip
   
    \bigskip
    \hrule
    \bigskip

  \textbf{Computing the subvalue map}:

     \begin{align*}
        \AlgDeduce \preci{\finalNode} &\AlgDef x > 0 & \text{(Goal)} \\
        \AlgGuess \preci{1} &\AlgDef \big\langle (((v := v + a)^* \cup (v:=v-1 \seq a := *)) \seq \{x'=v\}^d)^{\times} \big\rangle \, x > 0 & \text{(Seek Invariant)} \\
        &\AlgSubApprox \big\langle ((\ExRelaxedLoop) \cup (v\!:=\!v\!-\!1 \seq\! a\! \!:=\! *)) \seq\! \{x'\!=\!v\}^d \big\rangle \, x \!>\! 0 & \text{(Refinement)} \\ 
        &\AlgEqv x > 0 \,\wedge\, (v \ge 0 \vee a > 0)  & \text{(Axioms, QE)} \\
        \AlgDeduce \preci{2} &\AlgDef \langle \{x'=v\}^d \rangle \, \preci{1} \AlgEqv \forall t \, (t \ge 0 \limply \preci{1}(x \mapsto x + vt)) & \text{(Solve ODE)} \\
        &\AlgEqv x > 0 \wedge v \ge 0 & \text{(QE)} \\ 
        \AlgDeduce \preci{3} &\AlgDef \exists a \, \phi_2 \AlgEqv x > 0 \wedge v \ge 0, \quad  \preci{5} \AlgDef \preci{4} \AlgDef \phi_3(v\!-\!1 \mapsto v) \AlgEqv x > 0 \wedge v \ge 1 & \text{(Axioms, QE)} \\
        \AlgGuess \preci{6} &\AlgDef \big\langle (v := v + a)^*\big\rangle \,\preci{2} \,\AlgEqv\, \big\langle \ExRelaxedLoop \big\rangle \, \preci{2} & \text{(Invariant)} \\ 
        &\AlgEqv x > 0 \wedge (v \ge 0 \vee a > 0) & \text{(QE)} \\
        \AlgDeduce \preci{7} &\AlgDef \preci{6}(v\mapsto v\!+\!a) \AlgEqv x > 0 \wedge (v + a \ge 0 \vee a > 0) & \\
        \AlgCheck \preci{6} &\limply \big\langle (?\preci{7} \seq v := v + a)^* \seq ?\preci{2} \big\rangle \preci{2} \text{ valid} & \text{(Check Invariant)}\\
        \AlgDeduce \preci{9} &\AlgDef \preci{8} \AlgDef \preci{6} \vee \preci{5} \AlgEqv x > 0 \wedge (v \ge 0 \vee a > 0) \\
        \AlgCheck \preci{1} &\limply \preci{9} \land x > 0 \text{ valid} & \text{(Check Invariant)} 
    \end{align*}

    \hrule

    \medskip
   
  \caption{An \emph{inductive subvalue} map for a simple hybrid game and the associated policy. The subvalue formulas are numbered in the order in which they are derived by the algorithm presented in Section~\ref{sec:solving-algorithm}. We will revisit this example throughout the paper, providing more explanations.}
  \label{fig:overview}
\end{figure}

To illustrate our new concept of an \emph{inductive subvalue map} and how it captures a policy (set of control solutions), let us consider the simple hybrid game from \rref{fig:overview}.
In this game, Angel aims to ensure $x > 0$ after each iteration of a loop ($^\times$) controlled by Demon. At each iteration, Angel must choose ($\cup$) between two alternative paths, after which Demon continuously evolves $x$ at rate $v$ for a duration of his choice ($\{x'=v\}^d$). In the first path, Angel can increment $v$ by an amount of $a$ in a loop ($^\ast$) as many times as she desires. In the second path, $v$ is decremented by 1 and Angel gets an opportunity to arbitrarily set the value of $a$ ($a:=\ast$).

Our goal is to characterize not \emph{one}, but \emph{all} of the ways Angel can win (or at least \emph{as many} as possible). At every choice point, we must determine which actions are compatible with Angel still winning the game and which actions are not.
This way, Angel can pursue additional, unmodeled and possibly changing secondary goals without threatening her mission-critical objective of ensuring $x > 0$.
In the traditional setting of control theory \cite{Bellman:DynamicProgramming,BellmanDreyfus1962} and reinforcement learning \cite{rl,Sutton1988LearningTP,10.1162/089976699300016070}, this information is typically presented in the form of a \emph{value function} that splits the state space into winning and losing states. This work proposes a representation of such value functions for hybrid games which is \emph{sound}, \emph{symbolic}, and which respects their \emph{composable} structure.

\newcommand{\AdvExLoop}[1]{\textcolor{olive}{#1}}
\newcommand{\AdvExAssign}[1]{\textcolor{blue}{#1}}
\newcommand{\AdvExChoice}[1]{\textcolor{red}{#1}}

\paragraph{Subvalue maps (\rref{sec:subvalue-maps}).} A \emph{subvalue map} associates \emph{each} subgame of a hybrid game to a formula that \emph{conservatively} estimates the set of states from which the overall game can be won, \emph{starting from} this specific subgame. In the example from \rref{fig:overview}, $\preci{1}$ provides a sufficient initial condition for winning the game. More interestingly, $\preci{5}$ provides a sufficient condition for winning the subgame starting from the second Angel branch within the loop. It must respect
\begin{equation}
    \label{eq:cont-example}
    \vDash \  \preci{5} \,\limply\,  \big\langle v:=v-1 \seq a := * \seq \{x'=v\}^d \seq \alpha^\times \big\rangle \, x > 0
\end{equation}
where $\alpha$ stands for the full Demon loop body.
At every control point for Angel, the subvalue map indicates which actions can be taken without forfeiting winning.
For example, when reaching the $\cup$ decision point, Angel can safely take the left branch if $\preci{6}$ is true.
Similarly, she can safely take the right branch if $\preci{5}$ is.
What if none of these hold though?
Useful policies permit at least one action to be chosen.
For this reason, \rref{sec:solution-representation} introduces the refined concept of an \emph{inductive} subvalue map, whose local compatibility conditions ensure this property.
In our example, these conditions mandate in particular that $\preci{8} \limply \preci{6} \vee \preci{5}$ is valid, ensuring that a safe action is always available to Angel at subgame $\cup$.

\paragraph{Internalizing policy reasoning.} It is possible to reason about the nondeterministic policies represented by subvalue maps \emph{within \dGL itself}.
In particular, the \emph{universal projection} of a game onto a subvalue maps constructs a \dGL game that wins precisely when \emph{all plays that respect the subvalue-map policy} are guaranteed to win.
The \emph{existential projection} of a game onto a subvalue map wins when \emph{there exists some policy-compliant play} that wins.
These projections let us use \dGL to state and prove properties, keeping the theory tidy.
The safety of the monitor derived from an inductive subvalue map that keeps the game on track for winning can similarly be expressed in $\dGL$ itself.
In our example, the following formula is \emph{guaranteed} to be valid:
\begin{equation}
  \label{eq:overview-projection}
  \begin{aligned}
      \phi_1 \limply \big\langle \big( &\big(
      \AdvExChoice{?(\preci{6} \!\vee\! \preci{5})} \seq 
      \big(
          \AdvExChoice{!\preci{6}} \seq \AdvExLoop{?\preci{6}} \seq (\AdvExLoop{!\preci{7}} \seq v\!:=\!v\!+\!a \seq \AdvExLoop{?\preci{6}})^{\AdvExLoop{\times}} \seq \AdvExLoop{!\preci{2}}
      \big)
      \; \AdvExChoice{\cap} \\
      &\big(\AdvExChoice{!\preci{5}} \seq v\!:=\!v\!-\!1 \seq \AdvExAssign{?(\exists a \,\preci{2})} \seq a \!:=\! \AdvExAssign{\otimes} \seq \AdvExAssign{!\preci{2}}\big)\big) \seq \{x'\!=\!v\}^{\!d}\big)^{\times} \big\rangle \, x>0.
  \end{aligned}
  \end{equation}
In the game above, every Angel choice from the original game has been transferred to Demon, on the condition that he complies with Angel's subvalue strategy monitor.
Compliance is always possible, as ensured by proper Angel tests.

\paragraph{Existence of an optimal subvalue map (\rref{sec:ordering-solutions}).} Subvalue maps can be ordered by precision, and thus by how permissive the policy is.
\rref{sec:ordering-solutions} formalizes such an ordering, where the main subtlety is about ensuring that monitor guards are compared \emph{in the context} in which they are reachable.
An optimal subvalue map \emph{exists} that corresponds to the dominant \dGL strategy.
In our example, this solution maps $\preci{5}$ to the right-hand side of the implication in \rref{eq:cont-example}.
The optimal subvalue map construction achieves completeness of the optimal subvalue map synthesis problem.

\paragraph{Solving hybrid games via symbolic execution (\rref{sec:solving-algorithm}).}
For practical applications, it is useful for subvalue maps to compute the set of available actions at a given state efficiently.
An algorithmic framework for computing such subvalue maps of hybrid games is presented in Section~\ref{sec:solving-algorithm}.
Like precondition calculus \cite{10.1145/360933.360975}, it is parametrized on an invariant generation procedure.
We propose \emph{game rewriting} as a general principle for implementing this procedure (\rref{sec:oracle-implementation}).
We illustrate the resulting algorithm in \rref{fig:overview}.
The algorithm leverages the fact that any $\dGL$ formula that is free of loops and whose ODEs have polynomial solutions can be automatically rewritten into an equivalent formula of propositional real arithmetic, by using the axioms of $\dGL$ along with quantifier elimination (QE)~\cite{tarski}.
From there, it works backwards to bound the subvalue of all subgames.
Whenever a loop is hit, a guess is needed for a subvalue of the associated subgame.
Such a guess is typically made by using heuristics to conservatively rewrite it into a loop-free $\dGL$ formula that can be symbolically evaluated.
The algorithm then proceeds with the loop body, after which the validity of the guess is checked retrospectively.
\rref{sec:solving-algorithm} provides more details.
A subvalue map induces a \emph{lower} bound of a game's value function.
An \emph{upper} bound can also be computed by considering the game from the opposing player's perspective.
We can compare the dual subvalue maps and thereby symmetrically prove the optimality of the solution derived in \rref{fig:overview}.

The proofs in this paper largely follow from structural induction.
\emph{\rref{app:proofs} provides all proofs throughout this paper.}

\section{Prefixes and Suffixes}
\label{sec:stepping-through-dgl}

An \emph{Angelic subvalue map} maps every \dGL subgame to a formula denoting Angel's known winning subregion for the rest of the game starting at that subgame.
As it is not always possible to compute exact winning regions, the \emph{known} winning \emph{sub}regions in a subvalue map are conservative, safe subsets of the actual theoretical winning regions per \dGL semantics.
\rref{fig:overview} shows an example where each subgame, corresponding to nodes of the syntax tree, is mapped to its corresponding winning subregion ($\phi_1, \cdots \phi_9$).
In this example, the winning \emph{sub}regions are maximal, \emph{equal} to the winning regions.

This section develops the formalism necessary for defining subvalue maps.
A subvalue map maps each subgame to a formula that must logically imply the winning region for playing the rest of the game starting at that subgame.
To characterize the winning region starting from a subgame, we first introduce the \emph{game suffix}, or game remainder of a subgame.
The game suffix of a subgame features \emph{all} possible future behaviors of the overall game from \emph{any} point where an instance of the subgame is encountered.
Crucially, we characterize all these behaviors with a finite \dGL game.
Then, the exact winning region starting from a subgame will be defined as the winning region for the game suffix.

\subsection{Game Suffix Construction}
\label{sec:suffix-construction}

For disambiguation, each subgame is given a unique label.
Notation $\namedGame{\alpha}{a}$ refers to the \emph{named} game $\alpha$ where each subgame has a unique label, and the entire game $\alpha$ (corresponding to the root node of the AST) is labeled $a$.
Even subgames that are identical in structure but appear at different places in the game are given different labels to distinguish multiple occurrences.
$\nodes{\namedGame{\alpha}{a}}$ refers to the set of all named subgames of $\namedGame{\alpha}{a}$.
Game suffix $\fwd{b}{\namedGame{\alpha}{a}}$ is a game modeling what remains to be played after seeing some occurrence of subgame $\namedGame{\beta}{b}$ while playing $\namedGame{\alpha}{a}$ (\rref{def:execution-suffix})%
\footnote{
The notation is analogous to Python list slicing. $\fwd{b}{\namedGame{\alpha}{a}}$ represents the part of $\namedGame{\alpha}{a}$ appearing \emph{after} $b$. The complementary concept of \emph{game prefix} is written as $\prefix{b}{\namedGame{\alpha}{a}}$, and consists of the part of the game appearing before subgame $b$.}.
This is a complex idea: unboundedly many different plays of a game can lead to playing a given subgame \(b\), and even within a single play, \(b\) can occur an unbounded number of times.
The ability to succinctly characterize all suffix plays (leveraging \dGL) is crucial to the definition of subvalue maps.

\begin{example}
  Let the overall game of \rref{fig:overview} be $\namedGame{\alpha}{a}$, and let the label of each subgame be the index $i$ of the winning subregion formula $\phi_i$ shown on its corresponding node.
  Consider subgame 5 inside loop $\alpha$.
  All execution traces after any instance of subgame 5 have the remainder of the loop body ($\{x'=v\}^d$) followed by potential future iterations of the loop ($\alpha$).
  So the game suffix $\fwd{5}{\namedGame{\alpha}{a}}$ of subgame 5 is $\{x'=v\}^d\seq \alpha$.
\end{example}

\begin{definition}[Game suffix]
  \label{def:execution-suffix}
  The game suffix $\fwd{b}{\namedGame{\alpha}{a}}$ of subgame $\namedGame{\beta}{b}$ in $\nodes{\namedGame{\alpha}{a}}$ is constructed as follows.
  If $\namedGame{\alpha}{a}$ is $\namedGame{\beta}{b} \textrm{ then } \fwd{b}{\namedGame{\alpha}{a}} = \namedGame{\alpha}{a}$.
  Otherwise, if $\namedGame{\alpha}{a}$ has structure:
  \begin{equation*}
    \begin{aligned}
      \namedGame{(\namedGame{\gamma}{g})^\ast}{a} \textrm{ or } \namedGame{(\namedGame{\gamma}{g})^\times}{a} &\textrm{ then } \fwd{b}{\namedGame{\alpha}{a}} = (\fwd{b}{\namedGame{\gamma}{g}}) \seq \namedGame{\alpha}{a} \\
      \namedGame{(\namedGame{\gamma}{g} \cup \namedGame{\delta}{d})}{a} \textrm{ or } \namedGame{(\namedGame{\gamma}{g} \cap \namedGame{\delta}{d})}{a} &\textrm{ then } \begin{cases}
        \fwd{b}{\namedGame{\alpha}{a}} = \fwd{b}{\namedGame{\gamma}{g}} &b \in \nodes{\namedGame{\gamma}{g}} \\
        \fwd{b}{\namedGame{\alpha}{a}} = \fwd{b}{\namedGame{\delta}{d}} &b\in \nodes{\namedGame{\delta}{d}}
      \end{cases} \\
      \namedGame{(\namedGame{\gamma}{g} \seq \namedGame{\delta}{d})}{a} &\textrm{ then } \begin{cases}
        \fwd{b}{\namedGame{\alpha}{a}} = (\fwd{b}{\namedGame{\gamma}{g}})\seq \namedGame{\delta}{d} &b \in \nodes{\namedGame{\gamma}{g}} \\
        \fwd{b}{\namedGame{\alpha}{a}} = \fwd{b}{\namedGame{\delta}{d}} &b\in \nodes{\namedGame{\delta}{d}}
      \end{cases}
    \end{aligned}
  \end{equation*}
  In $\namedGame{(\namedGame{\gamma}{g} \seq \namedGame{\delta}{d})}{a}$ and $\namedGame{(\namedGame{\gamma}{g} \cup \namedGame{\delta}{d})}{a}$ cases, $b$ is either in $\nodes{\namedGame{\gamma}{g}}$ or $\nodes{\namedGame{\delta}{d}}$ because $b$ is a subgame of $\namedGame{\alpha}{a}$ (and not $a$ itself which is already handled initially).
  Cases where \(\namedGame{\alpha}{a}\) is atomic, i.e., \(\alpha \in \alpha \in \{x:=e, x:=*, \ptest{Q}, !Q, \{x'=f(x)\ \&\ Q\}, \{x'=f(x)\ \& \ Q\}^d\}\), \(b\) must be \(a\), so these cases are already handled initially.
\end{definition}

\begin{remark}
  The suffix characterizing behavior of a \dGL game after a subgame \emph{itself as a \dGL game} is analogous to how Brzozowski derivatives reason about the behavior seen in a regular expression after a given string as a new regular expression \cite{10.1145/321239.321249}.
\end{remark}

\subsection{Game Trees and Plays}
\label{sec:labeled-subgames}

This section discusses \dGL gameplay to precisely connect subgame suffixes to their defining property in \dGL operational semantics.
However, later in this paper we will rely only on the syntactic suffix definition that has already been introduced.
The background on \dGL game trees is also useful for understanding \emph{strategies}, which in \rref{sec:subvalue-maps} will provide a global view of the gameplay that a policy induces.
This section first recounts which plays lie within a game and makes precise which suffix plays occur after a subgame.
Then it defines a way to syntactically compare these play suffixes by the impact they have on state, \emph{ignoring minor game structure differences}, to ensure that the \emph{game suffix} has the same effect as these \emph{suffix plays}, establishing that game suffixes reconstruct all gameplay that can occur after a subgame is seen.
The definitions of game trees and plays that follow are from the existing operational semantics of \dGL \cite{DBLP:journals/tocl/Platzer15}[Appendix C], extended to account for labels.

A labeled game \(\namedGame{\alpha}{a}\), at a given initial state \(\sigma\) corresponds to a \emph{labeled game tree} \(\opSemantics{\namedGame{\alpha}{a}}{\sigma}\) where each node is associated with a subgame label and an action.
Using the conventions of descriptive set theory, the tree is represented as a prefix-closed set of the paths originating at the root (\rref{fig:tree-and-paths}).

\begin{example}
  In the game of \rref{fig:overview} let the label of each subgame be the index $i$ of the winning subregion formula $\phi_i$ shown on its corresponding node.
  Consider subgame 6, the inner loop.
  The game tree for this subgame has the paths \(\nodeLabel{6} \then \actstop\) for 0 iterations, \(\nodeLabel{6} \then \actgo \then \nodeLabel{7} \then (v:=v+a) \then \nodeLabel{6} \then \actstop\) for one iteration, and so on.
  Operator \(\then\) concatenates actions/sequences of actions.
  This tree is visually represented in \rref{fig:tree-and-paths}.
\end{example}

\begin{figure}[ht]
  \centering
  \begin{subfigure}[t]{0.48\textwidth}
    \centering
    \begin{align*}
      &\text{Paths} \AlgDef \{ \\
      &\quad (6{:}\,\actstop), (6{:}\,\actgo), \\
      &\quad (6{:}\,\actgo \then 7{:}\,v\coloneqq v{+}a), \\
      &\quad (6{:}\,\actgo \then 7{:}\,v\coloneqq v{+}a \then 6{:}\,\actstop),\\
      &\quad (6{:}\,\actgo \then 7{:}\,v\coloneqq v{+}a \then 6{:}\,\actgo \then 7{:}\,v\coloneqq v{+}a),\\
      &\quad \cdots \}
    \end{align*}
    \caption{Representation of game tree as (prefix-closed) set of paths through the tree.}
    \label{fig:paths-list}
  \end{subfigure}%
  \hspace{0.01\textwidth} 
  \begin{subfigure}[t]{0.45\textwidth}
    \centering
    \begin{tikzpicture}
      [
        baseline=(current bounding box.north),
        level distance={\astSpacing},
        every node/.style={draw,rectangle,align=center},
        sibling distance=3cm
      ]
      \node[draw=none, inner sep=0pt, minimum size=0pt] {} 
      child { node {\nodeLabel{6}:\actstop} }
      child {
        node {\nodeLabel{6}:\actgo}
        child {
            node {\nodeLabel{7}:(v:=v+a)}
            [sibling distance=1.8cm]
            child { node {\nodeLabel{6}:\actstop} }
            child {
                node {\nodeLabel{6}:\actgo}
                [sibling distance=2.80cm]
                child {
                  node {\nodeLabel{7}:(v:=v+a)}
                  child [draw=none] { edge from parent[dashed] }
                  child [draw=none] { edge from parent[dashed] }
                  }
            }
        }
      };
    \end{tikzpicture}
    \caption{Visualization of the labeled game tree}
    \label{fig:paths-tree}
  \end{subfigure}
  \caption{The game tree for subgame 6 of the game in \rref{fig:overview}, i.e., \namedGame{(\namedGame{v:=v+a}{7})^\ast}{6}}
  \label{fig:tree-and-paths}
\end{figure}

\rref{def:labeled-semantics} in \rref{app:labeled-game-trees} shows the full construction of trees%
\footnote{This is similar to the standard \dGL operational semantics \cite[Appendix C]{DBLP:journals/tocl/Platzer15}, but with the addition of labels to keep track of subgames and some notation changes.}.
A path in a tree can be represented as the sequence of labels and actions encountered while traversing the path, where the subgame label appears immediately before the action that subgame corresponds to.
Because of loops, subgame labels can repeat even within a single path.
For example, in \rref{fig:tree-and-paths} because of the Angelic loop, labels 6 and 7 repeat.

Tree nodes have actions, which correspond to transformations from old game state before the action to new state.
\rref{eq:action-semantics} shows the transformations each action corresponds to, where notation \(\runstrategy{\act}{\sigma}\) is the state reached by running action \(\act\) starting at state \(\sigma\).
Labels have no impact on state and correspond to the identity function.
For tests and ODEs, it is possible for the transition to be undefined. In this case, the player responsible for the action loses.

\begin{equation}
  \label{eq:action-semantics}
  \begin{aligned}
    &\runstrategy{x:=e}{\sigma} = \sigma(x \mapsto e)
    \qquad
    \runstrategy{\act}{\sigma} = \sigma \textrm{ for } \act \in \{\actgo, \actstop, \actleft, \actright, \actseq, \pdual{\actgo}, \pdual{\actstop}, \pdual{\actleft}, \pdual{\actright}, \textrm{label}\}
    \\
    &\runstrategy{\ptest{Q}}{\sigma} = \begin{cases}
      \sigma &\textrm{ if } \sigma\in \Sem{Q} \\
      \textrm{not defined} &\textrm{ otherwise}
    \end{cases} \\
    &\runstrategy{\{x'=f(x) \& Q @ t\}}{\sigma} = \varphi(t) \textrm{ for the unique (differentiable) }\varphi : [0, t] \rightarrow \textrm{States}, \varphi(0) = \sigma,\\
    &\ \textrm{for all }s \in [0,t]
    \left(\deriveWrt{\varphi(r)(x)}{r}(s)=f(\varphi(s)(x)) \land \varphi(s)\models Q\right).
    \textrm{ Not defined if no such \(\varphi\) exists.}
  \end{aligned}
\end{equation}

Some transformations have no impact on state.
These transformations correspond either to labels or to actions whose only role is to restrict tree structure.
This set of these actions is \(\skipActions = \{\actgo, \actstop, \actleft, \actright, \actseq, \pdual{\actgo}, \pdual{\actstop}, \pdual{\actleft}, \pdual{\actright}\}\).
A path is executed by inductively applying the actions: \(\runstrategy{a \then t}{\sigma} = \runstrategy{t}{\runstrategy{a}{\sigma}}\).

We syntactically compare paths and trees per the reachable states resulting from gameplay, ignoring structural differences, \emph{by erasing the nodes that have no impact on state}.
This is useful for characterizing subgame suffixes because they do not preserve structure, only the reachable states.
The action projection operator \(\eraseSkip\) returns the path with all labels and actions belonging to \(\skipActions\) removed (full construction in \rref{app:strategy-set-generation}, \rref{def:erase-noop}).
For paths \(p_1\) and \(p_2\), if \(\eraseSkip(p_1)=\eraseSkip(p_2)\), then \(\runstrategy{p_1}{\sigma} = \runstrategy{p_2}{\sigma} = \runstrategy{\eraseSkip(p_1)}{\sigma}\).

\(\eraseSkip\) is overloaded to apply similarly to game trees.
For game tree \(t\), \(\eraseSkip(t)\) is a game tree where all paths have only actions that can impact state, produced by applying \(\eraseSkip\) to every path in \(t\).
Two trees result in the same reachable states (even if their structure differs) if their projections are the same.
We are now ready to show the defining property of game suffixes.

\begin{lemmaE}[Game Suffixes are Operational Gameplay Suffixes]
  \label{lem:execution-suffix}
  For subgame \(\namedGame{\beta}{b}\) in overall game \(\namedGame{\alpha}{a}\), let the game suffix be \(\fwd{b}{\namedGame{\alpha}{a}}\).
  Let \(\sigma\) be some state.
  The states reachable from playing \(\namedGame{\alpha}{a}\) such that \(\sigma\) is reached at subgame \(\namedGame{\beta}{b}\) at some point during the gameplay are the same as the states reachable from playing \(\fwd{b}{\namedGame{\alpha}{a}}\) starting in \(\sigma\).
  That is,
  \[\{\eraseSkip(s) \with p \then b \then s \in \opSemantics{\namedGame{\alpha}{a}}{\sigma'},
  \runstrategy{p}{\sigma'}=\sigma
  \textrm{ and }
  \eraseSkip(s) \neq () \}
  =
  \eraseSkip(\opSemantics{\fwd{b}{\namedGame{\alpha}{a}}}{\sigma}),\]
  where \(()\) is the empty path.
\end{lemmaE}
\begin{proofE}
  We prove
  \[
  \{\eraseSkip(s) \with p \then b \then s \in \opSemantics{\namedGame{\alpha}{a}}{\sigma'},
  \runstrategy{p}{\sigma'}=\sigma
  \textrm{ and }
  \eraseSkip(s) \neq () \}
  =
  \eraseSkip(\opSemantics{\fwd{b}{\namedGame{\alpha}{a}}}{\sigma})\]
  by induction on the structure of \(\namedGame{\alpha}{a}\), and case analysis.
  \begin{itemize}
    \item When \(\namedGame{\alpha}{a}\) is atomic and not controlled by Angel,
    i.e., \(\alpha \in \{x:=e, x:=\otimes, \ptest{Q}, !Q, \{x'=f(x)\ \&\ Q\}^d\}\), then
    the only subgame \(b\) in \(\nodes{\namedGame{\alpha}{a}}\) is \(\namedGame{\alpha}{a}\) itself.
    Thus the only possibility is \(b=a\).
    \(\fwd{a}{\namedGame{\alpha}{a}}\) is the game \(\namedGame{\alpha}{a}\).
    \(\opSemantics{\fwd{a}{\namedGame{\alpha}{a}}}{\sigma}\) is then \(\opSemantics{\namedGame{\alpha}{a}}{\sigma}\).

    The left-hand side, \(\eraseSkip(\opSemantics{\namedGame{\alpha}{a}}{\sigma})\), is \(\opSemantics{\alpha}{\sigma}\), the label-free game tree of the original game (see \rref{def:label-free-game-tree}).
    Per the definition of labeled game trees for any atomic subgame \(\alpha\) (\rref{def:labeled-semantics}), \(\opSemantics{\namedGame{\alpha}{a}}{\sigma}\) is \(a \then t \with t \in \opSemantics{\alpha}{\sigma}\).
    Thus, the right-hand side is also \(\opSemantics{\alpha}{\sigma}\).

    \item Extending the previous case to compositional subgames, when \(b=a\) and \(\namedGame{\alpha}{a}\) is not a loop, then a similar argument holds.
    
    The \emph{left-hand side} is \(\eraseSkip(\opSemantics{\namedGame{\alpha}{a}}{\sigma})\).

    The \emph{right-hand side} is
    \(\{\eraseSkip(s) \with
      p \then a \then s \in \opSemantics{\namedGame{\alpha}{a}}{\sigma'},
      \eraseSkip(s) \neq ()
      \textrm{ and } \runstrategy{p}{\sigma'}=\sigma\}
    \).
    Per the definition of (non-loop) labeled game trees, the only occurrence of label \(a\) is at the beginning of all plays and thus the right-hand side is equal to \(\{\eraseSkip(s) \with a \then s \in \opSemantics{\namedGame{\alpha}{a}}{\sigma}, \eraseSkip(s) \neq ()\}\).
    This is just \(\eraseSkip(\opSemantics{\namedGame{\alpha}{a}}{\sigma})\).
 
    \item When \(\namedGame{\alpha}{a}\) is a loop, i.e., \(\namedGame{(\namedGame{\gamma}{g})^\ast}{a}\) or \(\namedGame{(\namedGame{\gamma}{g})^\times}{a}\), and \(b\neq a\), then
    \(\fwd{b}{\namedGame{\alpha}{a}} = (\fwd{b}{\namedGame{\gamma}{g}}) \seq \namedGame{\alpha}{a}\).
    
    The \emph{left-hand side} is then \(\eraseSkip(\opSemantics{\fwd{b}{\namedGame{\gamma}{g}} \seq \namedGame{\alpha}{a}}{\sigma})\).

    Per the inductive hypothesis, \(\eraseSkip(\opSemantics{\fwd{b}{\namedGame{\gamma}{g}}}{\sigma})\) is equal to \(\{ \eraseSkip(s) \with p \then b \then s \in \opSemantics{\namedGame{\gamma}{g}}{\sigma'} \textrm{ and } \runstrategy{p}{\sigma'}=\sigma\}\).
    Thus, per the definition of labeled game trees for sequential composition games,
    the left-hand side is
    \(
      \{
        \eraseSkip(s) \with p \then b \then s \in \opSemantics{\namedGame{\gamma}{g}}{\sigma'}
        \textrm{ and }
        \runstrategy{p}{\sigma'}=\sigma \}
      \cup
      \{
        \eraseSkip(s \then t) \with p \then b \then s \in \leaf(\opSemantics{\namedGame{\gamma}{g}}{\sigma'}),
        \runstrategy{p}{\sigma'}=\sigma,
        \runstrategy{s}{\sigma}=\sigma'' \textrm{ and }
        t \in \opSemantics{\namedGame{\alpha}{a}}{\sigma''}
      \}
    \), where \(\leaf\) returns the set of leaf nodes in a tree as discussed in \rref{app:labeled-game-trees}.

    The \emph{right-hand side} is \(\{\eraseSkip(s) \with p \then b \then s \in \opSemantics{\namedGame{\alpha}{a}}{\sigma'}, \eraseSkip(s)\neq () \textrm{ and } \runstrategy{p}{\sigma'}=\sigma\}\).
    Per the definition of labeled game trees for loops, we can break this down by unrolling the loop: 
    \(
      \{
        \eraseSkip(s) \with p \then b \then s \in \opSemantics{\namedGame{\gamma}{g}}{\sigma'}
        \textrm{ and }
        \runstrategy{p}{\sigma'}=\sigma \}
      \cup
      \{
        \eraseSkip(s \then t) \with p \then b \then s \in \leaf(\opSemantics{\namedGame{\gamma}{g}}{\sigma'}),
        \runstrategy{p}{\sigma'}=\sigma,
        \runstrategy{s}{\sigma}=\sigma'' \textrm{ and }
        t \in \opSemantics{\namedGame{\alpha}{a}}{\sigma''}
      \}
    \).

    \item Otherwise, when \(b=a\) and \(\namedGame{\alpha}{a}\) is a loop, i.e., \(\namedGame{(\namedGame{\gamma}{g})^\ast}{a}\) or \(\namedGame{(\namedGame{\gamma}{g})^\times}{a}\), then \(\fwd{a}{\namedGame{\alpha}{a}} = \namedGame{\alpha}{a}\).

    The \emph{left-hand side} is \(\eraseSkip(\opSemantics{\namedGame{\alpha}{a}}{\sigma})\).
    
    The \emph{right-hand side} is \(\{\eraseSkip(s) \with p \then a \then s \in \opSemantics{\namedGame{\alpha}{a}}{\sigma'}, \eraseSkip(s)\neq () \textrm{ and } \runstrategy{p}{\sigma'}=\sigma\}\).
    By the choice of \(\sigma'=\sigma\), \(p=()\) and the definition of labeled game trees, we can show that this is a superset of \(\eraseSkip(\opSemantics{\namedGame{\alpha}{a}}{\sigma})\).
    Moreover, by the definition of labeled game trees for loops, we know that any suffix \(s\) that follows label \(a\) must be generated by the loop fixed point function \(f\) from \rref{def:labeled-semantics} for initial state \(\sigma\).
    So, \(a \then s\) must be within the game tree \(\opSemantics{\namedGame{\alpha}{a}}{\sigma}\).
    Therefore, the right-hand side is also a subset of \(\eraseSkip(\opSemantics{\namedGame{\alpha}{a}}{\sigma})\).
    Thus, the left-hand side and right-hand side must be equal.
    \item When \(\namedGame{\alpha}{a}\) is a choice, i.e., \(\namedGame{(\namedGame{\gamma}{g} \cup \namedGame{\delta}{d})}{a}\) or \(\namedGame{(\namedGame{\gamma}{g} \cap \namedGame{\delta}{d})}{a}\), and \(b\neq a\), then there are two cases.
    \begin{enumerate}
      \item If \(b\in \nodes{\namedGame{\gamma}{g}}\), then \(\fwd{b}{\namedGame{\alpha}{a}} = \fwd{b}{\namedGame{\gamma}{g}}\).
      
      The \emph{left-hand side} is \(\eraseSkip(\opSemantics{\fwd{b}{\namedGame{\gamma}{g}}}{\sigma})\).
      Per the inductive hypothesis, this is \(\{ \eraseSkip(s) \with p \then b \then s \in \opSemantics{\namedGame{\gamma}{g}}{\sigma'}, \eraseSkip(s)\neq () \textrm{ and } \runstrategy{p}{\sigma'} = \sigma\}\).

      The \emph{right-hand side} is \(\{\eraseSkip(s) \with p \then b \then s \in \opSemantics{\namedGame{\alpha}{a}}{\sigma'}, \eraseSkip(s)\neq () \textrm{ and } \runstrategy{p}{\sigma'} = \sigma\}\).
      Per the definition of labeled game trees for choices, \(\opSemantics{\namedGame{\alpha}{a}}{\sigma'}\) is the union of the game trees of \(\namedGame{\gamma}{g}\) and \(\namedGame{\delta}{d}\).
      However, \(b\) does not appear in \(\namedGame{\delta}{d}\).
      Thus the right-hand side can be simplified to \(\{\eraseSkip(s) \with p \then b \then s \in \opSemantics{\namedGame{\gamma}{g}}{\sigma'}, \eraseSkip(s)\neq () \textrm{ and } \runstrategy{p}{\sigma'} = \sigma\}\).
      \item If \(b\in \nodes{\namedGame{\delta}{d}}\), then \(\fwd{b}{\namedGame{\alpha}{a}} = \fwd{b}{\namedGame{\delta}{d}}\), and a symmetric argument applies, with \(\namedGame{\gamma}{g}\) and \(\namedGame{\delta}{d}\) swapped.
    \end{enumerate}
    \item When \(\namedGame{\alpha}{a}\) is a sequential composition, i.e., \(\namedGame{(\namedGame{\gamma}{g} \seq \namedGame{\delta}{d})}{a}\), and \(b\neq a\), then there are two cases.
    \begin{enumerate}
      \item If \(b\in \nodes{\namedGame{\gamma}{g}}\), then \(\fwd{b}{\namedGame{\alpha}{a}} = (\fwd{b}{\namedGame{\gamma}{g}})\seq \namedGame{\delta}{d}\).
      The argument is similar to the loop case.
      
      The \emph{left-hand side} is \(\eraseSkip(\opSemantics{(\fwd{b}{\namedGame{\gamma}{g}})\seq \namedGame{\delta}{d}}{\sigma})\).
      As per the inductive hypothesis, \(\eraseSkip(\opSemantics{\fwd{b}{\namedGame{\gamma}{g}}}{\sigma})\) is equal to \(\{ \eraseSkip(s) \with p \then b \then s \in \opSemantics{\namedGame{\gamma}{g}}{\sigma'}, \eraseSkip(s)\neq () \textrm{ and } \runstrategy{p}{\sigma'}=\sigma\}\).
      Thus, per the definition of labeled game trees for sequential composition games,
      the left-hand side is
      \(
        \{
          \eraseSkip(s) \with 
          p \then b \then s \in \opSemantics{\namedGame{\gamma}{g}}{\sigma'},
          \runstrategy{p}{\sigma'}=\sigma,
          \eraseSkip(s)\neq ()
        \} \cup
        \{
          \eraseSkip(s \then t) \with 
          p \then b \then s \in \leaf(\opSemantics{\namedGame{\gamma}{g}}{\sigma'}),
          t \in \opSemantics{\namedGame{\delta}{d}}{\sigma''},
          \runstrategy{p}{\sigma'}=\sigma,
          \eraseSkip(s \then t)\neq () \textrm{ and }
          \runstrategy{s}{\sigma}=\sigma''
        \}
      \).

      The \emph{right-hand side} is \(\{\eraseSkip(s) \with p \then b \then s \in \opSemantics{\namedGame{\alpha}{a}}{\sigma'}, \eraseSkip(s)\neq () \textrm{ and } \runstrategy{p}{\sigma'}=\sigma\}\).
      Per the definition of labeled game trees for sequential composition, since \(b\not\in\nodes{\namedGame{\delta}{d}}\), this is:
      \(
        \{
          \eraseSkip(s) \with 
          p \then b \then s \in \opSemantics{\namedGame{\gamma}{g}}{\sigma'},
          \runstrategy{p}{\sigma'}=\sigma,
          \eraseSkip(s)\neq ()
        \} \cup
        \{
          \eraseSkip(s \then t) \with 
          p \then b \then s \in \leaf(\opSemantics{\namedGame{\gamma}{g}}{\sigma'}),
          t \in \opSemantics{\namedGame{\delta}{d}}{\sigma''},
          \runstrategy{p}{\sigma'}=\sigma,
          \eraseSkip(s \then t)\neq () \textrm{ and }
          \runstrategy{s}{\sigma}=\sigma''
        \}
      \).
      \item If \(b\in \nodes{\namedGame{\delta}{d}}\), then \(\fwd{b}{\namedGame{\alpha}{a}} = \fwd{b}{\namedGame{\delta}{d}}\).
      
      The \emph{left-hand side} is \(\eraseSkip(\opSemantics{\fwd{b}{\namedGame{\delta}{d}}}{\sigma})\).
      Per the inductive hypothesis, this is \(\{ \eraseSkip(s) \with p \then b \then s \in \opSemantics{\namedGame{\delta}{d}}{\sigma'}, \eraseSkip(s)\neq () \textrm{ and } \runstrategy{p}{\sigma'} = \sigma\}\).

      The \emph{right-hand side} is \(\{\eraseSkip(s) \with p \then b \then s \in \opSemantics{\namedGame{\alpha}{a}}{\sigma'}, \eraseSkip(s)\neq () \textrm{ and } \runstrategy{p}{\sigma'} = \sigma\}\).
      Per the definition of labeled game trees for sequential composition, using the fact that \(b\not\in\nodes{\namedGame{\gamma}{g}}\), this can be simplified to \(\{\eraseSkip(s) \with p \then b \then s \in \opSemantics{\namedGame{\delta}{d}}{\sigma'} \textrm{ and } \runstrategy{p}{\sigma'} = \sigma\}\).
    \end{enumerate}
  \end{itemize}
\end{proofE}

In \rref{lem:execution-suffix}, the quantity on the left-hand side considers all paths in the original game that contain \(b\).
These must be of the form \(p\then b\then s\) where \(p\) is a prefix of \(b\) and \(s\) is a suffix.
The suffix set is built by collecting the projection of all such suffixes (ignoring sequences that are empty after applying \(\eraseSkip\)).
The second quantity on the right-hand side is the projection of the game tree of the subgame suffix at state \(\sigma\).
Although suffix $\fwd{b}{\namedGame{\alpha}{a}}$ may change game structure, ignoring \(\skipActions\), it captures all remaining plays occurring after any occurrence of $\namedGame{\beta}{b}$ while playing $\namedGame{\alpha}{a}$.

\subsection{Game Prefix}
\label{sec:game-prefix}

We also define the complementary concept of the \emph{game prefix}, which will be useful to compute in \emph{what contexts} a subgame can be seen.
In \rref{sec:subvalue-maps}, for example, the prefix of a subgame will help show that while following an Angelic subvalue map, Angel only reaches the subgame within the region from which she can win the rest of the game.
The game prefix of subgame $b$ in the game $\namedGame{\alpha}{a}$, written $\prefix{b}{\namedGame{\alpha}{a}}$, features \emph{all} possible behaviors of game $\alpha$ that can happen \emph{before} an occurrence of subgame $b$ within $\namedGame{\alpha}{a}$.
It is constructed analogously to the game suffix (\rref{app:prefix}, \rref{def:execution-prefix}) and satisfies an analogous property (\rref{lem:execution-prefix}).

\begin{lemmaE}[Execution Prefixes are Original Gameplay Prefixes]
  \label{lem:execution-prefix}
  For subgame \(\namedGame{\beta}{b}\) in \(\namedGame{\alpha}{a}\), for any initial state \(\sigma\), the states reachable at \(\namedGame{\beta}{b}\) while playing \(\namedGame{\alpha}{a}\) are equal to those reachable by playing game prefix \(\prefix{b}{\namedGame{\alpha}{a}}\).
  That is, prefix set \(\{\eraseSkip(p) \with p \then b \then s \in \opSemantics{\namedGame{\alpha}{a}}{\sigma}\}\) is equal, under prefix closure, to prefix game tree \(\eraseSkip(\opSemantics{\prefix{b}{\namedGame{\alpha}{a}}}{\sigma})\)%
  \footnote{
    The prefix closure is required to recover tree structure. Unlike with the suffix construction, the prefix set builder collects only \emph{parent} subgames of \(\namedGame{\beta}{b}\) and \emph{not all predecessors}.
    Expanding the prefix closure results in the following condition:
    \[
      \{\eraseSkip(t) \with p \then b \then s \in \opSemantics{\namedGame{\alpha}{a}}{\sigma} \textrm{ and } p = t \then u\}
      =
      \eraseSkip(\opSemantics{\prefix{b}{\namedGame{\alpha}{a}}}{\sigma})
    \]
  }.
\end{lemmaE}
\begin{proofE}
  We show that
    \[
      \eraseSkip(\opSemantics{\prefix{b}{\namedGame{\alpha}{a}}}{\sigma})
      =
      \{\eraseSkip(t) \with p \then b \then s \in \opSemantics{\namedGame{\alpha}{a}}{\sigma} \textrm{ and } p = t \then u\}
    \]
  Similar to \rref{lem:execution-suffix}, we perform induction on the structure of \(\namedGame{\alpha}{a}\) and case analysis.
  \begin{itemize}
    \item When \(b=a\) and \(\namedGame{\alpha}{a}\) is not a loop, then \(\prefix{a}{\namedGame{\alpha}{a}}\) is the empty game \(\skp\).

    The \emph{left-hand side} \(\eraseSkip(\opSemantics{\prefix{a}{\namedGame{\alpha}{a}}}{\sigma})\) is then \(\emptyset\).

    Indeed, on the \emph{right-hand side}, all plays in \(\opSemantics{\namedGame{\alpha}{a}}{\sigma} \) are of the form \(a \then p\) where \(p\) does not contain \(a\), so the condition of the theorem holds.
    \item When \(\namedGame{\alpha}{a}\) is atomic and not controlled by Angel,
    i.e., \(\alpha \in \{x:=e, x:=\otimes, \ptest{Q}, !Q, \{x'=f(x)\ \&\ Q\}^d\}\), then
    the only subgame \(b\) in \(\nodes{\namedGame{\alpha}{a}}\) is \(\namedGame{\alpha}{a}\) itself.
    Thus the only possibility is \(b=a\), handled already by the first case.
    \item When \(\namedGame{\alpha}{a}\) is a choice, i.e., \(\namedGame{(\namedGame{\gamma}{g} \cup \namedGame{\delta}{d})}{a}\) or \(\namedGame{(\namedGame{\gamma}{g} \cap \namedGame{\delta}{d})}{a}\), then we analyze cases depending on the value of \(b \neq a\).
    \begin{enumerate}
      \item If \(b\) is in \(\namedGame{\gamma}{g}\), then \(\prefix{b}{\namedGame{\alpha}{a}} = \prefix{b}{\namedGame{\gamma}{g}}\).
      
      The \emph{left-hand side} is \(\eraseSkip(\opSemantics{\prefix{b}{\namedGame{\gamma}{g}}}{\sigma})\).
      Per the inductive hypothesis, this consists of 
      \[ \{ \eraseSkip(t) \with p \then b \then s \in \opSemantics{\namedGame{\gamma}{g}}{\sigma}, p=t \then u, \eraseSkip(t)\neq () \}.\]
      
      The \emph{right-hand side} is
      \( \{ \eraseSkip(t) \with p\then b \then s \in \opSemantics{\namedGame{\alpha}{a}}{\sigma}, p=t \then u, \eraseSkip(t)\neq () \} \).
      Per the definition of game trees for choices,
      since \(\opSemantics{\namedGame{\alpha}{a}}{\sigma} \) is the union of the game trees for \(\namedGame{\gamma}{g}\) and \(\namedGame{\delta}{d}\), and \(b\) does not appear in \(\namedGame{\delta}{d}\), this simplifies to
      \[ \{ \eraseSkip(t) \with p \then b \then s \in \opSemantics{\namedGame{\gamma}{g}}{\sigma}, p = t \then u, \eraseSkip(t) \neq () \}.\]
      \item If \(b\) is in \(\namedGame{\delta}{d}\), then the argument is symmetric to the previous case, but \(\namedGame{\gamma}{g}\) swapped with \(\namedGame{\delta}{d}\).
    \end{enumerate}
    \item When \(\namedGame{\alpha}{a}\) is a sequential composition, i.e., \(\namedGame{(\namedGame{\gamma}{g} \seq \namedGame{\delta}{d})}{a}\), then we analyze cases depending on the value of \(b \neq a\).
    \begin{enumerate}
      \item If \(b\) is in \(\namedGame{\gamma}{g}\), then \(\prefix{b}{\namedGame{\alpha}{a}} = \prefix{b}{\namedGame{\gamma}{g}}\).
      
      The \emph{left-hand side} is \(\eraseSkip(\opSemantics{\prefix{b}{\namedGame{\gamma}{g}}}{\sigma})\).
      Per the inductive hypothesis, this is
      \[ \{ \eraseSkip(t) \with p \then b \then s \in \opSemantics{\namedGame{\gamma}{g}}{\sigma}, p = t \then u, \eraseSkip(t)\neq () \}.\]

      The \emph{right-hand side} is \(\{\eraseSkip(t) \with p\then b \then s \in \opSemantics{\namedGame{\alpha}{a}}{\sigma}, p = t \then u, \eraseSkip(t)\neq ()\}\).
      Since \(b\not\in\nodes{\namedGame{\delta}{d}}\), using the definition of labeled game trees for sequential composition, this can be simplified to
      \[ \{\eraseSkip(t) \with p\then b \then s \in \opSemantics{\namedGame{\gamma}{g}}{\sigma}, p = t \then u, \eraseSkip(t)\neq ()\}.\]
      \item If \(b\) is in \(\namedGame{\delta}{d}\), then \(\prefix{b}{\namedGame{\alpha}{a}} = \namedGame{\gamma}{g} \seq \prefix{b}{\namedGame{\delta}{d}}\).
      
      The \emph{left-hand side} is \(\eraseSkip(\opSemantics{\namedGame{\gamma}{g} \seq \prefix{b}{\namedGame{\delta}{d}}}{\sigma})\).
      As per the inductive hypothesis, \(\eraseSkip(\opSemantics{\prefix{b}{\namedGame{\delta}{d}}}{\sigma})\) is equal to \(\{ \eraseSkip(t) \with p \then b \then s \in \opSemantics{\namedGame{\delta}{d}}{\sigma}, p = t \then u, \eraseSkip(t)\neq ()\}\).
      Thus, per the definition of labeled game trees for sequential composition games,
      the left-hand side is
      \[\eraseSkip(\opSemantics{\namedGame{\gamma}{g}}{\sigma}) \cup \{\eraseSkip(q \then t) \with q \in \leaf(\opSemantics{\namedGame{\gamma}{g}}{\sigma}),
      \runstrategy{q}{\sigma}=\sigma',
      p \then b \then s \in \opSemantics{\namedGame{\delta}{d}}{\sigma'},
      p= t \then u \}\]
      where \(\leaf\) indicates a leaf node a game tree, as indicated in \rref{app:labeled-game-trees}.

      The \emph{right-hand side} is \(\{\eraseSkip(t) \with p\then b \then s \in \opSemantics{\namedGame{\alpha}{a}}{\sigma}\}, p = t \then u, \eraseSkip(t)\neq ()\).
      Per the definition of labeled game trees for sequential composition, since \(b\not\in\nodes{\namedGame{\gamma}{g}}\), this can be written as
      \[\eraseSkip(\opSemantics{\namedGame{\gamma}{g}}{\sigma}) \cup \{\eraseSkip(q \then t) \with q \in \leaf(\opSemantics{\namedGame{\gamma}{g}}{\sigma}),
      \runstrategy{q}{\sigma}=\sigma',
      p \then b \then s \in \opSemantics{\namedGame{\delta}{d}}{\sigma'},
      p = t \then u\}.\]
    \end{enumerate}
    \item When \(\namedGame{\alpha}{a}\) is a loop, i.e., \(\namedGame{(\namedGame{\gamma}{g})^\ast}{a}\) or \(\namedGame{(\namedGame{\gamma}{g})^\times}{a}\), then we analyze cases depending on the value of \(b\).
    \begin{enumerate}
      \item If \(b = a\), then \(\prefix{b}{\namedGame{\alpha}{a}} = \namedGame{\alpha}{a}\).
      
      The \emph{left-hand side} is then \(\eraseSkip(\opSemantics{\namedGame{\alpha}{a}}{\sigma})\).

      The \emph{right-hand side} is \(\{\eraseSkip(t) \with p \then a \then s \in \opSemantics{\namedGame{\alpha}{a}}{\sigma}, p = t \then u, \eraseSkip(t)\neq ()\}\).
      By the definition of labeled game trees (\rref{def:labeled-semantics}), all plays in \(\opSemantics{\namedGame{\alpha}{a}}{\sigma}\) end with \(a \then \pdual{\actstop}\) or \(a \then \actstop\).
      Thus, from the elements where play \(s\) is \(\actstop\) or \(\pdual{\actstop}\), \(\{\eraseSkip(t) \with p \then a \then s \in \opSemantics{\namedGame{\alpha}{a}}{\sigma}, p = t \then u, \eraseSkip(t)\neq ()\}\) is a superset of \(\eraseSkip(\opSemantics{\namedGame{\alpha}{a}}{\sigma})\).
      Moreover, as labeled game trees are prefix-closed, \(\{\eraseSkip(t) \with p \then a \then s \in \opSemantics{\namedGame{\alpha}{a}}{\sigma}, p = t \then u, \eraseSkip(t)\neq ()\}\) is also a subset of \(\eraseSkip(\opSemantics{\namedGame{\alpha}{a}}{\sigma})\).
      Thus, the left-hand side and right-hand side must be equal.
      \item If \(b\neq a\), then \(\prefix{b}{\namedGame{\alpha}{a}} = (\namedGame{\alpha}{a} \seq \prefix{b}{\namedGame{\gamma}{g}})\).
      
      The \emph{left-hand side} is \(\eraseSkip(\opSemantics{\namedGame{\alpha}{a} \seq (\prefix{b}{\namedGame{\gamma}{g}})}{\sigma})\).
      As per the inductive hypothesis, \(\eraseSkip(\opSemantics{\prefix{b}{\namedGame{\gamma}{g}}}{\sigma})\) is equal to \(\{ \eraseSkip(t) \with p \then b \then s \in \opSemantics{\namedGame{\gamma}{g}}{\sigma}, p = t \then u, \eraseSkip(t)\neq ()\}\).
      Thus, per the definition of labeled game trees for sequential composition games,
      the left-hand side is 
      \[ \eraseSkip(\opSemantics{\namedGame{\alpha}{a}}{\sigma}) \cup \{\eraseSkip(q \then t) \with q \in \leaf(\opSemantics{\namedGame{\alpha}{a}}{\sigma}),
      \runstrategy{q}{\sigma}=\sigma',
      p \then b \then s \in \opSemantics{\namedGame{\gamma}{g}}{\sigma'},
      p = t \then u\}.\]

      The \emph{right-hand side} is \(\{\eraseSkip(t) \with p\then b \then s \in \opSemantics{\namedGame{\alpha}{a}}{\sigma}, p = t \then u, \eraseSkip(t)\neq ()\}\).
      This is a subset of \(\eraseSkip(\opSemantics{\namedGame{\alpha}{a}}{\sigma})\), and therefore the left-hand side, because game trees are prefix-closed.
      Next we show it is also a superset.
      Consider any play in the left-hand side set.
      For Angelic loops, it must be a prefix of some play of the form \(q \then p\) where \(q \in \leaf(f^n(a \caret \actstop, a \caret \actgo))\) for some \(n\), \(\runstrategy{q}{\sigma}=\sigma'\), and \(p \then b \then s \in \opSemantics{\namedGame{\gamma}{g}}{\sigma'}\) where \(f\) and \(\caret\) are from \rref{def:labeled-semantics}.
      Such a play and all its prefixes feature in the right-hand side set because they are in \(f^{n+1}(a \caret \actstop, a \caret \actgo)\) which is in \(\opSemantics{\namedGame{\alpha}{a}}{\sigma}\).
      For Demonic loops, the same argument applies but for \(f^n(a \caret \pdual{\actstop}, a \caret \pdual{\actgo})\) as defined in \rref{def:labeled-semantics}.
    \end{enumerate}
  \end{itemize}
\end{proofE}

\section{Subvalue Maps}
\label{sec:subvalue-maps}

This section defines subvalue maps, shows their interpretation as policies, and discusses the properties of these policies.

\subsection{Subvalue Map Definition}
\label{sec:subvalue-map-definition}

\rref{def:subvalue-map} characterizes a subvalue map.
Each subgame maps to a formula that must logically imply the winning region for playing the rest of the game starting at that subgame.
The exact winning region starting from a subgame can be defined as the winning region for the game suffix.
A special \(\finalNode\) subgame additionally stores the terminal winning subregion of the overall game that the policy induced by the subvalue map tries to reach.
Mathematically, when Angel is playing the overall game $\namedGame{\alpha}{a}$ with winning condition $S(\finalNode)$, her winning region at the point when she is about to play subgame $\namedGame{\beta}{b}$ to achieve $S(\finalNode)$ overall is $\ddiamond{\fwd{b}{\namedGame{\alpha}{a}}}{S(\finalNode)}$.

\begin{definition}[Subvalue maps]
  \label{def:subvalue-map}
  A map $S$ from the subgames of game $\namedGame{\alpha}{a}$ and special subgame \finalNode to formulas is an \emph{Angelic subvalue map} of $\namedGame{\alpha}{a}$ when $\models S(b) \limply \langle \fwd{b}{\namedGame{\alpha}{a}} \rangle S(\finalNode)$ for every subgame $\namedGame{\beta}{b}$ in $\nodes{\namedGame{\alpha}{a}}$.
  Dually, $S$ is a \emph{Demonic subvalue map} for game $\namedGame{\alpha}{a}$ when $\models S(\beta) \limply [\fwd{b}{\namedGame{\alpha}{a}}] S(\finalNode)$ for every subgame $\namedGame{\beta}{b}$ in $\nodes{\namedGame{\alpha}{a}}$.
\end{definition}

\paragraph{Interpretation as policy.}
Angelic subvalue maps effectively provide Angel with a (nondeterministic) policy that describes how to play the game.
For example, the map in \rref{fig:overview} says that Angel at her choice $\cup$ should not transition to $(v:=v+a)^\ast$ when $\phi_6$ is false, and not transition to $v:=v-1 \seq a:=\ast$ when $\phi_5$ is false.
The intuition is, at any given subgame where Angel has a choice (e.g. \rref{fig:overview} at $\cup$), the \emph{next subgame} that she chooses to run (e.g., left branch $(v:=v+a)^\ast$) should be such that the current state $\sigma$ satisfies the winning subregion for that subsequent subgame (for the example, $\sigma \models \phi_6$).
This next subgame to run after the current one ends is called the \emph{successor}, identified as \rref{def:successor} shows
\footnote{The successor is not just the next label to occur in the game suffix. This is because the successor should not be a child of the current subgame. It should only start when the current subgame ends.}.

\begin{definition}[Successor]
  \label{def:successor}
  The successor of subgame $\namedGame{\beta}{b}$ in overall game $\namedGame{\alpha}{a}$, written $\successor{b}{\namedGame{\alpha}{a}}$ is defined as follows.
  If $\namedGame{\alpha}{a}$ is $\namedGame{\beta}{b}$ then $\successor{b}{\namedGame{\alpha}{a}}$ is the special subgame \(\finalNode\). Otherwise, if $\namedGame{\alpha}{a}$ has structure:
  \begin{equation*}
    \begin{aligned}
      \namedGame{(\namedGame{\gamma}{g} \seq \namedGame{\delta}{d})}{a} \textrm{ then }&
      \begin{cases}
        \successor{b}{\namedGame{\alpha}{a}} = \namedGame{\delta}{d}
          & b \in \nodes{\namedGame{\gamma}{g}} \textrm{ and }\successor{b}{\namedGame{\gamma}{g}}=\finalNode  \\
        \successor{b}{\namedGame{\alpha}{a}} = \successor{b}{\namedGame{\gamma}{g}}
          & \textrm{otherwise if } b \in \nodes{\namedGame{\gamma}{g}} \\
        \successor{b}{\namedGame{\alpha}{a}} = \successor{b}{\namedGame{\delta}{d}}
          &\textrm{otherwise}
      \end{cases}
      \\
      \namedGame{(\namedGame{\gamma}{g})^\ast}{a} \textrm{ or } \namedGame{(\namedGame{\gamma}{g})^\times}{a} &\textrm{ then } \\
      & \begin{cases}
        \successor{b}{\namedGame{\alpha}{a}} = a & b \in \nodes{\namedGame{\gamma}{g}} \textrm{ and } \successor{b}{\namedGame{\gamma}{g}} = \finalNode \\
        \successor{b}{\namedGame{\alpha}{a}} = \successor{b}{\namedGame{\gamma}{g}} & \textrm{otherwise if } b \in \nodes{\namedGame{\gamma}{g}} \\
      \end{cases} \\
      \begin{aligned}
          \namedGame{(\namedGame{\gamma}{g} \cup \namedGame{\delta}{d})}{a} \textrm{ or }& \\
          \namedGame{(\namedGame{\gamma}{g} \cap \namedGame{\delta}{d})}{a}\ \ &
      \end{aligned}
      \textrm{ then }&
      \begin{cases}
        \successor{b}{\namedGame{\alpha}{a}} = \successor{b}{\namedGame{\gamma}{g}} &b \in \nodes{\namedGame{\gamma}{g}} \\
        \successor{b}{\namedGame{\alpha}{a}} = \successor{b}{\namedGame{\delta}{d}} &\textrm{otherwise}
      \end{cases}
    \end{aligned}
  \end{equation*}
\end{definition}

Using the successor definition, \rref{def:policy} shows how in general, Angel can interpret the subvalue map as a (nondeterministic) policy indicating what control decisions are safe to take to eventually win based on the current state.

\begin{definition}[Policy Interpretation]
  \label{def:policy}
  \emph{Policy interpretation function} \(\apolicy{\alpha}{a}(S)\) transforms an Angelic subvalue map $S$ for game $\namedGame{\alpha}{a}$ to a nondeterministic policy function that maps current state \(\sigma\) and subgame \(\namedGame{\beta}{b}\) to the set of acceptable control actions Angel can take.
  \(\apolicy{\alpha}{a}(S)(b, \sigma)\) is defined as below. When the current subgame $\namedGame{\beta}{b}$ has structure:
  \begin{equation*}
  \begin{aligned}
    \namedGame{(\namedGame{\delta}{d} \cup \namedGame{\gamma}{g})}{b} \text{ then }&
    \begin{cases}
      \{\mathfrak{l}, \mathfrak{r}\} & \sigma \models S(d) \land S(g) \\
      \{\mathfrak{l}\} & \text{else if } \sigma \models S(d) \\
      \{\mathfrak{r}\} & \text{else if } \sigma \models S(g) \\
      \emptyset & \text{otherwise}
    \end{cases} \\
    \namedGame{(\namedGame{\gamma}{g})^*}{b}\text{ then }&
    \begin{cases}
      \{\actstop, \actgo\} & \sigma \models S(g) \land S(\successor{b}{\namedGame{\alpha}{a}}) \\
      \{\actgo\} & \text{else if } \sigma \models S(g) \\
      \{\actstop\} & \text{else if } \sigma \models S(\successor{b}{\namedGame{\alpha}{a}}) \\
      \emptyset & \text{otherwise}
    \end{cases} \\
    \namedGame{x := *}{b} \text{ then }& \{ (x := e) \mid \sigma(x \mapsto e) \models S(\successor{b}{\namedGame{\alpha}{a}}) \} \\
    \namedGame{\{x' = f(x) \, \&\, Q\}}{b} &\text{ then } \{
      (x' = f(x) \, \&\, Q \,@\, t) \with t\geq 0, \forall s \in [0, t], \varphi(s) \models S(b) \text{ and } \\
      &\varphi(t) \models S(\successor{b}{\namedGame{\alpha}{a}})
      \text{ where } \varphi: [0, t] \to \mathcal{S} \text{ is differentiable, }\\ 
      &\varphi(0) = \sigma, \text{ and } 
      \forall s \in [0, t],\, \varphi(s) \models x' = f(x) \land Q
    \}
  \end{aligned}
  \end{equation*}
\end{definition}

\paragraph{Relationship with strategies.}
We have focused on \emph{policies} which provide local answers for what action to play next.
\emph{Strategies}, on the other hand, are global gameplay solutions.
The Angelic strategy at a given initial state corresponds to a game tree at that state that has been pruned so that at every subgame where Angel has to take a decision, she takes exactly one decision.
Thus, a strategy plans out all the future moves of its player for all potential actions of the opponent for all parts of the game at once.
Viewing policies from the lens of strategies is useful for understanding the global behaviors that a policy induces, and will be relevant in \rref{sec:subvalue-projection} and \rref{sec:universal-projection}. 
The set of strategies that Angel can play for the game \(\namedGame{\alpha}{a}\) while following the policy of a subvalue map \(S\) starting in a state \(\sigma\) is written as \(\astrategize{\namedGameSpace{\alpha}{a}}{S}(\sigma)\).
The construction is shown in \rref{app:strategy-set-generation}, \rref{def:strategy-set-generation}.
The intuition is to generate the game tree while restricting Angel to play per the policy.
When the policy has multiple options at a given Angelic decision point, \(\astrategize{\namedGameSpace{\alpha}{a}}{S}(\sigma)\) has different strategies for each way to proceed.
When the policy is deterministic, the set contains a single strategy.

\paragraph{Winning games.} Subvalue maps induce policies that are suitable to play games with \emph{compatible} winning conditions, i.e., winning conditions that contain the terminal subregion of the policy that \(\finalNode\) maps to (\rref{def:compatible-winning-conditions}).

\begin{definition}[Compatible winning conditions]
  \label{def:compatible-winning-conditions}
  An Angelic subvalue map $S$ for game $\namedGame{\alpha}{a}$ is said to be \emph{compatible with Angel winning condition} \(\phi\) when $\models S(\finalNode) \limply \phi$.
  Dually, a Demonic subvalue map $S$ for game $\namedGame{\alpha}{a}$ is compatible with Demon winning condition \(\phi\) when $\models S(\finalNode) \limply \phi$.
\end{definition}

If Angel plays a game for a compatible winning condition per a policy induced by an Angelic subvalue map, she always remains in a state where there \emph{exists} a way for her to win the remaining game (\rref{thm:subvalue-map-winning-region} at the end of this section).
However, the policy she needs to play to win \emph{might not lie within the subvalue map}.

To see what this can look like, for the example of \rref{fig:overview} at Angel choice $\cup$, consider the subvalue map in the figure except that \(\phi_6\) is \(\bot\) instead of \(x>0 \land (v\geq 0 \lor a>0)\).
This is still an Angelic subvalue map.
If we arrive at \(\phi_8\) in state \(\sigma \equiv \{x=0, v=0, a=0\}\), then Angel can win by choosing the left branch guarded by \(\phi_6\), but will lose if she chooses the right branch.
The policy induced by the map at this point does not let Angel into the losing play of choosing the right branch.
However it does not let her choose the left branch either since \(\sigma \not\models \bot\).
Its returned action set is \(\emptyset\). Angel is stuck.
This motivates the definition of an \emph{inductive} subvalue map (\rref{sec:solution-representation}), where not only does the policy keep Angel in a state where she has a way to win, but this victory can be achieved by continuing to follow the subvalue map policy.

We next develop the formalism to say \emph{for all ways to play the policy of the subvalue map} (\rref{sec:universal-projection}) up to any subgame, \emph{there exists a way to continue to play within the subvalue map} (\rref{sec:subvalue-projection}).
We formulate this quantification over the ways to play a subvalue map within \dGL.

\subsection{Existential Subvalue Map Projection}
\label{sec:subvalue-projection}

The \emph{existential projection} of game $\namedGame{\alpha}{a}$ onto Angelic subvalue map $S$, written \(\aproj{\namedGame{\alpha}{a}}{S}{}\), is the game that can be won exactly when there \emph{exists} a way for Angel to win $\namedGame{\alpha}{a}$ \emph{while restricted to stay within the (nondeterministic) policy of} $S$.
Projection changes the game so that Angel is put on rails, limited to the policy induced by $S$ while taking every decision, and is forced to manifest what she must do to win while playing within the policy.

\begin{lemmaE}[Existential projection correspondence][restate, text proof={}]
  \label{lem:exists-projection}
  For any Angelic subvalue map \(S\) for game \(\namedGame{\alpha}{a}\) and compatible winning condition \(\phi\), in initial state \(\sigma \models \ddiamond{\alpha}{\phi}\), Angel has a winning strategy for the game \(\aproj{\namedGame{\alpha}{a}}{S}{}\) if and only if Angel can win \(\ddiamond{\alpha}{\phi}\) while following the policy of \(S\).
  That is, \(\sigma \models \ddiamond{\aproj{\namedGame{\alpha}{a}}{S}{}}{\phi}\) iff \(\astrategize{\namedGameSpace{\alpha}{a}}{S}(\sigma)\neq \emptyset\).
\end{lemmaE}
\begin{proofE}
  We use structural induction on the game $\namedGame{\alpha}{a}$.
  The cases for $\namedGame{\alpha}{a}$ are as follows.
  \begin{itemize}
    \item When \(\alpha\) is atomic and not controlled by Angel, that is, \(\alpha \in \{x:=e, x:=\otimes, \ptest{Q}, !Q, \{x'=f(x)\ \&\ Q\}^d\}\), then:
    
    \(\ddiamond{\aproj{\namedGame{\alpha}{a}}{S}{}}{\phi}\) is\\
    \(\ddiamond{\alpha}{\phi}\) \(\quad\) because \(\aproj{\namedGame{\alpha}{a}}{S}{}\)= \(\alpha\),\\
    which holds because by assumption \(\sigma \models \ddiamond{\alpha}{\phi}\).

    The corresponding strategy sets are all nonempty and always contain exactly one strategy per \rref{def:strategy-set-generation} as Angel has no decisions to make and the Angelic strategy is just the full game tree.
    \item If $\namedGame{\alpha}{a}$ is a non-deterministic Angel assignment $\namedGame{x:=\ast}{a}$, then \(\astrategize{\namedGameSpace{\alpha}{a}}{S}(\sigma)\) has the strategy set \(\{\{(x:=e) \mid \sigma(x \mapsto e) \models S(\finalNode)\}\}\).
    
    \(\sigma \models \ddiamond{\aproj{\namedGame{\alpha}{a}}{S}{\phi}}{\phi}\) iff\\
    \(\sigma \models \ddiamond{x:=\ast \seq \ptest{S(\finalNode)}}{\phi}\) \(\quad\) because \(\aproj{\namedGame{\alpha}{a}}{S}{}\) is ${x:=\ast \seq \ptest{S(\finalNode)}}$\\
    iff \(\sigma \models \ddiamond{x:=\ast}{(S(\finalNode) \land \phi)}\) \(\quad\) applying axiom \irref{testd}\\
    iff \(\sigma \models \ddiamond{x:=\ast}{(S(\finalNode))}\) \(\quad\) \(S(\finalNode) \land \phi \leftrightarrow S(\finalNode)\) as \(\phi\) is a compatible condition\\
    iff \(\exists e\, \sigma(x \mapsto e) \models S(\finalNode)\) \(\quad\) by semantics of \dGL\\
    which is exactly the condition for which the strategy set is nonempty.
    \item If $\namedGame{\alpha}{a}$ is an Angelic ODE $\namedGame{\{x'=f(x)\ \&\ Q\}}{a}$, then
    \(\astrategize{\namedGameSpace{\alpha}{a}}{S}(\sigma)\) has the strategy set
    \(\{
      \{(x' = f(x) \, \&\, Q \,@\, t)\} \with
      t\geq 0, \forall s \in [0, t], \varphi(s) \models S(b) \text{ and } \\
      \varphi(t) \models S(\successor{b}{\namedGame{\alpha}{a}})
      \text{ where } \varphi: [0, t] \to \mathcal{S} \text{ is differentiable, }\\ 
      \varphi(0) = \sigma, \text{ and } 
      \forall s \in [0, t],\, \varphi(s) \models x' = f(x) \land Q
    \}\).

    \(\sigma \models \ddiamond{\aproj{\namedGame{\alpha}{a}}{S}{\phi}}{\phi}\) iff\\
    \(\sigma \models \ddiamond{\{x'=f(x)\ \&\ Q\} \seq \ptest{S(\finalNode)}}{\phi}\) \(\quad\) because \(\aproj{\namedGame{\alpha}{a}}{S}{}\) is \(\{x'=f(x)\ \&\ Q\} \seq \ptest{S(\finalNode)}\),
    iff \(\sigma \models \ddiamond{\{x'=f(x)\ \&\ Q\}}{(S(\finalNode) \land \phi)}\) \(\quad\) applying axiom \irref{testd}\\
    iff \(\sigma \models \ddiamond{\{x'=f(x)\ \&\ Q\}}{(S(\finalNode))}\) \(\quad\) \(S(\finalNode) \land \phi \leftrightarrow S(\finalNode)\) as \(\phi\) is a compatible winning condition\\
    Which per the semantics of \dGL is exactly when the strategy set is nonempty.
    \item If $\namedGame{\alpha}{a}$ is an Angelic choice $\namedGame{(\namedGame{\gamma}{g} \cup \namedGame{\delta}{d})}{a}$, then \(\astrategize{\namedGameSpace{\alpha}{a}}{S}(\sigma)\) there are four possibilities.
    \begin{enumerate}
      \item If \(\sigma \not\models S(g)\) and \(\sigma \not\models S(d)\), then the strategy set is empty.
      
      Accordingly, Angel cannot win the existential projection game:
      \(\sigma \models \ddiamond{\aproj{\namedGame{\alpha}{a}}{S}{\phi}}{\phi}\) iff\\
      \(\sigma \models \ddiamond{(\ptest{S(g)}\seq\namedGame{\gamma}{g}) \cup (\ptest{S(d)}\seq\namedGame{\delta}{d})}{\phi}\) \(\quad\) because \(\aproj{\namedGame{\alpha}{a}}{S}{}\) is \((\ptest{S(g)}\seq\namedGame{\gamma}{g}) \cup (\ptest{S(d)}\seq\namedGame{\delta}{d})\)\\
      iff \(\sigma \models \ddiamond{(\ptest{S(g)}\seq \aproj{\namedGame{\gamma}{g}}{S}{})}{\phi} \lor \ddiamond{(\ptest{S(d)}\seq \aproj{\namedGame{\delta}{d}}{S}{})}{\phi}\) \(\quad\) by axiom \irref{choiced}\\
      iff \(\sigma \models (S(g) \land \ddiamond{\aproj{\namedGame{\gamma}{g}}{S}{}}{\phi}) \lor (S(d) \land \ddiamond{\aproj{\namedGame{\delta}{d}}{S}{}}{\phi})\) \(\quad\) by axiom \irref{testd}, \irref{composed}\\
      Which is false by assumption \(\sigma \not\models S(g)\) and \(\sigma \not\models S(d)\).
      \item If \(\sigma \models S(g)\) and \(\sigma \not\models S(g)\), then the strategy set is \(\{\{\actleft \caret t\} \with t\in\astrategize{\namedGameSpace{\gamma}{g}}{S}( \sigma)\}\), which is nonempty when \(\astrategize{\namedGameSpace{\gamma}{g}}{S}(\sigma)\) is nonempty.
      By the inductive hypothesis, \(\astrategize{\namedGameSpace{\gamma}{g}}{S}(\sigma)\) is nonempty when \(\sigma \models \ddiamond{\aproj{\namedGame{\gamma}{g}}{S}{}}{\phi}\).

      \(\sigma \models \ddiamond{\aproj{\namedGame{\alpha}{a}}{S}{\phi}}{\phi}\) iff\\
      \(\sigma \models \ddiamond{((\ptest{S(g)}\seq\namedGame{\gamma}{g}) \cup (\ptest{S(d)}\seq\namedGame{\delta}{d}))}{\phi}\) \(\quad\) because \(\aproj{\namedGame{\alpha}{a}}{S}{}\) is \((\ptest{S(g)}\seq\namedGame{\gamma}{g}) \cup (\ptest{S(d)}\seq\namedGame{\delta}{d})\)\\
      iff \(\sigma \models \ddiamond{(\ptest{S(g)}\seq \aproj{\namedGame{\gamma}{g}}{S}{})}{\phi} \lor \ddiamond{(\ptest{S(d)}\seq \aproj{\namedGame{\delta}{d}}{S}{})}{\phi}\) \(\quad\) by axiom \irref{choiced}\\
      iff \(\sigma \models \ddiamond{\aproj{\namedGame{\gamma}{g}}{S}{}}{\phi} \lor (S(d)\land\ddiamond{\aproj{\namedGame{\delta}{d}}{S}{}}{\phi})\) \(\quad\) \(\sigma\models S(g)\), applying axioms \irref{testd}, \irref{composed}\\
      iff \(\sigma \models \ddiamond{\aproj{\namedGame{\gamma}{g}}{S}{}}{\phi}\) \(\quad\) \(\sigma \not\models S(g)\) by assumption\\
      which is exactly the condition for which the strategy set is nonempty.
      \item If \(\sigma \models S(d)\) and \(\sigma \not\models S(g)\), the argument is symmetric but with the roles of $\gamma$ and $\delta$ swapped.
      \item If \(\sigma \models S(g)\) and \(\sigma \models S(d)\), then the strategy set is \(\{\{\actleft \caret t\} \with t\in\astrategize{\namedGameSpace{\gamma}{g}}{S}( \sigma)\} \cup \{\{\actright \caret t\} \with t\in\astrategize{\namedGameSpace{\delta}{d}}{S}( \sigma)\}\).
      This set is nonempty when either \(\astrategize{\namedGameSpace{\gamma}{g}}{S}(\sigma)\) or \(\astrategize{\namedGameSpace{\delta}{d}}{S}(\sigma)\) is nonempty.
      By the inductive hypothesis, this is when \(\sigma \models \ddiamond{\aproj{\namedGame{\gamma}{g}}{S}{}}{\phi}\) or \(\sigma \models \ddiamond{\aproj{\namedGame{\delta}{d}}{S}{}}{\phi}\).

      \(\sigma \models \ddiamond{\aproj{\namedGame{\alpha}{a}}{S}{\phi}}{\phi}\) iff\\
      \(\sigma \models \ddiamond{((\ptest{S(g)}\seq\namedGame{\gamma}{g}) \cup (\ptest{S(d)}\seq \namedGame{\delta}{d}))}{\phi}\) \(\quad\) because \(\aproj{\namedGame{\alpha}{a}}{S}{}\) is \((\ptest{S(g)}\seq\namedGame{\gamma}{g}) \cup (\ptest{S(d)}\seq \namedGame{\delta}{d})\)\\
      iff \(\sigma \models \ddiamond{(\ptest{S(g)}\seq \aproj{\namedGame{\gamma}{g}}{S}{})}{\phi} \lor \ddiamond{(\ptest{S(d)}\seq \aproj{\namedGame{\delta}{d}}{S}{})}{\phi}\) \(\quad\) by axiom \irref{choiced}\\
      iff \(\sigma \models \ddiamond{\aproj{\namedGame{\gamma}{g}}{S}{}}{\phi} \lor \ddiamond{\aproj{\namedGame{\delta}{d}}{S}{}}{\phi}\) \(\quad\) \(\sigma\models S(g)\), \(\sigma\models S(d)\)
      which is exactly the condition for which the strategy set is nonempty.
    \end{enumerate}
    \item If $\namedGame{\alpha}{a}$ is an Angelic loop $\namedGame{(\namedGame{\gamma}{g})^\ast}{a}$.
    By the inductive hypothesis, \(\astrategize{\namedGameSpace{\gamma}{g}}{S}(\sigma)\) is nonempty when \(\sigma \models \ddiamond{\aproj{\namedGame{\gamma}{g}}{\modEnd{S}{S(a)}}{}}{S(a)}\).

    \(\sigma \models \ddiamond{\aproj{\namedGame{\alpha}{a}}{S}{\phi}}{\phi}\) iff\\
    \(\sigma \models \ddiamond{\aproj{\namedGame{\alpha}{a}}{S}{}}{S(\finalNode)}\) \(\quad\) \(\phi\) is a compatible winning condition\\

    We show that if there is a strategy in the subvalue map strategy set then there is a strategy to win \(\ddiamond{(\ptest{S(g)} \seq \aproj{\namedGame{\gamma}{g}}{\modEnd{S}{S(a)}}{})^\ast}{(S(\finalNode))}\) from \(\sigma\) and vice versa by induction over iterations.

    First we build up a translation function \(\projToStratI{Z}\) inductively builds up translations for sequences (elements) in strategy set \(Z\).
    It maps elements of a strategy for \(\ddiamond{(\ptest{S(g)} \seq \aproj{\namedGame{\gamma}{g}}{\modEnd{S}{S(a)}}{})^\ast}{(S(\finalNode))}\) to elements of a strategy in the subvalue map strategy set.
    For example, \(\projToStratI{\{\actgo\}}(\actgo) = \actgo\). 
    \(\projToStratI{Z}\) is overloaded to also apply to strategies by applying to every element of the strategy.

    \begin{align*}
      \projToStratI{f(Z)} &\define 
      \{\actstop \caret (\ptest{S(\finalNode)}) \} \mapsto \{\actstop\}; \qquad \{\actgo \} \mapsto \{\actgo\}; \\ 
      & \textrm{for \(t \then \actgo \in \leaf(Z)\)},
      (t \caret \actgo \caret (\ptest{S(g)}) \caret u_t) \mapsto 
      (\projToStratI{Z}(t \caret \actgo) \caret \projToStrat{\namedGameSpace{\gamma}{g}}(u_t)) \\
      & \quad (t \caret \actgo \caret (\ptest{S(g)}) \caret u_t \caret \actgo) \mapsto 
      (\projToStratI{Z}(t \caret \actgo) \caret \projToStrat{\namedGameSpace{\gamma}{g}}(u_t) \caret \actgo) \\
      & \quad (t \caret \actgo \caret (\ptest{S(g)}) \caret u_t \caret \actstop \caret (\ptest{S(\finalNode)})) \mapsto 
      (\projToStratI{Z}(t \caret \actgo) \caret \projToStrat{\namedGameSpace{\gamma}{g}}(u_t) \caret \actstop ) \\
      & \quad \textrm{where \(u_t\) is a winning strategy for game } \ddiamond{\aproj{\namedGame{\gamma}{g}}{\modEnd{S}{S(a)}}{}}{S(a)} \textrm{ and}\\
      & \quad \textrm{corresponding translation function \(\projToStrat{\namedGameSpace{\gamma}{g}}(u_t)\) for strategies of game} \\
      & \quad \textrm{\(\ddiamond{\aproj{\namedGame{\gamma}{g}}{\modEnd{S}{S(a)}}{}}{S(a)}\) to strategies induced by subvalue map} \\
      & \quad \textrm{\modEnd{S}{S(a)}.}
    \end{align*}
    Above, where for readability, the map implicitly respects prefix closure over \(caret\), i.e., if \(a\caret b \mapsto c \caret d\) then \(a \mapsto c\) and \(a \then b \mapsto c \then d\).
    The exception is for the extra \(\ptest{S(g)}\) and \(\ptest{S(\finalNode)}\) actions which have no corresponding mapped action.
    The mapping function \(\projToStrat{\namedGameSpace{\gamma}{g}}(u_t)\) can be recursively defined to map a strategy winning the existential projection of a game onto a subvalue map to a strategy from the strategy set generated by that subvalue map where both strategies result in the same final state, though we do not show the full construction here.

    Now we argue that the translated functions are amongst the subvalue map's induced strategies.
    \begin{itemize}
      \item As base case, if \(\{\actstop, \actstop \then (\ptest{S(\finalNode)})\}\) wins the projection game from initial state \(\sigma\) then \(\{\actstop\}\) is in \(\astrategize{\namedGameSpace{\alpha}{a}}{S}(\sigma)\).
      Otherwise if \(\actgo\) is the first step in the projection game's strategy, then \(\sigma \models S(g)\) of the Angelic test that follows, \(\{\actgo\}\) can start strategies in \(\astrategize{\namedGameSpace{\alpha}{a}}{S}(\sigma)\).
      \item For the recursive case, suppose that \(Z \cup (\bigcup_{ t \then \actgo \in \leaf(Z)} t \caret \actgo \caret (\ptest{S(g)}) \caret u_t  \caret \actstop \caret (\ptest{S(\finalNode)}))\) is a strategy that wins for the projection game.
      
      Then \((\projToStratI{Z}(t \caret \actgo) \caret \projToStrat{\namedGameSpace{\gamma}{g}}(u_t) \caret \actstop )\) should be in \(\astrategize{\namedGameSpace{\alpha}{a}}{S}(\sigma)\) for the original game.
      \(\projToStratI{Z}(t \caret \actgo)\) can begin a strategy in \(\astrategize{\namedGameSpace{\alpha}{a}}{S}(\sigma)\) because of the inductive hypothesis of the inner induction on iterations.
      \(\projToStrat{\namedGameSpace{\gamma}{g}}(u_t)\) is some strategy for game \(\ddiamond{\aproj{\namedGame{\gamma}{g}}{\modEnd{S}{S(a)}}{}}{S(a)}\) which must exist because of the inductive hypothesis from the outer proof.
      \(\actstop\) appears in \(\apolicy{\alpha}{a}(S)(\runstrategy{\sigma}{t \caret \actgo \caret \projToStrat{\namedGameSpace{\gamma}{g}}(u_t) \caret \actstop})\), appearing as a final action during strategy generation because in the original projection game, because in the existential projection it was in a state where the test \((\ptest{S(\finalNode)})\) is satisfied.

      Otherwise, if \((Z \cup \bigcup_{ t\then \actgo \in \leaf(Z)} t \caret \actgo \caret (\ptest{S(g)}) \caret u_t)\) or \((Z \cup \bigcup_{ t\then \actgo \in \leaf(Z)} t \caret \actgo \caret (\ptest{S(g)}) \caret u_t)\) appear in a winning strategy of the projection game and do not yet lose, then their translations can generate strategies in \(\astrategize{\namedGameSpace{\alpha}{a}}{S}(\sigma)\) as well because of a similar argument.
    \end{itemize}
    This establishes that if there is a way to win the projection game, then the subvalue map strategy set is nonempty.

    For the other direction, a similar translation function can be written by inverting the construction above. A similar argument can be made to show that the translated strategies are winning for the projection game.
    \item If $\namedGame{\alpha}{a}$ is a Demonic loop $\namedGame{(\namedGame{\gamma}{g})^\ast}{a}$, we use an approach similar to the Angelic loop case.
    By the inductive hypothesis, \(\astrategize{\namedGameSpace{\gamma}{g}}{S}(\sigma)\) is nonempty when \(\sigma \models \ddiamond{\aproj{\namedGame{\gamma}{g}}{S}{}}{S(a)}\).

    \(\sigma \models \ddiamond{\aproj{\namedGame{\alpha}{a}}{S}{\phi}}{\phi}\) iff\\
    \(\sigma \models \ddiamond{\aproj{\namedGame{\alpha}{a}}{S}{}}{S(\finalNode)}\) \(\quad\) \(\phi\) is a compatible winning condition\\

    We show that if there is a strategy in the set then there is a strategy to win \(\sigma \ddiamond{(\ptest{S(g)} \seq \aproj{\namedGame{\gamma}{g}}{\modEnd{S}{S(a)}}{})^\ast}{(S(\finalNode))}\) and vice versa by induction over iterations.

    We again build up a translation function \(\projToStratI{Z}\).It inductively builds up translations for sequences (elements) in strategy set \(Z\).
    It is similar to the previous case but with dual actions and without the added Angelic tests.
    
    \begin{align*}
      \projToStratI{f(Z)} &\define 
      \{\pdual{\actstop} \} \mapsto \{\pdual{\actstop} \}; \qquad \{ \pdual{\actgo} \} \mapsto \{ \pdual{\actgo} \}; \\ 
      & \textrm{for \(t \then \pdual{\actgo} \in \leaf(Z)\)},
      (t \caret \pdual{\actgo} \caret u_t) \mapsto 
      (\projToStratI{Z}(t \caret \pdual{\actgo}) \caret \projToStrat{\namedGameSpace{\gamma}{g}}(u_t)) \\
      & \quad (t \caret \pdual{\actgo} \caret u_t \caret \pdual{\actgo}) \mapsto 
      (\projToStratI{Z}(t \caret \pdual{\actgo}) \caret \projToStrat{\namedGameSpace{\gamma}{g}}(u_t) \caret (\ptest{S(a)}) \caret \pdual{\actgo}) \\
      & \quad (t \caret \pdual{\actgo} \caret u_t \caret \pdual{\actstop} \caret (\ptest{S(\finalNode)})) \mapsto 
      (\projToStratI{Z}(t \caret \pdual{\actgo}) \caret \projToStrat{\namedGameSpace{\gamma}{g}}(u_t) \caret \pdual{\actstop} ) \\
      & \qquad \textrm{where \(u_t\) is a winning strategy for game } \ddiamond{\aproj{\namedGame{\gamma}{g}}{\modEnd{S}{S(a)}}{}}{S(a)} \textrm{ and}\\
      & \qquad \textrm{corresponding \(\projToStrat{\namedGameSpace{\gamma}{g}}(u_t)\) for game \(\ddiamond{\gamma}{S(a)}\) must exist per inductive hypothesis.} \\
    \end{align*}
    Above, where for readability, the map implicitly respects prefix closure over \(caret\).
    A full strategy is translated by translating each sequence.

    The argument that the translated function is in the strategy set is also similar to the previous case.
    As with the previous case, an inverse translation function shows the opposite direction.
    \item If $\namedGame{\alpha}{a}$ is a Demonic choice $\namedGame{(\namedGame{\gamma}{g} \cap \namedGame{\delta}{d})}{a}$, then \(\astrategize{\namedGameSpace{\alpha}{a}}{S}(\sigma)\) is \(\{\pdual{\actleft} \caret t_l \cup \pdual{\actright} \caret t_r \with t_l\in\astrategize{\namedGameSpace{\gamma}{g}}{S}(g, \sigma) \textrm{ and } t_r\in\astrategize{\namedGameSpace{\delta}{d}}{S}(d, \sigma)\}\), which is nonempty when both \(\astrategize{\namedGameSpace{\gamma}{g}}{S}(\sigma)\) and \(\astrategize{\namedGameSpace{\delta}{d}}{S}(\sigma)\) are nonempty.
    By the inductive hypothesis, this is when \(\sigma \models \ddiamond{\aproj{\namedGame{\gamma}{g}}{S}{}}{\phi}\) and \(\sigma \models \ddiamond{\aproj{\namedGame{\delta}{d}}{S}{}}{\phi}\).

    \(\sigma \models \ddiamond{\aproj{\namedGame{\alpha}{a}}{S}{\phi}}{\phi}\) iff\\
    \(\sigma \models \ddiamond{(\aproj{\namedGame{\gamma}{g}}{S}{}) \cap (\aproj{\namedGame{\delta}{d}}{S}{})}{\phi}\) \(\quad\) by definition of \(\aproj{\namedGame{\alpha}{a}}{S}{}\)\\
    iff \(\sigma \models \ddiamond{\aproj{\namedGame{\gamma}{g}}{S}{}}{\phi} \land \ddiamond{\aproj{\namedGame{\delta}{d}}{S}{}}{\phi}\) \(\quad\) by axiom \irref{dchoiced}
    which is exactly the condition for which the strategy set is nonempty.
    \item If $\namedGame{\alpha}{a}$ is a sequence $\namedGame{(\namedGame{\gamma}{g};\namedGame{\delta}{d})}{a}$, then \(\astrategize{\namedGameSpace{\alpha}{a}}{S}(\sigma)\) is \(\{ t \cup \bigcup_{v \in \leaf(t)} (v\caret u) \with t \in \astrategize{\namedGameSpace{\gamma}{g}}{\modEnd{S}{S(d)}}(g, \sigma) \textrm{ and } u \in \astrategize{\namedGameSpace{\delta}{d}}{S}(d)(\runstrategy{v}{\sigma}) \}\), which is nonempty when both \(\astrategize{\namedGameSpace{\gamma}{g}}{S}(\sigma)\) is nonempty and for some choice of \(v\), \(\astrategize{\namedGameSpace{\delta}{d}}{S}(\runstrategy{v}{\sigma})\) is nonempty.
    By the inductive hypothesis, this is when \(\sigma \models \ddiamond{\aproj{\namedGame{\gamma}{g}}{\modEnd{S}{S(d)}}{}}{S(d)}\) and for some \(v \in \astrategize{\namedGameSpace{\gamma}{g}}{S}(\sigma)\), \(\runstrategy{v}{\sigma} \models \ddiamond{\aproj{\namedGame{\delta}{d}}{S}{}}{\phi}\).

    \(\sigma \models \ddiamond{\aproj{\namedGame{\alpha}{a}}{S}{\phi}}{\phi}\) iff\\
    \(\sigma \models \ddiamond{(\aproj{\namedGame{\gamma}{g}}{\modEnd{S}{S(d)}}{}) \seq (\aproj{\namedGame{\delta}{d}}{S}{})}{\phi}\) \(\quad\) by definition of \(\aproj{\namedGame{\alpha}{a}}{S}{}\)\\
    \(\sigma \models \ddiamond{(\aproj{\namedGame{\gamma}{g}}{\modEnd{S}{S(d)}}{})}{\ddiamond{(\aproj{\namedGame{\delta}{d}}{S}{})}{\phi}}\) \(\quad\) by axiom \irref{composed}\\
    which is matches the condition for which the strategy set is nonempty.
  \end{itemize}
\end{proofE}

\rref{def:subvalue-projection} defines existential projection by inserting Angelic tests to guard \emph{every} Angelic decision that would make Angel lose if
her choice would result in playing into the subsequent subgame in a state outside that subgame's winning subregion.
For example, in \rref{fig:overview}, the subvalue map says that $\preci{2}$ is the winning subregion of the subgame that is played after subgame $\namedGame{a:=\ast}{3}$.
The projection of the subvalue map onto subgame $3$ is $a:=\ast \seq \AdvExAssign{?\preci{2}}$, thus \emph{locally} ensuring that Angel makes a correct assignment to $a$ that will let Angel win the rest of the game.
When Angel has no decisions to take, e.g., $\namedGame{\gamma}{g} \cap \namedGame{\delta}{d}$ where \emph{Demon} chooses between $\namedGame{\gamma}{g}$ and $\namedGame{\delta}{d}$, projection does not interfere, only recursively adding guards for any Angelic decisions within subgames $\gamma$ and $\delta$.
During recursive calls, we sometimes have to update the \(\finalNode\) subvalue to reflect the appropriate terminal condition for the subgame under consideration.
Notation \(\modEnd{S}{Q}\) refers to \(S\) with the subvalue of \(\finalNode\) updated to \(Q\).

\begin{definition}[Existential projection]
  \label{def:subvalue-projection}
  The \emph{existential projection} of \dGL game $\namedGame{\alpha}{a}$ onto Angelic subvalue map $S$, written \(\aproj{\namedGame{\alpha}{a}}{S}{}\), is generated recursively from the structure of $\namedGame{\alpha}{a}$ as follows. If $\namedGame{\alpha}{a}$ has structure:
  \begin{align*}
    \namedGame{x:=\ast}{a}
    &\textrm{ then }
    x:= \ast \seq \AdvExAssign{\ptest{S(\finalNode)}}.
    \quad
    \namedGame{\{x'=f(x)\ \&\ Q\}}{a}
    \textrm{ then }
    \{x'=f(x)\ \&\ Q\} \seq \AdvExAssign{\ptest{S(\finalNode)}} \\
    \namedGame{(\namedGame{\gamma}{g} \cup \namedGame{\delta}{d})}{a} 
    &\textrm{ then } ( \AdvExChoice{?S(g)} \seq \aproj{\namedGame{\gamma}{g}}{S}{}) {\cup} (\AdvExChoice{?S(d)} \seq \aproj{\namedGame{\delta}{d}}{S}{}) \\
    \namedGame{(\namedGame{\gamma}{g})^\ast}{a} 
    &\textrm{ then } \left(\AdvExLoop{?S(g)} \seq \left(\aproj{\namedGame{\gamma}{g}}{(\modEnd{S}{S(a)})}{}\right)\right)^{\ast}; \AdvExLoop{\ptest{S(\finalNode)}} \\
    \namedGame{(\namedGame{\gamma}{g};\namedGame{\delta}{d})}{a} 
    &\textrm{ then } \left(\aproj{\namedGame{\gamma}{g}}{(\modEnd{S}{S(d)})}{}\right) \seq \left(\aproj{\namedGame{\delta}{d}}{S}{}\right) \\
    \namedGame{(\namedGame{\gamma}{g} \cap \namedGame{\delta}{d})}{a} 
    &\textrm{ then } \aproj{\namedGame{\gamma}{g}}{S}{} \cap \aproj{\namedGame{\delta}{d}}{S}{}
    \quad
    \namedGame{(\namedGame{\gamma}{g})^\times}{a}
    \textrm{ then } \big( \aproj{\namedGame{\gamma}{g}}{(\modEnd{S}{S(a)})}{} \big)^\times\\
    \textrm{atomic }&\textrm{and not controlled by Angel, i.e., }\\
    &\alpha \in \{x:=e, x:=\otimes, \ptest{Q}, !Q, \{x'=f(x)\ \&\ Q\}^d\}, \textrm{ then } \alpha
  \end{align*}
  where game labels are preserved and fresh labels are used for newly introduced subgames.
  The existential projection of a Demonic subvalue map is symmetric (\rref{app:subvalue-projection}), with the addition of Demonic assertions after subgames controlled by Demon.
\end{definition}

\begin{remark}
  $\aproj{\namedGame{\alpha}{a}}{S}{}$ is an \emph{Angelic refinement} of $\namedGame{\alpha}{a}$, i.e., for any winning condition $\psi$, $\models \ddiamond{\aproj{\namedGame{\alpha}{a}}{S}{}}{\psi} \limply \ddiamond{\alpha}{\psi}$ (\rref{lem:subvalue-projection-refinement}).
  Intuitively, if Angel has a winning strategy within $S$, she can win with the same strategy in the unconstrained game.
\end{remark}

\subsection{Universal Subvalue Map Projection}
\label{sec:universal-projection}
The \emph{universal projection}, written \(\asubst{\namedGame{\alpha}{a}}{S}{}\), explores a subvalue map maximally flexibly by making Demon play Angel's policy in Angelic subvalue map \(S\) to look for loopholes.
It changes the game to put \emph{Demon} on rails forcing him to play within Angel's $S$, as he
looks for \emph{any} strategy in the policy that falsifies Angel's winning condition.

\begin{lemmaE}[Universal projection correspondence][restate, text proof={}]
  \label{lem:forall-projection}
  For any Angelic subvalue map \(S\) for game \(\namedGame{\alpha}{a}\), in initial state \(\sigma \models \ddiamond{\alpha}{a}\), for compatible winning condition \(\phi\), Angel has a winning strategy for the game \(\asubst{\namedGame{\alpha}{a}}{S}{}\) if and only if all ways to pursue the policy of \(S\) for game \(\namedGame{\alpha}{a}\) that complete the game end in Angel's winning region.
  That is, \(\sigma \models \ddiamond{\asubst{\namedGame{\alpha}{a}}{S}{}}{\phi}\) iff
  all elements of set \(\astrategize{\namedGameSpace{\alpha}{a}}{S}(\sigma)\) are Angel winning strategies.
\end{lemmaE}
\begin{proofE}
  We show the proof by induction over the structure of \(\namedGame{\alpha}{a}\).

  \begin{itemize}
    \item If \(\namedGame{\alpha}{a}\) is atomic, i.e., \(\alpha \in \{x:=e, x:=\otimes, \ptest{Q}, !Q, \{x'=f(x)\ \&\ Q\}^d\}\), then \(\asubst{\namedGame{\alpha}{a}}{S}{}\) is \(\alpha\).
    \(\sigma \models \ddiamond{\asubst{\namedGame{\alpha}{a}}{S}{}}{\phi}\) iff \(\sigma \models \ddiamond{\alpha}{\phi}\) by definition of \(\asubst{\namedGame{\alpha}{a}}{S}{}\).
    This is true because \(\sigma\models \ddiamond{\alpha}{\phi}\).
    
    The strategy set consists of a single strategy with no Angelic decisions, that is just the entire game tree.
    This strategy is a winning strategy because we know there is one winning strategy because by assumption \(\sigma\models \ddiamond{\alpha}{\phi}\), and this strategy is the only candidate.
    \item If $\namedGame{\alpha}{a}$ is a non-deterministic Angel assignment $\namedGame{x:=\ast}{a}$, then \(\astrategize{\namedGameSpace{\alpha}{a}}{S}(\sigma)\) has the strategy set \(\{\{(x:=e) \mid \sigma(x \mapsto e) \models S(\finalNode)\}\}\), where all strategies are winning strategies when \(\sigma(x \mapsto e) \models \phi \mid \sigma(x \mapsto e) \models S(\finalNode)\), which is always, since \(\phi\) is a compatible winning condition so \(S(\finalNode)\limply \phi\).
    
    \(\sigma \models \ddiamond{\asubst{\namedGame{\alpha}{a}}{S}{\phi}}{\phi}\) iff\\
    \(\sigma \models \ddiamond{x:= \AdvExAssign{\otimes} \seq \AdvExAssign{!S(\finalNode)}}{\phi}\) \(\quad\) by definition of \(\asubst{\namedGame{\alpha}{a}}{S}{}\)\\
    iff \(\sigma \models \ddiamond{x:= \AdvExAssign{\otimes}}{(S(\finalNode) \limply \phi)}\) \(\quad\) by axioms \irref{composed}, \irref{testb}\\
    iff \(\sigma \models \ddiamond{x:= \AdvExAssign{\otimes}}{\top}\) \(\quad\) \(S(\finalNode)\limply \phi\) because \(\phi\) is a compatible winning condition\\
    But this always holds by the semantics of \dGL.
    \item If $\namedGame{\alpha}{a}$ is an Angelic ODE $\namedGame{\{x'=f(x)\ \&\ Q\}}{a}$, then
    \(\astrategize{\namedGameSpace{\alpha}{a}}{S}(\sigma)\) has the strategy set
    \(\{
      \{(x' = f(x) \, \&\, Q \,@\, t)\} \with t\geq 0, \forall s \in [0, t], \varphi(s) \models S(b) \text{ and }
      \varphi(t) \models S(\successor{b}{\namedGame{\alpha}{a}})
      \text{ where } \varphi: [0, t] \to \mathcal{S} \text{ is differentiable, }
      \varphi(0) = \sigma, \text{ and } 
      \forall s \in [0, t],\, \varphi(s) \models x' = f(x) \land Q
    \}\).
    All strategies are always winning strategies by a similar argument to for the previous case, since \(S(\finalNode)\limply \phi\).

    \(\sigma \models \ddiamond{\asubst{\namedGame{\alpha}{a}}{S}{\phi}}{\phi}\) iff\\
    \(\sigma \models \ddiamond{\{x'=f(x)\ \&\ Q\}^{\AdvExAssign{d}} \seq \AdvExAssign{!S(\finalNode)}}{\phi}\) \(\quad\) by definition of \(\asubst{\namedGame{\alpha}{a}}{S}{}\)\\
    iff \(\sigma \models \ddiamond{\{x'=f(x)\ \&\ Q\}^{\AdvExAssign{d}}}{(S(\finalNode) \limply \phi)}\) \(\quad\) by axioms \irref{composed}, \irref{testb}\\
    iff \(\sigma \models \ddiamond{\{x'=f(x)\ \&\ Q\}^{\AdvExAssign{d}}}{\top}\) \(\quad\) \(S(\finalNode)\limply \phi\)\\
    But this always holds by the semantics of \dGL.
    \item If $\namedGame{\alpha}{a}$ is an Angelic choice $\namedGame{(\namedGame{\gamma}{g} \cup \namedGame{\delta}{d})}{a}$, then
    
    \(\ddiamond{\asubst{\namedGame{\alpha}{a}}{S}{\phi}}{\phi}\) iff\\
    \(\ddiamond{(\AdvExChoice{!S(g)} \seq \asubst{\namedGame{\gamma}{g}}{S}{}) \AdvExChoice{\cap} (\AdvExChoice{!S(d)} \seq \asubst{\namedGame{\delta}{d}}{S}{})}{\phi}\) \(\quad\) by definition of \(\asubst{\namedGame{\alpha}{a}}{S}{}\)\\
    iff \(\sigma \models \ddiamond{\AdvExChoice{!S(g)} \seq \asubst{\namedGame{\gamma}{g}}{S}{}}{\phi} \land \ddiamond{\AdvExChoice{!S(d)} \seq \asubst{\namedGame{\delta}{d}}{S}{}}{\phi}\) \(\quad\) by axiom \irref{dchoiced}\\
    iff \(\sigma \models (S(g) \limply \ddiamond{\asubst{\namedGame{\gamma}{g}}{S}{}}{\phi}) \land (S(d) \limply \ddiamond{\asubst{\namedGame{\delta}{d}}{S}{}}{\phi})\) \(\quad\) by axiom \irref{testb}, \irref{composed}\\

    There are four cases:
    \begin{itemize}
      \item \(\sigma \not\models S(g)\) and \(\sigma \not\models S(d)\): then the Angel wins the projection game because of the vacuous premises of the implications, and the strategy set has only winning strategies vacuously since it has no strategies.
      \item \(\sigma \models S(g)\) and \(\sigma \not\models S(d)\): then the Angel wins the projection game iff
      
      \(\sigma \models (S(g) \limply \ddiamond{\asubst{\namedGame{\gamma}{g}}{S}{}}{\phi})\) \(\quad\) second conjunct true as \(\sigma \not\models S(d)\) \\
      iff \(\sigma \models \ddiamond{\asubst{\namedGame{\gamma}{g}}{S}{}}{\phi}\) \(\quad\) as \(\sigma \models S(g)\)\\

      By the inductive hypothesis, this is exactly when the strategy set \(\astrategize{\namedGameSpace{\gamma}{g}}{S}(\sigma)\) has only winning strategies for Angel in game \(\ddiamond{\gamma}{\phi}\).
      
      The strategy set for the overall game is \(\{\{\actleft \caret t\} \with t\in\astrategize{\namedGameSpace{\gamma}{g}}{S}(g, \sigma)\}\) which consists of winning strategies exactly when \(\astrategize{\namedGameSpace{\gamma}{g}}{S}(\sigma)\) has only winning strategies for Angel in game \(\ddiamond{\gamma}{\phi}\).
      \item \(\sigma \not\models S(g)\) and \(\sigma \models S(d)\), the argument is symmetric to the previous case, with the roles of \(g\) and \(d\) swapped.
      \item \(\sigma \models S(g)\) and \(\sigma \models S(d)\): then the Angel wins the projection game iff
      
      \(\sigma \models \ddiamond{\asubst{\namedGame{\gamma}{g}}{S}{}}{\phi} \land \ddiamond{\asubst{\namedGame{\delta}{d}}{S}{}}{\phi}\) \(\quad\)
      as \(\sigma \models S(g)\) and \(\sigma \models S(d)\)\\
      By the inductive hypothesis, this is exactly when both strategy sets \(\astrategize{\namedGameSpace{\gamma}{g}}{S}(\sigma)\) and \(\astrategize{\namedGameSpace{\delta}{d}}{S}(\sigma)\) have only winning strategies for Angel in games \(\ddiamond{\gamma}{\phi}\) and \(\ddiamond{\delta}{\phi}\) respectively.  

      The strategy set for the overall game is \(\{\{\actleft \caret t\} \with t\in\astrategize{\namedGameSpace{\gamma}{g}}{S}(g, \sigma)\}\ \cup \{\{\actright \caret t\} \with t\in\astrategize{\namedGameSpace{\delta}{d, \sigma}}{S}(d)\}\).

      This consists of winning strategies exactly when both \(\astrategize{\namedGameSpace{\gamma}{g}}{S}(\sigma)\) and \(\astrategize{\namedGameSpace{\delta}{d}}{S}(\sigma)\) have only winning strategies for Angel in games \(\ddiamond{\gamma}{\phi}\) and \(\ddiamond{\delta}{\phi}\) respectively.
    \end{itemize}
    \item If $\namedGame{\alpha}{a}$ is an Angelic loop $\namedGame{(\namedGame{\gamma}{g})^\ast}{a}$:
    
    \(\sigma \models \ddiamond{\asubst{\namedGame{\alpha}{a}}{S}{\phi}}{\phi}\) iff\\
    \(\sigma \models \ddiamond{\left(\AdvExLoop{!S(g)} \seq (\asubst{\namedGame{\gamma}{g}}{(\modEnd{S}{S(a)})}{}) \right)^{\AdvExLoop{\times}} \seq \AdvExLoop{!S(\finalNode)}}{\phi}\) \(\quad\) by definition of \(\asubst{\namedGame{\alpha}{a}}{S}{}\)\\
    iff \(\sigma \models \ddiamond{\left(\AdvExLoop{!S(g)} \seq (\asubst{\namedGame{\gamma}{g}}{(\modEnd{S}{S(a)})}{}) \right)^{\AdvExLoop{\times}}}{S(\finalNode)\limply\phi}\) \(\quad\) by axioms \irref{composed}, \irref{testb}\\
    iff \(\sigma \models \ddiamond{\left(\AdvExLoop{!S(g)} \seq (\asubst{\namedGame{\gamma}{g}}{(\modEnd{S}{S(a)})}{}) \right)^{\AdvExLoop{\times}}}{\top}\) \(\quad\) \(S(\finalNode)\limply\phi\)\\
    But this always holds applying the loop rule with invariant \(\ddiamond{\alpha}{\phi}\) and using the inductive hypothesis on the subgame \(\ddiamond{\gamma}{\modEnd{S}{S(a)}}\) and \rref{thm:subvalue-map-winning-region}.

    The strategy set also consists of only winning strategies.
    Consider any leaf \(t\) in any strategy tree in the set.
    If \(t\) ends with \(\actstop\) then by construction (\rref{def:strategy-set-generation}), \(\runstrategy{t}{\sigma} \models S(\finalNode)\) and \(t\) is a winning play for Angel.
    Otherwise, if \(t\) does not end in \(\actstop\), then it must have ended while playing loop body \(\gamma\).
    Let \(t = u \then v\) where \(v\) plays the final (incomplete) iteration of the loop body.
    \(v\) ends with Demon getting stuck because it is a winning strategy for \(\gamma\) by the inductive hypothesis in combination with \rref{thm:subvalue-map-winning-region}.
    Thus, \(t\) is also a winning strategy for Angel in the larger game.

    \item If $\namedGame{\alpha}{a}$ is a Demonic loop $\namedGame{(\namedGame{\gamma}{g})^\times}{a}$, then
    
    \(\sigma \models \ddiamond{\asubst{\namedGame{\alpha}{a}}{S}{\phi}}{\phi}\) iff\\
    \( \sigma\models\ddiamond{(\asubst{\namedGame{\gamma}{g}}{\modEnd{S}{S(a)}}{})^\times}{\phi} \)

    This always holds: we can apply the loop invariant rule with invariant \(\ddiamond{\alpha}{\phi}\) and use \rref{thm:subvalue-map-winning-region} on the subgame \(\ddiamond{\gamma}{\modEnd{S}{S(a)}}\).

    The argument that the strategy set has only winning strategies is similar to the previous case.
    Consider any leaf in any strategy tree in the set.
    If if ends with \(\actstop\), it is a winning play for Angel because \(\runstrategy{t}{\sigma} \models S(\finalNode)\).
    Otherwise, if it does not end in \(\actstop\), then it must have ended while playing loop body \(\gamma\).
    It must end with Demon getting stuck because by the inductive hypothesis in combination with \rref{thm:subvalue-map-winning-region}, and so must be a winning strategy for Angel in the larger game.
    \item If $\namedGame{\alpha}{a}$ is a Demonic choice $\namedGame{(\namedGame{\gamma}{g} \cap \namedGame{\delta}{d})}{a}$, then \(\astrategize{\namedGameSpace{\alpha}{a}}{S}(\sigma)\) has the strategy set \(\{\pdual{\actleft} \caret t_l \cup \pdual{\actright} \caret t_r \with t_l\in\astrategize{\namedGameSpace{\gamma}{g}}{S}(g, \sigma) \textrm{ and } t_r\in\astrategize{\namedGameSpace{\delta}{d}}{S}(d, \sigma)\}\), which contains only winning strategies exactly when \(\astrategize{\namedGameSpace{\gamma}{g}}{S}(\sigma)\) and \(\astrategize{\namedGameSpace{\delta}{d}}{S}(\sigma)\) have only winning strategies for Angel in games \(\ddiamond{\gamma}{\phi}\) and \(\ddiamond{\delta}{\phi}\) respectively.
    
    \(\sigma \models \ddiamond{\asubst{\namedGame{\alpha}{a}}{S}{\phi}}{\phi}\) iff\\
    \(\sigma \models \ddiamond{(\asubst{\namedGame{\gamma}{g}}{S}{}) \cap (\asubst{\namedGame{\delta}{d}}{S}{})}{\phi}\) \(\quad\) by definition of \(\asubst{\namedGame{\alpha}{a}}{S}{}\)\\
    iff \(\sigma \models \ddiamond{\asubst{\namedGame{\gamma}{g}}{S}{}}{\phi} \land \ddiamond{\asubst{\namedGame{\delta}{d}}{S}{}}{\phi}\) \(\quad\) by axiom \irref{dchoiced}\\
    which by the inductive hypothesis applied for \(\namedGame{\gamma}{g}\) and \(\namedGame{\delta}{d}\) is exactly when the overall strategy set has only winning Angel strategies.
    \item If $\namedGame{\alpha}{a}$ is sequential composition $\namedGame{(\namedGame{\gamma}{g};\namedGame{\delta}{d})}{a}$, then \(\astrategize{\namedGameSpace{\alpha}{a}}{S}(\sigma)\) has the strategy set \(\{ t \cup \bigcup_{v \in \leaf(t)} (v\caret u) \with t \in \astrategize{\namedGameSpace{\gamma}{g}}{\modEnd{S}{S(d)}}(g, \sigma) \textrm{ and } u \in \astrategize{\namedGameSpace{\delta}{d}}{S}(d)(\runstrategy{v}{\sigma}) \}\).
    These are always winning strategies: \(\astrategize{\namedGameSpace{\gamma}{g}}{\modEnd{S}{S(d)}}(g, \sigma)\) and \(\astrategize{\namedGameSpace{\delta}{d}}{S}(d)(\runstrategy{v}{\sigma})\) have only winning strategies using the inductive hypothesis along with \rref{thm:subvalue-map-winning-region}.
    All plays that get stuck win for Angel, and all plays that run till the end result in a state \(\sigma'\in S(\finalNode)\), and so win for Angel because \(S(\finalNode)\limply \phi\).

    \(\sigma \models \ddiamond{\asubst{\namedGame{\alpha}{a}}{S}{\phi}}{\phi}\) always holds by \rref{thm:subvalue-map-winning-region}.
  \end{itemize}
\end{proofE}

\rref{def:subvalue-ext} defines universal projection
by replacing all Angelic games ($\cup, x:=\ast, \alpha^*,
\{x'=f(x) \& Q\}$) by their dual Demonic games and inserting Demonic tests to ensure Demon plays these decisions within Angel's control envelope.
Like in projection, control choices that were already Demon's remain unguarded.
\rref{eq:overview-projection} shows universal projection for the example subvalue map and game in \rref{fig:overview}.
\begin{definition}[Universal projection]
  \label{def:subvalue-ext}
  The \emph{universal projection} of Angelic subvalue map $S$ onto \dGL game $\namedGame{\alpha}{a}$, written $\asubst{\namedGame{\alpha}{a}}{S}{}$, is generated recursively per the structure of $\namedGame{\alpha}{a}$ as follows. If $\namedGame{\alpha}{a}$ has structure:
  \begin{align*}
    \namedGame{x:=\ast}{a},
    \textrm{ then }&
    x:= \AdvExAssign{\otimes} \seq \AdvExAssign{!S(\finalNode)}.
    \quad
    \namedGame{(\namedGame{\gamma}{g};\namedGame{\delta}{d})}{a}, 
    \textrm{ then } \left(\asubst{\namedGame{\gamma}{g}}{(\modEnd{S}{S(d)})}{}\right) \seq \left(\asubst{\namedGame{\delta}{d}}{S}{}\right).\\
    \namedGame{\{x'=f(x)\ \&\ Q\}}{a},
    &\textrm{ then }
    \{x'=f(x)\ \&\ Q\}^{\AdvExAssign{d}} \seq \AdvExAssign{!S(\finalNode)}.\\
    \namedGame{(\namedGame{\gamma}{g} \cup \namedGame{\delta}{d})}{a}, 
    \textrm{ then }&
    ( \AdvExChoice{!S(g)} \seq \asubst{\namedGame{\gamma}{g}}{S}{}) \AdvExChoice{\cap} (\AdvExChoice{!S(d)} \seq \asubst{\namedGame{\delta}{d}}{S}{}).\\
    \namedGame{(\namedGame{\gamma}{g})^\ast}{a}, 
    \textrm{ then }& \left(\AdvExLoop{!S(g)} \seq (\asubst{\namedGame{\gamma}{g}}{(\modEnd{S}{S(a)})}{}) \right)^{\AdvExLoop{\times}} \seq \AdvExLoop{!S(\finalNode)}.\\
    \namedGame{(\namedGame{\gamma}{g} \cap \namedGame{\delta}{d})}{a}, 
    \textrm{ then }& (\asubst{\namedGame{\gamma}{g}}{S}{}) \cap (\asubst{\namedGame{\delta}{d}}{S}{}). 
    \quad
    \namedGame{(\namedGame{\gamma}{g})^\times}{a}, 
    \textrm{ then } (\asubst{\namedGame{\gamma}{g}}{\modEnd{S}{S(a)}}{})^\times.\\
    \textrm{atomic and not}&\textrm{ controlled by Angel, i.e., }\\
    \alpha \in \{x:=e,& x:=\otimes, \ptest{Q}, !Q, \{x'=f(x)\ \&\ Q\}^d\}, \textrm{ then } \alpha.
  \end{align*}
  where game labels are preserved (including for Angel-controlled subgames transformed to their Demon-controlled equivalent) and fresh labels are used for newly introduced subgames.
  The universal projection of a Demonic subvalue map is symmetric (see \rref{app:subvalue-externalization-demon}).
\end{definition}

We can now show the defining property of subvalue maps.

\begin{theoremE}[Subvalue Map stays in Winning Region][restate, text proof={}]
\label{thm:subvalue-map-winning-region}
  Suppose \(S\) is an Angelic subvalue map for game \(\namedGame{\alpha}{a}\) compatible with winning condition \(\phi\), i.e., for every subgame \(b \in \nodes{\namedGame{\alpha}{a}}\), \(\models S(b) \limply \ddiamond{\fwd{b}{\namedGame{\alpha}{a}}}{S(\finalNode)}\) and \(\models S(\finalNode)\limply \phi\).
  Upon starting in any state \(\sigma\in \Sem{S(a)}\),
  and reaching subgame \(b\) by following any strategy induced by the subvalue map \(S\), there exists a winning strategy for Angel to win the remainder of the game.
  That is, \(S(a) \models \ddiamond{\prefix{b}{(\asubst{\namedGame{\alpha}{a}}{S}{})}}{\ddiamond{\fwd{b}{\namedGame{\alpha}{a}}}{\phi}}\).
\end{theoremE}
\begin{proofE}
  We perform induction on the structure of \(\namedGame{\alpha}{a}\).
  \begin{itemize}
    \item When \(\namedGame{\alpha}{a}\) is atomic and not controlled by Angel,
    i.e., \(\alpha \in \{x:=e, x:=\otimes, \ptest{Q}, !Q, \{x'=f(x)\ \&\ Q\}^d\}\), then
    The only subgame \(b\) in \(\nodes{\namedGame{\alpha}{a}}\) is \(\namedGame{\alpha}{a}\) itself.
    Thus the only possibility is \(b=a\).
    \(\prefix{a}{\namedGame{\alpha}{a}}\) is the empty game \skp, and \(\fwd{a}{\namedGame{\alpha}{a}}\) is the game \(\namedGame{\alpha}{a}\), and \(\asubst{\namedGame{\alpha}{a}}{S}{\phi}\) is just \(\namedGame{\alpha}{a}\).
    Thus, \(\ddiamond{\asubst{\prefix{a}{\namedGame{\alpha}{a}}}{S}{\phi}}{\ddiamond{\fwd{a}{\namedGame{\alpha}{a}}}{\phi}}\) is equivalent to \(\ddiamond{\skp}{\ddiamond{{\alpha}}{\phi}}\).
    We need to show that \(S(a) \models {\ddiamond{\namedGame{\alpha}{a}}{\phi}}\) which holds because by definition of an Angelic subvalue map, \(\models S(a) \limply \ddiamond{\alpha}{\phi}\).
    \item When \(\namedGame{\alpha}{a}\) is an Angelic nondeterministic assignment, i.e., \(\namedGame{x:=\ast}{a}\), then again the only subgame \(b\) in \(\nodes{\namedGame{\alpha}{a}}\) is \(\namedGame{\alpha}{a}\) itself.
    Thus the only possibility is \(b=a\).
    \(\prefix{a}{\namedGame{\alpha}{a}}\) is the empty game \skp, and \(\asubst{\namedGame{\alpha}{a}}{S}{\phi}\) is just \(\namedGame{\alpha}{a} \seq \ptest{S(\finalNode)}\), and \(\fwd{a}{\namedGame{\alpha}{a} \seq \ptest{S(\finalNode)}}\) is the game \(\namedGame{\alpha}{a} \seq \ptest{S(\finalNode)}\).
    Thus, \(\ddiamond{\asubst{\prefix{a}{\namedGame{\alpha}{a}}}{S}{\phi}}{\ddiamond{\fwd{a}{\namedGame{\alpha}{a}}}{\phi}}\) is equivalent to \(\ddiamond{\skp}{\ddiamond{{\alpha \seq \ptest{S(\finalNode)}}}{\phi}}\).
    We need to show that \(S(a) \models \ddiamond{\alpha}{(\ptest{S(\finalNode)} \land \phi)}\).
    Since \(\phi\) is a compatible winning condition, this is \(S(a) \models {\ddiamond{\alpha}{\ptest{S(\finalNode)}}}\) which holds by the definition of an inductive Angelic subvalue map.
    \item When \(\namedGame{\alpha}{a}\) is an Angelic ODE, i.e., \(\namedGame{\{x'=f(x)\ \&\ Q\}}{a}\), then the argument is similar to the previous case.
    As before the only possibility is \(b=a\), and we need to show that \(S(a) \models {\ddiamond{\alpha}{\ptest{S(\finalNode)} \land \phi}}\).
    Since \(\phi\) is a compatible winning condition, this is \(S(a) \models {\ddiamond{{\namedGame{\alpha}{a}}}{\ptest{S(\finalNode)}}}\)
    which holds by the definition of an inductive Angelic subvalue map.
    \item When \(\namedGame{\alpha}{a}\) is an Angelic choice, \(\namedGame{(\namedGame{\gamma}{g} \cup \namedGame{\delta}{d})}{a}\), then we analyze cases depending on the value of \(b\).
    \begin{enumerate}
      \item If \(b=a\), then \(\prefix{b}{\asubst{\prefix{a}{\namedGame{\alpha}{a}}}{S}{\phi}} = \skp\), and \(\fwd{b}{\namedGame{\alpha}{a}} = \namedGame{\alpha}{a}\).
            Thus, \(\ddiamond{\asubst{\prefix{b}{\namedGame{\alpha}{a}}}{S}{\phi}}{\ddiamond{\fwd{b}{\namedGame{\alpha}{a}}}{\phi}}\) is equivalent to \(\ddiamond{\namedGame{\alpha}{a}}{\phi}\).
            \(S(a) \models \ddiamond{\namedGame{\alpha}{a}}{\phi}\)
            This can be proved by applying the axioms \irref{implyr} along with \(S(a) \limply \ddiamond{\namedGame{\alpha}{a}}{\phi}\) from the definition of an Angelic subvalue map.
      \item Otherwise, if \(b \in \nodes{\namedGame{\gamma}{g}}\), then \(\prefix{b}{\asubst{\prefix{a}{\namedGame{\alpha}{a}}}{S}{\phi}} = {!S(g)} \seq \prefix{b}{\asubst{\namedGame{\gamma}{g}}{S}{}}\) and \(\fwd{b}{\namedGame{\alpha}{a}} = \fwd{b}{\namedGame{\gamma}{g}}\).
      We can show that \(S(a) \models \ddiamond{{!S(g)} \seq \prefix{b}{\asubst{\namedGame{\gamma}{g}}{S}{}}}{\ddiamond{\fwd{b}{\namedGame{\gamma}{g}}}{\phi}}\) by first applying axioms \irref{composed}, \irref{testb}, \irref{box}, \irref{duald}, \irref{implyr} to get the goal \(S(a), S(g) \models \ddiamond{\prefix{b}{\asubst{\namedGame{\gamma}{g}}{S}{}}}{\ddiamond{\fwd{b}{\namedGame{\gamma}{g}}}{\phi}}\)
      Notice that \(S\) is also an inductive Angelic subvalue map for \(\namedGame{\gamma}{g}\) and winning condition \(\phi\).
      Thus by the inductive hypothesis, \(S(g) \models \ddiamond{\asubst{\prefix{b}{\namedGame{\gamma}{g}}}{S}{\phi}}{\ddiamond{\fwd{b}{\namedGame{\gamma}{g}}}{\phi}}\) completing the proof.
      \item Otherwise, if \(b \in \nodes{\namedGame{\delta}{d}}\), then the proof proceeds similarly to the previous case, but with \(\namedGame{\delta}{d}\) replacing \(\namedGame{\gamma}{g}\).
    \end{enumerate}
    \item When \(\namedGame{\alpha}{a}\) is an Angelic loop, i.e., \(\namedGame{(\namedGame{\gamma}{g})^\ast}{a}\), then we analyze cases depending on the value of \(b\).
    \begin{enumerate}
      \item If \(b=a\), then \(\fwd{b}{\namedGame{\alpha}{a}} = \namedGame{\alpha}{a}\) and \(\prefix{b}{\asubst{\namedGame{\alpha}{a}}{S}{\phi}} = {(!{S(g)} \seq \asubst{\namedGame{\gamma}{g}}{\modEnd{S}{S(a)}}{S(a)})^\times} \seq !S(\finalNode)\).
            Thus, 
            \[\ddiamond{\prefix{b}{\asubst{{\namedGame{\alpha}{a}}}{S}{\phi}}}{\ddiamond{\fwd{b}{\namedGame{\alpha}{a}}}{\phi}}\] is by definition 
            \[\ddiamond{{({!S(g)} \seq \asubst{\namedGame{\gamma}{g}}{\modEnd{S}{S(a)}}{S(a)})^\times} \seq !S(\finalNode)}{\ddiamond{\namedGame{\alpha}{a}}{\phi}}.\]
            By the inductive hypothesis, \(S(g) \models \ddiamond{\asubst{\namedGame{\gamma}{g}}{\modEnd{S}{S(a)}}{\phi}}{S(a)}\).
            Then applying \irref{loop} with invariant \(S(a)\) we have
            \[S(a) \models \ddiamond{{({!S(g)} \seq \asubst{\namedGame{\gamma}{g}}{\modEnd{S}{S(a)}}{S(a)})^\times}}{S(a)}.\]

            By the definition of an Angelic subvalue map, \(\models S(a) \limply \ddiamond{\namedGame{\alpha}{a}}{\phi}\).
            Thus by \irref{M}, \irref{implyr}, and \irref{testd}, we prove the desired goal
             \[S(a) \models \ddiamond{{({!S(g)} \seq \asubst{\namedGame{\gamma}{g}}{\modEnd{S}{S(a)}}{S(a)})^\times} \seq !S(\finalNode)}{\ddiamond{\namedGame{\alpha}{a}}{\phi}}.\]
      \item Otherwise, if \(b \in \nodes{\namedGame{\gamma}{g}}\), then
      
      \(\prefix{b}{\asubst{\namedGame{\alpha}{a}}{S}{\phi}} = ({!S(g)} \seq \asubst{\namedGame{\gamma}{g}}{\modEnd{S}{S(a)}}{})^\times \seq {!S(g)} \seq \prefix{b}{(\asubst{\namedGame{\gamma}{g}}{\modEnd{S}{S(a)}}{})}\),
      
      \(\fwd{b}{{\namedGame{\alpha}{a}}} = \fwd{b}{\namedGame{\gamma}{g}} \seq {\namedGame{\alpha}{a}}\).

      Want to show:
      \begin{multline*}
      S(a) \models \langle (!{S(g)} \seq \asubst{\namedGame{\gamma}{g}}{\modEnd{S}{S(a)}}{})^\times \seq {!{S(g)} \seq \prefix{b}{(\asubst{\namedGame{\gamma}{g}}{\modEnd{S}{S(a)}}{})}} \rangle
      {\ddiamond{\fwd{b}{({\namedGame{\gamma}{g}})} \seq \namedGame{\alpha}{a}}{\phi}}.
      \end{multline*}

      First, by \rref{lem:subvalue-universal}, we have \(S(a) \models \ddiamond{(!{S(g)} \seq {\asubst{\namedGame{\gamma}{g}}{\modEnd{S}{S(a)}}{}})^\times \seq !S(\finalNode)}{\top}\), so by axioms \irref{testb} and \irref{composed},
      \(S(a) \models \ddiamond{(!{S(g)} \seq {\asubst{\namedGame{\gamma}{g}}{\modEnd{S}{S(a)}}{}})^\times}{\top}\).

      Next, by the inductive hypothesis and applications of axioms \irref{composed}, \irref{testd} and \irref{implyr}, we have
      \[ \top \models \ddiamond{!{S(g)} \seq \prefix{b}{(\asubst{\namedGame{\gamma}{g}}{\modEnd{S}{S(a)}}{})}}{\ddiamond{\fwd{b}{{\namedGame{\gamma}{g}}}}{S(a)}} \]

      Finally, since \(S\) is an Angelic subvalue map, we have \(S(a) \models \ddiamond{{\namedGame{\alpha}{a}}}{\phi}\).
      Thus by \irref{M} and \irref{composed} we can conclude the desired result.
    \end{enumerate}
    \item When \(\namedGame{\alpha}{a}\) is a Demon loop, i.e., \(\namedGame{(\namedGame{\gamma}{g})^\times}{a}\), then we analyze cases depending on the value of \(b\).
    \begin{enumerate}
      \item If \(b=a\), then 
      \(\prefix{b}{(\asubst{\namedGame{\alpha}{a}}{S}{\phi})} = \asubst{\namedGame{\alpha}{a}}{S}{\phi}\) and \(\fwd{b}{(\namedGame{\alpha}{a})} = \namedGame{\alpha}{a}\).
      
      Thus we must show that \(S(a) \models \ddiamond{\asubst{\namedGame{\alpha}{a}}{S}{\phi}}{\ddiamond{\namedGame{\alpha}{a}}{\phi}}\).
      
      \(S(a) \models \ddiamond{\asubst{\namedGame{\alpha}{a}}{S}{}}{S(a)}\) because \(S(a)\) is a loop invariant of \((\asubst{\namedGame{\gamma}{g}}{\modEnd{S}{S(a)}}{})^\times\) by \rref{lem:subvalue-universal}.

      \(S(a) \models \ddiamond{\namedGame{\alpha}{a}}{\phi}\) by definition of a subvalue map.
      Thus by \irref{M}, we can conclude the desired result.
      \item Otherwise, if \(b \in \nodes{\namedGame{\gamma}{g}}\), then:
      
      \(\prefix{b}{(\asubst{\namedGame{\alpha}{a}}{S}{\phi})} = \asubst{\namedGame{\alpha}{a}}{S}{\phi} \seq \prefix{b}{(\asubst{\namedGame{\gamma}{g}}{\modEnd{S}{S(a)}}{})}\),
      
      \(\fwd{b}{(\namedGame{\alpha}{a})} = \fwd{b}{(\namedGame{\gamma}{g})} \seq \namedGame{\alpha}{a}\).

      Thus we must show that
      \[
        S(a) \models \langle \asubst{\namedGame{\alpha}{a}}{S}{\phi} \seq \prefix{b}{(\asubst{\namedGame{\gamma}{g}}{\modEnd{S}{S(a)}}{})} \rangle
        {\ddiamond{\fwd{b}{(\namedGame{\gamma}{g})} \seq \namedGame{\alpha}{a}}{\phi}}
      \]
      First, we have \(S(a) \models \ddiamond{\asubst{\namedGame{\alpha}{a}}{S}{}}{S(a)}\) because \(S(a)\) is a loop invariant of \((\asubst{\namedGame{\gamma}{g}}{\modEnd{S}{S(a)}}{})^\times\) by \rref{lem:subvalue-universal}.

      Next, by the inductive hypothesis, we have 
      \[S(a) \models \ddiamond{\prefix{b}{(\asubst{\namedGame{\gamma}{g}}{\modEnd{S}{S(a)}}{})} \seq \fwd{b}{(\namedGame{\gamma}{g})}}{S(a)}.\]

      By definition of a subvalue map, we have \(S(a) \models \ddiamond{\namedGame{\alpha}{a}}{\phi}\).

      Thus by \irref{M} and \irref{composed}, we can conclude the desired result.
    \end{enumerate}
    \item When \(\namedGame{\alpha}{a}\) is sequential composition, i.e., \(\namedGame{(\namedGame{\gamma}{g} \seq \namedGame{\delta}{d})}{a}\), then we analyze cases depending on the value of \(b\).
    \begin{enumerate}
      \item If \(b=a\), then \(\prefix{b}{(\asubst{\namedGame{\alpha}{a}}{S}{\phi})} = \skp\) and \(\fwd{b}{\namedGame{\alpha}{a}} = \namedGame{\alpha}{a}\).
            Thus we must show that \(S(a) \models \ddiamond{\skp}{\ddiamond{\namedGame{\alpha}{a}}{\phi}}\), i.e., \(S(a) \models \ddiamond{\namedGame{\alpha}{a}}{\phi}\), which hold by definition of a subvalue map.
      \item Otherwise, if \(b \in \nodes{\namedGame{\gamma}{g}}\), then \(\prefix{b}{(\asubst{\namedGame{\alpha}{a}}{S}{\phi})} = \prefix{b}{(\asubst{\namedGame{\gamma}{g}}{\modEnd{S}{S(d)}}{})}\)
      and \(\fwd{b}{\namedGame{\alpha}{a}} = \fwd{b}{(\namedGame{\gamma}{g})} \seq \namedGame{\delta}{d}\).
      
      Thus we must show that \(S(a) \models \ddiamond{\prefix{b}{\asubst{\namedGame{\gamma}{g}}{\modEnd{S}{S(d)}}{}}}{\ddiamond{\fwd{b}{(\namedGame{\gamma}{g})} \seq \namedGame{\delta}{d}}{\phi}}\).
      
      This is \(S(a) \models \ddiamond{\prefix{b}{(\asubst{\namedGame{\gamma}{g}}{\modEnd{S}{S(d)}}{})}}{\ddiamond{\fwd{b}{(\namedGame{\gamma}{g})}}{\ddiamond{\namedGame{\delta}{d}}{\phi}}}\) applying \irref{composed}.

      This is \(S(a) \models \ddiamond{\prefix{b}{(\asubst{\namedGame{\gamma}{g}}{\modEnd{S}{S(d)}}{})} \seq \fwd{b}{(\namedGame{\gamma}{g})}}{\ddiamond{\namedGame{\delta}{d}}{\phi}}\) again applying \irref{composed}.

      By the inductive hypothesis, \(S(a) \models \ddiamond{\prefix{b}{(\asubst{\namedGame{\gamma}{g}}{\modEnd{S}{S(d)}}{})}}{S(d)}\) since \(S(a) \limply S(g)\) in an inductive subvalue map.

      By definition of an Angelic subvalue map, \(S(d) \models \ddiamond{\namedGame{\delta}{d}}{\phi}\).

      By rule \irref{M}, we can conclude the desired result.
      \item Otherwise, if \(b \in \nodes{\namedGame{\delta}{d}}\), then a similar argument applies.
      \(\prefix{b}{\asubst{\namedGame{\alpha}{a}}{S}{\phi}} = \{\asubst{\namedGame{\gamma}{g}}{\modEnd{S}{S(d)}}{} \seq \prefix{b}{(\asubst{\namedGame{\delta}{d}}{S}{\phi})}\}\), and \(\fwd{b}{\namedGame{\alpha}{a}} = \fwd{b}{\namedGame{\delta}{d}}\).

      Thus we must show that \(S(a) \models \ddiamond{\asubst{\namedGame{\gamma}{g}}{\modEnd{S}{S(d)}}{} \seq \prefix{b}{(\asubst{\namedGame{\delta}{d}}{S}{\phi})}}{\ddiamond{\fwd{b}{\namedGame{\delta}{d}}}{\phi}}\).

      By the \irref{composed} axiom, this is
      \(S(a) \models \ddiamond{\asubst{\namedGame{\gamma}{g}}{\modEnd{S}{S(d)}}{}}{\ddiamond{\prefix{b}{(\asubst{\namedGame{\delta}{d}}{S}{})} \seq \fwd{b}{(\namedGame{\delta}{d})}}{\phi}}\).

      By \rref{lem:subvalue-universal}, and since in an inductive subvalue map, \(\models S(a)\limply S(g)\), \(S(a) \models \ddiamond{\asubst{\namedGame{\gamma}{g}}{\modEnd{S}{S(d)}}{}}{S(d)}\).

      By the inductive hypothesis, \(S(d) \models \ddiamond{\prefix{b}{(\asubst{\namedGame{\delta}{d}}{S}{})} \seq \fwd{b}{(\namedGame{\delta}{d})}}{\phi}\).

      By rule \irref{M}, we can conclude the desired result.
    \end{enumerate}
  \end{itemize}
\end{proofE}

As discussed in \rref{sec:subvalue-map-definition}, the property of \rref{thm:subvalue-map-winning-region} is critical but not alone strong enough, since it does not guarantee that following the policy will ensure Angel wins.
The property it maintains is only that Angel always continues to have some winning strategy (\(\ddiamond{\fwd{b}{\namedGame{\alpha}{a}}}{\phi}\)), not that such a strategy is within the policy (\(\ddiamond{\fwd{b}{\aproj{\namedGame{\alpha}{a}}{S}{}}}{\phi}\)).
Angel might get stuck without a valid choice permitted by the map \(S\).
The \emph{inductive} subvalue map introduced next prevents this situation, disqualifying subvalue maps where Angel gets stuck.

\section{Inductive Subvalue Maps}
\label{sec:solution-representation}

To correct for the shortcoming of subvalue maps by preventing players from getting stuck, we introduce \emph{inductive} subvalue maps.
Notation \(\avalid{\alpha}{a}{S}{}\) indicates \(S\) is an \emph{inductive} subvalue map for game \(\namedGame{\alpha}{a}\).
\rref{thm:inductive-subvalue-map-winning-actions} shows the characteristic property of an inductive subvalue map: its policy always keeps the player in a region where the player can win the rest of the game \emph{by following the subvalue map policy}.

\begin{theoremE}[Inductive subvalue map ensures winning actions][restate, text proof={}]
  \label{thm:inductive-subvalue-map-winning-actions}
  Suppose \(S\) is an inductive Angelic subvalue map for game \(\namedGame{\alpha}{a}\) and compatible with winning condition \(\phi\), i.e., \(\avalid{\alpha}{a}{S}{\phi}\) and \(S(\finalNode) \models \phi\).
  Upon starting in any state \(\sigma \models S(a)\), and reaching subgame \(b\) by following any strategy induced by the subvalue map \(S\), there is a way for Angel to win by continuing to follow \(S\).
  That is, \(S(a) \models \ddiamond{\prefix{b}{(\asubst{\namedGame{\alpha}{a}}{S}{\phi})}}{\ddiamond{\fwd{b}{(\aproj{\namedGame{\alpha}{a}}{S}{\phi})}}{\phi}}\).
\end{theoremE}
\begin{proofE}
  We perform induction on the structure of \(\namedGame{\alpha}{a}\).
  \begin{itemize}
    \item When \(\namedGame{\alpha}{a}\) is atomic and not controlled by Angel,
    i.e., \(\alpha \in \{x:=e, x:=\otimes, \ptest{Q}, !Q, \{x'=f(x)\ \&\ Q\}^d\}\), then
    The only subgame \(b\) in \(\nodes{\namedGame{\alpha}{a}}\) is \(\namedGame{\alpha}{a}\) itself.
    Thus the only possibility is \(b=a\).
    \(\prefix{a}{\namedGame{\alpha}{a}}\) is the empty game \skp, and \(\fwd{a}{\namedGame{\alpha}{a}}\) is the game \(\namedGame{\alpha}{a}\), and \(\asubst{\namedGame{\alpha}{a}}{S}{\phi}\) is just \(\namedGame{\alpha}{a}\).
    Thus, \(\ddiamond{\prefix{a}{(\asubst{{\namedGame{\alpha}{a}}}{S}{\phi})}}{\ddiamond{\fwd{a}{(\aproj{{\namedGame{\alpha}{a}}}{S}{\phi})}}{\phi}}\) is equivalent to \(\ddiamond{\skp}{\ddiamond{{\alpha}}{\phi}}\).
    We need to show that \(S(a) \models {\ddiamond{\aproj{{\namedGame{\alpha}{a}}}{S}{} }{\phi}}\) which holds because by definition of an inductive Angelic subvalue map along with the fact that \(S(\finalNode)\limply \phi\), \(S(a) \limply \ddiamond{{\namedGame{\alpha}{a}}}{\phi}\).
    \item When \(\namedGame{\alpha}{a}\) is an Angelic nondeterministic assignment, i.e., \(\namedGame{x:=\ast}{a}\), then again the only subgame \(b\) in \(\nodes{\namedGame{\alpha}{a}}\) is \(\namedGame{\alpha}{a}\) itself.
    Thus the only possibility is \(b=a\).
    \(\prefix{a}{\namedGame{\alpha}{a}}\) is the empty game \skp, and \(\asubst{\namedGame{\alpha}{a}}{S}{\phi}\) is just \(\namedGame{\alpha}{a} \seq \ptest{S(\finalNode)}\), and \(\fwd{a}{\namedGame{\alpha}{a} \seq \ptest{S(\finalNode)}}\) is the game \(\namedGame{\alpha}{a} \seq \ptest{S(\finalNode)}\).
    Thus, \(\ddiamond{\asubst{\prefix{a}{\namedGame{\alpha}{a}}}{S}{\phi}}{\ddiamond{\aproj{\fwd{a}{\namedGame{\alpha}{a}}}{S}{\phi}}{\phi}}\) is equivalent to \(\ddiamond{\skp}{\ddiamond{{\alpha \seq \ptest{S(\finalNode)}}}{\phi}}\).
    We need to show that \(S(a) \models \ddiamond{\alpha}{(\ptest{S(\finalNode)} \land \phi)}\).
    Since \(\phi\) is a compatible winning condition, this is \(S(a) \models {\ddiamond{\alpha}{\ptest{S(\finalNode)}}}\) which holds by the definition of an inductive Angelic subvalue map.
    \item When \(\namedGame{\alpha}{a}\) is an Angelic ODE, i.e., \(\namedGame{\{x'=f(x)\ \&\ Q\}}{a}\), then the argument is similar to the previous case.
    As before the only possibility is \(b=a\), and we need to show that \(S(a) \models {\ddiamond{\alpha}{\ptest{S(\finalNode)} \land \phi}}\).
    Since \(\phi\) is a compatible winning condition, this is \(S(a) \models {\ddiamond{\aproj{{\namedGame{\alpha}{a}}}{S}{} }{\ptest{S(\finalNode)}}}\)
    which holds by the definition of an inductive Angelic subvalue map.
    \item When \(\namedGame{\alpha}{a}\) is an Angelic choice, i.e., \(\namedGame{(\namedGame{\gamma}{g} \cup \namedGame{\delta}{d})}{a}\), then we analyze cases depending on the value of \(b\).
    \begin{enumerate}
      \item If \(b=a\), then \(\prefix{b}{(\asubst{\namedGame{\alpha}{a}}{S}{\phi})} = \skp\) and \(\fwd{b}{\aproj{\namedGame{\alpha}{a}}{S}{\phi}} = \aproj{\namedGame{\alpha}{a}}{S}{}\).
      Thus we must show that \(S(a) \models \ddiamond{\skp}{\ddiamond{\aproj{\namedGame{\alpha}{a}}{S}{}}{\phi}}\), i.e., \(S(a) \models \ddiamond{\aproj{\namedGame{\alpha}{a}}{S}{}}{\phi}\), which hold by \rref{lem:inductive-subvalue-existence}.
      \item Otherwise, if \(b \in \nodes{\namedGame{\gamma}{g}}\), then \(\prefix{b}{(\asubst{\namedGame{\alpha}{a}}{S}{\phi})} = {!S(g)} \seq \prefix{b}{\asubst{\namedGame{\gamma}{g}}{S}{}}\) and \(\fwd{b}{\aproj{\namedGame{\alpha}{a}}{S}{\phi}} = \fwd{b}{\namedGame{\gamma}{g}}\).
      Thus we must show that \(S(a) \models \ddiamond{!S(g) \seq \prefix{b}{\asubst{\namedGame{\gamma}{g}}{S}{}}}{\ddiamond{\fwd{b}{\namedGame{\gamma}{g}}}{\phi}}\).

      This is \(S(a) \land S(g) \models \ddiamond{\prefix{b}{\asubst{\namedGame{\gamma}{g}}{S}{}}}{\ddiamond{\fwd{b}{\namedGame{\gamma}{g}}}{\phi}}\) using axioms \irref{composed}, \irref{testb}.
      This holds by the inductive hypothesis.
      \item Otherwise, if \(b \in \nodes{\namedGame{\delta}{d}}\), then the proof proceeds similarly to the previous case, but with \(\namedGame{\delta}{d}\) replacing \(\namedGame{\gamma}{g}\).
      \end{enumerate}
    \item When \(\namedGame{\alpha}{a}\) is a Demonic choice, i.e., \(\namedGame{(\namedGame{\gamma}{g} \cap \namedGame{\delta}{d})}{a}\), then the argument is similar to Angelic choice, except that we get assumptions \(S(g)\) or \(S(d)\) not from the execution prefix but from \(S(a)\) because in an inductive subvalue map, \(\models S(a) \limply S(g) \land S(d)\).
    We analyze cases depending on the value of \(b\).
    \begin{enumerate}
      \item If \(b=a\), then \(\prefix{b}{(\asubst{\namedGame{\alpha}{a}}{S}{\phi})} = \skp\) and \(\fwd{b}{\aproj{\namedGame{\alpha}{a}}{S}{\phi}} = \aproj{\namedGame{\alpha}{a}}{S}{}\).
            Thus we must show that \(S(a) \models \ddiamond{\skp}{\ddiamond{\aproj{\namedGame{\alpha}{a}}{S}{}}{\phi}}\), i.e., \(S(a) \models \ddiamond{\aproj{\namedGame{\alpha}{a}}{S}{}}{\phi}\), which hold by \rref{lem:inductive-subvalue-existence}.
      \item Otherwise, if \(b \in \nodes{\namedGame{\gamma}{g}}\), then \(\prefix{b}{(\asubst{\namedGame{\alpha}{a}}{S}{\phi})} = \prefix{b}{\asubst{\namedGame{\gamma}{g}}{S}{}}\) and \(\fwd{b}{\aproj{\namedGame{\alpha}{a}}{S}{\phi}} = \fwd{b}{\namedGame{\gamma}{g}}\).
      Thus we must show that \(S(a) \models \ddiamond{\prefix{b}{\asubst{\namedGame{\gamma}{g}}{S}{}}}{\ddiamond{\fwd{b}{\namedGame{\gamma}{g}}}{\phi}}\).
      This holds by the inductive hypothesis since as discussed, \(S(a)\limply S(g)\).
      \item Otherwise, if \(b \in \nodes{\namedGame{\delta}{d}}\), then the proof proceeds similarly to the previous case, but with \(\namedGame{\delta}{d}\) replacing \(\namedGame{\gamma}{g}\).
    \end{enumerate}
    \item When \(\namedGame{\alpha}{a}\) is sequential composition, i.e., \(\namedGame{(\namedGame{\gamma}{g} \seq \namedGame{\delta}{d})}{a}\), then we analyze cases depending on the value of \(b\).
    \begin{enumerate}
      \item If \(b=a\), then \(\prefix{b}{(\asubst{\namedGame{\alpha}{a}}{S}{\phi})} = \skp\) and \(\fwd{b}{\aproj{\namedGame{\alpha}{a}}{S}{\phi}} = \aproj{\namedGame{\alpha}{a}}{S}{}\).
            Thus we must show that \(S(a) \models \ddiamond{\skp}{\ddiamond{\aproj{\namedGame{\alpha}{a}}{S}{}}{\phi}}\), i.e., \(S(a) \models \ddiamond{\aproj{\namedGame{\alpha}{a}}{S}{}}{\phi}\), which hold by \rref{lem:inductive-subvalue-existence}.
      \item Otherwise, if \(b \in \nodes{\namedGame{\gamma}{g}}\), then \(\prefix{b}{(\asubst{\namedGame{\alpha}{a}}{S}{\phi})} = \prefix{b}{(\asubst{\namedGame{\gamma}{g}}{\modEnd{S}{S(d)}}{})}\)
      and \(\fwd{b}{\aproj{\namedGame{\alpha}{a}}{S}{\phi}} = \fwd{b}{(\aproj{\namedGame{\gamma}{g}}{\modEnd{S}{S(d)}}{})} \seq \aproj{\namedGame{\delta}{d}}{S}{}\).
      Thus we must show that \(S(a) \models \ddiamond{\prefix{b}{\asubst{\namedGame{\gamma}{g}}{\modEnd{S}{S(d)}}{}}}{\ddiamond{\fwd{b}{(\aproj{\namedGame{\gamma}{g}}{\modEnd{S}{S(d)}}{})} \seq \aproj{\namedGame{\delta}{d}}{S}{}}{\phi}}\).
      
      This is \(S(a) \models \ddiamond{\prefix{b}{(\asubst{\namedGame{\gamma}{g}}{\modEnd{S}{S(d)}}{})}}{\ddiamond{\fwd{b}{(\aproj{\namedGame{\gamma}{g}}{\modEnd{S}{S(d)}}{})}}{\ddiamond{\aproj{\namedGame{\delta}{d}}{S}{}}{\phi}}}\) applying \irref{composed}.

      This is \(S(a) \models \ddiamond{\prefix{b}{(\asubst{\namedGame{\gamma}{g}}{\modEnd{S}{S(d)}}{})} \seq \fwd{b}{(\aproj{\namedGame{\gamma}{g}}{\modEnd{S}{S(d)}}{})}}{\ddiamond{\aproj{\namedGame{\delta}{d}}{S}{}}{\phi}}\) again applying \irref{composed}.

      By the inductive hypothesis, \(S(a) \models \ddiamond{\prefix{b}{(\asubst{\namedGame{\gamma}{g}}{\modEnd{S}{S(d)}}{})}}{S(d)}\) since \(S(a) \limply S(g)\) in an inductive subvalue map.

      By \rref{lem:inductive-subvalue-existence}, \(S(d) \models \ddiamond{\aproj{\namedGame{\delta}{d}}{S}{}}{\phi}\).

      By rule \irref{M}, we can conclude the desired result.
      \item Otherwise, if \(b \in \nodes{\namedGame{\delta}{d}}\), then a similar argument applies.
      \(\prefix{b}{\asubst{\namedGame{\alpha}{a}}{S}{\phi}} = \{\asubst{\namedGame{\gamma}{g}}{\modEnd{S}{S(d)}}{} \seq \prefix{b}{(\asubst{\namedGame{\delta}{d}}{S}{\phi})}\}\), and \(\fwd{b}{\aproj{\namedGame{\alpha}{a}}{S}{\phi}} = \fwd{b}{\aproj{\namedGame{\delta}{d}}{S}{}}\).

      Thus we must show that \(S(a) \models \ddiamond{\asubst{\namedGame{\gamma}{g}}{\modEnd{S}{S(d)}}{} \seq \prefix{b}{(\asubst{\namedGame{\delta}{d}}{S}{\phi})}}{\ddiamond{\fwd{b}{\aproj{\namedGame{\delta}{d}}{S}{}}}{\phi}}\).

      By the \irref{composed} axiom, this is
      \(S(a) \models \ddiamond{\asubst{\namedGame{\gamma}{g}}{\modEnd{S}{S(d)}}{}}{\ddiamond{\prefix{b}{(\asubst{\namedGame{\delta}{d}}{S}{})} \seq \fwd{b}{(\aproj{\namedGame{\delta}{d}}{S}{})}}{\phi}}\).

      By \rref{lem:inductive-subvalue-universal}, and since in an inductive subvalue map, \(\models S(a)\limply S(g)\), \(S(a) \models \ddiamond{\asubst{\namedGame{\gamma}{g}}{\modEnd{S}{S(d)}}{}}{S(d)}\).

      By the inductive hypothesis, \(S(d) \models \ddiamond{\prefix{b}{(\asubst{\namedGame{\delta}{d}}{S}{})} \seq \fwd{b}{(\aproj{\namedGame{\delta}{d}}{S}{})}}{\phi}\).

      By rule \irref{M}, we can conclude the desired result.
    \end{enumerate}
    \item When \(\namedGame{\alpha}{a}\) is an Angelic loop, i.e., \(\namedGame{(\namedGame{\gamma}{g})^*}{a}\), then we analyze cases depending on the value of \(b\).
    \begin{enumerate}
      \item If \(b=a\) then \(\prefix{b}{(\asubst{\namedGame{\alpha}{a}}{S}{\phi})} = (!{S(g)} \seq \asubst{\namedGame{\gamma}{g}}{\modEnd{S}{S(a)}}{})^\times\) (with the terminal test missing compared to \(\asubst{\namedGame{\alpha}{a}}{S}{\phi}\)) and \(\fwd{b}{\aproj{\namedGame{\alpha}{a}}{S}{\phi}} = \aproj{\namedGame{\alpha}{a}}{S}{}\).
      
      We need to show that \(S(a) \models \ddiamond{{(!{S(g)} \seq \asubst{\namedGame{\gamma}{g}}{\modEnd{S}{S(a)}}{})^\times}}{\ddiamond{\aproj{\namedGame{\alpha}{a}}{S}{}}{\phi}}\).

      We first show that \(S(a) \models \ddiamond{{(!{S(g)} \seq \asubst{\namedGame{\gamma}{g}}{\modEnd{S}{S(a)}}{})^\times}}{S(a)}\).
      This follows from applying the \irref{loop} rule with invariant \(S(a)\) and then using the fact hat per \rref{lem:inductive-subvalue-universal}, \(S(g)\models \ddiamond{\asubst{\namedGame{\gamma}{g}}{\modEnd{S}{S(a)}}{}}{S(a)}\).
      
      Then by \rref{lem:inductive-subvalue-existence},
      \(S(a) \models \ddiamond{\aproj{\namedGame{\alpha}{a}}{S}{}}{\phi}\).
      Thus by \irref{M}, we can conclude the desired result.
      \item Otherwise, if \(b\neq a\) then it is in \(\nodes{\namedGame{\gamma}{g}}\).
      
      \(\prefix{b}{(\asubst{\namedGame{\alpha}{a}}{S}{\phi})} = (!{S(g)} \seq \prefix{b}{\asubst{\namedGame{\gamma}{g}}{\modEnd{S}{S(a)}}{}})^\times \seq {!{S(g)} \seq \prefix{b}{(\asubst{\namedGame{\gamma}{g}}{\modEnd{S}{S(a)}}{})}}\).
      
      \(\fwd{b}{\aproj{\namedGame{\alpha}{a}}{S}{\phi}} = \fwd{b}{(\aproj{\namedGame{\gamma}{g}}{\modEnd{S}{S(a)}{}})} \seq \aproj{\namedGame{\alpha}{a}}{S}{}\).

      Want to show:
      \begin{multline*}
      S(a) \models \langle (!{S(g)} \seq \asubst{\namedGame{\gamma}{g}}{\modEnd{S}{S(a)}}{})^\times \seq {!{S(g)} \seq \prefix{b}{(\asubst{\namedGame{\gamma}{g}}{\modEnd{S}{S(a)}}{})}} \rangle \\
      {\ddiamond{\fwd{b}{(\aproj{\namedGame{\gamma}{g}}{\modEnd{S}{S(a)}}{})} \seq \aproj{\namedGame{\alpha}{a}}{S}{}}{\phi}}.
      \end{multline*}

      First, by \rref{lem:inductive-subvalue-universal}, we have \(S(a) \models \ddiamond{(!{S(g)} \seq {\asubst{\namedGame{\gamma}{g}}{\modEnd{S}{S(a)}}{}})^\times \seq !S(\finalNode)}{\top}\), so by axioms \irref{testb} and \irref{composed},
      \(S(a) \models \ddiamond{(!{S(g)} \seq {\asubst{\namedGame{\gamma}{g}}{\modEnd{S}{S(a)}}{}})^\times}{\top}\).

      Next, by the inductive hypothesis and applications of axioms \irref{composed}, \irref{testd} and \irref{implyr}, we have
      \[ \top \models \ddiamond{!{S(g)} \seq \prefix{b}{(\asubst{\namedGame{\gamma}{g}}{\modEnd{S}{S(a)}}{})}}{\ddiamond{\fwd{b}{\aproj{\namedGame{\gamma}{g}}{\modEnd{S}{S(a)}}{}}}{S(a)}} \]

      Finally, by \rref{lem:inductive-subvalue-existence}, we have \(S(a) \models \ddiamond{\aproj{\namedGame{\alpha}{a}}{S}{}}{\phi}\).
      Thus by \irref{M} and \irref{composed} we can conclude the desired result.
    \end{enumerate}
    \item When \(\namedGame{\alpha}{a}\) is a Demonic loop, i.e., \(\namedGame{(\namedGame{\gamma}{g})^\times}{a}\), then we analyze cases depending on the value of \(b\).
    \begin{enumerate}
      \item If \(b=a\), then 
      \(\prefix{b}{(\asubst{\namedGame{\alpha}{a}}{S}{\phi})} = \asubst{\namedGame{\alpha}{a}}{S}{\phi}\) and \(\fwd{b}{(\aproj{\namedGame{\alpha}{a}}{S}{\phi})} = \aproj{\namedGame{\alpha}{a}}{S}{}\).
      
      Thus we must show that \(S(a) \models \ddiamond{\asubst{\namedGame{\alpha}{a}}{S}{\phi}}{\ddiamond{\aproj{\namedGame{\alpha}{a}}{S}{}}{\phi}}\).
      
      \(S(a) \models \ddiamond{\asubst{\namedGame{\alpha}{a}}{S}{}}{S(a)}\) because \(S(a)\) is a loop invariant of \((\asubst{\namedGame{\gamma}{g}}{\modEnd{S}{S(a)}}{})^\times\) by \rref{lem:inductive-subvalue-universal}.

      \(S(a) \models \ddiamond{\aproj{\namedGame{\alpha}{a}}{S}{}}{\phi}\) by \rref{lem:inductive-subvalue-existence}.
      Thus by \irref{M}, we can conclude the desired result.
      \item Otherwise, if \(b \in \nodes{\namedGame{\gamma}{g}}\), then:
      
      \(\prefix{b}{(\asubst{\namedGame{\alpha}{a}}{S}{\phi})} = \asubst{\namedGame{\alpha}{a}}{S}{\phi} \seq \prefix{b}{(\asubst{\namedGame{\gamma}{g}}{\modEnd{S}{S(a)}}{})}\),
      
      \(\fwd{b}{(\aproj{\namedGame{\alpha}{a}}{S}{\phi})} = \fwd{b}{(\aproj{\namedGame{\gamma}{g}}{\modEnd{S}{S(a)}}{})} \seq \aproj{\namedGame{\alpha}{a}}{S}{}\).

      Thus we must show that
      \[
        S(a) \models \langle \asubst{\namedGame{\alpha}{a}}{S}{\phi} \seq \prefix{b}{(\asubst{\namedGame{\gamma}{g}}{\modEnd{S}{S(a)}}{})} \rangle
        {\ddiamond{\fwd{b}{(\aproj{\namedGame{\gamma}{g}}{\modEnd{S}{S(a)}}{})} \seq \aproj{\namedGame{\alpha}{a}}{S}{}}{\phi}}
      \]
      First, we have \(S(a) \models \ddiamond{\asubst{\namedGame{\alpha}{a}}{S}{}}{S(a)}\) because \(S(a)\) is a loop invariant of \((\asubst{\namedGame{\gamma}{g}}{\modEnd{S}{S(a)}}{})^\times\) by \rref{lem:inductive-subvalue-universal}.

      Next, by the inductive hypothesis, we have \(S(a) \models \ddiamond{\prefix{b}{(\asubst{\namedGame{\gamma}{g}}{\modEnd{S}{S(a)}}{})}}{S(a)}\).

      By \rref{lem:inductive-subvalue-existence}, we have \(S(a) \models \ddiamond{\aproj{\namedGame{\alpha}{a}}{S}{}}{\phi}\).

      Thus by \irref{M} and \irref{composed}, we can conclude the desired result.
    \end{enumerate}
  \end{itemize}
\end{proofE}

It is possible to characterize an inductive subvalue map using recursive conditions expressed via \dGL formulas (\rref{def:local-envelope-conditions}).

\begin{definition}[Inductive subvalue maps]
  \label{def:local-envelope-conditions}
  Let $S$ be a map from the subgames of $\namedGame{\alpha}{a}$ to winning subregions.
  $S$ is an \emph{inductive Angelic subvalue map} for game $\namedGame{\alpha}{a}$, written $\avalid{\alpha}{a}{S}{}$, when the following holds.
  If $\namedGame{\alpha}{a}$ has structure: 
  \begin{align*}
    \textrm{atomic, i.e., }\alpha \in &\{x:=e, x:=*, x:=\otimes, \ptest{Q}, !Q, \{x'=f(x)\ \&\ Q\}, \\
    & \phantom{\{} \{x'=f(x)\ \& \ Q\}^d\} \textrm{ then } \models S(a) \limply \ddiamond{\alpha}{S(\finalNode)} . \\
    \namedGame{(\namedGame{\gamma}{g} \cup \namedGame{\delta}{d})}{a} \textrm{ then }& \models S(a) \limply S(g) \lor S(d) \textrm{ and }
    \avalid{\gamma}{g}{S}{} \textrm{ and } \avalid{\delta}{d}{S}{} . \\
    \namedGame{(\namedGame{\gamma}{g} \cap \namedGame{\delta}{d})}{a} \textrm{ then }& \models S(a) \limply S(g) \land S(d) \textrm{ and }
    \avalid{\gamma}{g}{S}{} \textrm{ and } \avalid{\delta}{d}{S}{} . \\
    \namedGame{(\namedGame{\gamma}{g}\seq \namedGame{\delta}{d})}{a} \textrm{ then }& \models S(a) \limply S(g) \textrm{ and }
    \avalid{\gamma}{g}{\modEnd{S}{S(d)}}{} \textrm{ and } \avalid{\delta}{d}{S}{} .\\
    \namedGame{(\namedGame{\gamma}{g})^*}{a} \textrm{ then }& \models S(a) \limply \ddiamond{\aproj{\namedGame{\alpha}{a}}{S}{}}{S(\finalNode)} \textrm{ and } \avalid{\gamma}{g}{\modEnd{S}{S(a)}}{} . \\
    \namedGame{(\namedGame{\gamma}{g})^\times}{a} \textrm{ then }& \models S(a) \limply S(g) \land S(\finalNode)
    \textrm{ and } \avalid{\gamma}{g}{\modEnd{S}{S(a)}}{} .
  \end{align*}
  Symmetric conditions characterize when $S$ is an \emph{inductive Demonic subvalue map} for game $\namedGame{\alpha}{a}$ (\rref{app:inductive-subvalue-map}).
\end{definition}
\begin{example}
Consider again the Angelic choice in \rref{fig:overview}.
To establish $S \models 8\mspace{-6mu}:\mspace{-5mu} (\namedGame{(v:=a)^*}{6} \cup \namedGame{(v:=v-1 \seq a:=*)}{5})$,
\rref{def:local-envelope-conditions} first checks that if $S(8)$ (i.e., $\phi_8$) holds then either Angel can win by choosing to run subgame 6 (when $\phi_6$ holds) or subgame 5 (when $\phi_5$ holds).
It then also recursively ensures that after going to subgame 5 or 6, $S$ continues to provide a valid Angelic strategy ($\avalid{(v:=v+a)^\ast}{6}{S}{}$ and $\avalid{(v:=v-1\seq a:=\ast)}{5}{S}{}$).
\end{example}

We discuss the subtle loop cases of \rref{def:local-envelope-conditions}.
For Demonic loop validity $\avalid{(\namedGame{\gamma}{g})^\times}{a}{S}{}$, the following conditions should hold.
\begin{enumerate}
  \item $\models S(a) \limply S(\finalNode)$, and $\models S(a) \limply$ $\langle \gamma \rangle S(a)$, i.e. S(a) is an invariant that holds inductively and implies the postcondition.
  Inductiveness condition $\models S(a) \limply \langle \gamma \rangle S(a)$ does not appear in \rref{def:local-envelope-conditions} because it is already implied by later conditions $\models S(a) \limply S(g)$ and $\avalid{\gamma}{g}{\modEnd{S}{S(a)}}{}$.
  \item $\avalid{\gamma}{g}{\modEnd{S}{S(a)}}{}$, i.e. the subvalue map for the loop body ensures finishing within the safe invariant.
  Using instead the more liberal condition $\avalid{\gamma}{g}{S}{}$ to finish the loop body within the postcondition is unsound: Demon can choose to repeat the loop indefinitely, so Angel must play to remain within the inductive region.
  Alternative condition $\avalid{\gamma}{g}{\modEnd{S}{S(g)}}{}$ is also unsound since then $S(g)$ is merely inductive and does not have to imply reaching desired postcondition $S(\finalNode)$.
  \item $\models S(a) \limply S(g)$.
  In case Demon chooses to run the loop, Angel is guaranteed to not get stuck while playing the loop body in winning subregion $S(g)$, so overall loop winning subregion $S(a)$ should not exceed $S(g)$.
  On the other hand, $S(g)$ can be weaker than $S(a)$ as it can be one loop iteration away from implying $S(\finalNode)$, while $S(a)$ must already imply $S(\finalNode)$.
\end{enumerate}

For Angelic loop validity $\avalid{(\namedGame{\gamma}{g})^*}{a}{S}{}$, the following conditions must hold.
\begin{enumerate}
\item $\models S(a) \limply \ddiamond{\aproj{\namedGame{\alpha}{a}}{S}{}}{S(\finalNode)}$, i.e., Angelic loop invariant
$S(a)$ ensures that there is a \emph{well-founded} strategy to exit while staying within the control envelope.
Angel should not get stuck in a state where her only option is playing the loop forever.
Using instead the more liberal condition $\models S(a) \limply \ddiamond{\mathbf{\alpha}}{S(\finalNode)}$ is not enough, since Angel can get stuck in a state where there is a way to exit the loop but only when going outside the subvalue map (see \rref{app:counterexample-loop-proj}).
\item $\avalid{\gamma}{g}{\modEnd{S}{S(a)}}{}$, i.e., the subvalue map for the loop body ensures that after playing the loop body, there is still a well-founded exit strategy within the subvalue map $(S(a))$.
Using condition $\avalid{\gamma}{g}{\modEnd{S}{S(g)\lor S(\finalNode)}}{}$ instead is unsound as
then $S(g)$ is merely inductive and does not have to imply an exit strategy to reach postcondition $S(\finalNode)$ (see \rref{app:counterexample-loop-inner}).
\end{enumerate}

\rref{thm:consistancy} below proves that \emph{all} inductive subvalue maps are subvalue maps, indicating that the validity conditions compose to produce correct Angel and Demon winning subregions.
Consequently, for suitable initial states, while following an inductive subvalue map, the player will never reach a state from which it cannot win (\rref{thm:subvalue-map-winning-region}).
Additionally, \rref{thm:inductive-subvalue-map-winning-actions} at the start of this section shows that following the policy induced by an inductive subvalue map ensures that the agent will never reach a state where it is stuck.
In combination, these mean that starting in a suitable initial state, while following an inductive subvalue map, a player cannot lose.

\begin{theoremE}[Inductive subvalue maps are subvalue maps][restate, text proof={}]
  \label{thm:consistancy}
  For \dGL game $\namedGame{\alpha}{a}$, if $S$ is an inductive Angelic subvalue map ($\avalid{\alpha}{a}{S}{}$), then it is also an Angelic subvalue map for $\namedGame{\alpha}{a}$, i.e., for every subgame $\namedGame{\beta}{b}$ in $\nodes{\namedGame{\alpha}{a}}$, $\models S(b) \limply \langle \fwd{b}{\namedGame{\alpha}{a}} \rangle S(\finalNode)$.
  Dually, if \(S\) is an inductive Demonic subvalue map for \(\namedGame{\alpha}{a}\), then \(S\) is a Demonic subvalue map for $\namedGame{\alpha}{a}$.
\end{theoremE}
\begin{proofE}
    The proof uses structural induction.
    When $\namedGame{\alpha}{a}$ is atomic, i.e. $\alpha \in \{x:=e, x:=*, \ptest{Q}, !Q, \{x'=f(x)\ \&\ Q\}, \{x'=f(x)\ \& \ Q\}^d\}$, according to \rref{def:local-envelope-conditions}, $\avalid{\alpha}{a}{S}{}$ exactly when $\models S(a) \limply \langle \alpha \rangle S(\finalNode)$.
    Since $\fwd{a}{\namedGame{\alpha}{a}}=\namedGame{\alpha}{a}$ and $a$ is the only subgame in $\nodes{\namedGame{\alpha}{a}}$, we can conclude that $S$ is a subvalue map.
    For the recursive cases, if $\namedGame{\alpha}{a}$ has the structure:
    \begin{enumerate}
        \item $\namedGame{\gamma}{g}\cup\namedGame{\delta}{d}$,
        for any subgame $b$ in $\nodes{\namedGame{\gamma}{g}}$, $\fwd{b}{\namedGame{\alpha}{a}} = \fwd{b}{\namedGame{\gamma}{g}}$.
        Since the per inductive hypothesis, $\models S(b) \limply \langle \fwd{b}{\namedGame{\gamma}{g}} \rangle S(\finalNode)$, we can conclude that $\models S(b) \limply \langle \fwd{b}{\namedGame{\alpha}{a}} \rangle S(\finalNode)$.
        The same argument applies to all subgames in $\nodes{\namedGame{\delta}{d}}$.
        It remains to show that for the overall game $a$, $\models S(a) \limply \langle \fwd{a}{\namedGame{\alpha}{a}} \rangle S(\finalNode)$.
        By the inductive hypothesis, because $\avalid{\gamma}{g}{S}{}$, so $\models S(g) \limply \langle \gamma \rangle S(\finalNode)$ and because $\avalid{\delta}{d}{S}{}$, so $\models S(d) \limply \langle \delta \rangle S(\finalNode)$.
        Using the disjunction of implications, $\models S(g)\lor S(d) \limply \langle \gamma \rangle S(\finalNode) \lor \langle \delta \rangle S(\finalNode)$.
        But according to \dGL axiom \irref{choiced}, $\models \langle \gamma \rangle S(\finalNode) \lor \langle \delta \rangle S(\finalNode) \leftrightarrow \langle \gamma \cup \delta\rangle S(\finalNode)$,
        so $\models S(g)\lor S(d) \limply \langle \gamma \cup \delta\rangle S(\finalNode)$.
        Because $S$ is valid, $\models S(a) \limply S(g) \lor S(d)$.
        By transitivity of implication, $\models S(a) \limply \langle \gamma \cup \delta\rangle S(\finalNode)$, proving the desired result.
        \item $\namedGame{\gamma}{g}\cap\namedGame{\delta}{d}$,
        for any subgame $b$ in $\nodes{\namedGame{\gamma}{g}}$, $\fwd{b}{\namedGame{\alpha}{a}} = \fwd{b}{\namedGame{\gamma}{g}}$.
        Since per the inductive hypothesis, $\models S(b) \limply \langle \fwd{b}{\namedGame{\gamma}{g}} \rangle S(\finalNode)$, we can conclude that $\models S(b) \limply \langle \fwd{b}{\namedGame{\alpha}{a}} \rangle S(\finalNode)$.
        The same argument applies to all subgames in $\nodes{\namedGame{\delta}{d}}$.
        It remains to show that for the overall game $a$, $\models S(a) \limply \langle \fwd{a}{\namedGame{\alpha}{a}} \rangle S(\finalNode)$.
        By the inductive hypothesis, because $\avalid{\gamma}{g}{S}{}$, so $\models S(g) \limply \langle \gamma \rangle S(\finalNode)$ and because $\avalid{\delta}{d}{S}{}$, so $\models S(d) \limply \langle \delta \rangle S(\finalNode)$.
        Using the conjunction of implications, $\models S(g)\land S(d) \limply \langle \gamma \rangle S(\finalNode) \land \langle \delta \rangle S(\finalNode)$.
        But according to \dGL axiom $\irref{dchoiced}$, $\models \langle \gamma \rangle S(\finalNode) \land \langle \delta \rangle S(\finalNode) \leftrightarrow \langle \gamma \cap \delta\rangle S(\finalNode)$,
        so $\models S(g)\land S(d) \limply \langle \gamma \cap \delta\rangle S(\finalNode)$.
        Because $S$ is valid, $\models S(a) \limply S(g) \land S(d)$.
        By transitivity of implication, $\models S(a) \limply \langle \gamma \cap \delta\rangle S(\finalNode)$, proving the desired result.
        \item $\namedGame{\gamma}{g};\namedGame{\delta}{d}$, 
        for any subgame $b$ in $\nodes{\namedGame{\delta}{d}}$, $\fwd{b}{\namedGame{\delta}{d}} = \fwd{b}{\namedGame{\alpha}{a}}$.
        Since per the inductive hypothesis, $\models S(b) \limply \langle \fwd{b}{\namedGame{\delta}{d}} \rangle S(\finalNode)$, we can conclude that $\models S(b) \limply \langle \fwd{b}{\namedGame{\alpha}{a}} \rangle S(\finalNode)$.

        Consider now any subgame $b$ in $\nodes{\namedGame{\gamma}{g}}$.
        Because $\avalid{\gamma}{g}{S}{S(d)}$, $\models S(b) \limply \langle \fwd{b}{\namedGame{\gamma}{g}} \rangle S(d)$.
        Now, because $\avalid{\delta}{d}{S}{}$, $\models S(d) \limply \langle \delta \rangle S(\finalNode)$.
        According to the \dGL monotonicity rule \irref{M}, this means that $\models S(b) \limply \langle \fwd{b}{\namedGame{\gamma}{g}} \rangle \langle \delta \rangle S(\finalNode)$.
        According to the \dGL axiom \irref{composed}, $\models S(b) \limply \langle \fwd{b}{\namedGame{\gamma}{g}};\delta \rangle S(\finalNode)$.
        Since $\fwd{b}{\namedGame{\alpha}{a}} = \fwd{b}{\namedGame{\gamma}{g}};\namedGame{\delta}{d}$, we can conclude that $\models S(b) \limply \langle \fwd{b}{\namedGame{\alpha}{a}} \rangle S(\finalNode)$.

        It remains to show that for the overall game $a$, $\models S(a) \limply \langle \fwd{a}{\namedGame{\alpha}{a}} \rangle S(\finalNode)$.
        By the inductive hypothesis, because $\avalid{\gamma}{g}{S}{S(\delta)}$, so $\models S(g) \limply \langle \gamma \rangle S(\delta)$ and because $\avalid{\delta}{d}{S}{}$, so $\models S(d) \limply \langle \delta \rangle S(\finalNode)$.
        According to the \dGL monotonicity rule \irref{M}, this means that $\models S(g) \limply \langle \gamma \rangle \langle \delta\rangle S(\finalNode)$.
        According to the \dGL axiom \irref{composed}, $\models \langle \gamma \rangle \langle \delta \rangle S(\finalNode) \leftrightarrow \langle \gamma \seq \delta\rangle S(\finalNode)$.
        Thus, $\models S(g) \limply \langle \gamma \seq \delta\rangle S(\finalNode)$.
        Because $S$ is valid, $\models S(a) \limply S(g)$.
        By transitivity of implication, $\models S(a) \limply \langle \gamma \seq \delta\rangle S(\finalNode)$, proving the desired result.
        \item $\namedGame{\gamma}{g}^*$, then because $\avalid{(\namedGame{\gamma}{g})^\ast}{a}{S}{}$, then immediately as the first validity condition $\models S(a) \limply \langle \aproj{\namedGame{\alpha}{a}}{S}{} \rangle S(\finalNode)$.
        Per \rref{lem:subvalue-projection-refinement}, $\models \langle \aproj{\namedGame{\alpha}{a}}{S}{} \rangle S(\finalNode) \limply \langle \alpha \rangle S(\finalNode)$.
        By the transitivity of implication, $\models S(a) \limply \langle \alpha \rangle S(\finalNode)$.
        \item $\namedGame{\gamma}{g}^\times$, then because $\avalid{(\namedGame{\gamma}{g}^\times)}{a}{S}{}$, $\models S(a) \limply S(g)$ and $\models S(a) \limply S(\finalNode)$.
        Additionally, by the inductive hypothesis, because $\avalid{\gamma}{g}{S}{S(a)}, \models S(g) \limply \langle \gamma \rangle S(a)$.
        By the transitivity of implication, $\models S(a) \limply \langle \gamma \rangle S(a)$.
        Applying the \irref{invind} rule of \dGL, this permits us to infer $\models S(a) \limply \langle \gamma^\times \rangle S(a)$.
        Recall that because $\avalid{(\namedGame{\gamma}{g}^\times)}{a}{S}{}$, $\models S(a) \limply S(\finalNode)$.
        By the monotonicity rule \irref{M}, $\models S(a) \limply \langle \gamma^\times \rangle S(\finalNode)$, proving that $\models S(a) \limply \fwd{a}{\namedGame{\alpha}{a}}$.

        It remains to show that for any subgame $b$ in $\nodes{\namedGame{\gamma}{g}}$, $\models S(b) \limply \langle \fwd{b}{\namedGame{\alpha}{a}} \rangle S(\finalNode)$.
        Because $\avalid{\gamma}{g}{S}{S(a)}$, by the inductive hypothesis $\models S(b) \limply \langle \fwd{b}{\namedGame{\gamma}{g}} \rangle S(a)$.
        Since $\models S(a) \limply \langle \gamma^\times \rangle S(\finalNode)$, by the monotonicity rule \irref{M}, $\models S(b) \limply \langle \fwd{b}{\namedGame{\gamma}{g}} \rangle \langle \gamma^\times \rangle S(\finalNode)$.
        By \dGL axiom \irref{composed}, $\models S(b) \limply \langle \fwd{b}{\namedGame{\gamma}{g}};\gamma^\times \rangle S(\finalNode)$.
        But, $\fwd{b}{\namedGame{\alpha}{a}} = \fwd{b}{\namedGame{\gamma}{g}} \seq (\gamma^\times)$.
        Thus, $\models S(b) \limply \langle \fwd{b}{\namedGame{\alpha}{a}} \rangle S(\finalNode)$.
    \end{enumerate}
\end{proofE}
\begin{corollaryE}[Inductive subvalue map correctness][all end,text proof={}]
  \label{cor:soundness}
  For \dGL game $\namedGame{\alpha}{a}$ and inductive Angelic subvalue map $S$, if $\avalid{\alpha}{a}{S}{}$, then $\models S(a) \limply \ddiamond{\alpha}{S(\finalNode)}$.
  Dually, if $\dvalid{\alpha}{a}{S}{}$, then $\models S(a) \limply [\alpha] S(\finalNode)$.
\end{corollaryE}
\begin{proofE}
  Corollary of \rref{thm:consistancy}, which shows that inductive subvalue maps are subvalue maps.
  \rref{def:subvalue-map} says that the winning subregion for every subgame in the subvalue map implies the winning region for starting at that subgame and running the rest of the game.
  For overall game $a$, this means $\models S(a) \limply \langle \alpha \rangle S(\finalNode)$.
\end{proofE}

\rref{def:mpc-sol} shows how the dominant strategy to win a \dGL game can be captured as an inductive subvalue map, indicating that \rref{def:local-envelope-conditions} is sufficiently liberal.
Angel's \emph{dominant} strategy for a game is one that wins in the \emph{largest} possible set of states.
\rref{def:mpc-sol} constructs the inductive subvalue map inducing a dominant strategy indicating subvalues via modal \dGL formulas.
The winning subregion for the overall game is maximal.

\begin{definition}[Model predictive subvalue map]
\label{def:mpc-sol}
For the game $\namedGame{\alpha}{a}$ and win condition $\phi$, the \emph{model predictive Angelic subvalue map} $\{\finalNode \mapsto \phi\} \cup \bigcup_{b\in\nodes{\namedGame{\alpha}{a}}}\{b \mapsto \ddiamond{\fwd{b}{\namedGame{\alpha}{a}}}{\phi}\}$ maps every subgame $b$ of $\namedGame{\alpha}{a}$ to the optimal winning regions of its game suffix.
Dually, the \emph{model predictive Demonic subvalue map} is \(\{\finalNode \mapsto \phi\} \cup \bigcup_{b\in\nodes{\namedGame{\alpha}{a}}}\{b \mapsto [\fwd{b}{\namedGame{\alpha}{a}}] \phi\}\).
\end{definition}
\begin{example}
  For the game of \rref{fig:overview}, the model predictive Angelic subvalue map maps subgame $5$ to $\ddiamond{\{x'=v\}^d\seq \alpha}{x>0}$, which per the semantics of \dGL, precisely means the condition under which Angel can win by reaching her assigned winning conditions after playing the remainder of the game $\{x'=v\}^d\seq \alpha$.
\end{example}
\begin{remark}
  Model Predictive Control (MPC) \cite{GARCIA1989335,richalet1978model,Cutler1979DynamicMC} is a well-known controller design principle that
   uses a mathematical model to predict controller behavior
   and back-computes the control conditions that respect the desired system constraints.
  \rref{def:mpc-sol} extends this construction to hybrid games logic.
  game suffix $\langle \fwd{b}{\namedGame{\alpha}{a}}\rangle\phi$ serves as an exact model of future game behavior after subgame $b$\footnote{Though unlike standard MPC, as we are computing maximal control sets, there is no optimization step to identify the best control solution within the set.}.
\end{remark}

\begin{remark}
A runtime monitor (constructed similar to existing work on monitoring \dL \cite{DBLP:journals/fmsd/MitschP16} and constructive \dGL \cite{cdgl}) raises an alarm or intervenes to enforce a correct fallback when gameplay deviates from the control envelope.
It is possible to model the use of a subvalue map as a \emph{runtime monitor} in a single \dGL formula.
For an Angelic subvalue map, this \emph{monitored game} is constructed by 
\begin{enumerate*}
  \item giving Demon control over Angelic decisions while restricting him to play per the subvalue map (like in universal projection) to mimic an unverified controller that plays unreliably while being forced to stay within the control envelope by a runtime monitor and
  \item ensuring before each such monitored subgame that the subvalue map allows \emph{some} control action using Angelic tests.
\end{enumerate*}
\rref{eq:overview-projection} shows the monitored game for the running example.
The \emph{monitored game} for \dGL game $\namedGame{\alpha}{a}$ and Angelic subvalue map $S$, written $\aruntimem{\namedGame{\alpha}{a}}{S}$, is generated recursively per the structure of $\namedGame{\alpha}{a}$ as follows. If $\namedGame{\alpha}{a}$ has structure:
  \begin{align*}
    \namedGame{x:=\ast}{a},
    \textrm{ then }
    \AdvExAssign{\ptest{(\exists x \, S(\finalNode))}}& 
    \seq x:= \AdvExAssign{\otimes} \seq \AdvExAssign{!S(\finalNode)}.
    \quad
    \namedGame{(\namedGame{\gamma}{g};\namedGame{\delta}{d})}{a}, 
    \textrm{ then } \aruntimem{\namedGame{\gamma}{g}}{\modEnd{S}{S(d)}} \seq \aruntimem{\namedGame{\delta}{d}}{S}.\\
    \namedGame{\{x'=f(x)\ \&\ Q\}}{a},
    \textrm{ then }&
    \AdvExAssign{\ptest{\ddiamond{\alpha}{S(\finalNode)}}} \seq \{x'=f(x)\ \&\ Q\}^{\AdvExAssign{d}} \seq \AdvExAssign{!S(\finalNode)}.\\
    \namedGame{(\namedGame{\gamma}{g} \cup \namedGame{\delta}{d})}{a}, 
    \textrm{ then }&
    \AdvExChoice{\ptest{(S(g)\lor S(d))}} \seq \left(( \AdvExChoice{!S(g)} \seq \aruntimem{\namedGame{\gamma}{g}}{S}) \AdvExChoice{\cap} (\AdvExChoice{!S(d)} \seq \aruntimem{\namedGame{\delta}{d}}{S})\right).\\
    \namedGame{(\namedGame{\gamma}{g})^\ast}{a}, 
    \textrm{ then }& \AdvExLoop{\ptest{S(a)}} \seq \left(\AdvExLoop{!S(g)} \seq \aruntimem{\namedGame{\gamma}{g}}{\modEnd{S}{S(a)}} \seq \AdvExLoop{\ptest{S(a)}} \right)^{\AdvExLoop{\times}} \seq \AdvExLoop{!S(\finalNode)}.\\
    \namedGame{(\namedGame{\gamma}{g} \cap \namedGame{\delta}{d})}{a}, 
    \textrm{ then }& (\aruntimem{\namedGame{\gamma}{g}}{S}) \cap (\aruntimem{\namedGame{\delta}{d}}{S}). 
    \quad
    \namedGame{(\namedGame{\gamma}{g})^\times}{a}, 
    \textrm{ then } (\aruntimem{\namedGame{\gamma}{g}}{\modEnd{S}{S(a)}})^\times.\\
    \textrm{atomic and not}&\textrm{ controlled by Demon, i.e., }\\
    \alpha \in \{x:=e,& x:=\otimes, \ptest{Q}, !Q, \{x'=f(x)\ \&\ Q\}^d\}, \textrm{ then } \alpha.
  \end{align*}
  Monitoring a Demonic subvalue map is symmetric (see \rref{app:subvalue-externalization-demon}).
Every inductive subvalue map \(S\) is a sound runtime monitor, i.e., \(\models S(a) \limply \ddiamond{\aruntimem{\namedGame{\alpha}{a}}{S}}{S(\finalNode)}\) (\rref{thm:runtime-monitor}).

\begin{theoremE}[Subvalue maps are sound monitors][all end, text proof={}]
  \label{thm:runtime-monitor}
  For \dGL game $\namedGame{\alpha}{a}$ and subvalue map $S$, if $\avalid{\alpha}{a}{S}{\phi}$, then $\models S(a) \limply \langle \aruntimem{\namedGame{\alpha}{a}}{S}{\phi} \rangle \phi$.
  Dually, if $\dvalid{\alpha}{a}{S}{\phi}$ then $\models S(a) \limply [\druntimem{\namedGame{\alpha}{a}}{S}{\phi}] \phi$.
\end{theoremE}
\begin{proofE}
  The proof proceeds by induction on the structure of the game $\namedGame{\alpha}{a}$.
  We show this for the $\aruntimem{\namedGame{\alpha}{a}}{S}$ case; the dual case is analogous.
  If $\namedGame{\alpha}{a}$ has structure:
  \begin{itemize}
    \item atomic and Angel does not take any decision, i.e., $\alpha \in \{x:=e, x:=\otimes, \ptest{Q}, !Q, \{x'=f(x)\ \& \ Q\}^d\}$, then $\langle \aruntimem{\namedGame{\alpha}{a}}{S} \rangle \phi$ is $\langle \alpha \rangle \phi$.
    Because $\avalid{\alpha}{a}{S}{\phi}$, we know that $\models S(a) \limply \langle \alpha \rangle \phi$.
    \item $x:=*$, then $\langle \aruntimem{\namedGame{\alpha}{a}}{S} \rangle \phi$ is $\langle ?\langle\alpha\rangle\phi \seq x:=\otimes \seq !\phi \rangle \phi$.
    By the \irref{composed}, \irref{testb} and \irref{testd} axioms, this is equivalent to $\langle \alpha \rangle \phi \land \langle x:=\otimes \rangle(\phi \limply \phi)$.
    Because $\avalid{\alpha}{a}{S}{\phi}$, we know that $\models S(a) \limply \langle \alpha \rangle \phi$.
    By \irref{assignd}, $\langle x:=\otimes \rangle \top$ is $\forall x \top$ which is $\top$.
    \item $\{x'=f(x)\ \&\ Q\}$, then $\langle \aruntimem{\namedGame{\alpha}{a}}{S} \rangle \phi$ is $\langle  ?\langle\alpha\rangle\phi \seq \{x'=f(x)\ \&\ Q\}^d \seq !\phi\rangle \phi$.
    By the \irref{assignd}, \irref{testd} and \irref{testb} axioms, this is equivalent to $\langle \alpha \rangle \phi \land \langle \{x'=f(x)\ \&\ Q\}^d \rangle (\phi \limply \phi)$.
    Because $\avalid{\alpha}{a}{S}{\phi}$, we know that $\models S(a) \limply \langle \alpha \rangle \phi$.
    By \irref{evolveb} and \irref{duald}, $\langle \{x'=f(x)\ \&\ Q\}^d \rangle \top$ is $\forall t\geq 0 [x:=y(t)] \top$ where $y'(t)=f(y)$. By \irref{G}, $[x:=y(t)] \top$ is $\top$, so $\forall t\geq 0 \top$ is $\top$.
    \item $\namedGame{(\namedGame{\gamma}{g} \cup \namedGame{\delta}{d})}{a}$ then $\langle \aruntimem{\namedGame{\alpha}{a}}{S} \rangle \phi$ is $\langle ?S(g)\lor S(d) \seq \big((\aruntimem{\namedGame{\gamma}{g}}{S}) \cap (\aruntimem{\namedGame{\delta}{d}}{S})\big) \rangle \phi$.
    By the \irref{testd}, \irref{choiced}, \irref{testb} and \irref{composed} axioms, this is equivalent to $(S(g) \lor S(d)) \land (S(g) \limply \langle \aruntimem{\namedGame{\gamma}{g}}{S}\rangle \phi) \land (S(d) \limply \langle \aruntimem{\namedGame{\delta}{d}}{S}\rangle \phi)$.
    Because $\avalid{\alpha}{a}{S}{\phi}$, we know that $\models S(a) \limply S(g) \lor S(d)$, and further, $\avalid{\gamma}{g}{S}{\phi}$ and $\avalid{\delta}{d}{S}{\phi}$.
    By the inductive hypothesis, $\models S(g) \limply \langle \aruntimem{\namedGame{\gamma}{g}}{S} \rangle \phi$ and $\models S(d) \limply \langle \aruntimem{\namedGame{\delta}{d}}{S} \rangle \phi$.
    Thus, $\models S(a) \limply \langle \aruntimem{\namedGame{\alpha}{a}}{S} \rangle \phi$.
    \item $\namedGame{(\namedGame{\gamma}{g} \cap \namedGame{\delta}{d})}{a}$ then $\langle \aruntimem{\namedGame{\alpha}{a}}{S} \rangle \phi$ is $\langle \big((\aruntimem{\namedGame{\gamma}{g}}{S}) \cap (\aruntimem{\namedGame{\delta}{d}}{S})\big) \rangle \phi$.
    We want to show $S(a) \limply \langle (\aruntimem{\namedGame{\gamma}{g}}{S}) \cap (\aruntimem{\namedGame{\delta}{d}}{S}) \rangle \phi$.
    By the inductive hypothesis, $\models S(g) \limply \langle \aruntimem{\namedGame{\gamma}{g}}{S} \rangle \phi$ and $\models S(d) \limply \langle \aruntimem{\namedGame{\delta}{d}}{S} \rangle \phi$.
    By \irref{choiced} and \irref{M}, then, $\models (S(g) \land S(d)) \limply \langle (\aruntimem{\namedGame{\gamma}{g}}{S}) \cap (\aruntimem{\namedGame{\delta}{d}}{S}) \rangle \phi$.
    By the definition of a valid solution, $S(a) \limply S(g) \land S(d)$.
    By transitivity of implication, $\models S(a) \limply \langle (\aruntimem{\namedGame{\gamma}{g}}{S}) \cap (\aruntimem{\namedGame{\delta}{d}}{S}) \rangle \top$.
    \item $\namedGame{(\namedGame{\gamma}{g}\seq \namedGame{\delta}{d})}{a}$ then $\langle \aruntimem{\namedGame{\alpha}{a}}{S} \rangle \phi$ is $\langle (\aruntimem{\namedGame{\gamma}{g}}{\modEnd{S}{S(d)}}{}) \seq (\aruntimem{\namedGame{\delta}{d}}{S}{}) \rangle \phi$.
    After applying the \irref{composed} axiom, we want to show $S(a) \limply \langle \aruntimem{\namedGame{\gamma}{g}}{\modEnd{S}{S(d)}}{} \rangle \langle \aruntimem{\namedGame{\delta}{d}}{S}{} \rangle \phi$.
    By the inductive hypothesis, $\models S(g) \limply \langle \aruntimem{\namedGame{\gamma}{g}}{\modEnd{S}{S(d)}}{} \rangle S(d)$ and $\models S(d) \limply \langle \aruntimem{\namedGame{\delta}{d}}{S}{} \rangle \phi$.
    Since $\models S(d) \limply \langle \aruntimem{\namedGame{\delta}{d}}{S}{} \rangle \phi$, by \irref{M}, it remains to show that $\models S(a) \limply \langle \aruntimem{\namedGame{\gamma}{g}}{\modEnd{S}{S(d)}}{} \rangle S(d)$.
    Since $\avalid{\alpha}{a}{S}{\phi}$, we know that $\models S(a) \limply S(g)$.
    Thus, by transitivity of implication, $\models S(a) \limply \langle \aruntimem{\namedGame{\gamma}{g}}{\modEnd{S}{S(d)}}{} \rangle S(d)$, completing the proof.
    \item $\namedGame{(\namedGame{\gamma}{g})^*}{a}$ then $\langle \aruntimem{\namedGame{\alpha}{a}}{S}{} \rangle \phi$ is
    $\langle ?(S(a)) \seq (!S(g) \seq \aruntimem{\namedGame{\gamma}{g}}{\modEnd{S}{S(a)}}{} \seq ?S(a))^\times \seq !S(\finalNode) \rangle \phi$.
    By the \irref{testd}, \irref{testb} and \irref{composed} axioms, this is equivalent to

    \((S(a)) \land \langle (!S(g) \seq \aruntimem{\namedGame{\gamma}{g}}{\modEnd{S}{S(a)}}{} \seq ?S(a))^\times \rangle (\phi \limply \phi)\).

    Thus, we want to show $\models S(a) \limply (S(a)) \land \langle (!S(g) \seq \aruntimem{\namedGame{\gamma}{g}}{\modEnd{S}{S(a)}}{} \seq ?S(a))^\times \rangle \top$.

    We want to show the that $S(a) \limply \langle (!S(g) \seq \aruntimem{\namedGame{\gamma}{g}}{\modEnd{S}{S(a)}}{} \seq$ $?S(a))^\times \rangle \top$ holds.
    We prove this using the \irref{loop} rule.
    $\top$ serves as an invariant of this loop game since it holds initially ($\models \top$), implies the postcondition ($\models \top \limply \top$), and is inductive ($\langle !S(g) \seq \aruntimem{\namedGame{\gamma}{g}}{\modEnd{S}{S(a)}} \seq ?S(a) \rangle \top$) as argued next.
    Since $\avalid{\alpha}{a}{S}{\phi}$, so $\avalid{\gamma}{g}{\modEnd{S}{S(a)}}{}$,
    By the inductive hypothesis, $\models S(g) \limply \langle \aruntimem{\namedGame{\gamma}{g}}{\modEnd{S}{S(a)}}{} \rangle (S(a))$.
    Applying the axioms \irref{testb}, \irref{composed} and \irref{testd}, this is equivalent to
    
    $\models \langle !S(g) \seq \aruntimem{\namedGame{\gamma}{g}}{\modEnd{S}{S(a)}} \seq ?S(a) \rangle \top$
    
    showing that $\top$ is an inductive invariant.
    \item $\namedGame{(\namedGame{\gamma}{g})^\times}{a}$ then $\langle \aruntimem{\namedGame{\alpha}{a}}{S}{} \rangle \phi$ is $\langle (\aruntimem{\namedGame{\gamma}{g}}{S}{})^\times \rangle \phi$.
    
    We want to show
    $\models S(a) \limply $ $\langle (\aruntimem{\namedGame{\gamma}{g}}{\modEnd{S}{S(a)}}{})^\times \rangle \phi$.

    Since $\avalid{\alpha}{a}{S}{\phi}$, it must be the case that $\models S(a) \limply (S(g) \land \phi)$ and $\avalid{\gamma}{g}{S}{S(a)}$.
    By the inductive hypothesis, $\models S(g) \limply \langle \aruntimem{\namedGame{\gamma}{g}}{\modEnd{S}{S(a)}}{} \rangle S(a)$.
    We apply the \irref{loop} rule to prove the conclusion
    $\models S(a) \limply$ $\langle (\aruntimem{\namedGame{\gamma}{g}}{S}{})^\times \rangle \phi$.
    $S(a)$ serves as an invariant of this loop.
    It holds initially.
    It is inductive per the inductive hypothesis (using weakening and $\models S(a) \limply S(g)$), and implies the postcondition ($\models S(a) \limply \phi$).
  \end{itemize}
\end{proofE}
\end{remark}

In subsequent sections, 
\rref{def:ordering} provides a partial ordering on inductive subvalue maps where more permissive maps are better, and \rref{thm:mpc-maximal} shows that the dominant \dGL strategy corresponds to the maximally permissive inductive subvalue map.
Finally, \rref{sec:solving-algorithm} uses \rref{def:local-envelope-conditions} to define a natural, symbolic execution based algorithm that synthesizes subvalue maps that generate policy decisions in polynomial time.

\section{Maximal Solution}
\label{sec:ordering-solutions}
For a given game and winning condition, there are many possible correct inductive subvalue maps.
Most useful are the \emph{more permissive} subvalue maps that allow as many control actions as possible.
More permissive subvalue maps leave more control options open, permitting more aggressive, efficient control strategies.
We formalize subvalue map permissiveness with a partial ordering where more permissive subvalue maps are greater than less permissive subvalue maps.
A natural candidate for ordering the subvalue maps of game $\namedGame{\alpha}{a}$ would be to say that $S$ is at least as permissive as $S'$ ($S\sqsupseteq S'$) when for each subgame $b$ in $\nodes{\namedGame{\alpha}{a}}$, \emph{every state in subregion $S'(b)$ is also within subregion $S(b)$}, i.e., $\models S'(b) \limply S(b)$, so that $S(b)$ locally \emph{permits} the agent to play into subgame $b$ in all the states that $S'(b)$ does.
However, this ordering fails to order cases like the following.
Let $S'$ set the winning region for the overall game to be empty ($S'(a)=\bot$), meaning that Angel is never allowed to play at all.
For \emph{any} other inductive subvalue map $S$, a good ordering should indicate that $S\sqsupseteq S'$ since $S$ cannot possibly be less permissive,
but the proposed na\"ive ordering does not.
The winning subregion for some later subgame $b \neq a$ can be such that the check $\models S'(b) \limply S(b)$ fails (see \rref{app:counterexample} for details).

An improved ordering makes the following modified check for every subgame $b$.
$S\sqsupseteq S'$ when, \emph{amongst the states that are reachable at subgame $b$ while playing per subvalue map $S'$}, every state that is in subregion $S'(b)$ should also be within subregion $S(b)$.
To express this logically, we first characterize the states that are reachable at subgame $b$ while playing per subvalue map $S'$.

The game prefix (\rref{def:execution-prefix}) of subgame $b$ in the game $\namedGame{\alpha}{a}$, written $\prefix{b}{\namedGame{\alpha}{a}}$, features \emph{all} possible behaviors of game $\alpha$ that can happen \emph{before} an occurrence of subgame $b$ within $\namedGame{\alpha}{a}$.
To characterize all the states reachable at subgame $b$ while following the policy induced by $S'$, we use the universal projection of $S'$ onto the game prefix $\asubst{(\prefix{b}{\namedGame{\alpha}{a}})}{S'}{S'(b)}$.
\rref{def:ordering} then compares the winning subregions of $S$ and $S'$ after running the projected game to filter out unreachable states.

\begin{definition}[Permissiveness ordering]
  \label{def:ordering}
  For two inductive Angelic subvalue maps $S$ and $S'$ for the game $\namedGame{\alpha}{a}$, $S$ is at least as good as $S'$, written $S \sqsupseteq S'$, iff for each subgame $b$ of $\namedGame{\alpha}{a}$,
  $\models \langle\asubst{\prefix{b}{\namedGame{\alpha}{a}}}{S'}{S'(b)} \rangle(S'(b) \limply S(b))$ and \(\models \ddiamond{\asubst{\namedGame{\alpha}{a}}{S'}{S'(b)}}{(S'(\finalNode) \limply S(\finalNode))}\).
  Dually, for inductive Demonic subvalue maps, $S \sqsupseteq S'$ iff for subgame $b$ of $\alpha$, $\models [\dsubst{(\prefix{b}{\namedGame{\alpha}{a}})}{S'}{S'(b)}](S'(b) \limply S(b))$ and \(\models [\dsubst{\namedGame{\alpha}{a}}{S'}{S'(b)}](S'(\finalNode) \limply S(\finalNode))\).
\end{definition}

\begin{remark}
For every game and winning condition combination, the model predictive subvalue map (\rref{def:mpc-sol}) is a maximally permissive solution (\rref{thm:mpc-maximal}).
As \rref{def:mpc-sol} always constructs an optimal subvalue map, it optimally solves the \dGL control envelope synthesis problem, achieving completeness for symbolic hybrid games synthesis.
\end{remark}

\begin{theoremE}[Maximal inductive subvalue map][restate, text proof={}]
  \label{thm:mpc-maximal}
  Amongst the inductive Angelic subvalue maps for game $\namedGame{\alpha}{a}$ compatible with winning condition $\phi$, the model predictive Angelic map given by \rref{def:mpc-sol} is maximal under the ordering of \rref{def:ordering}.
  That is, for all $S$ such that $\avalid{\alpha}{a}{S}{\phi}$, model predictive Angelic subvalue map $S'$ satisfies $S' \sqsupseteq S$.
  Dually, the model predictive Demonic subvalue map (\rref{def:mpc-sol}) is maximal amongst inductive Demonic subvalue maps per the ordering of \rref{def:ordering}.
\end{theoremE}
\begin{proofE}
  \rref{lem:mpc-valid} shows that the MPC solution is valid.
  Next we show that it is maximal.
  We first show that $S' \succsim S$ for all $S$ such that $\avalid{\alpha}{a}{S}{\phi}$ under the weaker ordering of \rref{def:weak-ordering}.
  A symmetric proof applies when $\avalid{\alpha}{a}{S}{\phi}$ is replaced by $\dvalid{\alpha}{a}{S}{\phi}$.
  The proof uses structural induction.

  For \(\finalNode\), if $\avalid{\alpha}{a}{S}{\phi}$ then $\models S(\finalNode) \limply \phi$.
  But for the MPC solution, $S'(\finalNode)$ is $\phi$,
  so $\models S(a) \limply S'(a)$.

  When $\namedGame{\alpha}{a}$ is atomic, i.e. $\alpha \in \{x:=e, x:=*, \ptest{f}, !f, \{x'=f(x)\ \&\ Q\}, \{x'=f(x)\ \& \ Q\}^d\}$, the result is immediate: according to \rref{def:local-envelope-conditions}, if $\avalid{\alpha}{a}{S}{\phi}$ then $\models S(a) \limply \langle \alpha \rangle \phi$.
  But per the definition of the MPC solution, $S'(a)$ is $\langle \alpha \rangle \phi$.
  Thus, $\models S(a) \limply S'(a)$.

  For the recursive cases, if $\namedGame{\alpha}{a}$ has the structure:
  \begin{enumerate}
    \item $\namedGame{\gamma}{g}\cup\namedGame{\delta}{d}$, then observe that restrictions $\restrict{S'}{\nodes{\gamma}}$ and $\restrict{S'}{\nodes{\delta}}$ of $S'$ are the MPC solutions for Angel win condition $\phi$ for games $\namedGame{\gamma}{g}$ and $\namedGame{\delta}{d}$ respectively.
    By the inductive hypothesis, for all subgames $b \in \nodes{\namedGame{\gamma}{g}}$, $S(b) \limply S'(b)$.
    Similarly, for all subgames $b \in \nodes{\namedGame{\delta}{d}}$, $S(b) \limply S'(b)$.
    The only remaining subgame in $\nodes{\alpha}$ is $a$.
    Per the inductive hypothesis, $\models S(g)\limply S'(g)$ and $\models S(d) \limply S'(g)$.
    Taking the disjunction of these implications, we get $\models S(g)\lor S(d) \limply S'(g)\lor S'(d)$.
    Because $S$ is valid, $\models S(a) \limply S(g) \lor S(d)$.
    By transitivity of implication, $\models S(a) \limply S'(g)\lor S'(d)$.
    By definition, $S'(a)$ is $S'(g)\lor S'(d)$.
    Thus we can conclude that $\models S(a) \limply S'(a)$.
    \item $\namedGame{\gamma}{g}\cap\namedGame{\delta}{d}$, then observe that restrictions $\restrict{S'}{\nodes{\gamma}}$ and $\restrict{S'}{\nodes{\delta}}$ of $S'$ are the MPC solutions for Angel win condition $\phi$ for games $\namedGame{\gamma}{g}$ and $\namedGame{\delta}{d}$ respectively.
    By the inductive hypothesis, for all nodes $b \in \nodes{\namedGame{\gamma}{g}}$, $S(b) \limply S'(b)$.
    Similarly, for all subgames $b \in \nodes{\namedGame{\delta}{d}}$, $S(b) \limply S'(b)$.
    The only remaining subgame in $\nodes{\alpha}$ is $a$.
    Per the inductive hypothesis, $\models S(g)\limply S'(g)$ and $\models S(d) \limply S'(g)$.
    Taking the conjunction of these implications, we get $\models S(g)\land S(d) \limply S'(g)\land S'(d)$.
    Because $S$ is valid, $\models S(a) \limply S(g) \land S(d)$.
    By transitivity of implication, $\models S(a) \limply S'(g)\land S'(d)$.
    By definition, $S'(a)$ is $S'(g)\land S'(d)$.
    Thus we can conclude that $\models S(a) \limply S'(a)$.
    \item $\namedGame{\gamma}{g};\namedGame{\delta}{d}$, then observe that $\restrict{S'}{\nodes{\delta}}$ is the MPC solutions for Angel win condition $\phi$ for game $\namedGame{\delta}{d}$.
    By the inductive hypothesis, for all subgames $b \in \nodes{\namedGame{\delta}{d}}$, $\models S(b) \limply S'(b)$.
    Now, let $S''$ be the MPC solution for Angel win condition $S(d)$ for game $\namedGame{\gamma}{g}$.
    By the inductive hypothesis, for all subgames $b \in \nodes{\namedGame{\gamma}{g}}$, $\models S(b) \limply S''(b)$.
    By definition, for any subgame $b \in \nodes{\namedGame{\gamma}{g}}$, $S''(b)$ is $\langle \fwd{b}{\namedGame{\gamma}{g}} \rangle S(d)$ while $S'(b)$ is $\langle \fwd{b}{ \namedGame{\gamma}{g}} \rangle \phi$.
    As argued already, by the inductive hypothesis applied to game $\delta$, $\models S(d) \limply S'(d)$.
    By the \dGL monotonicity rule \irref{M}, $\models \langle \fwd{b}{\namedGame{\gamma}{g}}\rangle S(d) \limply \langle \fwd{b}{\namedGame{\gamma}{g}} \rangle S'(d)$.
    Thus, $\models S''(b) \limply S'(b)$.
    By the transitivity of implication, $\models S(b) \limply S'(b)$.
    We have shown the correct implication for all the nodes in $\nodes{\namedGame{\gamma}{g}}$ and $\nodes{\namedGame{\delta}{d}}$.
    The only subgame remaining is $a$.
    Because $\avalid{\alpha}{a}{S}{}$, $\models S(a) \limply S(g)$.
    We have shown that $\models S(g) \limply S'(g)$.
    By the definition of the MPC solution, $S'(a)$ is $S'(g)$.
    Thus, $\models S(a) \limply S'(a)$.
    \item $\namedGame{\gamma}{g}^*$, then because $\avalid{(\gamma^*)}{a}{S}{}$, $\models S(a) \limply \langle \aproj{\namedGame{\alpha}{a}}{S}{} \rangle \phi$.
    Further, by \rref{lem:subvalue-projection-refinement}, $\models S(a) \limply \langle \gamma^* \rangle \phi$.
    But by the definition of the MPC solution, $S'(a)$ is $\langle \gamma^* \rangle \phi$.
    By transitivity of implication, $\models S(a) \limply S'(a)$.
    The subgames that remain besides $a$ all belong to $\gamma$.
    Consider $S''$, the MPC solution for Angel win condition $S(a)$ for game $\namedGame{\gamma}{g}$.
    By the inductive hypothesis, for every subgame $b$ in $\gamma$, $\models S(b) \limply S''(b)$.
    Now, we argue that $\models S''(b) \limply S'(b)$, so that by transitivity of implication, $\models S(b) \limply S'(b)$ completing the proof.
    Per the definition of the MPC solution, $S''(b)$ is $\langle \fwd{b}{\namedGame{\gamma}{g}}\rangle S(a)$ while $S'(b)$ is $\langle \fwd{b}{\namedGame{\gamma}{g}}\rangle \phi$.
    By the monotonicity rule \irref{M}, since $\models S(a) \limply \phi$, we have $\models \langle \fwd{b}{\namedGame{\gamma}{g}}\rangle S(a) \limply \langle \fwd{b}{\namedGame{\gamma}{g}}\rangle \phi$.
    Thus, $\models S''(b) \limply S'(b)$.
    \item $\namedGame{\gamma}{g}^\times$, then because $\avalid{(\gamma^\times)}{a}{S}{}$, per \rref{thm:consistancy}, $\models S(a) \limply \langle \gamma^\times\rangle \phi$.
    By the definition of the MPC solution, $S'(a)$ is $\langle \gamma^\times \rangle \phi$.
    So, $\models S(a) \limply S'(a)$.
    The nodes that remain besides $a$ all belong to $\nodes{\namedGame{\gamma}{g}}$.
    Consider $S''$, the MPC solution for Angel win condition $S(a)$ for game $\namedGame{\gamma}{g}$.
    By the inductive hypothesis, for every subgame $b$ in $\gamma$, $\models S(b) \limply S''(b)$.
    Now, we argue that $\models S''(b) \limply S'(b)$, so that by transitivity of implication, $\models S(b) \limply S'(b)$, completing the proof.
    Per the definition of the MPC solution, $S''(b)$ is $\langle \fwd{b}{\namedGame{\gamma}{g}}\rangle S(a)$ while $S'(b)$ is $\langle \fwd{b}{\namedGame{\gamma}{g}}\rangle S'(a)$.
    Since $\models S(a) \limply S'(a)$, by the \dGL monotonicity rule \irref{M}, $\models \langle \fwd{b}{\namedGame{\gamma}{g}}\rangle S(a) \limply \langle \fwd{b}{\namedGame{\gamma}{g}}\rangle S'(a)$.
    Thus, $\models S''(b) \limply S'(b)$.
  \end{enumerate}
  We have shown that $S'\succsim S$ for all $S$ such that $\avalid{\alpha}{a}{S}{}$ under the weaker ordering of \rref{def:weak-ordering}.
  From \rref{lem:subsumption}, it also follows that $S'\sqsupseteq S$ under the stronger ordering of \rref{def:ordering}.
  Thus, $S'$ is maximal.
\end{proofE}

\section{Synthesis Framework for Inductive Subvalue Map with Efficient Checks}
\label{sec:solving-algorithm}

To check whether a given control decision lies within the policy of a subvalue map, \rref{def:policy} checks whether the resulting state is within a subvalue (e.g., for \(\namedGame{(\namedGame{\delta}{d} \cup \namedGame{\gamma}{g})}{b}\), the policy checks \(\sigma \models S(d)\) and \(\sigma \models S(g)\)).
This section synthesizes inductive subvalue maps where state checks are efficient, for quick responses in real systems.
In the maximal solution from the previous section (\rref{def:mpc-sol}), such checking is not efficient because of subvalues that are \dGL formulas, for which checking is undecidable \cite{DBLP:journals/tocl/Platzer15}.
In this section we restrict subvalues to formulas consisting of equations/inequalities of real polynomials along with propositional logic connectives $(\land, \lor, \neg, etc.)$ but no quantifiers or modalities, henceforth, \emph{propositional real arithmetic} formulas (\(\props\)).
For these, checking the truth value at a state takes polynomial time.
\emph{Given a \dGL game and winning condition in \(\props\), this section synthesizes inductive subvalue maps where subvalues are in \(\props\)}.

Synthesizing such a map is challenging because the subvalues must remain \emph{mutually compatible}, indicating only the winning regions achievable while following the other subvalues.
That is, they must follow the inductive subvalue map conditions \rref{def:local-envelope-conditions}, where there are circular dependencies between subvalues.
Consider, for example, the case of Demonic loop \(\namedGame{(\namedGame{\gamma}{g})^\times}{a}\).
Subvalue \(S(a)\) must imply \(S(g)\), but \(S(g)\) depends on \(S(a)\) via condition \(\avalid{\gamma}{g}{\modEnd{S}{S(a)}}{}\).
In particular, starting from the optimal subvalue map construction (\rref{def:mpc-sol}) and changing individual \dGL formulas to under-approximations in \(\props\) would not alone work because cyclically interdependent soundness conditions would need to be reestablished.
\begin{enumerate*}
  \item We address the challenge of interdependent conditions by isolating their effect with a hybrid games analog of loop invariants.
  \item To soundly lower formulas to \(\props\) while maintaining permissiveness, we use a \emph{predicate transformer} based approach that given a subgame and winning condition in \(\props\) computes subvalues in \(\props\).
\end{enumerate*}
This lets us present a backwards symbolic execution \cite{10.1145/360933.360975,AUMANN19956} based \emph{algorithmic framework} to compute inductive subvalue maps with subvalues in \(\props\).

\begin{remark}
  Weakest precondition calculi \cite{10.1145/360933.360975} generate weak \emph{preconditions}.
  To reason about loops, they generally \emph{check} loop invariants, and rely on external loop invariant generation.
  Analogously, our framework generates permissive \emph{subvalues}.
  To reason about loops, it \emph{provides checks for invariants of Angel and Demon loop subvalue maps}.
  It is \emph{parametric} in the generation of such invariants.
\end{remark}

We define predicate transformers over subgames such that \emph{given a postcondition in }\(\props\), they \emph{output a subvalue in }\(\props\) to recursively produce an entire subvalue map with subvalues in \(\props\).
The transformers for the loop-free fragment of \dGL are derived from existing work \cite{DBLP:conf/tacas/KabraLMP24}, which we discuss later in this section.
For loops, our framework uses the approach of heuristically generating \emph{invariant candidates} and checking their correctness, which has worked well for hybrid systems \cite{DBLP:journals/fmsd/SogokonMTCP22,DBLP:journals/fmsd/PlatzerC09}.
We characterize the necessary and sufficient conditions for a sound invariant (for Angel loops, \(\models \invexpr\limply \langle \aproj{\namedGame{\alpha}{a}}{S}{\phi} \rangle \phi\), for Demon loops, \(\models \invexpr \limply S(g) \land \phi\)).
The framework generates invariant candidate $\invexpr$ (details in (\rref{sec:oracle-implementation})), computes the subvalue map of the loop body, and then checks whether the invariant and loop body map are compatible\footnote{%
  Checking Angel loop invariants ($\models \invexpr \limply \langle \aproj{\namedGame{\alpha}{a}}{S}{\phi} \rangle \phi$) can be difficult to automate using traditional variant generation techniques.
  We use refinement based proofs in practice (\rref{app:refinements}).
}.
Proposed invariants must be in $\props$.

\rref{alg:solving} puts together the predicate transformers to define a recursive solving function $\asolve(\namedGame{\alpha}{a}, \phi)$ that computes an inductive Angelic subvalue map with winning regions in $\props$ for game $\namedGame{\alpha}{a}$ that is compatible with winning condition $\phi \in \props$.
Operator $\uplus$ computes the disjoint union of subvalue maps.
\(S \setminus \finalNode\) indicates the subvalue map $S$ with the mapping for \(\finalNode\) removed.
Such deletion is sometimes necessary to ensure that two maps to which we apply the disjoint union operator are genuinely disjoint.
The algorithm is parametrized on \emph{invariant candidate} generation function $\invgen$.
$\dsolve$ similarly computes an inductive Demonic subvalue map (\rref{app:solving-demon}).

\begin{algorithm}[h]
  \caption{Inductive subvalue map synthesis framework}\label{alg:solving}
  \begin{algorithmic}[1]
    \Function{$\asolve$}{$\namedGame{\alpha}{a}, \phi$} \Comment{Parametric Input: \invgen}
    \State \(S_e \gets \{ \finalNode \mapsto \phi \}\);
    \If{$\alpha \in \{x:=e, x:=*, x:=\otimes, \ptest{Q}, !Q, \{x'=f(x) \& Q\}, \{x'=f(x) \& Q\}^d\}$}
    \State \Return $\{ a \mapsto \exec(\langle \alpha \rangle \phi) \} \uplus S_e $
    \ElsIf{$\alpha = \namedGame{(\namedGame{\gamma}{g} \cup \namedGame{\delta}{d})}{a}$}
    \State $S_1 \gets \asolve(\namedGame{\gamma}{g}, \phi)$;
      $S_2 \gets \asolve(\namedGame{\delta}{d}, \phi)$;
      \Return $S_1 \uplus (S_2 \setminus \finalNode) \uplus \{ a \mapsto S_1(g)\lor S_2(d)\}$
    \ElsIf{$\alpha = \namedGame{(\namedGame{\gamma}{g} \cap \namedGame{\delta}{d})}{a}$}
    \State $S_1 \gets \asolve(\namedGame{\gamma}{g}, \phi)$;
      $S_2 \gets \asolve(\namedGame{\delta}{d}, \phi)$;
      \Return $S_1 \uplus (S_2 \setminus \finalNode) \uplus \{ a \mapsto S_1(g)\land S_2(d)\}$
    \ElsIf{$\alpha = \namedGame{(\namedGame{\gamma}{g};\namedGame{\delta}{d})}{a}$}
    \State $S_1 \gets \asolve(\namedGame{\delta}{d}, \phi)$;
      $S_2 \gets \asolve(\namedGame{\gamma}{g}, S_1(d))$;
      \Return $S_1 \uplus (S_2 \setminus \finalNode) \uplus \{a \mapsto S_2(g)\}$
    \ElsIf{$\alpha = \namedGame{(\namedGame{\gamma}{g})^*}{a}$}
      \While {true}
      \State \invexpr $\gets \invgen(\langle(\namedGame{\gamma}{g})^*\rangle\phi)$;
        $S \gets \asolve(\namedGame{\gamma}{g}, \invexpr \lor \phi)$
      \If{$\models \invexpr\limply \langle \aproj{\namedGame{\alpha}{a}}{S}{\phi} \rangle \phi$}\label{line:alg-angel-loop}
      \Return $(S \setminus \finalNode) \uplus \{a \mapsto \invexpr \} \uplus S_e $
      \EndIf
      \EndWhile
    \Else{ $\alpha = \namedGame{(\namedGame{\gamma}{g})^\times}{a}$}
      \While {true}
      \State \invexpr $\gets \invgen(\langle(\namedGame{\gamma}{g})^\times\rangle\phi)$;
        $S \gets \asolve(\namedGame{\gamma}{g}, \invexpr)$
      \If{$\models \invexpr \limply S(g) \land \phi$}\label{line:alg-demon-loop}
      \Return $(S \setminus \finalNode) \uplus \{a \mapsto \invexpr \} \uplus S_e $
      \EndIf
      \EndWhile
    \EndIf
    \EndFunction
  \end{algorithmic}
\end{algorithm}

$\asolve$ recursively computes the subvalue for each subgame.
The terminal winning subregion \(S(\finalNode)\) is always set to the target winning condition $\phi$.
In the atomic cases, there is only one subgame, and thus only one additional subvalue to compute.
Here, the predicate transformer sets the subvalue per \dGL winning region semantics (e.g., for $\namedGame{x:=e}{a}$, the winning subregion is $\langle x:=e \rangle \phi$).
But $\langle \alpha \rangle \phi$ is a modal \dGL formula which needs to be further simplified to a \(\props\) formula (in this case, to $\phi(x \mapsto e)$, i.e., the formula $\phi \in \props$ with every occurrence of variable $x$ substituted by $e$).
This simplification is written as $\exec(\langle \alpha \rangle \phi)$.
Existing work \cite{DBLP:conf/tacas/KabraLMP24} shows how to implement $\exec$\footnote{%
 The function $\reduce(\ddiamond{\alpha}{\phi}, A)$ \cite{DBLP:conf/tacas/KabraLMP24} produces a $\props$ formula that is equivalent to $\ddiamond{\alpha}{\phi}$ under assumption formula $A$. $\exec(\ddiamond{\alpha}{\phi})$ is $\reduce(\ddiamond{\alpha}{\phi}, \top)$.}.
While $\exec(\langle \alpha \rangle \phi)$ usually computes an exact expression in $\props$, this is sometimes not possible for subgames with complicated ODEs.
In such situations, a conservative, sound approximation is made such that $\models\exec(\langle \alpha \rangle \phi) \limply \langle \alpha \rangle \phi$, i.e. the simplified expression is true only in the states where the original expression was true.
\rref{thm:solve-valid} shows that the solving framework is robust to such conservative approximations, producing inductive subvalue maps despite them.
In the compositional cases, e.g., $\namedGame{(\namedGame{\gamma}{g}\cup\namedGame{\delta}{d})}{a}$, $\asolve$ first computes the subvalues of the composing subgames (in this case $S_1$ for $\namedGame{\gamma}{g}$ and $S_2$ for $\namedGame{\delta}{d}$), and then the subvalue for the outer subgame (in the case, $S(g) \lor S(d)$), following the rules of \(\exec\) \cite{DBLP:conf/tacas/KabraLMP24}.
For the loop cases, $\asolve$ generates invariants until it finds one that passes the check.

\rref{fig:overview} shows an example of computing a subvalue map using $\asolve$.
As before, label $i$ denotes the subgame with winning subregion $\phi_i$.
At the top level, subgame 1 is a Demon loop.
For this structure, after setting \(\preci{\finalNode}\), in second step $\AlgGuess \preci{1}$, $\asolve$ guesses invariant $\invexpr$ using a \emph{refinement} heuristic (discussed in \rref{sec:oracle-implementation}), computing answer $x>0 \land (v\geq 0 \lor a>0)$.
This invariant is used to recursively compute an inductive Angelic subvalue map for the loop body consisting of subgame 9.
A final $\AlgCheck$ step ensures that $\models \psi \limply \phi_9 \land x>0$, i.e., the invariant is compatible with the inner subvalue map so Angel won't get stuck regardless of whether Demon runs the loop or exits.

\begin{theoremE}[Sound solving][restate, text proof={}]
\label{thm:solve-valid}
$S := \asolve(\namedGame{\alpha}{a}, \phi)$ is an inductive Angelic subvalue map for game $\namedGame{\alpha}{a}$ ($\avalid{\alpha}{a}{S}{\phi}$) with all subvalues in \(\props\), compatible with Angel winning condition $\phi \in \props$.
Dually, $S := \dsolve(\namedGame{\alpha}{a}, \phi)$ is an inductive Demonic subvalue map with all subvalues in \(\props\), compatible with Demon winning condition \(\phi\).
\end{theoremE}
\begin{proofE}
  \label{proof:solve-valid}
  The proof is by structural induction on $\namedGame{\alpha}{a}$.
  We show the proof for $\asolve$, the proof for $\dsolve$ is symmetric.
  \(S(\finalNode)\) is compatible with \(\phi\) because $S(\finalNode)$ is set to $\phi$.
  In the case where $\namedGame{\alpha}{a}$ is atomic, i.e., when $\alpha\in \{x:=e, x:=*, \ptest{f}, !f, \{x'=f(x) \& Q\}, \{x'=f(x) \& Q\}^d\}$, the result is immediate.
  According to \rref{def:local-envelope-conditions}, $\avalid{\alpha}{a}{S}{\phi}$ exactly when $\models S(a) \limply \langle \alpha \rangle \phi$.
  But per the definition of $\asolve$, $S(a)$ is $\exec(\langle \alpha \rangle \phi)$ which by definition implies $\langle \alpha \rangle \phi$.
  
  In the compositional cases, we use structural induction.
  If $\namedGame{\alpha}{a}$ has the structure:
  \begin{itemize}
    \item $\namedGame{(\namedGame{\gamma}{g}\cup\namedGame{\delta}{d})}{a}$, then let $S_1:=\asolve(\namedGame{\gamma}{g}, \phi)$ and $S_2:=\asolve(\namedGame{\delta}{d}, \phi)$.
    $S$ is the disjoint union of $S_1$ and $S_2$ extended by mapping $a$ to $S_1(g)\lor S_2(d)$.
    We must show that $\models S(a) \limply S(g) \lor S(d) \textrm{ and }
    \avalid{\gamma}{g}{S}{\phi} \textrm{ and } \avalid{\delta}{d}{S}{\phi}$.
    The first implication is immediate as $S(g)=S_1(g)$ and $S(d)=S_2(d)$ by the construction of $S$.
    The second and third formulas hold by the inductive hypothesis.
    \item $\namedGame{(\namedGame{\gamma}{g} \cap \namedGame{\delta}{d})}{a}$, then let $S_1:=\asolve(\namedGame{\gamma}{g}, \phi)$ and $S_2:=\asolve(\namedGame{\delta}{d}, \phi)$.
    $S$ is the disjoint union of $S_1$ and $S_2$ extended by mapping $a$ to $S_1(g)\land S_2(d)$.
    We must show that $\models S(a) \limply S(g) \land S(d) \textrm{ and } \avalid{\gamma}{g}{S}{\phi} \textrm{ and } \avalid{\delta}{d}{S}{\phi}$.
    The first implication is immediate as $S(g)=S_1(g)$ and $S(d)=S_2(d)$ by the construction of $S$.
    The second and third formulas hold by the inductive hypothesis.
    \item $\namedGame{(\namedGame{\gamma}{g} \seq \namedGame{\delta}{d})}{a}$, then let $S_1:=\asolve(\namedGame{\delta}{d}, \phi)$ and $S_2:=\asolve(\namedGame{\gamma}{g}, S_1(d))$.
    $S$ is the disjoint union of $S_1$ and $S_2$ extended by mapping $a$ to $S_2(g)$.
    We must show that $\models S(a) \limply S(g) \textrm{ and } \avalid{\gamma}{g}{S}{\phi} \textrm{ and } \avalid{\delta}{d}{S}{\phi}$.
    The first implication is immediate as $S(g)=S_2(g)$ by the construction of $S$.
    The second and third formulas hold by the inductive hypothesis.
    \item $\namedGame{(\namedGame{\gamma}{g})^*}{a}$, 
    the let $S':=\asolve(\namedGame{\gamma}{g}, \invexpr \lor \phi)$.
    $S$ is constructed by extending $S'$ by mapping $a$ to $\invexpr$.
    Additionally, per the side condition $\models \invexpr \limply \langle \aproj{\namedGame{\alpha}{a}}{{S'}}{\phi} \rangle \phi$.
    We must show $S(a) \limply \langle \aproj{\namedGame{\alpha}{a}}{S}{\phi} \rangle \phi$ and $\avalid{\gamma}{g}{S}{S(a)\lor\phi}$.
    The first implication follows from the side condition along with the fact that $S(g)=S'(g)$ by the construction of $S$.
    The second formula holds by the inductive hypothesis as it applies to $S'$.
    \item $\namedGame{(\namedGame{\gamma}{g})^\times}{a}$,
    then let $S':=\asolve(\namedGame{\gamma}{g}, \invexpr)$.
    $S$ is constructed by extending $S'$ by mapping $a$ to $\invexpr$.
    Additionally, per the side condition $\models \invexpr \limply S(g) \land \phi$.
    We must show that $ S(a) \limply \phi \land S(g)$, and $\avalid{\gamma}{g}{S}{S(a)}$.
    The first implication follows from the side condition along with the fact that $S(g)=S'(g)$ by the construction of $S$.
    The second formula holds by the inductive hypothesis as it applies to $S'$.
  \end{itemize}
\end{proofE}

\subsection{Invariant Generation}
\label{sec:oracle-implementation}

The subvalue map synthesis framework of \rref{alg:solving} provides the conditions to check if a candidate formula is a sound invariant.
Any combination of invariant generation heuristics \cite{DBLP:journals/fmsd/PlatzerC09,DBLP:journals/fmsd/SogokonMTCP22,DBLP:conf/tacas/KabraLMP24} can be used to generate candidate invariants.
To solve complicated games where no techniques apply, we propose a \emph{game rewriting} framework.
This framework uses a library of case-based syntactic transformations to rewrite the loop game $\alpha^*$ or $\alpha^\times$ as a new simpler game $\beta$ that emulates aspects of its behavior.
Then, if $\beta$ is loop-free, $\exec(\langle \beta \rangle \phi)$ becomes an invariant guess, otherwise other invariant generation heuristics are tried on $\beta$.
Checks in \rref{alg:solving}, \rref{line:alg-angel-loop} and \rref{line:alg-demon-loop} recover soundness regardless.
For example, the step $\AlgGuess \preci{1}$ in \rref{fig:overview} shows a rewrite when computing an invariant of the outermost loop ($\preci{1}$).
The rewritten game (on the second line) emulates inner loop $(v:=v+a)^*$ by letting Angel decide in advance how many loop iterations to run and assign her decision freely to new variable $n$.
Then, setting $v := v + an$ conservatively estimates the effect of running the loop $\left\lceil n \right\rceil$ times (when $n$ is fractional, $v$ ends up smaller, and thus conservatively less favorable to Angel than after $\left\lceil n \right\rceil$ loop iterations).
The rewrite also eliminates the outer loop using a generalized form of one-shot refinement \cite{DBLP:conf/tacas/KabraLMP24} where running the ODE emulates the worst case behavior of multiple loop iterations.
In general, different rewrite rules are tried heuristically based on the shape of a problem until an answer is found.
New rewrite rules can be added to the library as needed.
\rref{app:rewrite-heuristics} shows the rewrite rules that solve the evaluation problems (\rref{sec:evaluation}).

\subsection{Recovering Solution Optimality}
\label{sec:solution-optimality}
\rref{alg:solving} independently constructs strategies for Angel and Demon.
For the same game $\namedGame{\alpha}{a}$, let $S_{\langle\!\rangle}$ be an Angelic subvalue map for winning condition $\phi$ and $S_{[]}$ be a Demonic subvalue map for the opposite objective $\neg\phi$.
Cross-checking these subvalue maps can recover whether they are optimal, capturing all possible ways for their players to win.
If for each subgame $b$, $S_{[]}(b) \lor S_{\langle\!\rangle}(b)$ is valid, then intuitively, $S_{\langle\!\rangle}(b)$ and $S_{[]}(b)$ together cover every possible state, identifying it either as one where Angel has a winning strategy or where Demon does.
Further weakening $S_{\langle\!\rangle}(b)$ would violate consistency (i.e., it is never possible for both Angel and Demon to win)\footnote{This is analogous to how in $\alpha$-$\beta$ pruning \cite{KNUTH1975293}, if MIN's action results in a position where MAX wins, then that action cannot be part of MIN's strategy.}.
So $S_{\langle\!\rangle}(b)$ must be a maximally permissive, optimal Angel winning region, and $S_{\langle\!\rangle}$, an optimal Angelic subvalue map.

\section{Evaluation}
\label{sec:evaluation}

To test the generality and flexibility of our approach, we demonstrate it on examples from classic control theory problem classes (event-triggered control, reach-avoid problems, nonlinear control, infinite horizon switching).
The existing tool addressing a problem space closest to ours is CESAR \cite{DBLP:conf/tacas/KabraLMP24}, which synthesizes symbolic control envelopes as well, but for hybrid \emph{systems}.
NYCS \cite{10.1145/3047500} is a tool that can synthesize control envelopes for hybrid games (expressed via hybrid automata), but does not solve parametrically for symbolic constraints%
\footnote{Our problems have symbolic parameters that can represent any real number, e.g., in event-triggered ETCS (\rref{model:event-triggered-etcs}), A, B and T. We generate an envelope that is parametric in these symbols, equivalent to making uncountably infinitely many calls to NYCS corresponding to every possible parameter value assignment. Such parametricity is crucial in some applications (e.g., when parameters are unknown statically and only identified at runtime \cite{10.1007/978-3-030-17462-0_28}).}.
We compare performance.

CESAR and NYCS cannot solve the new examples, each of which displays different problem features, demonstrating the greater generality of our approach, which solves them within a few minutes (\rref{fig:new-examples}).
We also evaluate our approach on CESAR's control envelope synthesis benchmark suite \cite{DBLP:conf/tacas/KabraLMP24}.
Our implementation's performance is similar to CESAR despite the significantly greater generality (\rref{fig:cesar-benchmarks}).
To compare to NYCS, we procedurally generate 25 instantiations of a reach-avoid problem. 
The implementation finds control envelopes for 17 of them, timing out for the rest, while NYCS solves all of them, suggesting potential for future optimizations for this problem class.

The experiments were run on a 32GB RAM M2 MacBook Pro.
The implementation uses Pegasus \cite{DBLP:journals/fmsd/SogokonMTCP22}, CESAR \cite{DBLP:conf/tacas/KabraLMP24}, and additional rewriting heuristics for invariant generation.
It uses Mathematica for simplification and quantifier elimination.
Our examples favor interesting control challenges over dynamic complexity to avoid the computer algebra bottleneck, which is an avenue for future research.
\dGL formulas for all the new examples are provided in \rref{app:benchmarks-listings}.

\begin{figure}
  \centering
  \begin{subfigure}[t]{0.9\textwidth}
    \rowcolors{3}{gray!15}{white}
    \begin{tabularx}{\textwidth}{>{\raggedright\arraybackslash}X|>{\raggedleft}p{1cm}|>{\raggedleft}p{1.4cm}|>{\raggedleft}p{1.4cm}|>{\raggedright\arraybackslash}X}
    \toprule
      \textbf{Benchmark} & \textbf{Ours} & \textbf{CESAR} & \textbf{NYCS} & \textbf{Problem Type} \\
    \midrule
      Event ETCS &  23s   & $\infty$ & $\infty$   & Event-triggered control \\
      Infinite Track        &  233s  & $\infty$ & $\infty$   & Infinite horizon switching \\
      Surgical Robot~\cite{DBLP:conf/hybrid/KouskoulasRPK13}      &  168s  & $\infty$ & $\infty$   & Nonlinear control  \\
      Highway Driving~\cite{DBLP:conf/fm/LoosPN11}     &  91s   & $\infty$ & $\infty$   & Time-triggered control with adversarial agent \\
      Reach-avoid Robot    &  21s   & $\infty$ & $\infty$   & Reach-avoid problem \\
      \bottomrule
    \end{tabularx}
    \caption{Summary of new examples. \rref{app:benchmarks} discusses them.}
    \label{fig:new-examples}
  \end{subfigure}\\
  \begin{subfigure}{0.55\textwidth}
    \centering
    \begin{tabularx}{\textwidth}{X|>{\raggedleft\arraybackslash}p{0.85cm}|>{\raggedleft\arraybackslash}p{1.4cm} |>{\raggedleft\arraybackslash}p{1.4cm}}
        \toprule
        \textbf{Benchmark} & \textbf{Ours} & \textbf{CESAR} & \textbf{NYCS}\\
    \midrule
        ETCS Train    &   12s   & 15s  & \(\infty\) \\
        Sled          &   64s   & 64s  & \(\infty\) \\
        Intersection  &   104s  & 104s & \(\infty\) \\
        Curvebot      &   111s  & 112s & \(\infty\) \\
        Parachute     &   113s  & 123s & \(\infty\) \\
        Corridor      &   40s   & 40s  & \(\infty\) \\
        Power Station &   56s   & 56s  & \(\infty\) \\
        Coolant       &   311s  & 312s & \(\infty\) \\
    \bottomrule
  \end{tabularx}
    \caption{Running time comparison on the CESAR benchmark suite \cite{DBLP:conf/tacas/KabraLMP24}}
    \label{fig:cesar-benchmarks}
  \end{subfigure} \hfill
  \begin{subfigure}{0.4\textwidth}
    \vtop{
    \begin{tabularx}{\textwidth}{>{\raggedright\arraybackslash}X|>{\raggedright\arraybackslash}X|>{\raggedright\arraybackslash}X}
      \toprule
      \textbf{Tool} & \textbf{Success} & \textbf{Failure}\\
      \midrule
      Ours & 17 & 8 \\
      NYCS & 25 & 0 \\
      CESAR & 0 & 25 \\
      \bottomrule
    \end{tabularx}
    \\
    \includegraphics[width=\linewidth]{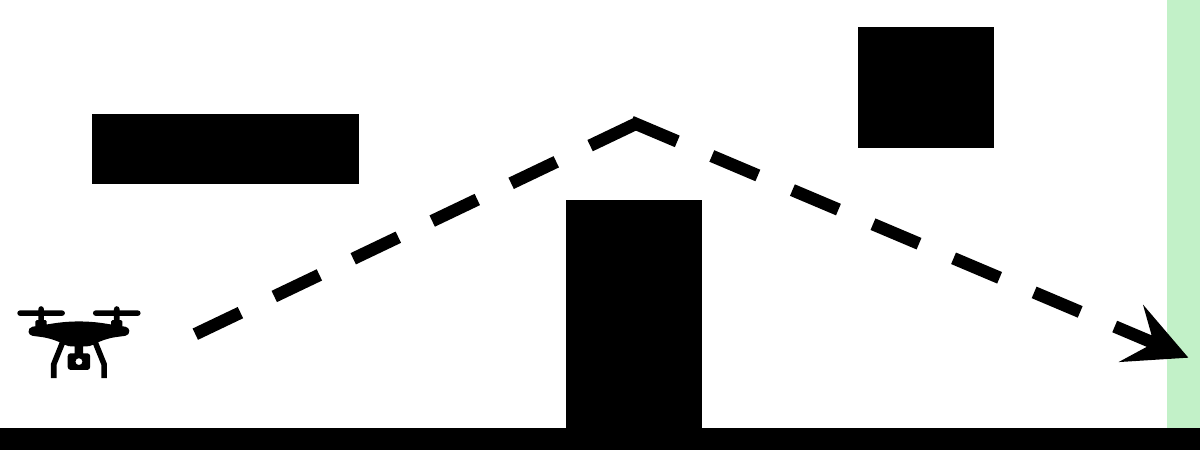}
    }
    \caption{Setup and outcome of the quadcopter reach-avoid problem with procedurally generated obstacle placement.}
    \label{fig:procedural-benchmarks}
  \end{subfigure}
\label{fig:benchmarks}
\end{figure}

\emph{Event-triggered ETCS} (\rref{model:event-triggered-etcs}) is an example of an \emph{event-triggered} control problem, in contrast with the time-triggered control handled by existing tools \cite{DBLP:conf/tacas/KabraLMP24}.
In time-triggered control, the controller repeatedly polls with some maximum time latency.
In event-triggered control, the controller is instead triggered at \emph{events}, when specific conditions are met.
Identifying what these events are that should trigger controller intervention is a part of the synthesis challenge.
Our example problem (\rref{model:event-triggered-etcs}) translates the ETCS Train benchmark of the CESAR benchmark suite from the time-triggered to the event-triggered paradigm.
It describes a train controller that must stop the train by point $e$, and has the choice to either accelerate or to brake.
\rref{line:etcs-event-plant} shows the plant which Angel (emulating the controller) chooses how long to run.
The synthesized subvalue map indicates when Angel must stop continuous evolution to enforce braking, characterizing the event at which to trigger the controller.
Modeling an event-triggered system as a game is subtle: the controller has the ability to choose how long to run the plant, and should not be permitted to win by interrupting infinitely often and inducing Zeno behavior.
\dGL's flexibility allows a correct encoding using nested loops and time progress assertions.

\begin{model}
  \setcounter{modelline}{0}
  \caption{Event-triggered ETCS}
  \label{model:event-triggered-etcs}
  \begin{align*}
    \text{\assumptions} &\,\big| \mline{line:etcs-event-init} A > 0 \land B > 0 \land T > 0 \limply \big\langle \\
    \text{\kwd{choice}} &\,\big| \mline{line:etcs-event-choices} \Big( t := 0 \seq \big( \left( a := A \cup a := -B \right) \seq \\
    \text{\kwd{plant}} &\,\big| \mline{line:etcs-event-plant} \left\{ p' = v, v' = a, t' = 1\ \&\ e - p > 0 \right\} \big)^* \seq \ptest{t \geq 1} \Big)^\times \rangle \kwd{True}
  \end{align*}
\end{model}

\emph{Surgical Robot} shows an example of nonlinear control (\rref{model:surgical-robot}).
It expresses the core control challenge in a case study from the literature \cite{DBLP:conf/hybrid/KouskoulasRPK13} modeling force feedback in surgical robots.
The objective is to dampen the force applied by the surgeon to ensure that a surgical tool stays within bounding planes called \emph{virtual fixtures}.
The damping factor $g$ (\rref{line:surgical-damping}) that we synthesize for can be \emph{any} nonnegative real number, implying infinitely many control possibilities.

\begin{model}[h]
  \setcounter{modelline}{0}
  \caption{Surgical Robot example based on \cite{DBLP:conf/hybrid/KouskoulasRPK13}}
  \label{model:surgical-robot}
  \begin{align*}
    \text{\assumptions} &\,\big| \mline{line:surgical-init} T > 0 \land K > 0 \land n_x^2 + n_y^2 = 1 \limply \big\langle \\
    \text{\kwd{input}} &\,\big| \mline{line:surgical-force-plan} \big( f_{xp} := \otimes \seq f_{yp} := \otimes \seq \\
    \text{\kwd{damping}} &\,\big| \mline{line:surgical-damping} g := \ast; ?g \geq 0 \seq \\
    \text{\kwd{plant}} &\,\big| \mline{line:surgical-motion} t := 0; \left\{ q_x' = K f_x g, q_y' = K f_y g, f_x' = f_{xp}, f_y' = f_{yp}, t' = 1\ \&\ t \leq T \right\}^d \\
    \text{\safe} &\,\big| \mline{line:surgical-safety} \big)^\times \big\rangle \left( (q_x - p_x) n_x + (q_y - p_y) n_y \geq 0 \right)
  \end{align*}
\end{model}

\emph{Highway Driving} (\rref{model:highway-driving}) expresses the core control challenge of a case study \cite{DBLP:conf/fm/LoosPN11} where on a highway, a controlled car must follow the car ahead at a safe distance.
The car ahead can behave adversarially by braking suddenly.
\emph{Reach-avoid Robot} (\rref{model:reach-avoid}) demonstrates an example of envelope synthesis for reach-avoid problems, where an agent must \emph{reach} an objective while \emph{avoiding} unsafe situations.
\emph{Infinite Track} (\rref{model:looping-track}) is an example of infinite horizon switching.
To remain safe for arbitrarily long time, the controlled vehicle must keep switching directions at the right moments, and make \emph{infinitely} many switches.
\rref{app:benchmarks-discussion} discusses these examples and their synthesized control envelopes further, and \rref{app:benchmarks-listings} lists the models.

To compare the algorithm to NYCS, we procedurally generate instantiations of \emph{quadcopter}, replacing symbolic parameters with concrete numbers, to create a reach-avoid challenge.
Each instantiation describes a scenario like in \rref{fig:procedural-benchmarks}, where a quadcopter must reach a target safe area while avoiding obstacles placed randomly in the environment.
The quadcopter constantly moves forward and can choose to either move upwards or downwards, but cannot revise its decision, once made, for one second.
The objective is to find a control envelope showing from which points the quadcopter has a way to reach the target safe zone while avoiding the obstacles.
The implementation finds control envelopes for $17$ problems.
The other $8$ problems time out after 20 minutes.
CESAR cannot solve any of these problems because they do not fit in its template.
NYCS, with its polytope representation and optimizations for this problem class, solves all of the problems.
\rref{app:quadcopter} provides further details.

\section{Related Work}
\label{sec:related-work}

Subvalue maps represent nondeterministic policies for hybrid games in a way that is compositional at the level of imperative program constructs such as loops and branches.
They can be verified/synthesized leveraging program analysis techniques like loop invariants, making it possible to answer hybrid game control envelope synthesis questions that no previous work solves.
\begin{itemize}
  \item \emph{Program logics} like game logic and differential game logic express games as imperative programs with Angelic and Demonic nondeterminism \cite{parikhGL,10.1093/logcom/12.1.149,DBLP:journals/tocl/Platzer15}.
  Their semantics compositionally define when \emph{there exists} a way to win a game, but not when a given policy wins.
  \item Refinement \cite{10.5555/184737,10.1145/1706299.1706339}, Hoare logic \cite{10.1007/978-3-662-46678-0_2}, and constructive proof \cite{cdgl} based approaches soundly construct deterministic policies (i.e., programs with Angelic nondeterminism fully resolved) but do not study \emph{maximal sets of solutions}.
  They do not have a direct representation of nondeterministic policies from which to derive \emph{sets} of permissible control actions, or a formal comparison of permissiveness of such policies, both of which are crucial parts of the permissive policy synthesis problem that we solve.
  \item A line of research solves games specified by \emph{hybrid automata} with objectives expressed in logics like LTL \cite{10.1007/3-540-48320-9_23,10.1007/978-3-662-46078-8_19,10.1145/3047500,10.1007/978-981-96-0617-7_11,ijcai2022p386}.
  Our approach targets specifications at a higher level of abstraction, where the game has imperative programming constructs like loops and branches.
  Compositionality is also at the level of these constructs, and the theoretical development required is different (e.g., the subtle Angel loop case in \rref{def:local-envelope-conditions}).
  This higher level of abstraction permits synthesis to use program analysis techniques like loop invariants and their continuous analog, differential invariants.
  The latter has the practical advantage of being complete for reasoning about the very general class of dynamics whose solutions are Noetherian functions \cite{platzer2020differential}, thus making it possible to solve problems that these previous techniques do not.
  \item Traditional \emph{game theory} \cite{Neumann1928,Shannon1950} and \emph{reinforcement learning} \cite{rl} do have recursive solving techniques and representations of nondeterministic policies (value function) \cite{Sutton1988LearningTP,10.1162/089976699300016070}.
  But these works also operate at a different level of abstraction than our imperative program-like hybrid games.
  Also, unlike these works, the subvalue maps of this paper are logically formalized, suitable for computerized verification and formally justified synthesis.
  \item \emph{Compositional game theory} \cite{10.1145/3209108.3209165,atkey:hal-04470659} also has a compositional representation of nondeterministic policies but also operates at a different level of abstraction, solving not hybrid games but abstract games with sequential and parallel composition.
\item There is also an important distinction between this work and \emph{controller synthesis} \cite{LIU202230,Tabuada09,Belta17}, for which approaches include the use of games \cite{DBLP:conf/cdc/NerodeY92,DBLP:journals/IEEE/TomlinLS00} and CEGIS-based \cite{DBLP:journals/sttt/Solar-Lezama13} guidance from counterexamples \cite{DBLP:journals/acta/AbateBCDKKP20,DBLP:conf/emsoft/RavanbakhshS16,Dai2020CounterexampleGS}.
We tackle control \emph{envelope} synthesis, where the goal is to characterize \emph{all} the correct control strategies, requiring an extra universal quantifier compared to controller synthesis.
Controller synthesis techniques do not solve control envelope synthesis.
For example, \dGL control envelope synthesis does not fit the CEGIS quantifier pattern of $\exists \forall$.
\item \emph{Safety shields} computed by numerical methods \cite{10.1609/aaai.v32i1.11797,10.1109/tac.2018.2876389,kochenderfer2012next} share the objective of characterizing the space of safe control. They can handle dynamical systems that are hard to analyze symbolically,
but do not support unbounded state spaces, require a hopeless increase in dimensionality to computing shields parameterized on the model's parameters, and can lose rigorous formal safety guarantees when they discretize continuous systems.
\item Compared to safe set and \emph{barrier function} approaches \cite{10.1007/978-3-540-24743-2_32,10.1007/978-3-540-31954-2_31,DBLP:conf/eucc/AmesCENST19,cohen20,ARTSTEIN19831163,doi:10.1137/0321028,1047016}, we again solve for the higher level of abstraction of \dGL, with program structures such as loops, as well as the adversarial dynamics of games.
\item \emph{Differential Floyd-Hoare logic} \cite{9815834} identifies the relationship between preconditions and runtime monitoring, but deals with hybrid \emph{systems}, not games, which lack the complexity of adversarial dynamics.
Additionally it uses preconditions to synthesize \emph{deterministic} policies (proper responses), different from our nondeterministic policy question.
\item The \emph{CESAR algorithm} \cite{DBLP:conf/tacas/KabraLMP24} also synthesizes symbolic control envelopes for hybrid \emph{systems}, for a \emph{fixed, time-triggered template}.
However, it does not provide a general, compositional representation for nondeterministic policies, and does not deal with the complexities brought by the adversarial dynamics of full \dGL.
\end{itemize}

\section{Conclusion and Future Work}
\label{sec:conclusion}
We introduce a generalized symbolic control envelope synthesis approach for hybrid games.
Such envelopes identify what control action is safe to take when.
We represent the envelopes using inductive subvalue maps and provide an algorithmic framework to generate them.
Our approach allows the expression and synthesis of a variety of control problems.
The approach should generalize for non-\dGL Markovian games, where the policy depends only on state and not the history of previous moves, which is an avenue for future research.
Another direction for future work is to integrate other tools, e.g., NYCS \cite{10.1145/2185632.2185645} as invariant generators to improve performance on the problem classes that they solve.

\section{Acknowledgments}
\label{sec:acknowledgments}
This work was partially funded by the National Science Foundation (NSF) under Award Number 2220311, and by an Alexander von Humboldt Professorship.
We thank Ruben Martins and Stephanie Balzer for helpful feedback on drafts of this paper. 


\bibliographystyle{ACM-Reference-Format}
\bibliography{refs}

\appendix
\clearpage


\section{Na\"ive Ordering Counterexample}
\label{app:counterexample}

For game $\namedGame{\alpha}{a}$, let $S'$ set the winning region for the overall game to be empty ($S'(a)=\bot$), meaning that Angel is never allowed to play at all.
Now the winning subregions for later subgames are irrelevant because Angel cannot reach them,
but $S'$ nevertheless sets them to the theoretical optimal (for every subgame $b\neq a$, $S(b)\mapsto\langle \fwd{b}{\namedGame{\alpha}{a}} \rangle \phi$).
For \emph{any} other inductive subvalue map $S$, a good ordering should indicate that $S\sqsupseteq S'$ since $S$ cannot possibly be less permissive,
but the proposed na\"ive ordering does not.
When $S$ sets the winning subregion for any of the later subgames $b \neq a$ to anything less than the optimal solution, the check $\models S'(b) \limply S(b)$ fails.
\begin{figure}
  \centering
  \begin{tikzpicture}
    [level distance=1.2cm, sibling distance=1.5cm,
    every node/.style={draw,rectangle,align=center}]
    \node {
      $\cup$ \\ \blue{$x=1$}
    }
      child {node {$\alpha$ \\ \blue{$y=1$}} }
      child {node {$\beta$ \\ \blue{$\phi_\beta$}} };
  \end{tikzpicture}%
  \quad\quad%
  \begin{tikzpicture}
    [level distance=1.2cm, sibling distance=1.5cm,
    every node/.style={draw,rectangle,align=center}]
    \node {$\cup$ \\ \textcolor{red}{$\bot$}}
      child {node {$\alpha$ \\ \textcolor{red}{$\top$}} }
      child {node {$\beta$ \\ \textcolor{red}{$\phi_\beta$}} };
  \end{tikzpicture}
  \caption{Inductive subvalue maps \blue{$S$} and \red{$S'$}. Locally at subgame $\alpha$, \red{$S'$} seems more permissive, but global analysis reveals \blue{$S$} is more permissive.}
  \label{fig:counterexample}
\end{figure}
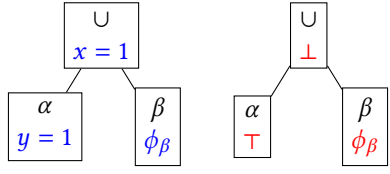

For a concrete visualization, consider the two inductive subvalue maps shown in \rref{fig:counterexample}: \blue{$S$} on the left and \red{$S'$} on the right.
Because $S'$ immediately blocks off \emph{all} strategies at the root, a good ordering should indicate $S\sqsupseteq S'$, but the proposed na\"ive ordering does not because at the subgame $\alpha$, $S$ has a strictly stronger formula ($y=1 \limply \top$ is false).
The problem in the na\"ive, point wise definition (\rref{def:weak-ordering}) is it does not recognize that the $\top$ at subgame $\alpha$ is vacuous: $S'$ permits no traces to reach subgame $\alpha$ in the first place, so the permissiveness of $\top$ is unimpressive.

\section{Inductive Subvalue Map Loop Condition Discussion}
\label{app:counterexample-loop}

We discuss the details of the definition of an inductive subvalue maps (\rref{def:local-envelope-conditions}) for the case $\avalid{(\namedGame{\gamma}{g})^*}{a}{S}{\phi}$, and correspondingly the synthesis framework $\asolve$ (\rref{alg:solving}) for the case of $\asolve(\namedGame{(\namedGame{\gamma}{g})^*}{a}, \phi)$.

\subsection{Counterexample for Projection Condition Weakening}
\label{app:counterexample-loop-proj}

\rref{def:local-envelope-conditions} and \rref{alg:solving} for the case $(\namedGame{\gamma}{g})^*$ use the check $\models S(a) \limply \langle \aproj{\namedGame{\alpha}{a}}{S}{\phi} \rangle \phi$, i.e., winning subregion $S(a)$ should ensure that there is a \emph{finite} strategy to exit while staying within the control envelope.
The subvalue map should ensure that Angel does not get stuck in a state where her only option is playing the loop forever.
It might seem that the more liberal, simpler condition $\models S(a) \limply \langle \alpha \rangle \phi$ also ensures Angel has a finite exit strategy, but this condition is subtly incorrect.
Angel can get stuck in a state where there is theoretically a way to exit the loop but this requires going outside the control envelope.
\rref{eq:counterexample-loop} shows a counterexample.
\begin{equation}
  \label{eq:counterexample-loop}
  \begin{aligned}
    &\langle \namedGame{(\namedGame{(\namedGame{x:=x+1}{c} \cup \namedGame{x:=x-1}{d})}{b})*}{a} \rangle x\geq 0 \\
    &\quad S:=\{ a \mapsto \top , b \mapsto \top , c \mapsto \bot, d\mapsto \top\}
  \end{aligned}
\end{equation}
For this game, if \rref{def:local-envelope-conditions} were modified such that $\langle \aproj{\namedGame{\alpha}{a}}{S}{\phi} \rangle \phi$ is replaced by $\langle \alpha \rangle \phi$, the subvalue map $S$ would be inductive.
However, it would also let Angel get stuck iterating the loop forever.
Consider starting state $\{x\mapsto -1\}$.
This state belongs to the winning subregion $S(a)$ of the overall game, so Angel should be able to win.
And indeed she can win the game by running one iteration of the loop where she increments $x$.
The new state would be $\{x\mapsto 0\}$, so Angel would exit the loop and win since winning condition $x\geq 0$ holds.
However, if Angel is forced to actually respect the subvalue map (modeled by projection of $S$ onto the game as shown in \rref{eq:projection-counterexample}), she cannot win.
\begin{equation}
  \label{eq:projection-counterexample}
    \langle ?\top \seq \left(?\top \seq \left( (?\bot \seq x:=x+1) \cup (?\top \seq x:=x-1) \right) \right)^* \seq ? x\geq 0 \rangle x\geq 0
\end{equation}
Starting at $\{x\mapsto -1\}$, Angel has to repeat the loop since she has not reached the winning region $x\geq 0$.
But inside the loop, the subvalue map says she can never choose to increment $x$ (subgame $c$).
So she is forced to decrement $x$, reaching new state $\{x\mapsto -2\}$.
In principle, if she were allowed to ignore the subvalue map, she could still exit with finite iterations by incrementing $x$ twice.
But with subgame $c$ blocked off, she is only ever allowed to decrement $x$ and is stuck in the loop forever.
To prevent such subvalue maps from counting as inductive, the check $\models S(a) \limply \langle \aproj{\namedGame{\alpha}{a}}{S}{\phi} \rangle \phi$ is necessary.

\subsection{Counterexample for Recursive Condition Weakening}
\label{app:counterexample-loop-inner}

\rref{def:local-envelope-conditions} and \rref{alg:solving} for the case $(\namedGame{\gamma}{g})^*$ also use the check $\avalid{\gamma}{g}{S}{(S(a)\lor\phi)}$, i.e., the subvalue map for the loop body ensures that after playing the loop body, either there is still a finite exit strategy or Angel can already exit the loop.
\rref{eq:counterexample-loop-inner} shows a counterexample where using $\avalid{\gamma}{g}{S}{(\mathbf{S(g)}\lor\phi)}$ instead is not sound, since the inner subvalue map permit reaching a state without a finite exit strategy.
\begin{equation}
  \label{eq:counterexample-loop-inner}
  \begin{aligned}
    &\langle \namedGame{(\namedGame{(\namedGame{x:=x+1}{c} \cup \namedGame{x:=x-1}{d})}{b})*}{a} \rangle x\geq 0 \\
    &\quad S:=\{a\mapsto x \geq -1, b\mapsto x\geq -1, c\mapsto x \geq -1, d\mapsto \top\}
  \end{aligned}
\end{equation}
Starting at $\{x\mapsto -1\}$, it is possible for Angel to choose subgame $d$ and decrement $x$.
The new state is $\{x\mapsto -2\}$, from which there is now no finite exit strategy because to increment $x$, per the guard of subgame $c$, $x$ must be at least $-1$.
Angel is now stuck decrementing $x$ forever.

\section{Example of Subvalue Map Shortfall for Runtime Monitoring}
\label{app:counterexample-subvalue-projection}

We show how for a (non-inductive) subvalue map, the property $\models S(a) \limply \langle \aruntimem{\namedGame{\alpha}{a}}{S}{\phi} \rangle \phi$ does not always hold.
Consider the subvalue map of \rref{fig:overview}, but with $\phi_6$ set to the empty subregion $\bot$ instead of $x>0 \land (v\geq 0 \lor a > 0)$.
As before, let the overall game be $\namedGame{\alpha}{a}$, and subgame labels be the numbers indexing the winning subregions.
The new map is still a subvalue map since $\models \bot \limply \langle \fwd{5}{\namedGame{\alpha}{a}} \rangle x>0$.
The untrusted controller monitoring correctness property is \rref{eq:overview-projection} with $\phi_6$ replaced by $\bot$.
The result has shape
\begin{equation*}
  \begin{aligned}
      x>0\land (v\geq0\lor a>0) \limply \big\langle \big( \big(
      \AdvExChoice{?(\bot \!\vee\! (x>0 \land v\geq 1))} \seq 
       \cdots 
      \; \AdvExChoice{\cap} 
       \cdots 
      \big) \seq \cdots \big)^{\times} \big\rangle \, x>0.
  \end{aligned}
\end{equation*}
This formula is not valid. Consider the starting state $\{x\mapsto1, v\mapsto0\}$.
Demon first chooses to run the loop, then Angel must pass the test $x>0 \land v\geq 1$, which fails.
The \emph{inductive} subvalue map prevents this situation by requiring that $\phi_8$ implies $\phi_6 \lor \phi_5$, thus disqualifying this counterexample where Angel gets stuck.

\section{Axiomatization of \dGL}
\label{app:dl-semantics}
\rref{fig:dGL-calculus} summarizes the axioms and proof rules of \dGL.
A full explanation is in \cite{Platzer18}.
Note that the dual operator $\alpha^d$ switches the role of players in $\alpha$. For example, $(\alpha \cup \beta)^d$ is the same as $\alpha^d \cap \beta^d$.
The semantics \cite{Platzer18} of \dGL is as follows.

{\renewcommand{\D}[2][]{\ifthenelse{\equal{#1}{}}{#2'}{\frac{d#2}{d{#1}}}}%

\begin{definition}[\dGL semantics] \label{def:dGL-semantics}
The \emph{semantics of a \dGL formula} $\phi$ is the subset \m{\imodel{\I}{\phi}\subseteq\linterpretations{\Sigma}{V}} of states in which $\phi$ is true.
It is defined inductively as follows
\begin{enumerate}
\item \(\imodel{\I}{p(\theta_1,\dots,\theta_k)} = \{\iportray{\I} \in \linterpretations{\Sigma}{V} \with (\ivaluation{\I}{\theta_1},\dots,\ivaluation{\I}{\theta_k})\in\iget[const]{\I}(p)\}\)
\item \(\imodel{\I}{\theta_1\sim\theta_2} = \{\iportray{\I} \in \linterpretations{\Sigma}{V} \with \ivaluation{\I}{\theta_1}\sim\ivaluation{\I}{\theta_2}\} \text{ where }{\sim}\in\{<,\leq,=,\geq,>\}\)
\item \(\imodel{\I}{\lnot\phi} = \scomplement{(\imodel{\I}{\phi})}\)
\item \(\imodel{\I}{\phi\land\psi} = \imodel{\I}{\phi} \cap \imodel{\I}{\psi}\)
\item
{\def\Im{\imodif[state]{\I}{x}{r}}%
\(\imodel{\I}{\lexists{x}{\phi}} =  \{\iportray{\I} \in \linterpretations{\Sigma}{V} \with \iget[state]{\Im} \in \imodel{\I}{\phi} ~\text{for some}~r\in\reals\}\)
}
\item \(\imodel{\I}{\ddiamond{\alpha}{\phi}} = \strategyfor[\alpha]{\imodel{\I}{\phi}}\)
\item \(\imodel{\I}{\dbox{\alpha}{\phi}} = \dstrategyfor[\alpha]{\imodel{\I}{\phi}}\)
\end{enumerate}
A \dGL formula $\phi$ is \emph{valid}, written \m{\entails{\phi}}, iff it is true in all states, i.e.\ \m{\imodel{\I}{\phi}=\linterpretations{\Sigma}{V}}.
\end{definition}

\begin{definition}[Semantics of hybrid games] \label{def:HG-semantics}
The \emph{semantics of a hybrid game} $\alpha$ is a function \m{\strategyfor[\alpha]{\cdot}} that, for each set of Angel's winning states \m{X\subseteq\linterpretations{\Sigma}{V}}, gives the \emph{winning region}, i.e.\ the set of states \m{\strategyfor[\alpha]{X}} from which Angel has a winning strategy to achieve $X$ in $\alpha$ (whatever strategy Demon chooses). It is defined inductively as follows
\begin{enumerate}
\item \(\strategyfor[\pupdate{\pumod{x}{\theta}}]{X} = \{\iportray{\I} \in \linterpretations{\Sigma}{V} \with \modif{\iget[state]{\I}}{x}{\ivaluation{\I}{\theta}} \in X\}\)
\item \(\strategyfor[\pevolvein{\D{x}=\genDE{x}}{\ivr}]{X} = \{\varphi(0) \in \linterpretations{\Sigma}{V} \with
      \varphi(r)\in X\)
      for some $r\in\reals_{\geq0}$ and (differentiable)
      \m{\varphi:[0,r]\to\linterpretations{\Sigma}{V}}
      such that
      \(\varphi(\zeta)\in\imodel{\I}{\ivr}\)
      and
      \m{\D[t]{\,\varphi(t)(x)} (\zeta) =
      \ivaluation{\iconcat[state=\varphi(\zeta)]{\I}}{f(x)}}
      for all \(0\leq\zeta\leq r\}\)
\item \(\strategyfor[\ptest{\ivr}]{X} = \imodel{\I}{\ivr}\cap X\)
\item \(\strategyfor[\pchoice{\alpha}{\beta}]{X} = \strategyfor[\alpha]{X}\cup\strategyfor[\beta]{X}\)
\item \(\strategyfor[\alpha;\beta]{X} = \strategyfor[\alpha]{\strategyfor[\beta]{X}}\)
\item \(\strategyfor[\prepeat{\alpha}]{X} = \capfold\{Z\subseteq\linterpretations{\Sigma}{V} \with X\cup\strategyfor[\alpha]{Z}\subseteq Z\}\)

\item \(\strategyfor[\pdual{\alpha}]{X} = \scomplement{(\strategyfor[\alpha]{\scomplement{X}})}\)
\end{enumerate}
The \emph{winning region} of Demon, i.e.\ the set of states \m{\dstrategyfor[\alpha]{X}} from which Demon has a winning strategy to achieve $X$ in $\alpha$ (whatever strategy Angel chooses) is defined inductively as follows
\begin{enumerate}
\item \(\dstrategyfor[\pupdate{\pumod{x}{\theta}}]{X} = \{\iportray{\I} \in \linterpretations{\Sigma}{V} \with \modif{\iget[state]{\I}}{x}{\ivaluation{\I}{\theta}} \in X\}\)
\item \(\dstrategyfor[\pevolvein{\D{x}=\genDE{x}}{\ivr}]{X} = \{\varphi(0) \in \linterpretations{\Sigma}{V} \with
      \varphi(r)\in X\)
      for all $r\in\reals_{\geq0}$ and (differentiable)
      \m{\varphi:[0,r]\to\linterpretations{\Sigma}{V}}
      such that
      \(\varphi(\zeta)\in\imodel{\I}{\ivr}\)
      and
      \m{\D[t]{\,\varphi(t)(x)} (\zeta) =
      \ivaluation{\iconcat[state=\varphi(\zeta)]{\I}}{\theta}}
      for all $0\leq\zeta\leq r\}$
\item \(\dstrategyfor[\ptest{\ivr}]{X} = \scomplement{(\imodel{\I}{\ivr})}\cup X\)
\item \(\dstrategyfor[\pchoice{\alpha}{\beta}]{X} = \dstrategyfor[\alpha]{X}\cap\dstrategyfor[\beta]{X}\)
\item \(\dstrategyfor[\alpha;\beta]{X} = \dstrategyfor[\alpha]{\dstrategyfor[\beta]{X}}\)
\item \(\dstrategyfor[\prepeat{\alpha}]{X} = \cupfold\{Z\subseteq\linterpretations{\Sigma}{V} \with Z\subseteq X\cap\dstrategyfor[\alpha]{Z}\}\)
\item \(\dstrategyfor[\pdual{\alpha}]{X} = \scomplement{(\dstrategyfor[\alpha]{\scomplement{X}})}\)
\end{enumerate}
\end{definition}
}

\providecommand{\axkey}[1]{\textcolor{vblue}{#1}}%
\cinferenceRuleStore[diamond|$\didia{\cdot}$]{diamond axiom}
{\linferenceRule[equiv]
  {\lnot\dbox{\ausprg}{\lnot \ausfml}}
  {\axkey{\ddiamond{\ausprg}{\ausfml}}}
}
{}
\cinferenceRuleStore[diamondax|$\didia{\cdot}$]{diamond axiom}
{\linferenceRule[equiv]
  {\lnot\dbox{\ausprgax}{\lnot \ausfmlax}}
  {\axkey{\ddiamond{\ausprgax}{\ausfmlax}}}
}
{}
\cinferenceRuleStore[assignb|$\dibox{:=}$]{assignment / substitution axiom}
{\linferenceRule[equiv]
  {p(\genDJ{x})}
  {\axkey{\dbox{\pupdate{\umod{x}{\genDJ{x}}}}{p(x)}}}
}
{}%
\cinferenceRuleStore[assignbax|$\dibox{:=}$]{assignment / substitution axiom}
{\linferenceRule[equiv]
  {p(\aconst)}
  {\axkey{\dbox{\pupdate{\umod{x}{\aconst}}}{p(x)}}}
}
{}%
\cinferenceRuleStore[Dassignb|$\dibox{:=}$]{differential assignment}
{\linferenceRule[equiv]
{p(\astrm)}
{\axkey{\dbox{\Dupdate{\Dumod{\D{x}}{\astrm}}}{p(\D{x})}}}
}
{}%
\cinferenceRuleStore[testb|$\dibox{?}$]{test}
{\linferenceRule[equiv]
  {(\ivr \limply \ausfml)}
  {\axkey{\dbox{\ptest{\ivr}}{\ausfml}}}
}{}%
\cinferenceRuleStore[testbax|$\dibox{?}$]{test}
{\linferenceRule[equiv]
  {(q \limply p)}
  {\axkey{\dbox{\ptest{q}}{p}}}
}{}%
\cinferenceRuleStore[evolveb|$\dibox{'}$]{evolve}
{\linferenceRule[equiv]
  {\lforall{t{\geq}0}{\dbox{\pupdate{\pumod{x}{\solf(t)}}}{p(x)}}}
  {\axkey{\dbox{\pevolve{\D{x}=\genDE{x}}}{p(x)}}}
}{\m{\D{\solf}(t)=\genDE{\solf}}}%
\cinferenceRuleStore[choiceb|$\dibox{\cup}$]{axiom of nondeterministic choice}
{\linferenceRule[equiv]
  {\dbox{\ausprg}{\ausfml} \land \dbox{\busprg}{\ausfml}}
  {\axkey{\dbox{\pchoice{\ausprg}{\busprg}}{\ausfml}}}
}{}%
\cinferenceRuleStore[choicebax|$\dibox{\cup}$]{axiom of nondeterministic choice}
{\linferenceRule[equiv]
  {\dbox{\ausprgax}{\ausfmlax} \land \dbox{\busprgax}{\ausfmlax}}
  {\axkey{\dbox{\pchoice{\ausprgax}{\busprgax}}{\ausfmlax}}}
}{}%
\cinferenceRuleStore[evolveinb|$\dibox{'}$]{evolve}
{\linferenceRule[equiv]
  {
        \lforall{t{\geq}0}{\big(
          (\lforall{0{\leq}s{\leq}t}{q(\solf(s))})
          \limply
          \dbox{\pupdate{\pumod{x}{\solf(t)}}}{p(x)}
        \big)}
      }
  {
        \dbox{\pevolvein{\D{x}=\genDE{x}}{q(x)}}{p(x)}
  }
}{}%
\cinferenceRuleStore[composeb|$\dibox{{;}}$]{composition} %
{\linferenceRule[equiv]
  {\dbox{\ausprg}{\dbox{\busprg}{\ausfml}}}
  {\axkey{\dbox{\ausprg;\busprg}{\ausfml}}}
}{}%
\cinferenceRuleStore[composebax|$\dibox{{;}}$]{composition} %
{\linferenceRule[equiv]
  {\dbox{\ausprgax}{\dbox{\busprgax}{\ausfmlax}}}
  {\axkey{\dbox{\ausprgax;\busprgax}{\ausfmlax}}}
}{}%
\cinferenceRuleStore[iterateb|$\dibox{{}^*}$]{iteration/repeat unwind} %
{\linferenceRule[equiv]
  {\ausfml \land \dbox{\ausprg}{\dbox{\prepeat{\ausprg}}{\ausfml}}}
  {\axkey{\dbox{\prepeat{\ausprg}}{\ausfml}}}
}{}%
\cinferenceRuleStore[iteratebax|$\dibox{{}^*}$]{iteration/repeat unwind} %
{\linferenceRule[equiv]
  {\ausfmlax \land \dbox{\ausprgax}{\dbox{\prepeat{\ausprgax}}{\ausfmlax}}}
  {\axkey{\dbox{\prepeat{\ausprgax}}{\ausfmlax}}}
}{}%
\cinferenceRuleStore[K|K]{K axiom / modal modus ponens}
{\linferenceRule[impl]
  {\dbox{\ausprg}{(\ausfml\limply\busfml)}}
  {(\dbox{\ausprg}{\ausfml}\limply\axkey{\dbox{\ausprg}{\busfml}})}
}{}%
\cinferenceRuleStore[Kax|K]{K axiom / modal modus ponens}
{\linferenceRule[impl]
  {\dbox{\ausprgax}{(\ausfmlax\limply\busfmlax)}}
  {(\dbox{\ausprgax}{\ausfmlax}\limply\axkey{\dbox{\ausprgax}{\busfmlax}})}
}{}%
\cinferenceRuleStore[I|II]{loop induction}
{\linferenceRule[impl]
  {\dbox{\prepeat{\ausprg}}{(\ausfml\limply\dbox{\ausprg}{\ausfml})}}
  {(\ausfml\limply\axkey{\dbox{\prepeat{\ausprg}}{\ausfml}})}
}{}%
\cinferenceRuleStore[Ieq|I]{loop induction}
{\linferenceRule[equiv]
  {\ausfml \land \dbox{\prepeat{\ausprg}}{(\ausfml\limply\dbox{\ausprg}{\ausfml})}}
  {\axkey{\dbox{\prepeat{\ausprg}}{\ausfml}}}
}{}%
\cinferenceRuleStore[Ieqax|I]{loop induction}
{\linferenceRule[equiv]
  {\ausfmlax \land \dbox{\prepeat{\ausprgax}}{(\ausfmlax\limply\dbox{\ausprgax}{\ausfmlax})}}
  {\axkey{\dbox{\prepeat{\ausprgax}}{\ausfmlax}}}
}{}%
\dinferenceRuleStore[backiterateb|\usebox{\backiterateb}]{backwards iteration/repeat unwind}
{\linferenceRule[equiv]
  {\ausfml \land \dbox{\prepeat{\ausprg}}{\dbox{\ausprg}{\ausfml}}}
  {\axkey{\dbox{\prepeat{\ausprg}}{\ausfml}}}
}{}%
\dinferenceRuleStore[iterateiterateb|$\dibox{{}^*{}^*}$]{double iteration}
{\linferenceRule[equiv]
  {\dbox{\prepeat{\ausprg}}{\ausfml}}
  {\axkey{\dbox{\prepeat{\ausprg};\prepeat{\ausprg}}{\ausfml}}}
}{}%
\dinferenceRuleStore[iterateiterated|$\didia{{}^*{}^*}$]{double iteration}
{\linferenceRule[equiv]
  {\ddiamond{\prepeat{\ausprg}}{\ausfml}}
  {\axkey{\ddiamond{\prepeat{\ausprg};\prepeat{\ausprg}}{\ausfml}}}
}{}%
\cinferenceRuleStore[B|B]{Barcan and converse}
{\linferenceRule[equiv]
        {\ddiamond{\ausprg}{\lexists{x}{\ausfml}}}
        {\lexists{x}{\ddiamond{\ausprg}{\ausfml}}}
}{\m{x{\not\in}\ausprg}}
\cinferenceRuleStore[V|V]{vacuous $\dbox{}{}$}
{\linferenceRule[impl]
  {p}
  {\axkey{\dbox{\ausprg}{p}}}
}{\m{FV(p)\cap BV(\ausprg)=\emptyset}}%
\cinferenceRuleStore[Vax|V]{vacuous $\dbox{}{}$}
{\linferenceRule[impl]
  {p}
  {\axkey{\dbox{a}{p}}}
}{}%
\cinferenceRuleStore[G|G]{$\dbox{}{}$ generalization} %
{\linferenceRule[formula]
  {\ausfml}
  {\dbox{\ausprg}{\ausfml}}
}{}%
\cinferenceRuleStore[Gax|G]{$\dbox{}{}$ generalization} %
{\linferenceRule[formula]
  {\ausfmlax}
  {\dbox{\ausprgax}{\ausfmlax}}
}{}%
\cinferenceRuleStore[genaax|$\forall{}$]{$\forall{}$ generalisation}
{\linferenceRule[formula]
  {p(x)}
  {\lforall{x}{p(x)}}
}{}%
\cinferenceRuleStore[MPax|MP]{modus ponens}
{\linferenceRule[formula]
  {p\limply q \quad p}
  {q}
}{}%
\cinferenceRuleStore[Mb|M${\dibox{\cdot}}$]{$\dbox{}{}$ monotone}
{\linferenceRule[formula]
  {\ausfml\limply \busfml}
  {\dbox{\ausprg}{\ausfml}\limply\dbox{\ausprg}{\busfml}}
}{}%
\cinferenceRuleStore[M|M]{$\ddiamond{}{}$ monotone / $\ddiamond{}{}$-generalization}
{\linferenceRule[formula]
  {\ausfml\limply\busfml}
  {\ddiamond{\ausprg}{\ausfml}\limply\ddiamond{\ausprg}{\busfml}}
}{}%

\dinferenceRuleStore[Mbr|M\rightrule]%
{$\ddiamond{}{}/\dbox{}{}$ generalization=M=G+K} 
{\linferenceRule[sequent]
  {\lsequent[L]{} {\dbox{\ausprg}{\busfml}} 
  &\lsequent[g]{\busfml} {\ausfml}}
  {\lsequent[L]{} {\dbox{\ausprg}{\ausfml}}}
}{}%

\dinferenceRuleStore[loop|loop]{inductive invariant}
{\linferenceRule[sequent]
  {\lsequent[L]{} {\inv}
  &\lsequent[g]{\inv} {\dbox{\ausprg}{\inv}}
  &\lsequent[g]{\inv} {\ausfml}}
  {\lsequent[L]{} {\dbox{\prepeat{\ausprg}}{\ausfml}}}
}{}%
\dinferenceRuleStore[invind|ind]{inductive invariant}
{\linferenceRule[sequent]
  {\lsequent[\globalrule]{\ausfml}{\dbox{\ausprg}{\ausfml}}}
  {\lsequent{\ausfml}{\dbox{\prepeat{\ausprg}}{\ausfml}}}
}{}%
\cinferenceRuleStore[con|con]{loop convergence right} 
{\linferenceRule[formula]
  {\lsequent[G]{\mapply{\var}{v}\land v>0}{\ddiamond{\ausprg}{\mapply{\var}{v-1}}}}
  {\lsequent[L]{\lexists{v}{\mapply{\var}{v}}}
      {\axkey{\ddiamond{\prepeat{\ausprg}}{\lexists{v{\leq}0}{\mapply{\var}{v}}}}}}
}{v\not\in\ausprg}%
\dinferenceRuleStore[congen|con]{loop convergence}
{\linferenceRule[sequent]
  {\lsequent[L]{}{\lexists{v}{\mapply{\var}{v}}}
  &\lsequent[G]{}{\lforall{v{>}0}{({\mapply{\var}{v}}\limply{\ddiamond{\ausprg}{\mapply{\var}{v-1}})}}}
  &\lsequent[G]{\lexists{v{\leq}0}{\mapply{\var}{v}}}{\busfml}
  }
  {\lsequent[L]{}{\ddiamond{\prepeat{\ausprg}}{\busfml}}}
}{v\not\in\ausprg}

\dinferenceRuleStore[band|${[]\land}$]{$\dbox{\cdot}{\land}$}
{\linferenceRule[equiv]
  {\dbox{\ausprg}{\ausfml} \land \dbox{\ausprg}{\busfml}}
  {\axkey{\dbox{\ausprg}{(\ausfml\land\busfml)}}}
}{}%

\dinferenceRuleStore[Hoarecompose|H${;}$]{Hoare $;$}
{\linferenceRule
  {A\limply\dbox{\ausprg}{E} & E\limply\dbox{\busprg}{B}}
  {A \limply \dbox{\ausprg;\busprg}{B}}
}{}%
\dinferenceRuleStore[composebrexplicit|$\dibox{{;}}$\rightrule]{$;$}
{\linferenceRule
  {A\limply\dbox{\ausprg}{\dbox{\busprg}{B}}}
  {A \limply \dbox{\ausprg;\busprg}{B}}
}{}
\cinferenceRuleStore[notr|$\lnot$\rightrule]{$\lnot$ right}
{\linferenceRule[sequent]
  {\lsequent[L]{\asfml}{}}
  {\lsequent[L]{}{\lnot \asfml}}
}{}%
\cinferenceRuleStore[notl|$\lnot$\leftrule]{$\lnot$ left}
{\linferenceRule[sequent]
  {\lsequent[L]{}{\asfml}}
  {\lsequent[L]{\lnot \asfml}{}}
}{}%
\cinferenceRuleStore[andr|$\land$\rightrule]{$\land$ right}
{\linferenceRule[sequent]
  {\lsequent[L]{}{\asfml}
    & \lsequent[L]{}{\bsfml}}
  {\lsequent[L]{}{\asfml \land \bsfml}}
}{}%
\cinferenceRuleStore[andl|$\land$\leftrule]{$\land$ left}
{\linferenceRule[sequent]
  {\lsequent[L]{\asfml , \bsfml}{}}
  {\lsequent[L]{\asfml \land \bsfml}{}}
}{}%
\cinferenceRuleStore[orr|$\lor$\rightrule]{$\lor$ right}
{\linferenceRule[sequent]
  {\lsequent[L]{}{\asfml, \bsfml}}
  {\lsequent[L]{}{\asfml \lor \bsfml}}
}{}%
\cinferenceRuleStore[orl|$\lor$\leftrule]{$\lor$ left}
{\linferenceRule[sequent]
  {\lsequent[L]{\asfml}{}
    & \lsequent[L]{\bsfml}{}}
  {\lsequent[L]{\asfml \lor \bsfml}{}}
}{}%
\cinferenceRuleStore[implyr|$\limply$\rightrule]{$\limply$ right}
{\linferenceRule[sequent]
  {\lsequent[L]{\asfml}{\bsfml}}
  {\lsequent[L]{}{\asfml \limply \bsfml}}
}{}%
\cinferenceRuleStore[implyl|$\limply$\leftrule]{$\limply$ left}
{\linferenceRule[sequent]
  {\lsequent[L]{}{\asfml}
    & \lsequent[L]{\bsfml}{}}
  {\lsequent[L]{\asfml \limply \bsfml}{}}
}{}%
\cinferenceRuleStore[equivr|$\lbisubjunct$\rightrule]{$\lbisubjunct$ right}
{\linferenceRule[sequent]
  {\lsequent[L]{\asfml}{\bsfml}
   & \lsequent[L]{\bsfml}{\asfml}}
  {\lsequent[L]{}{\asfml \lbisubjunct \bsfml}}
}{}%
\cinferenceRuleStore[equivl|$\lbisubjunct$\leftrule]{$\lbisubjunct$ left}
{\linferenceRule[sequent]
  {\lsequent[L]{\asfml\limply\bsfml, \bsfml\limply\asfml}{}}
  {\lsequent[L]{\asfml \lbisubjunct \bsfml}{}}
}{}%
\cinferenceRuleStore[id|id]{identity}
{\linferenceRule[sequent]
  {}
  {\lsequent[L]{\asfml}{\asfml}}
}{}%
\cinferenceRuleStore[cut|cut]{cut}
{\linferenceRule[sequent]
  {\lsequent[L]{}{\cusfml}
  &\lsequent[L]{\cusfml}{}}
  {\lsequent[L]{}{}}
}{}%
\cinferenceRuleStore[weakenr|W\rightrule]{weakening right}
{\linferenceRule[sequent]
  {\lsequent[L]{}{}}
  {\lsequent[L]{}{\asfml}}
}{}%
\cinferenceRuleStore[weakenl|W\leftrule]{weakening left}
{\linferenceRule[sequent]
  {\lsequent[L]{}{}}
  {\lsequent[L]{\asfml}{}}
}{}%
\cinferenceRuleStore[exchanger|P\rightrule]{exchange right}
{\linferenceRule[sequent]
  {\lsequent[L]{}{\bsfml,\asfml}}
  {\lsequent[L]{}{\asfml,\bsfml}}
}{}%
\cinferenceRuleStore[exchangel|P\leftrule]{exchange left}
{\linferenceRule[sequent]
  {\lsequent[L]{\bsfml,\asfml}{}}
  {\lsequent[L]{\asfml,\bsfml}{}}
}{}%
\cinferenceRuleStore[contractr|c\rightrule]{contract right}
{\linferenceRule[sequent]
  {\lsequent[L]{}{\asfml,\asfml}}
  {\lsequent[L]{}{\asfml}}
}{}%
\cinferenceRuleStore[contractl|c\leftrule]{contract left}
{\linferenceRule[sequent]
  {\lsequent[L]{\asfml,\asfml}{}}
  {\lsequent[L]{\asfml}{}}
}{}
\cinferenceRuleStore[closeTrue|$\top$\rightrule]{close by true}
{\linferenceRule[sequent]
  {}
  {\lsequent[L]{}{\ltrue}}
}{}%
\cinferenceRuleStore[closeFalse|$\bot$\leftrule]{close by false}
{\linferenceRule[sequent]
  {}
  {\lsequent[L]{\lfalse}{}}
}{}%

\cinferenceRuleStore[CE|CE]{congequiv congruence of equivalences on formulas}
{\linferenceRule[formula]
  {\ausfml \lbisubjunct \busfml}
  {\contextapp{C}{\ausfml} \lbisubjunct \contextapp{C}{\busfml}}
}{}%
\dinferenceRuleStore[CEr|CE\rightrule]{congequiv congruence of equivalences on formulas}
{\linferenceRule[formula]
  {\lsequent[L]{} {\contextapp{C}{\busfml}}
  &\lsequent[g]{} {\ausfml \lbisubjunct \busfml}}
  {\lsequent[L]{} {\contextapp{C}{\ausfml}}}
}{}%
\dinferenceRuleStore[CEl|CE\leftrule]{congequiv congruence of equivalences on formulas}
{\linferenceRule[formula]
  {\lsequent[L]{\contextapp{C}{\busfml}} {}
  &\lsequent[g]{} {\ausfml \lbisubjunct \busfml}}
  {\lsequent[L]{\contextapp{C}{\ausfml}} {}}
}{}%

\cinferenceRuleStore[allr|$\forall$\rightrule]{$\lforall{}{}$ right}
{\linferenceRule[sequent]
  {\lsequent[L]{}{p(y)}}
  {\lsequent[L]{}{\lforall{x}{p(x)}}}
}{\m{y\not\in\Gamma{,}\Delta{,}\lforall{x}{p(x)}}}%
\cinferenceRuleStore[alll|$\forall$\leftrule]{$\lforall{}{}$ left instantiation}
{\linferenceRule[sequent]
  {\lsequent[L]{p(\astrm)}{}}
  {\lsequent[L]{\lforall{x}{p(x)}}{}}
}{arbitrary term $\astrm$}%
\cinferenceRuleStore[existsr|$\exists$\rightrule]{$\lexists{}{}$ right}
{\linferenceRule[sequent]
  {\lsequent[L]{}{p(\astrm)}}
  {\lsequent[L]{}{\lexists{x}{p(x)}}}
}{arbitrary term $\astrm$}%
\cinferenceRuleStore[existsl|$\exists$\leftrule]{$\lexists{}{}$ left}
{\linferenceRule[sequent]
  {\lsequent[L]{p(y)} {}}
  {\lsequent[L]{\lexists{x}{p(x)}} {}}
}{\m{y\not\in\Gamma{,}\Delta{,}\lexists{x}{p(x)}}}%

\cinferenceRuleStore[qear|\usebox{\Rval}]{quantifier elimination real arithmetic}
{\linferenceRule[sequent]
  {}
  {\lsequent[g]{\Gamma}{\Delta}}
}{\text{if}~\landfold_{\ausfml\in\Gamma} \ausfml \limply \lorfold_{\busfml\in\Delta} \busfml ~\text{is valid in \LOS[\reals]}}%

\dinferenceRuleStore[allGi|i$\forall$]{inverse universal generalization / universal instantiation}
{\linferenceRule[sequent]
  {\lsequent[L]{} {\lforall{x}{\ausfml}}}
  {\lsequent[L]{} {\ausfml}}
}{}%

\cinferenceRuleStore[applyeqr|=\rightrule]{apply equation}
{\linferenceRule[sequent]
  {\lsequent[L]{x=\astrm}{p(\astrm)}}
  {\lsequent[L]{x=\astrm}{p(x)}}
}{}%
\cinferenceRuleStore[applyeql|=\leftrule]{apply equation}
{\linferenceRule[sequent]
  {\lsequent[L]{x=\astrm,p(\astrm)}{}}
  {\lsequent[L]{x=\astrm,p(x)}{}}
}{}%

\dinferenceRuleStore[alldupl|$\forall\forall$\leftrule]{$\lforall{}{}$ left instantiation retaining duplicates}
{\linferenceRule[sequent]
  {\lsequent[L]{\lforall{x}{p(x)},p(\astrm)}{}}
  {\lsequent[L]{\lforall{x}{p(x)}}{}}
}{}%

\dinferenceRuleStore[choicebrinsist|$\dibox{\cup}\rightrule$]{}
{\linferenceRule
  {\lsequent[L]{}{\dbox{\asprg}{\ausfml}\land\dbox{\bsprg}{\ausfml}}}
  {\lsequent[L]{}{\dbox{\pchoice{\asprg}{\bsprg}}{\ausfml}}}
}{}
\dinferenceRuleStore[choiceblinsist|$\dibox{\cup}\leftrule$]{}
{\linferenceRule
  {\lsequent[L]{\dbox{\asprg}{\ausfml}\land\dbox{\bsprg}{\ausfml}}{}}
  {\lsequent[L]{\dbox{\pchoice{\asprg}{\bsprg}}{\ausfml}}{}}
}{}
\dinferenceRuleStore[choicebrinsist2|$\dibox{\cup}\rightrule2$]{}
{\linferenceRule
  {\lsequent[L]{}{\dbox{\asprg}{\ausfml}}
  &\lsequent[L]{}{\dbox{\bsprg}{\ausfml}}}
  {\lsequent[L]{}{\dbox{\pchoice{\asprg}{\bsprg}}{\ausfml}}}
}{}
\dinferenceRuleStore[choiceblinsist2|$\dibox{\cup}\leftrule2$]{}
{\linferenceRule
  {\lsequent[L]{\dbox{\asprg}{\ausfml},\dbox{\bsprg}{\ausfml}}{}}
  {\lsequent[L]{\dbox{\pchoice{\asprg}{\bsprg}}{\ausfml}}{}}
}{}
\dinferenceRuleStore[cutr|cut\rightrule]{cut right}
{\linferenceRule[sequent]
  {\lsequent[L]{}{\bsfml}
  &\lsequent[L]{}{\bsfml\limply\asfml}}
  {\lsequent[L]{}{\asfml}}
}{}
\dinferenceRuleStore[cutl|cut\leftrule]{cut left}
{\linferenceRule[sequent]
  {\lsequent[L]{\bsfml} {}
  &\lsequent[L]{}{\asfml\limply\bsfml}}
  {\lsequent[L]{\asfml} {}}
}{}

\cinferenceRuleStore[Dplus|$+'$]{derive sum}
{\linferenceRule[eq]
  {\der{\asdtrm}+\der{\bsdtrm}}
  {\axkey{\der{\asdtrm+\bsdtrm}}}
}
{}
\cinferenceRuleStore[Dplusax|$+'$]{derive sum}
{\linferenceRule[eq]
  {\der{\asdtrmax}+\der{\bsdtrmax}}
  {\axkey{\der{\asdtrmax+\bsdtrmax}}}
}
{}
\cinferenceRuleStore[Dminus|$-'$]{derive minus}
{\linferenceRule[eq]
  {\der{\asdtrm}-\der{\bsdtrm}}
  {\axkey{\der{\asdtrm-\bsdtrm}}}
}
{}
\cinferenceRuleStore[Dminusax|$-'$]{derive minus}
{\linferenceRule[eq]
  {\der{\asdtrmax}-\der{\bsdtrmax}}
  {\axkey{\der{\asdtrmax-\bsdtrmax}}}
}
{}
\cinferenceRuleStore[Dtimes|$\cdot'$]{derive product}
{\linferenceRule[eq]
  {\der{\asdtrm}\cdot \bsdtrm+\asdtrm\cdot\der{\bsdtrm}}
  {\axkey{\der{\asdtrm\cdot \bsdtrm}}}
}
{}
\cinferenceRuleStore[Dtimesax|$\cdot'$]{derive product}
{\linferenceRule[eq]
  {\der{\asdtrmax}\cdot \bsdtrmax+\asdtrmax\cdot\der{\bsdtrmax}}
  {\axkey{\der{\asdtrmax\cdot \bsdtrmax}}}
}
{}
\cinferenceRuleStore[Dquotient|$/'$]{derive quotient}
{\linferenceRule[eq]
  {\big(\der{\asdtrm}\cdot \bsdtrm-\asdtrm\cdot\der{\bsdtrm}\big) / \bsdtrm^2}
  {\axkey{\der{\asdtrm/\bsdtrm}}}
}
{}
\cinferenceRuleStore[Dquotientax|$/'$]{derive quotient}
{\linferenceRule[eq]
  {\big(\der{\asdtrmax}\cdot \bsdtrmax-\asdtrmax\cdot\der{\bsdtrmax}\big) / \bsdtrmax^2}
  {\axkey{\der{\asdtrmax/\bsdtrmax}}}
}
{}
\cinferenceRuleStore[Dconst|$c'$]{derive constant}
{\linferenceRule[eq]
  {0}
  {\axkey{\der{\aconst}}}
  \hspace{3cm}
}
{\text{for numbers or constants~$\aconst$}}%
\cinferenceRuleStore[Dvar|$x'$]{derive variable}
{\linferenceRule[eq]
  {\D{x}}
  {\axkey{\der{x}}}
}
{\text{for variable~$x\in\allvars$}}%

\cinferenceRuleStore[DE|DE]{differential effect} %
{\linferenceRule[viuqe]
  {\axkey{\dbox{\pevolvein{\D{x}=\genDE{x}}{\ivr}}{\ausfml}}}
  {\dbox{\pevolvein{\D{x}=\genDE{x}}{\ivr}}{\dbox{\axeffect{\Dupdate{\Dumod{\D{x}}{\genDE{x}}}}}{\ausfml}}}
}
{}%
\cinferenceRuleStore[DEax|DE]{differential effect} %
{\linferenceRule[viuqe]
  {\axkey{\dbox{\pevolvein{\D{x}=\genDE{x}}{q(x)}}{\ausfmlax}}}
  {\dbox{\pevolvein{\D{x}=\genDE{x}}{q(x)}}{\dbox{\axeffect{\Dupdate{\Dumod{\D{x}}{\genDE{x}}}}}{\ausfmlax}}}
}
{}%

\cinferenceRuleStore[Dand|${\land}'$]{derive and}
{\linferenceRule[equiv]
  {\der{\asfml}\land\der{\bsfml}}
  {\axkey{\der{\asfml\land\bsfml}}}
}
{}
\cinferenceRuleStore[Dor|${\lor}'$]{derive or}
{\linferenceRule[equiv]
  {\der{\asfml}\land\der{\bsfml}}
  {\axkey{\der{\asfml\lor\bsfml}}}
}
{}
\cinferenceRuleStore[diffweaken|DW]{differential evolution domain} %
{\linferenceRule[viuqe]
  {\axkey{\dbox{\pevolvein{\D{x}=\genDE{x}}{\ivr}}{\ousfml[x]}}}
  {\dbox{\pevolvein{\D{x}=\genDE{x}}{\ivr}}{(\axeffect{\ivr}\limply \ousfml[x])}}
}
{}%
\cinferenceRuleStore[diffweakenax|DW]{differential evolution domain} 
{\linferenceRule[viuqe]
  {\axkey{\dbox{\pevolvein{\D{x}=\genDE{x}}{q(x)}}{p(x)}}}
  {\dbox{\pevolvein{\D{x}=\genDE{x}}{q(x)}}{(\axeffect{q(x)}\limply p(x))}}
}
{}%
\cinferenceRuleStore[dW|dW]{differential weakening}
{\linferenceRule[sequent]
  {\lsequent[g]{\ivr} {\ousfml[x]}}
  {\lsequent[g]{\Gamma} {\dbox{\pevolvein{\D{x}=f(x)}{\ivr}}{\ousfml[x]},\Delta}}
}
{}%
\cinferenceRuleStore[DI|DI]{differential induction}
{\linferenceRule[lpmi]
  {\big(\axkey{\dbox{\pevolvein{\D{x}=\genDE{x}}{\ivr}}{\ousfml[x]}}
  \lbisubjunct \dbox{\ptest{\ivr}}{\ousfml[x]}\big)}
  {(\ivr\limply\dbox{\pevolvein{\D{x}=\genDE{x}}{\ivr}}{\axeffect{\der{\ousfml[x]}}})}
}
{}%
\cinferenceRuleStore[DIax|DI]{differential induction}
{\linferenceRule[lpmi]
  {\big(\axkey{\dbox{\pevolvein{\D{x}=\genDE{x}}{q(x)}}{p(x)}}
  \lbisubjunct \dbox{\ptest{q(x)}}{p(x)}\big)}
  {(q(x)\limply\dbox{\pevolvein{\D{x}=\genDE{x}}{q(x)}}{\axeffect{\der{p(x)}}})}
}
{}%
\cinferenceRuleStore[DIlight|DI]{differential induction}
{\linferenceRule[lpmi]
  {\big(\axkey{\dbox{\pevolvein{\D{x}=\genDE{x}}{\ivr}}{\ousfml[x]}}
  \lbisubjunct \dbox{\ptest{\ivr}}{\ousfml[x]}\big)}
  {\dbox{\pevolvein{\D{x}=\genDE{x}}{\ivr})}{\axeffect{\der{\ousfml[x]}}}}
}
{}%

\cinferenceRuleStore[dI|dI]{differential invariant}
{\linferenceRule[sequent]
  {\lsequent[g]{\ivr}{\Dusubst{\D{x}}{\genDE{x}}{\der{F}}}}
  {\lsequent{F}{\dbox{\pevolvein{\D{x}=\genDE{x}}{\ivr}}{F}}}
}{}
\cinferenceRuleStore[DC|DC]{differential cut}
{\linferenceRule[lpmi]
  {\big(\axkey{\dbox{\pevolvein{\D{x}=\genDE{x}}{\ivr}}{\ousfml[x]}} \lbisubjunct \dbox{\pevolvein{\D{x}=\genDE{x}}{\ivr\land \axeffect{\ousfmlc[x]}}}{\ousfml[x]}\big)}
  {\dbox{\pevolvein{\D{x}=\genDE{x}}{\ivr}}{\axeffect{\ousfmlc[x]}}}
}
{}%
\cinferenceRuleStore[DCax|DC]{differential cut}
{\linferenceRule[lpmi]
  {\big(\axkey{\dbox{\pevolvein{\D{x}=\genDE{x}}{q(x)}}{p(x)}} \lbisubjunct \dbox{\pevolvein{\D{x}=\genDE{x}}{q(x)\land \axeffect{r(x)}}}{p(x)}\big)}
  {\dbox{\pevolvein{\D{x}=\genDE{x}}{q(x)}}{\axeffect{r(x)}}}
}
{}%
\cinferenceRuleStore[dC|dC]{differential cut}%
{\linferenceRule[sequent]
  {\lsequent[L]{}{\dbox{\pevolvein{\D{x}=\genDE{x}}{\ivr}}{\axeffect{\cusfml}}}
  &\lsequent[L]{}{\dbox{\pevolvein{\D{x}=\genDE{x}}{(\ivr\land \axeffect{\cusfml})}}{\ousfml[x]}}}
  {\lsequent[L]{}{\dbox{\pevolvein{\D{x}=\genDE{x}}{\ivr}}{\ousfml[x]}}}
}{}
\cinferenceRuleStore[DGanyode|DG]{differential ghost variables (unsound!)}
{\linferenceRule[viuqe]
  {\axkey{\dbox{\pevolvein{\D{x}=\genDE{x}}{\ivr}}{\ousfml[x]}}}
  {\lexists{y}{\dbox{\pevolvein{\D{x}=\genDE{x}\syssep\axeffect{\D{y}=g(x,y)}}{\ivr}}{\ousfml[x]}}}
}
{}
\cinferenceRuleStore[DG|DG]{differential ghost variables}
{\linferenceRule[viuqe]
  {\axkey{\dbox{\pevolvein{\D{x}=\genDE{x}}{\ivr}}{\ousfml[x]}}}
  {\lexists{y}{\dbox{\pevolvein{\D{x}=\genDE{x}\syssep\axeffect{\D{y}=a(x)\cdot y+b(x)}}{\ivr}}{\ousfml[x]}}}
}
{}
\cinferenceRuleStore[DGax|DG]{differential ghost variables}
{\linferenceRule[viuqe]
  {\axkey{\dbox{\pevolvein{\D{x}=\genDE{x}}{q(x)}}{p(x)}}}
  {\lexists{y}{\dbox{\pevolvein{\D{x}=\genDE{x}\syssep\axeffect{\D{y}=a(x)\cdot y+b(x)}}{q(x)}}{p(x)}}}
}
{}
\cinferenceRuleStore[dG|dG]{dG}%
{\linferenceRule[sequent]
  {\lsequent[L]{} {\lexists{y}{\dbox{\pevolvein{\D{x}=f(x)\syssep\axeffect{\D{y}=a(x)\cdot y+b(x)}}{\oivr[x]}}{\ousfml[x]}}}
  }
  {\lsequent[L]{} {\dbox{\pevolvein{\D{x}=f(x)}{\oivr[x]}}{\ousfml[x]}}}
}
{}%

\dinferenceRuleStore[assignbeqr|$\dibox{:=}_=$]{assignb}%
  {\linferenceRule[sequent]
    {\lsequent[L]{y=\austrm} {p(y)}}
    {\lsequent[L]{} {\dbox{\pupdate{\umod{x}{\austrm}}}{p(x)}}}
  }
  {\text{$y$ new}}

\cinferenceRuleStore[DSax|DS]{(constant) differential equation solution} 
{\linferenceRule[viuqe]
  {\axkey{\dbox{\pevolvein{\D{x}=\aconst}{q(x)}}{p(x)}}}
  {\lforall{t{\geq}0}{\big((\lforall{0{\leq}s{\leq}t}{q(x+\aconst\itimes s)}) \limply \dbox{\pupdate{\pumod{x}{x+\aconst\itimes t}}}{p(x)}\big)}}
}
{}

\dinferenceRuleStore[DIeq0|DI]{differential invariant axiom}
{\linferenceRule[lpmi]
  {\big(\axkey{\dbox{\pevolve{\D{x}=\genDE{x}}}{\,\astrm=0}} \lbisubjunct \astrm=0\big)}
  {\dbox{\pevolve{\D{x}=\genDE{x}}}{\,\axeffect{\der{\astrm}=0}}}
}
{}
\dinferenceRuleStore[diffindeq0|dI]{differential invariant $=0$ case}
{\linferenceRule[sequent]
  {\lsequent{~}{\Dusubst{\D{x}}{\genDE{x}}{\der{\astrm}}=0}}
  {\lsequent{\astrm=0}{\dbox{\pevolve{\D{x}=\genDE{x}}}{\astrm=0}}}
}{}
\cinferenceRuleStore[Liec|dI$_c$]{}
{\linferenceRule
  {\lsequent{\ivr}{\Dusubst{\D{x}}{\genDE{x}}{\der{\astrm}}=0}}
  {\lsequent{}{\lforall{c}{\big(\astrm=c \limply \dbox{\pevolvein{\D{x}=\genDE{x}}{\ivr}}{\astrm=c}\big)}}}
}{}

\cinferenceRuleStore[DIeq|DI$_=$]{differential induction $=$ case}
{\linferenceRule[lpmi]
  {\big(\axkey{\dbox{\pevolvein{\D{x}=\genDE{x}}{\ivr}}{\asdtrm=\bsdtrm}}
  \lbisubjunct \dbox{\ptest{\ivr}}{\asdtrm=\bsdtrm}\big)}
  {\dbox{\pevolvein{\D{x}=\genDE{x}}{\ivr})}{\axeffect{\der{\asdtrm}=\der{\bsdtrm}}}}
}
{}%

\dinferenceRuleStore[diffindgen|dI']{differential invariant}
{\linferenceRule[sequent]
  {\lsequent[L]{}{\inv}
  &\lsequent[g]{\ivr}{\Dusubst{\D{x}}{\genDE{x}}{\der{\inv}}}
  &\lsequent[g]{\inv}{\psi}
  }
  {\lsequent[L]{}{\dbox{\pevolvein{\D{x}=\genDE{x}}{\ivr}}{\psi}}}
}{}
\cinferenceRuleStore[diffindunsound|dI$_{??}$]{unsound}
{\linferenceRule[sequent]
  {\lsequent{\ivr\land\inv}{\Dusubst{\D{x}}{\genDE{x}}{\der{\inv}}}}
  {\lsequent{\inv}{\dbox{\pevolvein{\D{x}=\genDE{x}}{\ivr}}{\inv}}}
}{}

\dinferenceRuleStore[introaux|iG]{introduce discrete ghost variable}
{\linferenceRule[sequent]
  {\lsequent[L]{}{\dbox{\axeffect{\pupdate{\pumod{y}{\astrm}}}}{p}}}
  {\lsequent[L]{} {p}}
}{\text{$y$ new}}%
\dinferenceRuleStore[diffaux|dA]{differential auxiliary variables}
{\linferenceRule[sequent]
  {\lsequent[\globalrule]{}{\inv\lbisubjunct\lexists{y}{G}}
  &\lsequent{G} {\dbox{\pevolvein{\D{x}=\genDE{x}\syssep\axeffect{\D{y}=a(x)\cdot y+b(x)}}{\ivr}}{G}}}
  {\lsequent{\inv} {\dbox{\pevolvein{\D{x}=\genDE{x}}{\ivr}}{\inv}}}
}{}%
\cinferenceRuleStore[randomd|$\didia{{:}*}$]{nondeterministic assignment}
{\linferenceRule[equiv]
  {\lexists{x}{\ousfml[x]}}
  {\axkey{\ddiamond{\prandom{x}}{\ousfml[x]}}}
}{}
\cinferenceRuleStore[randomb|$\dibox{{:}*}$]{nondeterministic assignment}
{\linferenceRule[equiv]
  {\lforall{x}{\ousfml[x]}}
  {\axkey{\dbox{\prandom{x}}{\ousfml[x]}}}
}{}

\cinferenceRuleStore[box|$\dibox{\cdot}$]{box axiom}
{\linferenceRule[equiv]
  {\lnot\ddiamond{\ausprg}{\lnot\ausfml}}
  {\axkey{\dbox{\ausprg}{\ausfml}}}
}
{}
\cinferenceRuleStore[assignd|$\didia{:=}$]{assignment / substitution axiom}
{\linferenceRule[equiv]
  {p(\genDJ{x})}
  {\axkey{\ddiamond{\pupdate{\umod{x}{\genDJ{x}}}}{p(x)}}}
}
{}%
\cinferenceRuleStore[evolved|$\didia{'}$]{evolve}
{\linferenceRule[equiv]
  {\lexists{t{\geq}0}{\ddiamond{\pupdate{\pumod{x}{\solf(t)}}}{p(x)}}\hspace{1cm}}
  {\axkey{\ddiamond{\pevolve{\D{x}=\genDE{x}}}{p(x)}}}
}{\m{\D{\solf}(t)=\genDE{\solf}}}%
\cinferenceRuleStore[evolveind|$\didia{'}$]{evolve}
{\linferenceRule[equiv]
  {\lexists{t{\geq}0}{\big((\lforall{0{\leq}s{\leq}t}{q(\solf(s))}) \land 
  \ddiamond{\pupdate{\pumod{x}{\solf(t)}}}{p(x)}\big)}}
  {\axkey{\ddiamond{\pevolvein{\D{x}=\genDE{x}}{q(x)}}{p(x)}}}
}{\m{\D{\solf}(t)=\genDE{\solf}}}%
\cinferenceRuleStore[testd|$\didia{?}$]{test}
{\linferenceRule[equiv]
  {\ivr \land \ausfml}
  {\axkey{\ddiamond{\ptest{\ivr}}{\ausfml}}}
}{}
\cinferenceRuleStore[choiced|$\didia{\cup}$]{axiom of nondeterministic choice}
{\linferenceRule[equiv]
  {\ddiamond{\ausprg}{\ausfml} \lor \ddiamond{\busprg}{\ausfml}}
  {\axkey{\ddiamond{\pchoice{\ausprg}{\busprg}}{\ausfml}}}
}{}
\cinferenceRuleStore[composed|$\didia{{;}}$]{composition}
{\linferenceRule[equiv]
  {\ddiamond{\ausprg}{\ddiamond{\busprg}{\ausfml}}}
  {\axkey{\ddiamond{\ausprg;\busprg}{\ausfml}}}
}{}
\cinferenceRuleStore[iterated|$\didia{{}^*}$]{iteration/repeat unwind pre-fixpoint, even fixpoint}
{\linferenceRule[equiv]
  {\ausfml \lor \ddiamond{\ausprg}{\ddiamond{\prepeat{\ausprg}}{\ausfml}}}
  {\axkey{\ddiamond{\prepeat{\ausprg}}{\ausfml}}}
}{}
\cinferenceRuleStore[duald|$\didia{{^d}}$]{dual}
{\linferenceRule[equiv]
  {\lnot\ddiamond{\ausprg}{\lnot\ausfml}}
  {\axkey{\ddiamond{\pdual{\ausprg}}{\ausfml}}}
}{}
\cinferenceRuleStore[dualb|$\dibox{{^d}}$]{dual}
{\linferenceRule[equiv]
  {\lnot\dbox{\ausprg}{\lnot\ausfml}}
  {\axkey{\dbox{\pdual{\ausprg}}{\ausfml}}}
}{}
\cinferenceRuleStore[FP|FP]{iteration is least fixpoint / reflexive transitive closure RTC, equivalent to invind in the presence of M}
{\linferenceRule[formula]
  {\ausfml \lor \ddiamond{\ausprg}{\busfml} \limply \busfml}
  {\ddiamond{\prepeat{\ausprg}}{\ausfml} \limply \busfml}
}{}
\cinferenceRuleStore[invindg|ind]{inductive invariant for games}
{\linferenceRule[formula]
  {\ausfml\limply\dbox{\ausprg}{\ausfml}}
  {\ausfml\limply\dbox{\prepeat{\ausprg}}{\ausfml}}
}{}%

\dinferenceRuleStore[dchoiced|$\didia{{\cap}}$]{Demon's choice}
{
\axkey{\ddiamond{\dchoice{\ausprg}{\busprg}}{\ausfml}} \lbisubjunct \ddiamond{\ausprg}{\ausfml} \land \ddiamond{\busprg}{\ausfml}
}{}
\dinferenceRuleStore[dchoiceb|$\dibox{{\cap}}$]{Demon's choice}
{
\axkey{\dbox{\dchoice{\ausprg}{\busprg}}{\ausfml}} \lbisubjunct \dbox{\ausprg}{\ausfml} \lor \dbox{\busprg}{\ausfml}
}{}
\dinferenceRuleStore[diterateb|$\dibox{\drepeat{}}$]{Demon's repetition}
{\linferenceRule[equiv]
  {\ausfml \lor \dbox{\ausprg}{\dbox{\drepeat{\ausprg}}{\ausfml}}}
  {\axkey{\dbox{\drepeat{\ausprg}}{\ausfml}}}
}{}
\dinferenceRuleStore[diterated|$\didia{\drepeat{}}$]{Demon's repetition}
{\linferenceRule[equiv]
  {\ausfml \land \ddiamond{\ausprg}{\ddiamond{\drepeat{\ausprg}}}{\ausfml}}
  {\axkey{\ddiamond{\drepeat{\ausprg}}{\ausfml}}}
}{}
\dinferenceRuleStore[dinvindg|ind$\drepeat{}$]{inductive invariant for games}
{\linferenceRule[formula]
  {\ausfml\limply\ddiamond{\ausprg}{\ausfml}}
  {\ausfml\limply\ddiamond{\drepeat{\ausprg}}{\ausfml}}
}{}
\dinferenceRuleStore[dFP|FP$\drepeat{}$]{dual iteration is least fixpoint in Demon's winning strategy}
{\linferenceRule[formula]
  {\ausfml \lor \dbox{\ausprg}{\busfml} \limply \busfml}
  {\dbox{\drepeat{\ausprg}}{\ausfml} \limply \busfml}
}{}

\cinferenceRuleStore[US|US]{uniform substitution}
{\linferenceRule[formula]
  {\phi}
  {\applyusubst{\sigma}{\phi}}
}{}%

\cinferenceRuleStore[linequs|$\exists$lin]{linear equation uniform substitution}
{\linferenceRule[impl]
  {b\neq0}
  {\big(\lexists{x}{(b\cdot x+c=0 \land q(x))}
  \lbisubjunct {q(-c/b)}\big)}
}{}

\dinferenceRuleStore[FA|FA]{First arrival}
{\ddiamond{\prepeat{\ausprg}}{\ausfml} \limply \ausfml \lor \ddiamond{\prepeat{\ausprg}}{(\lnot\ausfml\land\ddiamond{\ausprg}{\ausfml})}
}{}
\dinferenceRuleStore[Mor|M]{monotonicity axiom}
{\ddiamond{\ausprg}{(\ausfml\lor\busfml)}
\lbisubjunct
\ddiamond{\ausprg}{\ausfml} \lor \ddiamond{\ausprg}{\busfml}
}{}
\dinferenceRuleStore[VK|VK]{vacuous possible $\dbox{}{}$}
{\linferenceRule[impl]
  {p}
  {(\dbox{\ausprg}{\ltrue}{\limply}\dbox{\ausprg}{p})}
  \qquad
}{\m{\freevars{p}\cap \boundvars{\ausprg}=\emptyset}}
\dinferenceRuleStore[R|R]{Regular}
{\linferenceRule[formula]
  {\ausfml_1\land\ausfml_2\limply\busfml}
  {\dbox{\ausprg}{\ausfml_1} \land \dbox{\ausprg}{\ausfml_2} \limply \dbox{\ausprg}{\busfml}}
}{}

\newcommand{\solf}{y}
\begin{figure}[tbhp]
  \centering
  \begin{calculuscollections}{\columnwidth}
    \begin{calculus}
      \cinferenceRuleQuote{box}
      \cinferenceRuleQuote{assignd}
      \cinferenceRuleQuote{evolved}
      \cinferenceRuleQuote{testd}
      \cinferenceRuleQuote{choiced}
      \cinferenceRuleQuote{composed}
      \cinferenceRuleQuote{iterated}
      \cinferenceRuleQuote{duald}
      \cinferenceRuleQuote{dchoiced}
      \cinferenceRuleQuote{diterated}
      \cinferenceRuleQuote{assignb}
      \cinferenceRuleQuote{evolveb}
      \cinferenceRuleQuote{testb}
      \cinferenceRuleQuote{choiceb}
      \cinferenceRuleQuote{composeb}
      \cinferenceRuleQuote{iterateb}
      \cinferenceRuleQuote{dualb}
      \cinferenceRuleQuote{dchoiceb}
      \cinferenceRuleQuote{diterateb}
    \end{calculus}

    \begin{calculus}
      \cinferenceRuleQuote{loop}
      \cinferenceRuleQuote{M}
      \cinferenceRuleQuote{Mb}
    \end{calculus}
    \begin{calculus}
      \dinferenceRuleQuote{invindg}
      \cinferenceRuleQuote{FP}
      \cinferenceRuleQuote{dFP}
    \end{calculus}
  \end{calculuscollections}
  \caption{\dGL axiomatization and derived axioms and rules}
  \label{fig:dGL-calculus}
\end{figure}

\section{Benchmarks}\label{app:benchmarks}

This appendix discusses the benchmarks, provides their full listings, and discusses the rewrite heuristics used to solve them.

\subsection{New Examples Discussion}\label{app:benchmarks-discussion}

We discuss the examples representing diverse control challenges introduced in this paper to demonstrate the flexibility of our \dGL control envelope synthesis approach.

\emph{Event-triggered ETCS} (\rref{model:event-triggered-etcs}) models a train from a case study modeling the European Train Control System (ETCS) \cite{DBLP:conf/icfem/PlatzerQ09}, but with a key modification: the train is modeled event triggered instead of time triggered.
The control envelope computed by our implementation finds the winning subregion for the overall loop to be $p < e \land (2Ae+v^2 < 2Ap \lor v \leq 0) \land (A \leq 0 \lor v > 0) \lor p < e \land (2(-B)e+v^2 < 2(-B)p \lor v \leq 0) \land (-B \leq 0 \lor v > 0)$.
After simplification, the key term ends up being $e-p + v^2/2B > 0$, i.e., there is still enough time for the train to come to a stop before the end of motion authority $e$ if it starts braking now.
Hence, the control envelope has a strategy to keep the train safe whenever it is possible for the train to be safe by braking.

\emph{Surgical Robot} (\rref{model:surgical-robot}) is based on a case study from the literature \cite{DBLP:conf/hybrid/KouskoulasRPK13} modeling force feedback in surgical robots.
The control envelope computed by our implementation finds the winning subregion for the overall loop to be $(q_x-p_x)n_x+(q_y-p_y)n_y>=0$, i.e., so long as the robot is not already outside the fixture, the subvalue map has a strategy for safe control.

\emph{Infinite Track} (\rref{model:looping-track}) is an example of infinite horizon switching.
A vehicle in a looping track can choose to move in any of the four coordinate directions, but must remain moving at all times.
To remain safe for arbitrarily long time, this vehicle must keep switching directions at the right moments, and make \emph{infinitely} many switches.
Our approach synthesizes a control envelope that indicates when it is safe to switch to any given direction.
CESAR does not solve this problem as it does not handle infinite control strategies.
The control envelope computed by our implementation finds the winning subregion for the overall loop to be $2R > x \land x > -2R \land 2R > y \land y > -2R \land  (x > R \lor x < -R \lor y > R \lor y < -R)$, i.e., the subvalue map has a safe strategy for the vehicle at every point in the track.

\emph{Reach-avoid Robot} (\rref{model:reach-avoid}) demonstrates an example of envelope synthesis for reach-avoid problems, where an agent must \emph{reach} an objective while \emph{avoiding} unsafe situations.
Such problems can model the simultaneous requirement of safety and liveness.
In this example, the robot is safe in the square $[-2R, 2R] \times [-2R, 2R]$ (enforced in domain constraint of \rref{line:reach-avoid-plant}) and must reach the target region $[R, 2R] \times [R, 2R]$ (on \rref{line:reach-avoid-safety}).
The robot can either travel upward ($v_x=0, v_y=V$) or leftward ($v_x=-V, v_y=0$).
The synthesized subvalue map identifies that it is safely possible to reach the target region starting in $[R, 2R] \times [-2R, 2R]$.

\emph{Highway Driving} (\rref{model:highway-driving}) is a time-triggered problem based on the core control challenge from a case study on highway driving \cite{DBLP:conf/fm/LoosPN11}.
Two cars, are driving on a highway.
The lead car, with velocity $v_l$ and acceleration $a_l$ is driven by some external agent with the physical limitation that the car's accelerates is bounded above by $A$ and its braking is bounded below by $-B$ (\rref{line:highway-lead}).
The following car, with velocity $v_f$ and acceleration $a_f$, is driven by a controller that must follow the lead car while avoiding collisions (\rref{line:highway-follow}).
The controllers revise decisions in a time triggered fashion, with maximum latency $T$.
Regardless of the lead car's behavior, the following car must always be able to stop in time to avoid collision (\rref{line:highway-safety}).
This problem does not fit CESAR's template, which does not support an adversarial agent like the lead car.
In the computed control envelope, the winning region for the overall game is $p_f < p_l \land v_f \leq v_l$, i.e. the subvalue map has a strategy to avoid the collisions whenever that velocity of the controlled car is less than that of the lead car and hasn't already collided with it.

The descriptions of the CESAR benchmarks can be found in \cite{DBLP:conf/tacas/KabraLMP24}.
The envelopes that we synthesize for these benchmarks are the same as those computed by CESAR.

\subsection{Quadcopter}\label{app:quadcopter}

The procedurally generated Quadcopter benchmark suite is based on the following template with the addition of randomly generated values for velocity and obstacles.

\begin{model}[h]
  \setcounter{modelline}{0}
  \caption{Quadcopter Suite Template}
  \label{model:quadcopter-template}
  \begin{align*}
    \text{\assumptions} &\,\big| \mline{line:quadcopter-assumptions} x > 0 \land V = \texttt{<velocity>} \limply \big\langle \\
    \text{\kwd{direction}} &\,\big| \mline{line:quadcopter-direction} \Big( \left( v_y := 1 \right) \cup \left( v_y := -1 \right) \big) \seq \\
    \text{\kwd{plant}} &\,\big| \mline{line:quadcopter-plant} \phantom{\Big(}t := 0; \left\{ x' = V, y' = v_y, t' = 1\ \&\ \texttt{<avoid obstacles>} \land y\geq 0 \right\} \seq \\
    \text{\safe} &\,\big| \mline{line:quadcopter-safety} \phantom{\Big(} \ptest{t \geq 1} \Big)^\ast \big\rangle\, x > 20
  \end{align*}
\end{model}

The Quadcopter chooses whether to go up or down and runs continuous dynamics.
The catch is that the quadcopter has a slow processor and cannot make a new choice in less than 1 second (\rref{line:quadcopter-safety}).
However, it can choose to run the dynamics for any time of its choice that is greater than 1 second (\rref{line:quadcopter-plant} has an Angel ODE).
It can repeat this process of choosing and then sticking with a choice any number of times (Angel loop in \rref{line:quadcopter-safety}).
While it flies, it must avoid obstacles at all times.
It must also not crash into the floor at $y>0$.
It must eventually reach the target region $x > 20$.
CESAR cannot solve this problem since it lies outside the template.
NYCS solves all instances.

An example of the formula that can be generated for \texttt{<avoid obstacles>} obstacles is $(x<2 \lor 2<x \lor y<0 \lor 3<y) \land (x<5 \lor 6<x \lor y<0 \lor 3<y) \land (x<5 \lor 8<x \lor y<9 \lor 9<y) \land (x<8 \lor 10<x \lor y<7 \lor 7<y)$.
This formula has four obstacles.
The first one is a line segment from $(2,0)$ to $(2,3)$.
The second is a rectangle from $(5,0)$ to $(6,3)$.
The third is a rectangle from $(5,9)$ to $(8,9)$.
The fourth is a rectangle from $(8,7)$ to $(10,7)$.
An example of a generated number that fills \texttt{<velocity>} is $1$.
The procedural generation parameters allow up to 5 obstacles with widths and heights of at most 3 distributed in the region $[0,10]\times[0,10]$.
Velocity is an integer between 1 and 5.

As expected, the implementation solves problems with fewer obstacles quickly, with the fastest solutions for one obstacle completing in a couple of seconds.
However, the time taken increases with the number of obstacles, and more importantly, when obstacles are placed in such a way that they interact to make the solution space more complex.
\rref{tab:quadcopter} summarizes the generated problems and their solving outcomes.

\begin{table}
  \caption{Quadcopter benchmark suite generated parameters and solving outcomes. In the result column, $\checkmark$ indicates that a nonempty control envelope was computed, while \ding{55} indicates that algorithm timed out after 20 minutes, failing to compute an envelope.}
  \label{tab:quadcopter}
  \rowcolors{2}{white}{gray!25}
  \begin{tabularx}{\textwidth}{c|c|X|c}
    \textbf{No.} & \textbf{Velocity} & \textbf{Obstacles} & \textbf{Result} \\
    \hline
    1 & 5 & $x<9 \lor 11<x \lor y<1 \lor 3<y$ & $\checkmark$ \\
    2 & 4 & $x<9 \lor 9<x \lor y<3 \lor 3<y$ & $\checkmark$ \\
    3 & 1 & $(x<2 \lor 2<x \lor y<0 \lor 3<y)$ $\land$ $(x<5 \lor 6<x \lor y<0 \lor 3<y)$ $\land$ $(x<5 \lor 8<x \lor y<9 \lor 9<y)$ $\land$ $(x<8 \lor 10<x \lor y<7 \lor 7<y)$ & \ding{55} \\
    4 & 1 & $x<3 \lor 6<x \lor y<9 \lor 10<y$ & $\checkmark$ \\
    5 & 2 & $x<0 \lor 1<x \lor y<2 \lor 4<y$ & $\checkmark$ \\
    6 & 2 & $(x<7 \lor 10<x \lor y<8 \lor 8<y)$ $\land$ $(x<3 \lor 5<x \lor y<8 \lor 9<y)$ & $\checkmark$ \\
    7 & 4 & $(x<7 \lor 8<x \lor y<1 \lor 3<y)$ $\land$ $(x<5 \lor 5<x \lor y<6 \lor 8<y)$ & \ding{55} \\
    8 & 1 & $(x<3 \lor 3<x \lor y<8 \lor 11<y)$ $\land$ $(x<0 \lor 0<x \lor y<1 \lor 1<y)$ & $\checkmark$ \\
    9 & 5 & $(x<1 \lor 3<x \lor y<7 \lor 9<y)$ $\land$ $(x<4 \lor 4<x \lor y<0 \lor 3<y)$ & \ding{55} \\
    10 & 1 & $\top$ & $\checkmark$ \\
    11 & 2 & $x<3 \lor 4<x \lor y<6 \lor 8<y$ & $\checkmark$ \\
    12 & 2 & $(x<10 \lor 13<x \lor y<10 \lor 11<y)$ $\land$ $(x<5 \lor 5<x \lor y<7 \lor 8<y)$ $\land$ $(x<7 \lor 10<x \lor y<10 \lor 12<y)$ $\land$ $(x<3 \lor 4<x \lor y<7 \lor 8<y)$ $\land$ $(x<2 \lor 2<x \lor y<2 \lor 4<y)$ & \ding{55} \\
    13 & 3 & $(x<2 \lor 2<x \lor y<7 \lor 7<y)$ $\land$ $(x<8 \lor 11<x \lor y<7 \lor 8<y)$ $\land$ $(x<1 \lor 2<x \lor y<9 \lor 12<y)$ & \ding{55} \\
    14 & 4 & $x<9 \lor 11<x \lor y<4 \lor 7<y$ $\land$ $x<3 \lor 3<x \lor y<8 \lor 8<y$ $\land$ $x<9 \lor 11<x \lor y<8 \lor 11<y$ $\land$ $x<8 \lor 11<x \lor y<4 \lor 7<y$ $\land$ $x<7 \lor 8<x \lor y<9 \lor 10<y$ & $\checkmark$ \\
    15 & 3 & $(x<1 \lor 4<x \lor y<9 \lor 10<y)$ $\land$ $(x<7 \lor 10<x \lor y<0 \lor 3<y)$ $\land$ $(x<8 \lor 10<x \lor y<8 \lor 11<y)$ $\land$ $(x<3 \lor 6<x \lor y<3 \lor 4<y)$ & \ding{55} \\
    16 & 4 & $(x<1 \lor 1<x \lor y<8 \lor 11<y)$ $\land$ $(x<8 \lor 10<x \lor y<6 \lor 9<y)$ $\land$ $(x<5 \lor 8<x \lor y<8 \lor 11<y)$ $\land$ $(x<6 \lor 6<x \lor y<3 \lor 4<y)$ $\land$ $(x<10 \lor 13<x \lor y<10 \lor 13<y)$ & \ding{55} \\
    17 & 4 & $\top$ & $\checkmark$ \\
    18 & 3 & $(x<5 \lor 5<x \lor y<4 \lor 4<y)$ $\land$ $(x<1 \lor 2<x \lor y<5 \lor 8<y)$ $\land$ $(x<8 \lor 11<x \lor y<9 \lor 12<y)$ $\land$ $(x<2 \lor 2<x \lor y<1 \lor 3<y)$ $\land$ $(x<6 \lor 7<x \lor y<1 \lor 2<y)$ & \ding{55} \\
    19 & 3 & $(x<4 \lor 6<x \lor y<0 \lor 0<y)$ $\land$ $(x<8 \lor 8<x \lor y<9 \lor 10<y)$ $\land$ $(x<2 \lor 2<x \lor y<2 \lor 2<y)$ $\land$ $(x<3 \lor 3<x \lor y<0 \lor 2<y)$ $\land$ $(x<6 \lor 9<x \lor y<8 \lor 11<y)$ & $\checkmark$ \\
    20 & 2 & $(x<2 \lor 3<x \lor y<9 \lor 11<y)$ $\land$ $(x<10 \lor 13<x \lor y<4 \lor 5<y)$ $\land$ $(x<9 \lor 11<x \lor y<8 \lor 11<y)$ & $\checkmark$ \\
    21 & 3 & $(x<4 \lor 4<x \lor y<5 \lor 5<y )$ $\land$ $( x<5 \lor 6<x \lor y<7 \lor 9<y )$ $\land$ $( x<10 \lor 13<x \lor y<3 \lor 4<y )$ & $\checkmark$ \\
    22 & 4 & $(x<6 \lor 8<x \lor y<3 \lor 4<y)$ & $\checkmark$ \\
    23 & 3 & $(x<5 \lor 8<x \lor y<8 \lor 9<y)$ $\land$ $(x<0 \lor 3<x \lor y<3 \lor 6<y)$ $\land$ $(x<6 \lor 8<x \lor y<6 \lor 9<y)$ & $\checkmark$ \\
    24 & 4 & $(x<3 \lor 3<x \lor y<7 \lor 10<y)$ $\land$ $(x<9 \lor 9<x \lor y<1 \lor 4<y)$ $\land$ $(x<3 \lor 5<x \lor y<10 \lor 12<y)$ $\land$ $(x<4 \lor 4<x \lor y<6 \lor 9<y)$ & $\checkmark$ \\
    25 & 1 & $(x<6 \lor 6<x \lor y<9 \lor 10<y)$ & $\checkmark$ \\
  \end{tabularx}
\end{table}

As a demonstrative exmaple, we discuss the control envelope computed for problem the 23 in the benchmark suite.
\rref{fig:quadcopter-23} shows a Mathematica plot of the winning subregion for the overall loop of this example.
The problem has three obstacles, two of which overlap.
It is naturally unsafe to start at an obstacle. But it is further unsafe to start right before an obstacle because collision with the obstacle becomes unavoidable.
Notice a triangle of unsafe space right after the lowest obstacle.
This unsafe tringle is an example how obstacles can interact to create unsafe regions. In this triangle, it is neither safe to go upwards because of collision with the upper obstacle, nor is it safe to go downwards because of collision with the floor.
This envelope is computed in 68 seconds.

\begin{figure}
  \centering
  \includegraphics[width=0.8\textwidth]{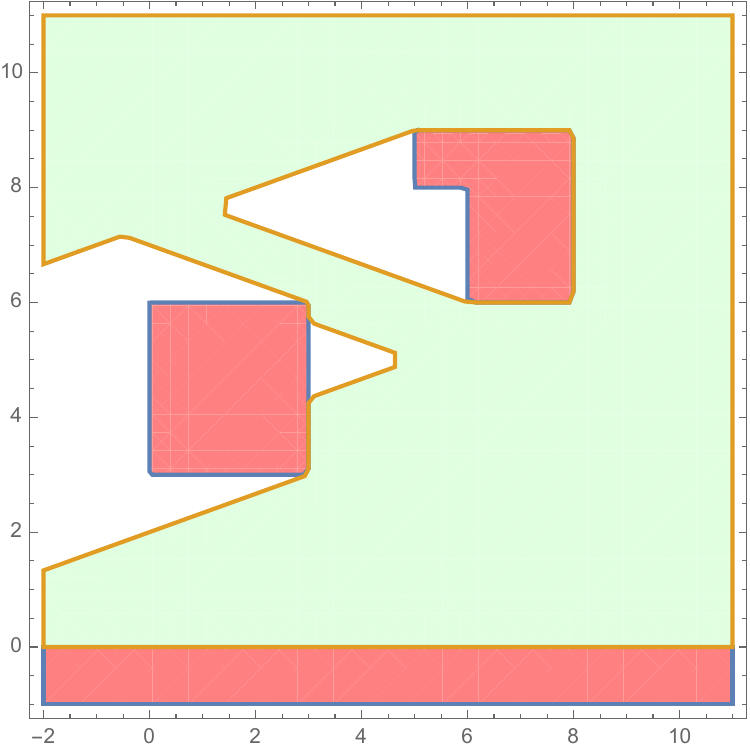}
  \caption{Computed control envelope for quadcopter problem 23. Red regions are unsafe (obstacles or floor). Green regions are safe starting points per the computed control envelope. The $x$ and $y$ coordinates represent the quadcopter's position.}
  \label{fig:quadcopter-23}
\end{figure}

\subsection{Full Benchmark Listings}\label{app:benchmarks-listings}

We list the \dGL formulas for all the new examples in \rref{fig:new-examples}.

\begin{model}[h]
  \setcounter{modelline}{0}
  \caption{Infinite Track}
  \label{model:looping-track}
  \begin{align*}
    \text{\assumptions} &\,\big| \mline{line:track-init} T > 0 \land V > 0 \land R > 0 \land 2V T < R \limply \big\langle \\
    \text{\kwd{east/west}} &\,\big| \mline{line:track-ew-dir} \Big( \big( \left( v_x := V; v_y := 0 \right) \cup \left( v_x := -V; v_y := 0 \right) \cup \\
    \text{\kwd{north/south}} &\,\big| \mline{line:track-ns-dir} \phantom{\Big(\big(} \left( v_x := 0; v_y := -V \right) \cup \left( v_x := 0; v_y := V \right) \big) \seq \\
    \text{\kwd{plant}} &\,\big| \mline{line:track-plant} \phantom{\Big(} t := 0; \left\{ x' = v_x, y' = v_y, t' = 1\ \&\ t \leq T \right\}^d \Big)^\times \\
    \text{\safe} &\,\big| \mline{line:track-safety} \big\rangle \left( 2R > x \land x > -2R \land 2R > y \land y > -2R \land  \right. \\
    &\,\big| \mline{line:track-safety-2} \left. \quad (x > R \lor x < -R \lor y > R \lor y < -R) \right)
  \end{align*}
\end{model}

\begin{model}[h]
  \setcounter{modelline}{0}
  \caption{Reach-avoid Robot}
  \label{model:reach-avoid}
  \begin{align*}
    \text{\assumptions} &\,\big| \mline{line:reach-avoid-init} V > 0 \land R > 0 \limply \big\langle \\
    \text{\kwd{direction}} &\,\big| \mline{line:reach-avoid-choices} \Big( \big( \left( v_x := -V; v_y := 0 \right) \cup \left( v_x := 0; v_y := V \right)\big) \seq \\
    \text{\kwd{plant}} &\,\big| \mline{line:reach-avoid-plant} \left\{ x' = v_x, y' = v_y \& \left( 2R \geq x \land x \geq -2R \land 2R \geq y \land y \geq -2R \right) \right\} \Big)^\ast \big\rangle \\
    \text{\safe} &\,\big| \mline{line:reach-avoid-safety} R \leq x \land x \leq 2R \land R \leq y \land y \leq 2R
  \end{align*}
\end{model}

\begin{model}[h]
  \setcounter{modelline}{0}
  \caption{Highway Driving}
  \label{model:highway-driving}
  \begin{align*}
    \text{\assumptions} &\,\big| \mline{line:highway-assumptions} A > 0 \land B > 0 \land T>0 \land v_f>0 \land v_l>0 \limply \big\langle \\
    \text{\kwd{lead car}} &\,\big| \mline{line:highway-lead} \Big(  a_l := \otimes; ! {-B < a_l \land a_l < A} \seq \\
    \text{\kwd{controlled car}} &\,\big| \mline{line:highway-follow} \phantom{\Big(}  a_f := *; \ptest{-B < a_f \land a_f < A} \seq \\
    \text{\kwd{plant}} &\,\big| \mline{line:highway-plant} \phantom{\Big(} t := 0; \left\{ p_f' = v_f, v_f' = a_f, p_l' = v_l, v_l' = a_l, t' = 1\ \&\ t \leq T \right\}^d \Big)^\times \big\rangle \\
    \text{\safe} &\,\big| \mline{line:highway-safety} p_f < p_l
  \end{align*}
\end{model}

\rref{fig:cesar-benchmarks} compares the performance of our approach to CESAR on the CESAR benchmark suite.
The performance is similar at most benchmarks with only a few seconds of overhead.
An exception is Intersection, where CESAR is significantly more expensive.
The observed differences are likely a consequence of different simplification heuristics rather than any fundamental algorithmic reasons.
Simplification is itself an expensive operation, but if done at the right juncture, can make the even more expensive quantifier eliminations that occur later in symbolic execution cheaper.
Thus simplification has a large influence on performance.
Further code tuning and heuristics optimizations can likely reduce the overheads.
LLM assistance was used to generate some of the boilerplate code of the implementation.

\subsection{Rewrite Heuristics}\label{app:rewrite-heuristics}

Our implementation uses rewriting heuristics to simplify problems into shapes where they can either be symbolically executed or solved using other invariant generation heuristics.
\rref{tab:rewrite-heuristics} lists the rewrite heuristics and which evaluation problem each is used in.

\begin{table}
  \caption{Rewrite heuristics used in the evaluation problems.}
  \label{tab:rewrite-heuristics}
  \begin{tabularx}{\textwidth}{c|X}
    \hline
    \textbf{Heuristc} & \textbf{Problems} \\
    \hline
    Extremal assignment rewrite & Surgical robot, Highway driving \\
    Adversarial one-shot rewrite & Event-triggered ETCS \\
    Loop unrolling & Reach-avoid robot, Quadcopter \\
    CESAR rewrite & Highway driving, CESAR benchmarks \\
    \hline
  \end{tabularx}
\end{table}

We discuss how the rewrite heuristics work.
\begin{itemize}
  \item \emph{Extremal assignment rewrite.} This rewrite replaces a guarded assignment (e.g., $x := \otimes; \ptest{x > 0}$) with an assignment to an extremal value identified from the guard (e.g., $x := 0$).
  The rewritten program is now much easier to reason about, since a free assignment otherwise requiring a quantifier to reason about has been replaced by a constant expression.
  \item \emph{Adversarial one-shot rewrite.} This rewrite is similar in spirit to the \emph{one-shot unrolling} of CESAR \cite{DBLP:conf/tacas/KabraLMP24}.
  It takes an Angel loop ending with an Angel ODE (e.g., the inner loop of \rref{model:event-triggered-etcs}) and replaces it with a single iteration of the loop but with the ODE changed to a Demon ODE guarded at the end by an Angel test for the domain constraint of the original ODE.
  The intuition is that the rewritten program emulates running any number of iterations of the original loop with the ODE run for any amount of time, so a precondition for the harder rewritten program is likely to also be a precondition for the original program.
  \item \emph{Loop unrolling.} This heuristic unrolls loops a fixed number of times. For example, $\langle \alpha^\ast \rangle \phi$ is rewritten as $\langle \alpha \rangle \phi$ or $\langle \alpha \rangle \langle \alpha \rangle \phi$ or $\langle \alpha \rangle \langle \alpha \rangle \langle \alpha \rangle \phi$ and so on.
  \item \emph{CESAR rewrite.} This is the one-shot refinement and multi-shot unrolling form CESAR \cite{DBLP:conf/tacas/KabraLMP24} where a time-triggered control loop is replaced by running the ODE for an unbounded amount of time.
  Because of the of our framework's retrospective checks, the requirement that the original program should have an idempotence property (``action permanence'') which was essential to the soundness of CESAR is relaxed.
\end{itemize}

\subsection{Refinements}\label{app:refinements}

While synthesizing subvalue maps for loops, \rref{alg:solving} uses checks to test invariant candidates before accepting them.
It is possible to use \emph{refinements} to perform these checks.
A game $\beta$ is an \emph{Angelic refinement} of game $\alpha$ when for any formula $\psi$, $\models \langle \beta \rangle \psi \limply \langle \alpha \rangle \psi$.
Some rewrites result in refinements, making the invariant candidate generated correct by construction.
For instance, for the Angelic loop in Reach-avoid robot (\rref{model:reach-avoid}), let $\alpha$ be the loop body (\rref{line:reach-avoid-choices} and \rref{line:reach-avoid-plant}), and $\phi$ the postcondition (\rref{line:reach-avoid-safety}).
The inductiveness check is $\invexpr \limply \langle \aproj{(\namedGameSpace{\alpha}{a}^*)}{S}{\phi} \rangle \phi$.
If $\invexpr$ is computed by rewriting loop $\langle \alpha^\ast \rangle \phi$ as $\langle \alpha \rangle \phi$ (loop unrolling rewrite) and then setting $\invexpr = \exec(\langle \alpha \rangle \phi)$ while writing all the intermediate symbolic execution steps of $\exec$ into the subvalue map $S$, then $\invexpr$ will pass the check by construction.
The reason is that $\langle \aproj{\namedGame{\alpha}{a}}{S}{\phi} \rangle \phi$ is a refinement of $\langle \aproj{(\namedGameSpace{\alpha}{a})^*}{S}{\phi} \rangle \phi$.
Since $\models \invexpr \limply \langle \aproj{\namedGame{\alpha}{a}}{S}{\phi} \rangle \phi$ by construction, it follows that $\models \invexpr \limply \langle \aproj{(\namedGameSpace{\alpha}{a})^*}{S}{\phi} \rangle \phi$.
We exploit such refinement properties in our implementation to soundly reduce the number of expensive checks performed.

\section{Proofs}
\label{app:proofs}

We prove the correctness of the theorems in the paper.
Notation $\restrict{S}{U}$ is the restriction of the subvalue map $S$ to the set of subgames $U \cup \{\finalNode\}$.

\begin{definition}[Weak ordering]
  \label{def:weak-ordering}
  We define a weaker ordering relation $\succsim$ on inductive subvalue maps that compares formulas at each subgame separately.
  For two inductive subvalue maps $S_1$ and $S_2$ for the game $\namedGame{\alpha}{a}$, $S_1$ is at least as good as $S_2$, written $S_1 \succsim S_2$, iff for each subgame $b$ in $\nodes{\namedGame{\alpha}{a}}$, $\models S_2(b) \limply S_1(b)$ and $\models S_2(\finalNode) \limply S_1(\finalNode)$.
\end{definition}

\begin{lemma}[Ordering subsumption]
  \label{lem:subsumption}
  If $S \succsim S'$, then $S \supseteq S'$.
\end{lemma}
\begin{proof}
  By definition of $\succsim$, for each subgame $b \in \nodes{\alpha}\cup\{\finalNode\}$, $\models S'(b) \limply S(b)$.
  Using the \dGL G\"odel generalization rule \irref{G}, we infer $\models \langle \asubst{(\prefix{b}{\namedGame{\alpha}{a}})}{S'}{\phi} \rangle (S(b) \limply S'(b))$ showing that $S \supseteq S'$ when interpreted as an Angelic strategy.
  When interpreted as a Demonic strategy, the same argument holds: using the \dGL G\"odel generalization rule \irref{G} we can infer that $\models [\dsubst{(\prefix{b}{\namedGame{\alpha}{a}})}{S'}{\phi}] (S'(b) \limply S(b))$.
\end{proof}

\begin{lemma}[Subvalue Map Projection is a Refinement]
  \label{lem:subvalue-projection-refinement}
  The projection of an Angelic subvalue map onto a game $\alpha$ is an Angelic refinement of the game, i.e., 
  \begin{equation*}
    \models\langle\aproj{\namedGame{\alpha}{a}}{S}{}\rangle \psi \limply \langle \alpha \rangle \psi.
  \end{equation*}
  Dually, the projection of a Demonic subvalue map is a Demonic refinement.
\end{lemma}
\begin{proof}
  We show this for Angelic subvalue maps, the proof for Demonic subvalue maps is symmetric.
  Intuitively this lemma is true because the projection of an Angelic subvalue map only adds Angelic tests, which make the game monotonically harder.
  The proof uses structural induction based on the shape of $\namedGame{\alpha}{a}$.
  When $\namedGame{\alpha}{a}$ is:
  \begin{itemize}
    \item Atomic and not controlled by Angel, i.e., $\alpha \in \{x:=e, x:=*, \ptest{Q}, !Q, \{x'=f(x)\ \&\ Q\}^d\}$, the result is immediate since $\aproj{\namedGame{\alpha}{a}}{S}{} = \alpha$.
    \item Atomic and controlled by Angel, i.e., $\alpha \in \{x:=\otimes, \{x'=f(x)\ \&\ Q\}\}$, then $\aproj{\namedGame{\alpha}{a}}{S}{} = \alpha \seq \ptest{S(\finalNode)}$.
    By the \dGL axioms \irref{testd} and \irref{composed}, $\models \langle \alpha \seq \ptest{S(\finalNode)} \rangle \psi \leftrightarrow \langle \alpha \rangle (\psi \land S(\finalNode))$.
    By the \dGL axiom \irref{M}, $\models \langle \alpha \rangle (\psi \land S(\finalNode)) \limply \langle \alpha \rangle \psi$.
    By transitivity of implication, $\models \langle \aproj{\namedGame{\alpha}{a}}{S}{} \rangle \psi \limply \langle \alpha \rangle \psi$.
    \item $\namedGame{\gamma}{g} \cup \namedGame{\delta}{d}$, then $\aproj{\namedGame{\alpha}{a}}{S}{} = (\AdvExChoice{?S(g)} \seq \aproj{\namedGame{\gamma}{g}}{S}{}) {\cup} (\AdvExChoice{?S(d)} \seq \aproj{\namedGame{\delta}{d}}{S}{})$.
    By the \dGL axiom \irref{choiced}, $\models \langle (\AdvExChoice{?S(g)} \seq \aproj{\namedGame{\gamma}{g}}{S}{}) {\cup} (\AdvExChoice{?S(d)} \seq \aproj{\namedGame{\delta}{d}}{S}{}) \rangle \psi \leftrightarrow \langle \AdvExChoice{?S(g)} \seq$ $\aproj{\namedGame{\gamma}{g}}{S}{} \rangle \psi \lor \langle \AdvExChoice{?S(d)} \seq \aproj{\namedGame{\delta}{d}}{S}{} \rangle \psi$.
    By the \dGL axioms \irref{composed} and \irref{choiced}, $\models \langle \AdvExChoice{?S(g)} \seq$ $\aproj{\namedGame{\gamma}{g}}{S}{} \rangle \psi \lor \langle \AdvExChoice{?S(d)} \seq \aproj{\namedGame{\delta}{d}}{S}{} \rangle \psi 
    \leftrightarrow 
    ( \AdvExChoice{S(g)} \land \langle \aproj{\namedGame{\gamma}{g}}{S}{} \rangle \psi ) \lor ( \AdvExChoice{S(d)} \land \langle \aproj{\namedGame{\delta}{d}}{S}{} \rangle \psi)$.
    By conjunction elimination, $\models ( \AdvExChoice{S(g)} \land \langle \aproj{\namedGame{\gamma}{g}}{S}{} \rangle \psi ) \lor ( \AdvExChoice{S(d)} \land \langle \aproj{\namedGame{\delta}{d}}{S}{} \rangle \psi) \limply \langle \aproj{\namedGame{\gamma}{a}}{S}{} \rangle \psi \lor \langle \aproj{\namedGame{\delta}{d}}{S}{} \rangle \psi$.
    By the inductive hypothesis, $\models \langle \aproj{\namedGame{\gamma}{g}}{S}{} \rangle \psi \limply \langle \gamma \rangle \psi$ and $\models \langle \aproj{\namedGame{\delta}{d}}{S}{} \rangle \psi \limply \langle \delta \rangle \psi$.
    By the transitivity of implication, $\models \langle \aproj{\namedGame{\gamma}{g}}{S}{} \rangle \psi \lor \langle \aproj{\namedGame{\delta}{d}}{S}{} \rangle \psi \limply \langle \gamma \rangle \psi \lor \langle \delta \rangle \psi$.
    By the \dGL axiom \irref{choiced}, $\models \langle \gamma \rangle \psi \lor \langle \delta \rangle \psi \leftrightarrow \langle \gamma \cup \delta \rangle \psi$.
    \item $\namedGame{\gamma}{g} \cap \namedGame{\delta}{d}$, then $\aproj{\namedGame{\alpha}{a}}{S}{} = \aproj{\namedGame{\gamma}{g}}{S}{} \cap \aproj{\namedGame{\delta}{d}}{S}{}$.
    By the \dGL axiom \irref{choiceb}, $\models \langle \aproj{\namedGame{\gamma}{g}}{S}{} \cap \aproj{\namedGame{\delta}{d}}{S}{} \rangle \psi \leftrightarrow \langle \aproj{\namedGame{\gamma}{g}}{S}{} \rangle \psi \land \langle \aproj{\namedGame{\delta}{d}}{S}{} \rangle \psi$.
    By the inductive hypothesis, $\models \langle \aproj{\namedGame{\gamma}{g}}{S}{} \rangle \psi \limply \langle \gamma \rangle \psi$ and $\models \langle \aproj{\namedGame{\delta}{d}}{S}{} \rangle \psi \limply \langle \delta \rangle\psi$.
    Thus, $\models \langle \aproj{\namedGame{\gamma}{g}}{S}{} \rangle \psi \land \langle \aproj{\namedGame{\delta}{d}}{S}{} \rangle \psi \limply \langle \gamma \rangle \psi \land \langle \delta \rangle \psi$.
    By the \dGL axiom \irref{choiceb}, $\models \langle \gamma \rangle \psi \land \langle \delta \rangle \psi \leftrightarrow \langle \gamma \cap \delta \rangle \psi$.
    \item $\namedGame{\gamma}{g};\namedGame{\delta}{d}$, then $\aproj{\namedGame{\alpha}{a}}{S}{} = \aproj{\namedGame{\gamma}{g}}{\modEnd{S}{S(d)}}{} \seq \aproj{\namedGame{\delta}{d}}{S}{}$.
    
    Using the \dGL axiom \irref{composed},
    
    $\models \langle \aproj{\namedGame{\gamma}{g}}{\modEnd{S}{S(d)}}{} \seq \aproj{\namedGame{\delta}{d}}{S}{} \rangle \psi \leftrightarrow \langle \aproj{\namedGame{\gamma}{g}}{\modEnd{S}{S(d)}}{} \rangle \langle \aproj{\namedGame{\delta}{d}}{S}{} \rangle \psi$.
    By the inductive hypothesis, $\models \langle \aproj{\namedGame{\delta}{d}}{S}{} \rangle \psi \limply \langle \delta \rangle \psi$.\\
    Thus,
    $\models \langle \aproj{\namedGame{\gamma}{g}}{\modEnd{S}{S(d)}}{} \rangle \langle \aproj{\namedGame{\delta}{d}}{S}{} \rangle \psi \limply \langle \aproj{\namedGame{\gamma}{g}}{\modEnd{S}{S(d)}}{} \rangle$ $\langle \delta \rangle \psi$.
    Further, by the inductive hypothesis, $\models \langle \aproj{\namedGame{\gamma}{g}}{\modEnd{S}{S(d)}}{} \rangle \langle \delta \rangle \psi \limply \langle \gamma \rangle \langle \delta \rangle \psi$.
    By the \dGL axiom \irref{composed}, $\models \langle \gamma \rangle \langle \delta \rangle \psi \leftrightarrow \langle \gamma;\delta \rangle \psi$.
    \item $\namedGame{(\namedGame{\gamma}{g})^\ast}{a}$, then $\aproj{\namedGame{\alpha}{a}}{S}{} = (\AdvExLoop{?S(g)} \seq \aproj{\namedGame{\gamma}{g}}{\modEnd{S}{S(s)}}{})^{\ast}; \AdvExLoop{?S(\finalNode)}$.
    By \dGL axioms \irref{composed}, \irref{testd}, and \irref{M}, $\models \langle \aproj{\namedGame{\alpha}{a}}{S}{} \rangle \psi \limply \ddiamond{\prepeat{(\AdvExLoop{?S(g)} \seq \aproj{\namedGame{\gamma}{g}}{\modEnd{S}{S(a)}}{})}}{\psi}$.
    By the \dGL proof rule \irref{FP},

    \begin{calculus}
    {\linferenceRule[formula]
      {\psi \lor \ddiamond{(\AdvExLoop{?S(g)} \seq \aproj{\namedGame{\gamma}{g}}{\modEnd{S}{S(a)}})}{\langle \gamma^* \rangle \psi} \limply \langle \gamma^* \rangle \psi}
      {\ddiamond{\prepeat{(\AdvExLoop{?S(g)} \seq \aproj{\namedGame{\gamma}{g}}{\modEnd{S}{S(a)}}{})}}{\psi} \limply \langle \gamma^* \rangle \psi}
    }{}
    \end{calculus}.

    Now, $\ddiamond{(\AdvExLoop{?S(g)} \seq \aproj{\namedGame{\gamma}{g}}{\modEnd{S}{S(a)}}{})}{\langle \gamma^* \rangle \psi}$, per the inductive hypothesis, implies $\langle \gamma \rangle \langle \gamma^* \rangle \psi$.
    Per \dGL axiom \irref{iterated}, $\langle \gamma \rangle \langle \gamma^* \rangle \psi \leftrightarrow \langle \gamma^* \rangle \psi$.
    This allows us to show the premise, since $\models \psi \lor \langle \gamma^* \rangle \psi \limply \langle \gamma^* \rangle \psi$.
    \item $\namedGame{(\namedGame{\gamma}{g})^\times}{a}$, then $\aproj{\namedGame{\alpha}{a}}{S}{} = \big(\aproj{\namedGame{\gamma}{g}}{S}{S(a)}\big)^\times$.
    By the \dGL proof rule \irref{loop},\\
    \begin{calculus}
      {\linferenceRule[sequent]
        {\lsequent[g]{\ddiamond{\big(\aproj{\namedGame{\gamma}{g}}{S}{S(a)}\big)^\times}{\psi}} {\inv}
        &\lsequent[g]{\inv} {\ddiamond{\gamma}{\inv}}
        &\lsequent[g]{\inv} {\psi}}
        {\lsequent[g]{\ddiamond{\big(\aproj{\namedGame{\gamma}{g}}{S}{S(a)}\big)^\times}{\psi}} {\ddiamond{\drepeat{\gamma}}{\psi}}}
      }{}%
    \end{calculus}. \\
    Setting invariant $J$ to $\ddiamond{\big(\aproj{\namedGame{\gamma}{g}}{S}{S(a)}\big)^\times}{\psi}$, the first premise is immediate.

    The second premise is now $\lsequent[g]{\ddiamond{\big(\aproj{\namedGame{\gamma}{g}}{S}{S(a)}\big)^\times}{\psi}} {\ddiamond{\gamma}{\ddiamond{\big(\aproj{\namedGame{\gamma}{g}}{S}{S(a)}\big)^\times}{\psi}}}$.
    By \irref{diterated}, 
    $\ddiamond{\big(\aproj{\namedGame{\gamma}{g}}{S}{S(a)}\big)^\times}{\psi} \limply \ddiamond{\aproj{\namedGame{\gamma}{g}}{S}{S(a)}}{\ddiamond{\big(\aproj{\namedGame{\gamma}{g}}{S}{S(a)}\big)^\times}{\psi}}$.\\
    By the inductive hypothesis, $\models \ddiamond{\aproj{\namedGame{\gamma}{g}}{S}{S(a)}\big}{\ddiamond{\big(\aproj{\namedGame{\gamma}{g}}{S}{S(a)}\big)^\times}{\psi}} \limply \langle \gamma \rangle$ $\ddiamond{\big(\aproj{\namedGame{\gamma}{g}}{S}{S(a)}\big)^\times}{\psi}$, thus proving the second premise.

    The third premise is proved as follows.
    \begin{equation*}
      {\linfer[diterated]
        {
          \linfer[weakenl]
          {
            \linfer[id]
            {\lclose}
            {\lsequent[g]{\psi}{\psi}}
          }
          {\lsequent[g]{\psi \land \ddiamond{\aproj{\namedGame{\gamma}{g}}{S}{S(a)}\big}{\ddiamond{\big(\aproj{\namedGame{\gamma}{g}}{S}{S(a)}\big)^\times}{\psi}}} {\psi}}
        }
        {\lsequent[g]{\ddiamond{\big(\aproj{\namedGame{\gamma}{g}}{S}{S(a)}\big)^\times}{\psi}} {\psi}}
      }{}
    \end{equation*}
  \end{itemize}
\end{proof}

\begin{lemma}[MPC Subvalue Map Projection Equivalence]
  \label{lem:mpc-proj-id}
  The projection game $\langle \aproj{\namedGame{\alpha}{a}}{S}{} \rangle \phi$ of the MPC Angelic subvalue map (\rref{def:mpc-sol}) is equivalent to the original game $\langle \alpha \rangle \phi$.
  Dually, the projection game $[\dproj{\namedGameSpace{\alpha}{a}}{S}{}] \phi$ of the MPC Demonic subvalue map (\rref{def:mpc-sol}) is equivalent to the original game $[\alpha] \phi$.
\end{lemma}
\begin{proof}
  We show this for Angelic subvalue maps, the proof for Demonic subvalue maps is symmetric.
  Intuitively this lemma is true because the projection of an Angelic subvalue map only adds Angelic tests, which make the game monotonically harder.
  The proof uses structural induction based on the shape of $\namedGame{\alpha}{a}$.
  Note that for the MPC solution, \(\finalNode\) is \(\phi\).
  When $\namedGame{\alpha}{a}$ is:
  \begin{itemize}
    \item Atomic and not controlled by Angel, i.e., $\alpha \in \{x:=e, x:=*, \ptest{Q}, !Q, \{x'=f(x)\ \&\ Q\}^d\}$, the result is immediate since $\aproj{\namedGame{\alpha}{a}}{S}{} = \alpha$.
    \item Atomic and controlled by Angel, i.e., $\alpha \in \{x:=\otimes, \{x'=f(x)\ \&\ Q\}\}$, then $\aproj{\namedGame{\alpha}{a}}{S}{} = \alpha \seq \ptest{\phi}$.
    By the \dGL axioms \irref{testd} and \irref{composed}, $\models \langle \alpha \seq \ptest{\phi} \rangle \phi \leftrightarrow \langle \alpha \rangle (\phi \land \phi) \leftrightarrow \langle \alpha \rangle \phi$.
    \item $\namedGame{\gamma}{g} \cup \namedGame{\delta}{d}$, then $\aproj{\namedGame{\alpha}{a}}{S}{} = (\AdvExChoice{?S(g)} \seq \aproj{\namedGame{\gamma}{g}}{S}{}) {\cup} (\AdvExChoice{?S(d)} \seq \aproj{\namedGame{\delta}{d}}{S}{})$.
    
    $\phantom{\leftrightarrow} \langle (\AdvExChoice{?S(g)} \seq \aproj{\namedGame{\gamma}{g}}{S}{}) {\cup} (\AdvExChoice{?S(d)} \seq \aproj{\namedGame{\delta}{d}}{S}{}) \rangle \phi $ \\
    $\leftrightarrow \langle \AdvExChoice{?S(g)} \seq \aproj{\namedGame{\gamma}{g}}{S}{} \rangle \phi \lor \langle \AdvExChoice{?S(d)} \seq \aproj{\namedGame{\delta}{d}}{S}{} \rangle \phi \qquad (\textrm{\dGL axiom }\irref{choiced})$ \\
    $\leftrightarrow ( \AdvExChoice{S(g)} \land \langle \aproj{\namedGame{\gamma}{g}}{S}{} \rangle \phi ) \lor ( \AdvExChoice{S(d)} \land \langle \aproj{\namedGame{\delta}{d}}{S}{} \rangle \phi) \qquad (\textrm{\dGL axioms }\irref{composed}, \irref{choiced})$ \\
    $\leftrightarrow (\AdvExChoice{S(g)} \land \langle \gamma \rangle \phi ) \lor ( \AdvExChoice{S(d)} \land \langle \delta \rangle \phi) \qquad (\textrm{inductive hypothesis})$ \\
    $\leftrightarrow (\langle \gamma \rangle \phi \land \langle \gamma \rangle \phi) \lor (\langle \delta \rangle \phi \land \langle \delta \rangle \phi) \qquad (\textrm{MPC definition})$ \\
    $\leftrightarrow \langle \gamma \rangle \phi \lor \langle \delta \rangle \phi$.
    \item $\namedGame{\gamma}{g} \cap \namedGame{\delta}{d}$, then $\aproj{\namedGame{\alpha}{a}}{S}{} = \aproj{\namedGame{\gamma}{g}}{S}{} \cap \aproj{\namedGame{\delta}{d}}{S}{}$.

    $\phantom{\leftrightarrow} \langle \aproj{\namedGame{\gamma}{g}}{S}{} \cap \aproj{\namedGame{\delta}{d}}{S}{} \rangle \phi$ \\
    $\leftrightarrow \langle \aproj{\namedGame{\gamma}{g}}{S}{} \rangle \phi \land \langle \aproj{\namedGame{\delta}{d}}{S}{} \rangle \phi \qquad (\textrm{\dGL axiom }\irref{choiceb})$ \\
    $\leftrightarrow \langle \gamma \rangle \phi \land \langle \delta \rangle \phi \qquad (\textrm{inductive hypothesis})$ \\
    $\leftrightarrow \langle \gamma \cap \delta \rangle \phi \qquad (\textrm{\dGL axiom }\irref{choiceb})$.
    \item $\namedGame{\gamma}{g};\namedGame{\delta}{d}$, then $\aproj{\namedGame{\alpha}{a}}{S}{} = \aproj{\namedGame{\gamma}{g}}{\modEnd{S}{S(d)}}{} \seq \aproj{\namedGame{\delta}{d}}{S}{}$.
    
    $\phantom{\leftrightarrow} \langle \aproj{\namedGame{\gamma}{g}}{\modEnd{S}{S(d)}}{} \seq \aproj{\namedGame{\delta}{d}}{S}{} \rangle \phi$ \\
    $\leftrightarrow \langle \aproj{\namedGame{\gamma}{g}}{\modEnd{S}{S(d)}}{} \rangle \langle \aproj{\namedGame{\delta}{d}}{S}{} \rangle \phi \qquad (\textrm{\dGL axiom }\irref{composeb})$ \\
    $\leftrightarrow \langle \aproj{\namedGame{\gamma}{g}}{\modEnd{S}{S(d)}}{} \rangle \langle \delta \rangle \phi \qquad (\textrm{inductive hypothesis})$ \\
    $\leftrightarrow \langle \gamma \rangle \langle \delta \rangle \phi \qquad (\textrm{inductive hypothesis})$ \\
    $\leftrightarrow \langle \gamma;\delta \rangle \phi \qquad (\textrm{\dGL axiom }\irref{composeb})$.
    \item $\namedGame{(\namedGame{\gamma}{g})^\ast}{a}$.
    By \rref{lem:subvalue-projection-refinement}, $\models \ddiamond{\aproj{\namedGame{\alpha}{a}}{S}{}}{\phi} \limply \ddiamond{\gamma^*}{\phi}$.

    We now show the opposite direction, i.e., $\models \ddiamond{\gamma^*}{\phi} \limply \ddiamond{\aproj{\namedGame{\alpha}{a}}{S}{}}{\phi}$.

    Since $S(a) = \ddiamond{\gamma^*}{\phi}$, we have $\models \ddiamond{\gamma^*}{\phi} \limply S(a)$.
    From $\avalid{\alpha}{a}{S}{\phi}$, we get $\models S(a) \limply \ddiamond{\aproj{\namedGame{\alpha}{a}}{S}{}}{\phi}$.

    Thus, by transitivity of implication, $\models \ddiamond{\gamma^*}{\phi} \limply \ddiamond{\aproj{\namedGame{\alpha}{a}}{S}{}}{\phi}$.
    \item $\namedGame{(\namedGame{\gamma}{g})^\times}{a}$
    
    By \rref{lem:subvalue-projection-refinement}, $\models \ddiamond{\aproj{\namedGame{\alpha}{a}}{S}{}}{\phi} \limply \ddiamond{\gamma^\times}{\phi}$.
    We now show the opposite direction, i.e., $\models \ddiamond{\gamma^\times}{\phi} \limply \ddiamond{\aproj{\namedGame{\alpha}{a}}{S}{}}{\phi}$.

    Since $S(a) = \ddiamond{\gamma^\times}{\phi}$, we have $\models \ddiamond{\gamma^\times}{\phi} \limply S(a)$.
    We now show that $\models S(a) \limply \ddiamond{\aproj{\namedGame{\alpha}{a}}{S}{}}{\phi}$, so that by transitivity of implication, $\models \ddiamond{\gamma^\times}{\phi} \limply \ddiamond{\aproj{\namedGame{\alpha}{a}}{S}{}}{\phi}$.

    By the \dGL proof rule \irref{loop},\\
    \begin{calculus}
      {\linferenceRule[sequent]
        {\lsequent[g]{S(a)} {\inv}
        &\lsequent[g]{\inv} {\ddiamond{\aproj{\namedGame{\gamma}{g}}{\modEnd{S}{S(a)}}{}}{\inv}}
        &\lsequent[g]{\inv} {\phi}}
        {\lsequent[g]{S(a)} {\ddiamond{\big(\aproj{\namedGame{\gamma}{g}}{\modEnd{S}{S(a)}}{}\big)^\times}{\phi}}}
      }{}%
    \end{calculus}. \\
    Setting invariant $J$ to $S(a)$, the first premise is immediate.
    The last premise is $\lsequent[g]{S(a)}{\phi}$ which follows from $S(a)\leftrightarrow \phi \land \ddiamond{\gamma}{\ddiamond{\gamma^\times}{\phi}}$ per \dGL axiom \irref{diterated}.
    Only the middle premise remains. It is
    $\lsequent[g]{S(a)}{\ddiamond{\aproj{\namedGame{\gamma}{g}}{\modEnd{S}{S(a)}}{}}{S(a)}}$.
    Per the inductive hypothesis, $\models \ddiamond{\aproj{\namedGame{\gamma}{g}}{\modEnd{S}{S(a)}}{}}{S(a)} \leftrightarrow \langle \gamma \rangle S(a)$.
    Thus, the middle premise can be rewritten as $\lsequent[g]{S(a)}{\langle \gamma \rangle S(a)}$.
    Recall that $S(a) = \ddiamond{\gamma^\times}{\phi}$, so that the middle premise is $\lsequent[g]{\ddiamond{\gamma^\times}{\phi}}{\langle \gamma \rangle \ddiamond{\gamma^\times}{\phi}}$.
    By the \dGL axiom \irref{diterated}, $\ddiamond{\gamma^\times}{\phi} \limply \langle \gamma \rangle \ddiamond{\gamma^\times}{\phi}$, thus proving the middle premise.
  \end{itemize}
\end{proof}

\begin{lemma}[MPC Valid]
  \label{lem:mpc-valid}
  The MPC solution (\rref{def:mpc-sol}) $S$ for a game $\namedGame{\alpha}{a}$ for Angel winning condition $\phi$ is an inductive subvalue map compatible with \(\phi\): $\avalid{\alpha}{a}{S}{}$.
  Dually, the MPC solution $S$ for a game $\namedGame{\alpha}{a}$ for Demon winning condition $\phi$ is an inductive subvalue map: $\dvalid{\alpha}{a}{S}{\phi}$.
\end{lemma}
\begin{proof}
  Proved by induction on the structure of the game $\namedGame{\alpha}{a}$.
  We show this for the Angel MPC solution; the proof for the Demonic MPC solution is analogous.
  Firstly, for the MPC solution, $S(\finalNode) = \phi$, so \(\phi\) is a compatible winning condition.
  The base case is when $\namedGame{\alpha}{a}$ is atomic, i.e., $\alpha \in \{x:=e, x:=*, \ptest{f}, !f, \{x'=f(x)\ \&\ Q\}, \{x'=f(x)\ \& \ Q\}^d\}$.
  In this case the MPC subvalue map consists of the mapping from $a$ to $\langle \alpha \rangle \phi$ (along with \(\finalNode \mapsto \phi\)), which matches the corresponding condition for $\avalid{\alpha}{a}{S}{}$ in \rref{def:local-envelope-conditions}.
  In the recursive case, if $\namedGame{\alpha}{a}$ has the structure:
  \begin{enumerate}
    \item $\namedGame{\gamma}{g} \cup \namedGame{\delta}{d}$: by the inductive hypothesis, as $\restrict{S}{\nodes{\namedGameSpace{\gamma}{g}}}$ is the MPC solution for game $\gamma$ and $\restrict{S}{\nodes{\namedGameSpace{\delta}{d}}}$ is the MPC solution for game $\delta$,
    $\avalid{\gamma}{g}{\restrict{S}{\nodes{\namedGameSpace{\gamma}{g}}}}{}$ and 
    $\avalid{\delta}{d}{\restrict{S}{\nodes{\namedGameSpace{\delta}{d}}}}{}$.
    So, $\avalid{\gamma}{g}{S}{}$ and $\avalid{\delta}{d}{S}{}$.
    It remains to show that $S(a) \limply S(g) \lor S(d)$.
    By the definition of the MPC solution, $S(a) = \langle \namedGame{\gamma}{g} \cup \namedGame{\delta}{d} \rangle \phi$, $S(g)=\langle \namedGame{\gamma}{g}\rangle\phi$ and $S(d)=\langle \namedGame{\delta}{d}\rangle \phi$.
    By the \dGL rule \irref{choiced}, $\langle \namedGame{\gamma}{g} \cup \namedGame{\delta}{d} \rangle \phi = \langle \namedGame{\gamma}{g} \rangle \phi \lor \langle \namedGame{\delta}{d} \rangle \phi$.
    Thus, $S(a) \limply S(g) \lor S(d)$.
    \item $\namedGame{\gamma}{g} \cap \namedGame{\delta}{d}$: by the inductive hypothesis, as $\restrict{S}{\nodes{\namedGameSpace{\gamma}{g}}}$ is the MPC solution for game $\gamma$ and $\restrict{S}{\nodes{\namedGameSpace{\delta}{d}}}$ is the MPC solution for game $\delta$,
    $\avalid{\gamma}{g}{\restrict{S}{\nodes{\namedGameSpace{\gamma}{g}}}}{}$ and $\avalid{\delta}{d}{\restrict{S}{\nodes{\namedGameSpace{\delta}{d}}}}{}$.
    So, $\avalid{\gamma}{g}{S}{}$ and $\avalid{\delta}{d}{S}{}$.
    It remains to show that $S(a) \limply S(g) \land S(d)$.
    By the definition of the MPC solution, $S(a) = \langle \namedGame{\gamma}{g} \cap \namedGame{\delta}{d} \rangle \phi$, $S(g)=\langle \namedGame{\gamma}{g}\rangle\phi$ and $S(d)=\langle \namedGame{\delta}{d}\rangle \phi$.
    By the \dGL rule \irref{dchoiced}, $\langle \namedGame{\gamma}{g} \cap \namedGame{\delta}{d} \rangle \phi = \langle \namedGame{\gamma}{g} \rangle \phi \land \langle \namedGame{\delta}{d} \rangle \phi$.
    Thus, $S(a) \limply S(g) \land S(d)$.
    \item $\namedGame{\gamma}{g} \seq \namedGame{\delta}{d}$: by the inductive hypothesis, as $\restrict{S}{\nodes{\namedGameSpace{\delta}{d}}}$ is the MPC solution for game $\delta$ and Angel winning condition $\phi$, $\avalid{\delta}{d}{\restrict{S}{\nodes{\namedGameSpace{\delta}{d}}}}{}$.
    So $\avalid{\delta}{d}{S}{}$.
    Further, as $\restrict{S}{\nodes{\namedGameSpace{\gamma}{g}}}$ is the MPC solution for game $\gamma$ and Angel winning condition $\langle \namedGame{\delta}{d} \rangle \phi = S(\delta)$, by the inductive hypothesis, $\avalid{\gamma}{g}{\restrict{\modEnd{S}{S(\delta)}}{\nodes{\namedGameSpace{\gamma}{g}}}}{}$.
    So $\avalid{\gamma}{g}{\modEnd{S}{S(\delta)}}{}$.
    It remains to show that $S(a) \limply S(g)$ which is immediate from the definition of the MPC solution since $\fwd{a}{\namedGame{\alpha}{a}} = \fwd{g}{\namedGame{\alpha}{a}}$ in this case.
    \item $\namedGame{\gamma}{g}^\ast$: observe that $\restrict{S}{\nodes{\namedGameSpace{\gamma}{g}}}$ is the MPC solution for game $\gamma$ and Angel winning condition $\langle \gamma^* \rangle \phi$.
    $S(a) = \langle \gamma^* \rangle \phi$ and $S(a) \limply \phi$, so $\langle \gamma^* \rangle \phi = S(a)\lor\phi$, so $\restrict{S}{\nodes{\namedGameSpace{\gamma}{g}}}$ is the MPC solution for game $\gamma$ and Angel winning condition $S(a)\lor\phi$.
    Then, by the inductive hypothesis, $\avalid{\gamma}{g}{\restrict{\modEnd{S}{S(a)}}{\nodes{\namedGameSpace{\gamma}{g}}}}{}$.
    So $\avalid{\gamma}{g}{\modEnd{S}{S(a)}}{}$.
    The remaining condition is $S(a) \limply \langle \aproj{\namedGame{\alpha}{a}}{S}{\phi} \rangle \phi$.
    For MPC solution $S$, $\ddiamond{\aproj{\namedGame{\alpha}{a}}{S}{\phi}}{\phi} = \langle \gamma^* \rangle \phi$ (\rref{lem:mpc-proj-id}).
    Since $S(a) = \langle \gamma^* \rangle \phi$, the condition is satisfied.
    \item $\namedGame{\gamma}{g}^{\times}$: observe that $\restrict{S}{\nodes{\namedGameSpace{\gamma}{g}}}$ is the MPC solution for game $\gamma$ and Angel winning condition $\langle \gamma^{\times} \rangle \phi$.
    By definition, $S(a) = \langle \gamma^{\times} \rangle \phi$.
    So, by the inductive hypothesis, $\avalid{\gamma}{g}{\restrict{\modEnd{S}{S(a)}}{\nodes{\namedGameSpace{\gamma}{g}}}}{}$.
    Thus $\avalid{\gamma}{g}{\modEnd{S}{S(a)}}{}$.
    Next we show that $\models S(a) \limply S(g) \land \phi$.
    Using the \dGL axiom \irref{diterated}, $\langle \gamma^{\times} \rangle \phi = \langle \gamma \rangle \langle \gamma^{\times} \rangle \phi \land \phi$.
    Substituting back the MPC definitions of $S(a)=\langle \gamma^{\times} \rangle \phi$ and $S(g)=\langle \gamma \rangle S(a)$, this means $S(a) = S(g) \land \phi$.
    So, $\models S(a) \limply S(g) \land \phi$.
  \end{enumerate}
\end{proof}

\begin{lemma}[Subvalue maps win]
  \label{lem:subvalue-universal}
  For any game $\namedGame{\alpha}{a}$, Angelic subvalue map \(S\), and compatible Angel winning condition $\phi$, \(\ddiamond{\alpha}{\phi} \models \ddiamond{\asubst{\namedGame{\alpha}{a}}{S}{}}{\phi}\).
\end{lemma}
\begin{proof}
  The proof follows from structural induction along with the application of the usual \dGL axioms and proof rules.

  \begin{itemize}
    \item If \(\namedGame{\alpha}{a}\) is atomic and not controlled by Angel, i.e., \(\alpha \in \{x:=e, x:=\otimes, \ptest{Q}, !Q, \{x'=f(x)\ \&\ Q\}^d\}\), then \(\ddiamond{\alpha}{\phi} = \ddiamond{\asubst{\namedGame{\alpha}{a}}{S}{}}{\phi}\) and the result is immediate.
    \item If \(\namedGame{\alpha}{a} = x := \ast\), then \(\asubst{\namedGame{\alpha}{a}}{S}{}\) is \(x:=\otimes \seq ! S(\finalNode)\).
    Thus we must show that
    \[\ddiamond{\alpha}{\phi} \models \ddiamond{x:=\otimes \seq ! S(\finalNode)}{\phi}.\]
    By the \dGL axioms \irref{composed} and \irref{testb}, this is
    \[\ddiamond{\alpha}{\phi} \models \ddiamond{x:=\otimes}{(S(\finalNode) \limply \phi)}.\]
    Since \(\phi\) is a compatible Angel winning condition, \(S(\finalNode) \limply \phi\) holds.
    Thus, we need to show
    \[\ddiamond{\alpha}{\phi} \models \ddiamond{x:=\otimes}{\top}.\]
    By the semantics of nondeterministic assignments in \dGL, the succedent is \(\top\), completing the proof.
    \item If \(\namedGame{\alpha}{a}\) is \(\{x'=f(x)\ \&\ Q\}\), then the argument is similar to the previous case.
    \(\asubst{\namedGame{\alpha}{a}}{S}{}\) is \(\{x'=f(x)\ \&\ Q\}^d \seq ! S(\finalNode)\).
    Thus we must show that
    \[\ddiamond{\alpha}{\phi} \models \ddiamond{\{x'=f(x)\ \&\ Q\}^d \seq ! S(\finalNode)}{\phi}.\]
    Applying the \dGL axioms \irref{composed}, \irref{duald}, and \irref{testb}, this is
    \[\ddiamond{\alpha}{\phi} \models \ddiamond{\{x'=f(x)\ \&\ Q\}^d}{(S(\finalNode) \limply \phi)}.\]
    Since \(\phi\) is a compatible Angel winning condition, \(S(\finalNode) \limply \phi\) holds.
    Thus, we need to show
    \[\ddiamond{\alpha}{\phi} \models \ddiamond{\{x'=f(x)\ \&\ Q\}^d}{\top}.\]
    By the semantics of differential games in \dGL, the succedent is \(\top\), completing the proof.
    \item If \(\namedGame{\alpha}{a} = \namedGame{\gamma}{g} \cup \namedGame{\delta}{d}\), then \(\asubst{\namedGame{\alpha}{a}}{S}{}\) is \((!S(g) \seq \asubst{\namedGame{\gamma}{g}}{S}{}) \cap (!S(d) \seq \asubst{\namedGame{\delta}{d}}{S}{})\).
    Thus we must show that
    \[ \ddiamond{\alpha}{\phi} \models \ddiamond{(!S(g) \seq \asubst{\namedGame{\gamma}{g}}{S}{}) \cap (!S(d) \seq \asubst{\namedGame{\delta}{d}}{S}{})}{\phi}.\]
    By the \dGL axioms \irref{composed} and \irref{choiceb}, this is
    \[ \ddiamond{\alpha}{\phi} \models (S(g) \limply \ddiamond{\asubst{\namedGame{\gamma}{g}}{S}{}}{\phi}) \land (S(d) \limply \ddiamond{\asubst{\namedGame{\delta}{d}}{S}{}}{\phi}).\]
    But by the inductive hypothesis, we have \(\ddiamond{\gamma}{\phi} \models \ddiamond{\asubst{\namedGame{\gamma}{g}}{S}{}}{\phi}\) and \(\ddiamond{\delta}{\phi} \models \ddiamond{\asubst{\namedGame{\delta}{d}}{S}{}}{\phi}\).
    \(S(g) \models \ddiamond{\gamma}{\phi}\) and \(S(d) \models \ddiamond{\delta}{\phi}\) by the definition of a subvalue map.
    By the proof rules \irref{implyr} and \irref{weakenl}, the following two formulas hold:
    \(S(g) \limply \ddiamond{\asubst{\namedGame{\gamma}{g}}{S}{}}{\phi}\) and \(S(d) \limply \ddiamond{\asubst{\namedGame{\delta}{d}}{S}{}}{\phi}\).
    Thus, we have
    \(\ddiamond{\alpha}{\phi} \models (S(g) \limply \ddiamond{\asubst{\namedGame{\gamma}{g}}{S}{}}{\phi}) \land (S(d) \limply \ddiamond{\asubst{\namedGame{\delta}{d}}{S}{}}{\phi})\),
    completing the proof for this case.
    \item 
    If \(\namedGame{\alpha}{a} = \namedGame{\gamma}{g}^\ast\), then \(\ddiamond{\asubst{\namedGame{\alpha}{a}}{S}{}}{\phi} = \ddiamond{(!S(g) \seq \asubst{\namedGame{\gamma}{g}}{\modEnd{S}{S(a)}}{})^\times \seq !S(\finalNode)}{\phi}\).
    
    Applying axioms \irref{composed} and \irref{testb}, this is
    \[\ddiamond{(!S(g) \seq \asubst{\namedGame{\gamma}{g}}{\modEnd{S}{S(a)}}{})^\times}{(S(\finalNode) \limply \phi)}\].

    Since \(\phi\) is a compatible Angel winning condition, \(\models S(\finalNode) \limply \phi\).

    Thus, we need to show
    \[\ddiamond{\alpha}{\phi} \models \ddiamond{(!S(g) \seq \asubst{\namedGame{\gamma}{g}}{\modEnd{S}{S(a)}}{})^\times}{(\top)}\]

    This we do using the loop rule with invariant \(\ddiamond{\alpha}{\phi}\).
    The inductive stop is shown using the inductive hypothesis, i.e., \(\ddiamond{\gamma}{\phi} \models \ddiamond{\asubst{\namedGame{\gamma}{g}}{\modEnd{S}{S(a)}}{}}{\ddiamond{\alpha}{\phi}}\).
    Here, \(\ddiamond{\alpha}{\phi}\) is a compatible winning condition because \(S(a) \limply \ddiamond{\alpha}{\phi}\) holds by the definition of a subvalue map.
    \item If \(\namedGame{\alpha}{a} = \namedGame{\gamma}{g} \cap \namedGame{\delta}{d}\), then \(\asubst{\namedGame{\alpha}{a}}{S}{}\) is \((\asubst{\namedGame{\gamma}{g}}{S}{}) \cap (\asubst{\namedGame{\delta}{d}}{S}{})\).
    We must show that
    \[\ddiamond{\gamma\cap\delta}{\phi} \models \ddiamond{(\asubst{\namedGame{\gamma}{g}}{S}{}) \cap (\asubst{\namedGame{\delta}{d}}{S}{})}{\phi}.\]
    By the \dGL axiom \irref{choiced}, this is
    \[\ddiamond{\gamma}{\phi} \land \ddiamond{\delta}{\phi} \models \ddiamond{\asubst{\namedGame{\gamma}{g}}{S}{}}{\phi} \land \ddiamond{\asubst{\namedGame{\delta}{d}}{S}{}}{\phi}.\]
    Applying the proof rules \irref{andr} and \irref{weakenl}, the goals to show are
    \[\ddiamond{\gamma}{\phi} \models \ddiamond{\asubst{\namedGame{\gamma}{g}}{S}{}}{\phi} \textrm{ and } \ddiamond{\delta}{\phi} \models \ddiamond{\asubst{\namedGame{\delta}{d}}{S}{}}{\phi}.\]
    Both of these follow from the inductive hypotheses.
    \item If \(\namedGame{\alpha}{a} = \namedGame{\gamma}{g} \seq \namedGame{\delta}{d}\), then \(\asubst{\namedGame{\alpha}{a}}{S}{}\) is \((\asubst{\namedGame{\gamma}{g}}{\modEnd{S}{S(d)}}{}) \seq (\asubst{\namedGame{\delta}{d}}{S}{})\).
    We must show that
    \[\ddiamond{\gamma\seq\delta}{\phi} \models \ddiamond{(\asubst{\namedGame{\gamma}{g}}{\modEnd{S}{S(d)}}{}) \seq (\asubst{\namedGame{\delta}{d}}{S}{})}{\phi}.\]
    By the \dGL axiom \irref{composeb}, this is
    \[\ddiamond{\gamma}{\ddiamond{\delta}{\phi}} \models \ddiamond{\asubst{\namedGame{\gamma}{g}}{\modEnd{S}{S(d)}}{}}{\ddiamond{\asubst{\namedGame{\delta}{d}}{S}{}}{\phi}}.\]
    By the inductive hypothesis, we know that
    \(\ddiamond{\gamma}{\ddiamond{\delta}{\phi}} \limply \ddiamond{\asubst{\namedGame{\gamma}{g}}{\modEnd{S}{S(d)}}{}}{S(d)}\) holds because \(\ddiamond{\delta}{\phi}\) is compatible with \(\modEnd{S}{S(d)}\), since by the definition of a subvalue map, we have \(S(d) \limply \ddiamond{\delta}{\phi}\).
    Per inductive hypothesis, we also know that
    \( \ddiamond{\delta}{\phi} \models \ddiamond{\asubst{\namedGame{\delta}{d}}{S}{}}{\phi}\).
    Thus by the transitivity of implication along with the monotonicity rule, we complete the proof.
    \item If \(\namedGame{\alpha}{a} = \namedGame{\gamma}{g}^\times\), then \(\asubst{\namedGame{\alpha}{a}}{S}{}\) is \((\asubst{\namedGame{\gamma}{g}}{\modEnd{S}{S(a)}}{}^\times)\).
    We must show that
    \[\ddiamond{\gamma^\times}{\phi} \models \ddiamond{(\asubst{\namedGame{\gamma}{g}}{\modEnd{S}{S(a)}}{})^\times}{\phi}.\]
    We show this by using the \dGL loop rule \irref{loop} with invariant \(\ddiamond{\gamma^\times}{\phi}\).
    The inductive stop is shown using the inductive hypothesis, i.e., 
    \[\ddiamond{\gamma}{\ddiamond{\gamma^\times}{\phi}} \models \ddiamond{\asubst{\namedGame{\gamma}{g}}{\modEnd{S}{S(a)}}{}}{\ddiamond{\gamma^\times}{\phi}},\]
    where \(\ddiamond{\gamma^\times}{\phi}\) is a compatible winning condition because \(S(a) \limply \ddiamond{\gamma^\times}{\phi}\) holds by the definition of a subvalue map.
    The invariant holds initially by \irref{id} and implies the postcondition by the \dGL axiom \irref{diterated}.
  \end{itemize}
\end{proof}

\begin{lemma}[Inductive subvalue maps have some strategy]
  \label{lem:inductive-subvalue-existence}
  For any game $\namedGame{\alpha}{a}$, Angelic inductive subvalue map \(S\), and compatible Angel winning condition $\phi$, \(S(a) \models \ddiamond{\aproj{\namedGame{\alpha}{a}}{S}{}}{\phi}\).
\end{lemma}
\begin{proof}
  This follows from structural induction, unwrapping the inductive subvalue map definition, and applying the usual \dGL axioms and proof rules.
  \begin{itemize}
  \item If \(\namedGame{\alpha}{a}\) is atomic and not controlled by Angel, i.e., \(\alpha \in \{x:=e, x:=\otimes, \ptest{Q}, !Q, \{x'=f(x)\ \&\ Q\}^d\}\), then \(\aproj{\namedGame{\alpha}{a}}{S}{} = \namedGame{\alpha}{a}\).
  In these cases, by the definition of a inductive subvalue map, \(\models S(a) \limply \ddiamond{\namedGame{\alpha}{a}}{\finalNode}\).
  Further since \(\phi\) is a compatible winning condition, \(\models S(\finalNode) \limply \phi\).
  By the monotonicity rule, \(S(a) \models \ddiamond{\namedGame{\alpha}{a}}{\phi}\).
  This allows us to conclude that \(S(a) \models \ddiamond{\aproj{\namedGame{\alpha}{a}}{S}{}}{\phi}\).
  \item If \(\namedGame{\alpha}{a} = \namedGame{x := \ast}{a}\), then \(\aproj{\namedGame{\alpha}{a}}{S}{} = x:=\ast \seq \ptest{S(\finalNode)}\).
  Thus we must show that
  \[S(a) \models \ddiamond{x:=\ast \seq \ptest{S(\finalNode)}}{\phi}.\]
  By the \dGL axioms \irref{composed} and \irref{testd}, this is
  \[S(a) \models \ddiamond{x:=\ast}{S(\finalNode) \land \phi}.\]
  Since \(\phi\) is a compatible Angel winning condition, \(S(\finalNode) \limply \phi\) holds. Applying the monotonicity rule, we must show that
  \[S(a) \models \ddiamond{x:=\ast}{S(\finalNode)}.\]
  By the definition of an inductive subvalue map for Angelic free assignments along with the \irref{implyr} rule, we conclude that \(S(a) \models \ddiamond{x:=\ast}{S(\finalNode)}\).
  \item If \(\namedGame{\alpha}{a} = \namedGame{\{x'=f(x)\ \&\ Q\}}{a}\), the proof follows exactly the same steps as the previous case.
  \item If \(\namedGame{\alpha}{a} = \namedGame{\gamma}{g} \cup \namedGame{\delta}{d}\), then \(\aproj{\namedGame{\alpha}{a}}{S}{} = (\ptest{S(g)} \seq \aproj{\namedGame{\gamma}{g}}{S}{}) \cup (\ptest{S(d)} \seq \aproj{\namedGame{\delta}{d}}{S}{})\).
  Thus we must show that
  \[S(a) \models \ddiamond{(\ptest{S(g)} \seq \aproj{\namedGame{\gamma}{g}}{S}{}) \cup (\ptest{S(d)} \seq \aproj{\namedGame{\delta}{d}}{S}{})}{\phi}.\]
  By the \dGL axioms \irref{composed} and \irref{choiced}, this is
  \[S(a) \models (S(g) \limply \ddiamond{\aproj{\namedGame{\gamma}{g}}{S}{}}{\phi}) \lor (S(d) \limply \ddiamond{\aproj{\namedGame{\delta}{d}}{S}{}}{\phi}).\]
  By the definition of an inductive subvalue map, we have \(\models S(a) \limply S(g)\lor S(d)\).
  Thus by proof rules \irref{cut}, \irref{orl}, \irref{orr}, and \irref{weakenr} it suffices to show the following two goals:
  \[S(g) \models S(g) \limply \ddiamond{\aproj{\namedGame{\gamma}{g}}{S}{}}{\phi} \textrm{ and } S(d) \models S(d) \limply \ddiamond{\aproj{\namedGame{\delta}{d}}{S}{}}{\phi}.\]
  Applying the rule \irref{implyr}, the remaining goals hold per the inductive hypotheses, \(S(g) \models \ddiamond{\aproj{\namedGame{\gamma}{g}}{S}{}}{\phi}\) and \(S(d) \models \ddiamond{\aproj{\namedGame{\delta}{d}}{S}{}}{\phi}\).
  \item We show the case of Angel loop, which follows directly from the definition of an inductive subvalue map.
  If \(\namedGame{\alpha}{a} = \namedGame{\gamma}{g}^\ast\), then by the definition of an inductive subvalue map,
  \(S(a) \models \ddiamond{\aproj{\namedGame{\alpha}{a}}{S}{}}{S(\finalNode)}\).
  Since \(\phi\) a compatible winning condition, we have \(S(\finalNode) \models \phi\).
  Using \dGL rule \irref{M}, we can conclude that
  \(S(a) \models \ddiamond{\aproj{\namedGame{\alpha}{a}}{S}{}}{\phi}\).
  \item If \(\namedGame{\alpha}{a} = \namedGame{\gamma}{g} \seq \namedGame{\delta}{d}\), then \(\aproj{\namedGame{\alpha}{a}}{S}{} = \aproj{\namedGame{\gamma}{g}}{\modEnd{S}{S(d)}}{} \seq \aproj{\namedGame{\delta}{d}}{S}{}\).
  Thus we must show that
  \[S(a) \models \ddiamond{\aproj{\namedGame{\gamma}{g}}{\modEnd{S}{S(d)}}{} \seq \aproj{\namedGame{\delta}{d}}{S}{}}{\phi}.\]
  By the \dGL axioms \irref{composed}, this is
  \[S(a) \models \ddiamond{\aproj{\namedGame{\gamma}{g}}{\modEnd{S}{S(d)}}{}}{\ddiamond{\aproj{\namedGame{\delta}{d}}{S}{}}{\phi}}.\]
  By the inductive hypothesis along with the fact that \(\models S(a)\limply S(g)\) from the definition of inductive subvalue maps,
  \[S(a) \models \ddiamond{\aproj{\namedGame{\gamma}{g}}{\modEnd{S}{S(d)}}{}}{S(d)}.\]
  Additionally, by the inductive hypothesis,
  \[S(d) \models \ddiamond{\aproj{\namedGame{\delta}{d}}{S}{}}{\phi}.\]
  Applying the \irref{M} rule, we can conclude that
  \[S(a) \models \ddiamond{\aproj{\namedGame{\gamma}{g}}{\modEnd{S}{S(d)}}{}}{\ddiamond{\aproj{\namedGame{\delta}{d}}{S}{}}{\phi}},\]
  completing the proof for this case.
  \item If \(\namedGame{\alpha}{a} = \namedGame{\gamma}{g} \cap \namedGame{\delta}{d}\), then \(\aproj{\namedGame{\alpha}{a}}{S}{} = \aproj{\namedGame{\gamma}{g}}{S}{} \cap \aproj{\namedGame{\delta}{d}}{S}{}\).
  Thus we must show that
  \[S(a) \models \ddiamond{\aproj{\namedGame{\gamma}{g}}{S}{} \cap \aproj{\namedGame{\delta}{d}}{S}{}}{\phi}.\]
  By the \dGL axiom \irref{dchoiced}, this is
  \[S(a) \models \ddiamond{\aproj{\namedGame{\gamma}{g}}{S}{}}{\phi} \land \ddiamond{\aproj{\namedGame{\delta}{d}}{S}{}}{\phi}.\]
  By the definition of an inductive subvalue map, we have
  \(S(a) \models S(g) \land S(d)\).
  Applying the proof rules \irref{andr}, \irref{cut}, \irref{andl} and \irref{weakenl}, the goals to show are
  \[S(g) \models \ddiamond{\aproj{\namedGame{\gamma}{g}}{S}{}}{\phi} \textrm{ and } S(d) \models \ddiamond{\aproj{\namedGame{\delta}{d}}{S}{}}{\phi}.\]
  Both of these follow from the inductive hypotheses, completing the proof for this case.
  \item If \(\namedGame{\alpha}{a} = \namedGame{\gamma}{g}^\times\), then \(\aproj{\namedGame{\alpha}{a}}{S}{} = (\aproj{\namedGame{\gamma}{g}}{(\modEnd{S}{S(a)})}{})^\times\).
  Thus we must show that
  \[S(a) \models \ddiamond{(\aproj{\namedGame{\gamma}{g}}{(\modEnd{S}{S(a)})}{})^\times}{S(\finalNode)}.\]
  We use the \irref{loop} rule with invariant \(S(a)\).
  \(S(a)\) holds initially by assumption.
  We show that is holds inductively, i.e., that
  \[S(a) \models \ddiamond{\aproj{\namedGame{\gamma}{g}}{(\modEnd{S}{S(a)})}{}}{S(a)}.\]
  By the definition of an inductive subvalue map, \(\models S(a) \limply S(g)\). Thus, after applying \irref{cut}, we must show
  \[S(g) \models \ddiamond{\aproj{\namedGame{\gamma}{g}}{(\modEnd{S}{S(a)})}{}}{S(a)}.\]
  This holds by the inductive hypothesis.
  Finally, the invariant implies the postcondition by the definition of an inductive subvalue map, which requires that \(\models S(a) \limply S(\finalNode)\).
  \end{itemize}
\end{proof}

\begin{lemma}[Inductive subvalue maps win]
  \label{lem:inductive-subvalue-universal}
  For any game $\namedGame{\alpha}{a}$, Angelic inductive subvalue map \(S\), and compatible Angel winning condition $\phi$, \(S(a) \models \ddiamond{\asubst{\namedGame{\alpha}{a}}{S}{}}{\phi}\).
\end{lemma}
\begin{proof}
  Follows from \rref{lem:subvalue-universal} because an inductive subvalue map is a subvalue map (\rref{thm:consistancy}) and \(S(a) \limply \ddiamond{\alpha}{\phi}\) holds per the definition of a subvalue map.
\end{proof}

\printProofs

\section{Additional Definitions}
\label{app:defs}
This appendix provides additional definitions and constructions including the Demonic version for definitions whose Angelic subvalue map version is already defined in the main text, and \dGL operational semantics concepts.

\subsection{Labeled Game Trees}
\label{app:labeled-game-trees}

We construct the labeled game tree, extending \cite{DBLP:journals/tocl/Platzer15}[Appendix C] with labels.
The operator \(\caret\) denotes concatenation under prefix closure. That is:
\begin{enumerate}
  \item For two actions/labels \(\act_1\) and \(\act_2\), \(\{\act_1\caret \act_2\}\) is the set \(\{\act_1, \act_1 \then \act_2\}\).
  \item For an action \(\act\) and a sequence of actions \(b\), \(\{\act \caret b\}\) is the set \(\{ \act, \act \then b\}\).
  \item For two actions/labels \(\act_1\) and \(\act_2\), and sequence \(b\), \(\{\act_1 \caret \act_2 \caret b\}\) is the set \(\{\act_1, \act_1 \then \act_2, \act_1 \then \act_2 \then b\}\).
  \item For an action \(\act\) and \emph{a set of sequences} \(\{a, b\}\), \(\{\act \caret \{a, b\}\}\) is the set \(\{\act, \act \then a, \act \then b\}\).
\end{enumerate}
Thus, \(\caret\) is overloaded to denote the concatentation of actions, sequences, and sets of sequences under prefix closure.

\begin{definition}[Labeled Semantics]
  \label{def:labeled-semantics}
  For any game \(\namedGame{\alpha}{a}\) and state \(\sigma\), the labeled semantics \(g(\namedGame{\alpha}{a})(\sigma)\) is defined as follows, where \(\caret\) denotes concatenation under prefix closure (\rref{app:labeled-game-trees}).
  \begin{align*}
    \opSemantics{x:=\theta}{\sigma} &\define \{\nodeLabel{a} \caret x:=\theta\} \\
    \opSemantics{x:=\otimes}{\sigma} &\define \{\nodeLabel{a} \caret \pdual{x:=\theta} \with \theta \in \mathbb{R}\} \\
    \opSemantics{x:=\ast}{\sigma} &\define \{\nodeLabel{a} \caret x:=\theta \with \theta \in \mathbb{R}\} \\
    \opSemantics{\{x'=f(x)\ \&\ Q\}}{\sigma} &\define \{\nodeLabel{a} \caret (x'=f(x)\&Q @ r) \with r \in \mathbb{R}, r \geq 0, \varphi(0) = \sigma \text{ for some } \\
    & \textrm{(differentiable) } \varphi : [0, r] \to \mathcal{S} \text{ such that } \frac{d\varphi(t)(x)}{dt}(\zeta) = \llbracket \theta \rrbracket_{\varphi(\zeta)}\\ 
    & \text{and } 
    \varphi(\zeta) \in \llbracket Q \rrbracket^I \text{ for all } \zeta \leq r \} \\
    \opSemantics{\ptest{Q}}{\sigma} &\define \{ \nodeLabel{a} \caret \ptest{Q} \} \\
    \opSemantics{!Q}{\sigma} &\define \{ \nodeLabel{a} \caret !Q \} \\
    \opSemantics{\namedGame{\gamma}{g} \cup \namedGame{\delta}{d}}{s} &\define \{ \nodeLabel{a} \caret \actleft \caret b \with b \in g(\namedGame{\gamma}{g})(\sigma) \} \cup \\
    & \qquad \{ \nodeLabel{a} \caret \actright \caret b \with b \in g(\namedGame{\delta}{d})(\sigma) \} \\
    \opSemantics{\namedGame{\gamma}{g} \cap \namedGame{\delta}{d}}{\sigma} &\define \{ \nodeLabel{a} \caret \pdual{\actleft} \caret b \with b \in g(\namedGame{\gamma}{g})(\sigma) \} \cup \\
    & \qquad \{ \nodeLabel{a} \caret \pdual{\actright} \caret b \with b \in g(\namedGame{\delta}{d})(\sigma) \} \\
    \opSemantics{\namedGame{\gamma}{g} \seq \namedGame{\delta}{d}}{\sigma} &\define \nodeLabel{a} \caret \actseq \caret \opSemantics{\gamma}{g} \cup \bigcup_{t \in \leaf{(\nodeLabel{a} \caret \opSemantics{\gamma}{g})}}\opSemantics{\namedGame{\delta}{d}}{\runstrategy{t}{\sigma}} \\
    \opSemantics{(\namedGame{\gamma}{g})^\ast}{\sigma} &\define \bigcup_{n < \omega} f^n(\{(\nodeLabel{a} \caret \actstop), (\nodeLabel{a} \caret \actgo)\}) \\
    & \text{where } f^n \text{ is the $n$-fold composition of the function} \\
    & f(Z) \overset{\text{def}}{=} Z \cup \bigcup_{t \caret \actgo \in \leaf(Z)} t \caret \actgo \caret \opSemantics{\alpha}{\runstrategy{t \caret \actgo}{\sigma}} \caret \{(\nodeLabel{a} \caret \actstop), (\nodeLabel{a} \caret \actgo)\} \\
    \opSemantics{(\namedGame{\gamma}{g})^\times}{\sigma} &\define \bigcup_{n < \omega} f^n(\{(\nodeLabel{a} \caret \pdual{\actstop}), (\nodeLabel{a} \caret \pdual{\actgo})\}) \\
    & \text{where } f^n \text{ is the $n$-fold composition of the function} \\
    & f(Z) \overset{\text{def}}{=} Z \cup \bigcup_{t \caret \pdual{\actgo} \in \leaf(Z)} t \caret \pdual{\actgo} \caret \opSemantics{\alpha}{\runstrategy{t \caret \pdual{\actgo}}{\sigma}}\\
  \end{align*}
\end{definition}

Leaves are the nodes of a tree with no children, i.e., any sequence \(l\) such that there is no sequence \(l \then s\) where \(s\) is a single action/label in the tree.

\subsection{Label-free Game Trees}
\label{app:label-free-game-trees}

To get back unlabeled \dGL games operational semantics, we erase labels from the labeled semantics.

\begin{definition}[Erasing Labels]
  \label{def:erase-label}
  The function \(\eraseLabel\) takes in a node or game tree and returns it with all subgame labels removed.
  First we define erasing labels for a single node. 
  Here, \(\eraseLabel\) accepts a node \(s\) (which is a sequence) and returns a label-free node (the sequence with labels dropped).
  \[ \eraseLabel(s) \define \left( a_i \with a_i \in s,\, a_i \notin \text{subgame labels} \right) \]
  Erasing labels for a game tree \(t\) proceeds by erasing labels for every node, producing a label-free tree.
  \[ \eraseLabel(t) \define \{ \eraseLabel(s) \with s \in t \textrm{ and } \eraseLabel(s)\neq () \} \]
  If erasing labels results in an empty sequence (), it is discarded.
\end{definition}

\begin{definition}[Label-free Semantics]
  \label{def:label-free-game-tree}
  The label-free game tree of a \dGL game \(\namedGame{\alpha}{a}\) results from erasing labels in the labeled sematics, i.e., \(\eraseLabel(\opSemantics{\namedGame{a}{a}}{\sigma})\).
\end{definition}

\subsection{Demonic Subvalue Projection}
\label{app:subvalue-projection}

\begin{definition}[Demonic subvalue projection]
  \label{def:subvalue-projection-demon}
  The \emph{projection} of Demonic subvalue map $S$ onto \dGL game $\namedGame{\alpha}{a}$ with Demon winning condition $\phi$, written $\dproj{\namedGameSpace{\alpha}{a}}{S}{\phi}$, is generated recursively per the structure of $\namedGame{\alpha}{a}$ as follows. If $\namedGame{\alpha}{a}$ has structure:
  \begin{align*}
    \namedGame{x:=\times}{a}
    &\textrm{ then }
    x:= \times \seq \AdvExAssign{!\phi}. \\
    \namedGame{\{x'=f(x)\ \&\ Q\}^d}{a}
    &\textrm{ then }
    \{x'=f(x)\ \&\ Q\}^d \seq \AdvExAssign{!\phi} \\
    \namedGame{(\namedGame{\gamma}{g} \cap \namedGame{\delta}{d})}{a}
    &\textrm{ then } ( \AdvExChoice{!S(g)} \seq \aproj{\namedGame{\gamma}{g}}{S}{\phi}) {\cap} (\AdvExChoice{!S(d)} \seq \aproj{\namedGame{\delta}{d}}{S}{\phi}) \\
    \namedGame{(\namedGame{\gamma}{g})^\times}{a}
    &\textrm{ then } (\AdvExLoop{!S(g)} \seq \dproj{\namedGameSpace{\gamma}{g}}{S}{S(g) \lor \phi})^{\times}; \AdvExLoop{!\phi} \\
    \namedGame{(\namedGame{\gamma}{g};\namedGame{\delta}{d})}{a}
    &\textrm{ then } \dproj{\namedGameSpace{\gamma}{g}}{S}{S(d)} \seq \dproj{\namedGameSpace{\delta}{d}}{S}{\phi} \\
    \namedGame{(\namedGame{\gamma}{g} \cup \namedGame{\delta}{d})}{a}
    &\textrm{ then } \dproj{\namedGameSpace{\gamma}{g}}{S}{\phi} \cup \dproj{\namedGameSpace{\delta}{d}}{S}{\phi}
    \quad
    \namedGame{(\namedGame{\gamma}{g})^\ast}{a}
    \textrm{ then } \dproj{\namedGameSpace{\gamma}{g}}{S}{S(a)}^\ast\\
    \textrm{atomic and}&\textrm{ not controlled by Demon, i.e., }\\
    \alpha \in &\{x:=e, x:=\otimes, \ptest{Q}, !Q, \{x'=f(x)\ \&\ Q\}^d\}, \textrm{ then } \alpha
  \end{align*}
\end{definition}

\subsection{Universal Projection of Demonic Subvalue Maps}
\label{app:subvalue-externalization-demon}

\begin{definition}[Universal Projection]
  \label{def:subvalue-ext-demon}
  The \emph{universal projection} of Demonic subvalue map $S$ onto \dGL game $\namedGame{\alpha}{a}$ with Demon winning condition $\phi$, written $\dsubst{\namedGame{\alpha}{a}}{S}{\phi}$, is generated recursively per the structure of $\namedGame{\alpha}{a}$ as follows. If $\namedGame{\alpha}{a}$ has structure:
  \begin{align*}
    & \namedGame{x:=\otimes}{a} \textrm{ then }
    \AdvExAssign{!(\exists x \, \phi)}
    \seq x:= \AdvExAssign{\ast} \seq \AdvExAssign{?\phi}.\\
    & \namedGame{\{x'=f(x)\ \&\ Q\}^d}{a},\ \textrm{ then }
    \AdvExAssign{!\langle\alpha\rangle\phi} \seq \{x'=f(x)\ \&\ Q\} \seq \AdvExAssign{?\phi}.\\
    &\namedGame{(\namedGame{\gamma}{g} \cap \namedGame{\delta}{d})}{a}, \textrm{ then } \AdvExChoice{!(S(g)\lor S(d))} \seq \left(( \AdvExChoice{?S(g)} \seq \dsubst{\namedGame{\gamma}{g}}{S}{\phi}) \AdvExChoice{\cup} (\AdvExChoice{?S(d)} \seq \dsubst{\namedGame{\delta}{d}}{S}{\phi})\right).\\
    &\namedGame{(\namedGame{\gamma}{g})^\times}{a}, \textrm{ then } \AdvExLoop{!(S(g)\lor \phi)} \seq \left(\AdvExLoop{?S(g)} \seq \dsubst{\namedGame{\gamma}}{S}{S(a) \lor \phi} \seq ?S(a) \lor \phi \right)^{\AdvExLoop{\ast}} \seq \AdvExLoop{?\phi}.\\
    &\namedGame{(\namedGame{\gamma}{g};\namedGame{\delta}{d})}{a}, \textrm{ then } \dsubst{\namedGame{\gamma}{g}}{S}{S(d)} \seq \dsubst{\namedGame{\delta}{d}}{S}{\phi}.\\
    &\namedGame{(\namedGame{\gamma}{g} \cap \namedGame{\delta}{d})}{a}, \textrm{ then } \dsubst{\namedGame{\gamma}{g}}{S}{\phi} \cap \dsubst{\namedGame{\delta}{d}}{S}{\phi}.\\
    &\namedGame{(\namedGame{\gamma}{g})^\times}{a}, \textrm{ then } \dsubst{\namedGame{\gamma}{g}}{S}{S(a)}^\times.\\
    &\textrm{atomic and not controlled by Demon, i.e., }\\
    & \qquad \alpha \in \{x:=e, x:=\ast, \ptest{Q}, !Q, \{x'=f(x)\ \&\ Q\}\}, \textrm{ then } \alpha.
  \end{align*}
\end{definition}

\subsection{Strategy Set Generation}
\label{app:strategy-set-generation}
\rref{def:policy} shows how an agent's Subvalue map lets it decide what actions to take at a given state. This section shows how to generate a set of strategies (as defined in \dGL operational semantics) given a subvalue map by tracing through the game while keeping track of which strategy leads to what state, and using the policy to predict what actions to take at decision points.

\begin{definition}[Transformation to Strategy Set]
\label{def:strategy-set-generation}
The function \(\astrategize{\namedGameSpace{\alpha}{a}}{S}(b, \sigma)\) produces the set of Angel strategies to play game \(\namedGame{\alpha}{a}\) following subvalue map \(S\) starting at subgame \(b\), where \(\sigma\) represents the current state, where \(\caret\) is prefix-closed concatenation as shown in \rref{app:labeled-game-trees}.
If the structure of \(b\) is:
\begin{align*}
  \namedGame{x:=e}{b} &\textrm{ then }
    \{\{(x:=e)\}\} \\
  \namedGame{x:=\ast}{b} &\textrm{ then } \{{p} \with p \in \apolicy{\alpha}{a}(S)(b, \sigma)\} \\
  \namedGame{x:=\otimes}{b} &\textrm{ then } \{ \{(\pdual{x:=e})  \with e \in \mathbb{R} \} \} \\
  & \mspace{-65mu} \namedGame{\{x'=f(x)\ \&\ Q\}}{b} \textrm{ then } \{\{p\} \with \apolicy{\alpha}{a}(S)(b, \sigma)\}\\
  & \mspace{-65mu} \namedGame{\{x'=f(x)\ \&\ Q\}^d}{b} \textrm{ then } \\
  & \qquad \{ \{ (\pdual{x'=f(x) \& Q} \,@\, t) \with t \in \mathbb{R} \textrm{ where} \\
  & \qquad \textrm{\(\varphi(0) = \sigma(x)\) and \(\forall s \in [0,t] \, \varphi'(s)=f(\varphi(s))\)} \} \} \\
  \namedGame{\ptest{Q}}{b} &\textrm{ then } \{\{ (\ptest{Q}) \}\} \\
  \namedGame{!Q}{b} &\textrm{ then } \{\{ (!Q) \}\} \\
  \namedGame{(\namedGame{\gamma}{g} \cup \namedGame{\delta}{d})}{b} &\textrm{ then }\\
  &\begin{cases}
    \{\} & \textrm{ if }\apolicy{\alpha}{a}(S)(b) = \{\} \\
    \{\{\actleft \caret t\} \with t\in\astrategize{\namedGameSpace{\gamma}{g}}{S}(g, \sigma)\} &\textrm{ if } \apolicy{\alpha}{a}(S)(b) = \{\actleft\} \\
    \{\{\actright \caret t\} \with t\in\astrategize{\namedGameSpace{\delta}{d}}{S}(d, \sigma)\} &\textrm{ if } \apolicy{\alpha}{a}(S)(b) = \{\actright\} \\
    \{\{\actleft \caret t\} \with t\in\astrategize{\namedGameSpace{\gamma}{g}}{S}(g, \sigma)\}\ \cup \qquad & \textrm{ otherwise if } \apolicy{\alpha}{a}(S)(b) = \{\actleft, \actright\}\\
    \phantom{\{} \{\{\actright \caret t\} \with t\in\astrategize{\namedGameSpace{\delta}{d, \sigma}}{S}(d)\} &
  \end{cases} \\
  \namedGame{(\namedGame{\gamma}{g} \cap \namedGame{\delta}{d})}{b} &\textrm{ then } 
    \{\pdual{\actleft} \caret t_l \cup \pdual{\actright} \caret t_r \with t_l\in\astrategize{\namedGameSpace{\gamma}{g}}{S}(g, \sigma) \textrm{ and } t_r\in\astrategize{\namedGameSpace{\delta}{d}}{S}(d, \sigma)\}
  \\
  \namedGame{(\namedGame{\gamma}{g} \seq \namedGame{\delta}{d})}{b} &\textrm{ then } \{ t \cup \bigcup_{v \in \leaf(t)} (v\caret u) \with t \in \astrategize{\namedGameSpace{\gamma}{g}}{\modEnd{S}{S(d)}}(g, \sigma) \textrm{ and }\\
  & \qquad u \in \astrategize{\namedGameSpace{\delta}{d}}{S}(d)(\runstrategy{v}{\sigma}) \} \\
  \namedGame{(\namedGame{\gamma}{g})^\ast}{b} &\textrm{ then } \bigcup_{n<\omega} \{t \with t\in f^n (\apolicy{\alpha}{a}(S)(b, \sigma)) \textrm{ and } \not\exists u (u \caret \actgo \in \leaf(t))\} \\
  & \qquad \textrm{ where } f^n \textrm{ is \(n\)-fold composition of the function}: \\
  & \qquad f(Z) = \bigcup_{t\in Z} \{t \cup \bigcup_{u \caret \actgo \in \leaf(t)} 
  \left\{\, 
    \begin{array}{ll}
    \{u \caret \actgo \caret v_u\} & \text{if }\runstrategy{u \caret \actgo \caret v_u}{\sigma}=\kwd{undefined}  \\[4pt]
    \{u \caret \actgo \caret v_u \caret o_{uv} \} & \text{otherwise}
    \end{array}
  \right\}
   \with \\
  & \qquad v_u \in \astrategize{\namedGameSpace{\gamma}{g}}{\modEnd{S}{S(a)}}(g)(\runstrategy{u}{\sigma}) \textrm{ and }\\
  & \qquad o_{uv} \in \apolicy{\alpha}{a}(S)(b, \runstrategy{u \caret \actgo \caret v_u}{\sigma})\} \\
  \namedGame{(\namedGame{\gamma}{g})^\times}{b} &\textrm{ then } \bigcup_{n<\omega} f^n (\{\{\pdual{\actstop}, \pdual{\actgo}\}\}) \\
  & \qquad \textrm{ where } f^n \textrm{ is \(n\)-fold composition of the function}: \\
  & \qquad f(Z) =  \bigcup_{t\in Z} \{t \cup \bigcup_{u \caret \pdual{\actgo} \in \leaf(t)} \{u \caret \pdual{\actgo} \caret v \caret \{\pdual{\actstop}, \pdual{\actgo} \} \with v \in \astrategize{\namedGameSpace{\alpha}{a}}{S}(g)(\runstrategy{u}{\sigma})\} \}
\end{align*}
with implicit closure under prefixes.
\end{definition}

To reason only about the states reachable by playing the game trees, we should introduce a version of the game trees free of actions that do nothing.
\begin{definition}[Erasing skip Actions]
  \label{def:erase-noop}
  Let the set of skip actions \skipActions be
  \[\skipActions = \{\actgo, \actstop, \actleft, \actright, \actseq, \pdual{\actgo}, \pdual{\actstop}, \pdual{\actleft}, \pdual{\actright}\}\].
  The function \(\eraseSkip\) takes in a node or game tree and returns it with all labels and skip actions removed.
  First we define \(\eraseSkip\) for a single node. 
  Here, \(\eraseSkip\) accepts a node \(s\) (which is a sequence) and returns a node (the sequence with labels dropped).
  \[ \eraseSkip(s) \define \left( a_i \with a_i \in s,\, a_i \notin \text{subgame labels} \textrm{ and } a_i \notin \skipActions \right) \]
  Erasing labels and skip actions for a game tree \(t\) proceeds by erasing these for every node, producing a new tree.
  \[ \eraseSkip(t) \define \{ \eraseSkip(s) \with s \in t \textrm{ and } \eraseSkip(s)\neq () \}. \]
\end{definition}

\subsection{Inductive Demonic Subvalue Maps}
\label{app:inductive-subvalue-map}

\begin{definition}[Inductive Demonic subvalue maps]
  \label{def:inductive-subvalue-map-demon}
  Let $S$ be a map from the subgames of $\namedGame{\alpha}{a}$ to winning subregions.
  $S$ is an \emph{inductive Demonic subvalue map} for game $\namedGame{\alpha}{a}$, written $\dvalid{\alpha}{a}{S}{}$, when the following holds.
  If $\namedGame{\alpha}{a}$ has structure:
  \begin{align*}
    \textrm{atomic, i.e., }\alpha \in &\{x:=e, x:=*, x:=\otimes, \ptest{Q}, !Q, \{x'=f(x)\ \&\ Q\}, \\
    & \phantom{\{} \{x'=f(x)\ \& \ Q\}^d\} \textrm{ then } \models S(a) \limply [ {\alpha} ] S(\finalNode) . \\
    \namedGame{(\namedGame{\gamma}{g} \cup \namedGame{\delta}{d})}{a} \textrm{ then }& \models S(a) \limply S(g) \land S(d) \textrm{ and }
    \dvalid{\gamma}{g}{S}{} \textrm{ and } \dvalid{\delta}{d}{S}{} . \\
    \namedGame{(\namedGame{\gamma}{g} \cap \namedGame{\delta}{d})}{a} \textrm{ then }& \models S(a) \limply S(g) \lor S(d) \textrm{ and }
    \dvalid{\gamma}{g}{S}{} \textrm{ and } \dvalid{\delta}{d}{S}{} . \\
    \namedGame{(\namedGame{\gamma}{g}\seq \namedGame{\delta}{d})}{a} \textrm{ then }& \models S(a) \limply S(g) \textrm{ and }
    \dvalid{\gamma}{g}{\modEnd{S}{S(d)}}{} \textrm{ and } \dvalid{\delta}{d}{S}{} .\\
    \namedGame{(\namedGame{\gamma}{g})^*}{a} \textrm{ then }&
    \models S(a) \limply S(\finalNode) \land S(g)
    \textrm{ and } \dvalid{\gamma}{g}{\modEnd{S}{S(a)}}{} . \\
    \namedGame{(\namedGame{\gamma}{g})^\times}{a} \textrm{ then }&
    \models S(a) \limply [\dproj{\namedGameSpace{\alpha}{a}}{S}{}]S(\finalNode)
    \textrm{ and } \dvalid{\gamma}{g}{\modEnd{S}{S(a)}}{} .
  \end{align*}
\end{definition}

\subsection{Game Prefix}
\label{app:prefix}

This definition uses \emph{the empty game} \(\skp\) which can be treated like the game \(\ptest{true}\) in terms of effect on plays, except with no corresponding action in the operational semantics.

\begin{definition}[Game Prefix]
  \label{def:execution-prefix}
  The game suffix $\prefix{b}{\namedGame{\alpha}{a}}$ of subgame $\namedGame{\beta}{b}$ in $\nodes{\namedGame{\alpha}{a}}$ is constructed as below.
  If $\namedGame{\alpha}{a}$ is $\namedGame{\beta}{b}$ and \(\alpha\) is not a loop, then $\prefix{b}{\namedGame{\alpha}{a}}$ is the empty game, $\skp$.
  If \(b=a\) and $\namedGame{\alpha}{a}$ is a loop, then $\prefix{b}{\namedGame{\alpha}{a}}$ is \(\namedGame{\alpha}{a}\).
  If \(b\) is \(\finalNode\) then $\prefix{b}{\namedGame{\alpha}{a}}$ is the entire game $\namedGame{\alpha}{a}$.
  Otherwise, if $\namedGame{\alpha}{a}$ has structure:
  \begin{equation*}
    \begin{aligned}
      \namedGame{(\namedGame{\gamma}{g})^\ast}{a} \textrm{ or } \namedGame{(\namedGame{\gamma}{g})^\times}{a} &\textrm{ then } \prefix{b}{\namedGame{\alpha}{a}} = \namedGame{\alpha}{a} \seq (\prefix{b}{\namedGame{\gamma}{g}}) \\
      \namedGame{(\namedGame{\gamma}{g} \cup \namedGame{\delta}{d})}{a} \textrm{ or } \namedGame{(\namedGame{\gamma}{g} \cap \namedGame{\delta}{d})}{a} &\textrm{ then } \begin{cases}
        \prefix{b}{\namedGame{\alpha}{a}} = \prefix{b}{\namedGame{\gamma}{g}} &b \in \nodes{\namedGame{\gamma}{g}} \\
        \prefix{b}{\namedGame{\alpha}{a}} = \prefix{b}{\namedGame{\delta}{d}} &\textrm{otherwise}
      \end{cases} \\
      \namedGame{(\namedGame{\gamma}{g} \seq \namedGame{\delta}{d})}{a} &\textrm{ then } \begin{cases}
        \prefix{b}{\namedGame{\alpha}{a}} = \prefix{b}{\namedGame{\gamma}{g}} &b \in \nodes{\namedGame{\gamma}{g}} \\
        \prefix{b}{\namedGame{\alpha}{a}} = \namedGame{\gamma}{g} \seq \prefix{b}{\namedGame{\delta}{d}} &\textrm{otherwise}
      \end{cases}
    \end{aligned}
  \end{equation*}
\end{definition}

\subsection{Solving Function for Inductive Demonic Subvalue Maps}
\label{app:solving-demon}

$\dsolve(\namedGame{\alpha}{a}, \phi)$, defined below, computes an inductive Demonic subvalue map for game $\namedGame{\alpha}{a}$ and Demon winning condition $\phi$.
Operator $\uplus$ denotes the disjoint union of two subvalue maps.
\begin{equation*}
      \begin{aligned}
    &\dsolve(\namedGame{\alpha}{a}, \phi) := \{ a \mapsto \exec([ \alpha ] \phi), \finalNode \mapsto \phi \} \textrm{ where }\alpha\in\{x:=e, x:=\ast, x:=\otimes,\\
      &\quad  \ptest{Q}, !Q,\{x'=f(x) \& Q\}, \{x'=f(x) \& Q\}^d \} \\
    &\dsolve(\namedGame{(\namedGame{\gamma}{g} \cap \namedGame{\delta}{d})}{a}, \phi) :=\ 
      S_1 \uplus S_2 \uplus \{ a \mapsto S_1(g)\lor S_2(d)\} \\
      &\quad \textrm{ where }S_1:=\dsolve(\namedGame{\gamma}{g}, \phi),\ S_2:=\dsolve(\namedGame{\delta}{d}, \phi) \setminus \finalNode \\
    &\dsolve(\namedGame{(\namedGame{\gamma}{g} \cup \namedGame{\delta}{d})}{a}, \phi) :=\ 
      S_1 \uplus S_2 \uplus \{ a \mapsto S_1(g)\land S_2(d)\} \\
      &\quad \textrm{ where }S_1:=\dsolve(\namedGame{\gamma}{g}, \phi),\ S_2:=\dsolve(\namedGame{\delta}{d}, \phi) \setminus \finalNode \\
    &\dsolve(\namedGame{(\namedGame{\gamma}{g};\namedGame{\delta}{d})}{a}, \phi) :=\  S_1 \uplus S_2 \uplus \{a \mapsto S_2(g)\} \\
      &\quad \textrm{ where }S_1:=\dsolve(\namedGame{\delta}{d}, \phi),\ S_2:=\dsolve(\namedGame{\gamma}{g}, S_1(d)) \setminus \finalNode \\
    &\dsolve(\namedGame{(\namedGame{\gamma}{g})^\times}{a}, \phi) :=\  S \uplus \{a \mapsto \invexpr, \finalNode \mapsto \phi \} \textrm{ if }
        \langle \aproj{\namedGame{\alpha}{a}}{S}{\phi} \rangle \phi \\
        & \quad \textrm{ where }
        S:=\dsolve(\namedGame{\gamma}{g}, \invexpr \lor \phi) \setminus \finalNode \\
    &\dsolve(\namedGame{(\namedGame{\gamma}{g})^\ast}{a}, \phi) :=\  S \uplus \{a \mapsto \invexpr, \finalNode \mapsto \phi \}
        \textrm{ if }\models \invexpr \limply S(g) \land \phi \\
        & \quad \textrm{ where } S := \dsolve(\namedGame{\gamma}{g}, \invexpr) \setminus \finalNode
    \end{aligned}
  \end{equation*}

\end{document}